\documentclass{tlp}
\usepackage{latexsym}
\usepackage{amssymb}
\usepackage{graphicx}
\usepackage{amsmath}
\usepackage{verbatim}
\usepackage{varioref}

\usepackage{slashbox}
\usepackage{multirow}
\usepackage{rotating}

\newtheorem{lemma}{Lemma} 
\newtheorem{definition}{Definition} 
\newtheorem{example}{Example} 
\newtheorem{theorem}{Theorem} 
\newtheorem{postulate}{Postulate} 
\newtheorem{proposition}{Proposition} 

\newcommand{\red}{\vdash\!\!\vdash} 

\newcommand{\nc}{\newcommand}

\nc{\tp}[1]{\overline{#1}}   

\nc{\tpp}[2]{\tp{#1}_{#2}}  

\nc{\sig}{\Sigma}

\nc{\sigbot}{\Sigma_{\bot}}

\nc{\fs}[1]{FS^{#1}}   

\nc{\dc}[1]{DC^{#1}}  

\nc{\ure}{\mathcal{U}}

\nc{\var}{\mathcal V\!\!ar}     
\nc{\varx}{\mathcal{X}}    %
\nc{\vary}{\mathcal{Y}}    
\nc{\varz}{\mathcal{Z}}    %

\nc{\tvar}{\mathcal {T\!V}\!\!ar}   

\nc{\expr}[1]{Exp_{\bot}(#1)}    
\nc{\texpr}[1]{Exp(#1)}          
\nc{\gexpr}[1]{GExp_{\bot}(#1)}  
\nc{\gtexpr}[1]{GExp(#1)}        

\nc{\pat}[1]{Pat_{\bot}(#1)} 

\nc{\tpat}[1]{Pat(#1)}       

\nc{\gpat}[1]{GPat_{\bot}(#1)} 

\nc{\gtpat}[1]{GPat(#1)}       

\nc{\sub}[1]{Sub_{\bot}(#1)}    
\nc{\tsub}[1]{Sub(#1)}          
\nc{\gsub}[1]{GSub_{\bot}(#1)}  
\nc{\gtsub}[1]{GSub(#1)}        

\nc{\restrict}{\upharpoonright}

\nc{\leqinfo}{\sqsubseteq} \nc{\lsinfo}{\sqsubset}
\nc{\geqinfo}{\sqsupseteq} \nc{\gtinfo}{\sqsupset}

\nc{\cdom}{\mathcal{D}}

\nc{\ccdom}{\mathcal{C}}

\nc{\mdom}{\mathcal{M}}

\nc{\sdom}{\mathcal{S}}

\nc{\trunc}[2]{\, \mid #1 \mid_{#2}} 

\nc{\csig}{\Gamma}                            
\nc{\pfun}[1]{PF^{#1}}                         
\nc{\df}[1]{DF^{#1}}                             
\nc{\uni}[1]{\mathcal{U}_{#1}}            
\nc{\dom}[1]{\mathcal{D}_{#1}}          

\nc{\hoare}[1]{\emph{HP}(#1)} 

\nc{\nat}{\mathbb{N}} \nc{\ent}{\mathbb{Z}} \nc{\real}{\mathbb{R}}

\nc{\herbrand}{\mathcal{H}}            
\nc{\herbrandseq}{\mathcal{H}_{seq}}   
\nc{\rdom}{\mathcal{R}}                
\nc{\rt}{\mathcal{RT}}                 
\nc{\ft}{\mathcal{FT}}                 
\nc{\fd}{\mathcal{FD}}                 

\nc{\yields}{\to!\,}

\nc{\anyarrow}{\to^{?}\,}

\nc{\patom}{p\, \overline{t}_{n} \to!\, t\,}

\nc{\true}{\lozenge} \nc{\false}{\blacklozenge}

\nc{\pcon}[1]{PCon_{\bot}(#1)}    
\nc{\ptcon}[1]{PCon(#1)}          
\nc{\pgcon}[1]{PGCon_{\bot}(#1)}  
\nc{\pgtcon}[1]{PGCon(#1)}        

\nc{\val}[1]{Val_{\bot}(#1)}  
\nc{\tval}[1]{Val(#1)}        

\nc{\sol}[2]{Sol_{#1}(#2)}      

\nc{\wtsol}[2]{WTSol_{#1}(#2)}      

\nc{\myneq}{/\hspace*{-2mm}=}   

\nc{\sat}[2]{Sat_{#1}(#2)}      
\nc{\unsat}[2]{Unsat_{#1}(#2)}  

\nc{\clp}[1]{CLP(#1)}    
\nc{\fp}[1]{FP(#1)}      
\nc{\cflp}[1]{CFLP(#1)}  

\nc{\datom}{p\, \overline{e}_n \to!\, t\,}

\nc{\dcon}[1]{DCon_{\bot}(#1)}    
\nc{\dtcon}[1]{DCon(#1)}          
\nc{\dgcon}[1]{DGCon_{\bot}(#1)}  
\nc{\dgtcon}[1]{DGCon(#1)}        

\nc{\prog}{\mathcal{P}}  

\nc{\prule}{f\, \tpp{t}{n}\, \to\, r\, \Leftarrow\, P\, \Box\, \Delta}

\nc{\sfrule}{f\, \tpp{t}{n}\, \to\, r\, \Leftarrow\, \Delta}

\nc{\sprule}{f\, \tpp{t}{n}\, \to\, true\, \Leftarrow\, \Delta}

\nc{\prules}[1]{[#1]_{\bot}} 

\nc{\calg}{\mathcal{A}}

\nc{\I}{\emph{I}}                   
\nc{\J}{\emph{J}}                   
\nc{\Int}[1]{\mathbb{I}_{#1}}          
\nc{\Bot}{\bot\hspace*{-2.35mm}\bot}   
\nc{\Top}{\top\hspace*{-2.35mm}\top}   
\nc{\FI}{\mathfrak{I}}                 
\nc{\FJ}{\mathfrak{J}}                 

\nc{\grounding}[2]{[#1]_{#2}} 

\nc{\lub}[1]{\LARGE{\sqcup\,}{#1}}      
\nc{\glb}[1]{\LARGE{\sqcap\,}{#1}}      

\nc{\ufact}{f\, \tpp{t}{n} \to\, t}                    
\nc{\tfact}{f\, \tpp{t}{n} \to\, \bot}                 
\nc{\cfact}{f\, \tpp{t}{n} \to\, t\ \Leftarrow\, \Pi}  

\nc{\entails}[1]{\succcurlyeq_{#1}}       
\nc{\entailedby}[1]{\preccurlyeq_{#1}}    

\nc{\forces}[1]{\vdash\hspace*{-2.0mm}\vdash_{#1}} 

\nc{\ld}{[\![} \nc{\rd}{]\!]}
\nc{\den}[3]{\ld{#1}\rd^{#2}_{#3}} 
\nc{\gden}[2]{\ld{#1}\rd^{#2}}     

\nc{\smodels}[1]{\models_{#1}^s}   
\nc{\wmodels}[1]{\models_{#1}^w}   

\nc{\spt}[1]{preST_{#1}}             
\nc{\st}[1]{ST_{#1}}                 
\nc{\wpt}[1]{preWT_{#1}}             
\nc{\wt}[1]{WT_{#1}}                 
\nc{\ist}[2]{ST_{#1}\uparrow^{#2}}   
\nc{\iwt}[2]{WT_{#1}\uparrow^{#2}}   

\nc{\slfp}[1]{\bigcup_{k \in \nat} \ist{#1}{k}(\Bot)}  
\nc{\wlfp}[1]{\bigcup_{k \in \nat} \iwt{#1}{k}(\Bot)}  

\nc{\smod}[1]{\emph{S}_{#1}}  
\nc{\wmod}[1]{\emph{W}_{#1}}  

\nc{\crwl}[1]{CRWL(#1)} 

\nc{\clnc}[1]{CLNC(#1)} 

\nc{\cclnc}[1]{CCLNC(#1)} 

\nc{\cproves}[1]{\vdash_{#1}} 

\nc{\size}[1]{\|#1\|}       
\nc{\rsize}[1]{\mid#1\mid}  

\nc{\sts}[2]{\vdash\!\!\vdash_{#1,\, #2}} 

\nc{\ests}[1]{\vdash\!\!\vdash_{#1}} 

\nc{\toy}{\mathcal{TOY}}              
\nc{\toye}{$\mathcal{TOY \;}$}    

\submitted {26 September 2008}
\revised {6 February 2009}
\accepted{6 April 2009}

\begin{document}

\title[On the Cooperation of the Constraint Domains $\herbrand$, $\rdom$ and $\fd$ in $CFLP$]
      {On the Cooperation of the Constraint Domains $\herbrand$, $\rdom$ and $\fd$ in $CFLP$}

\author[Est\'evez, Fern\'{a}ndez, Hortal\'{a}, Rodr\'iguez,
S\'{a}enz and del Vado] {S. EST\'EVEZ-MART\'IN, T.
HORTAL\'{A}-GONZ\'{A}LEZ, \and M.
RODR\'IGUEZ-ARTALEJO and R. DEL VADO-V\'IRSEDA\footnotemark\\
      Dpto. de Sistemas Inform\'aticos y Computaci\'on \\
      Universidad Complutense de Madrid, Spain \\
      E-mails: {\{s.estevez,teresa,mario,rdelvado\}@sip.ucm.es}
\and F. S\'AENZ-P\'EREZ\addtocounter{footnote}{-1}
\thanks{The work of these authors has been partially supported by projects MERIT-FORMS (TIN2005-09207-C03-03),
PROMESAS-CAM (S-0505/TIC0407) and STAMP (TIN2008-06622-C03-01).}\\
      Dpto. de Ingenier\'ia del Software e Inteligencia Artificial \\
      Universidad Complutense de Madrid, Spain \\
      E-mail: fernan@sip.ucm.es
\and A. J. FERN\'{A}NDEZ
\thanks {The work of this author has been partially supported by projects TIN2008-05941 (from
Spanish Ministry of Innovation and Science) and P06-TIC2250 (from Andalusia Regional Government).} \\
      Dpto. de Lenguajes y Ciencias de la Computaci\'{o}n, \\
      Universidad de M\'{a}laga, Spain \\
      E-mail: afdez@lcc.uma.es
}

\pagerange{\pageref{firstpage}--\pageref{lastpage}}
\volume{\textbf{??} (??):} \jdate{Junio 2004 2002}
\setcounter{page}{1} \pubyear{2004}

\maketitle

\label{firstpage}

%
%

\begin{abstract}
This paper presents a computational model  for the cooperation of
constraint domains and an implementation for a particular case of
practical importance. The computational model supports declarative
programming with lazy and possibly higher-order functions,
predicates, and the cooperation of different constraint domains
equipped with their res\-pective solvers, relying on a so-called
{\em Constraint Functional Logic Programming} ($CFLP$) scheme. The
implementation has been developed on top of the $CFLP$ system
$\mathcal{TOY}$, supporting the cooperation of the three domains
$\herbrand$, $\rdom$ and $\fd$, which supply equality and
disequality constraints over symbolic terms, arithmetic constraints
over the real  numbers, and finite domain constraints over the
integers, respectively. The computational model has been proved
sound and complete w.r.t. the declarative semantics provided by the
$CFLP$ scheme, while the implemented system has  been tested with a
set of benchmarks and shown to behave quite efficiently in
comparison to the closest related approach we are aware of.\\
{\em To appear in Theory and Practice of Logic Programming (TPLP)}
\end{abstract}

\begin{keywords}
Cooperating Constraint Domains,
Constraint Functional Logic Programming,
Constrained Lazy Narrowing,
Implementation.
\end{keywords}

%
%

\section{Introduction} \label{introduction}


Constraint Programming relies on {\em constraint solving} as a powerful mechanism for tackling practical applications.
The well-known $CLP$ Scheme  \cite{jaffar+:clp-popl87,JM94,JMM+98} provides a powerful and practical
framework for constraint programming which inherits the clean semantics and
declarative style of logic programming.
Moreover, the combination of $CLP$ with functional
programming has given rise to various so-called $CFLP$
(Constraint Functional Logic Programming)
schemes,
developed since 1991 and aiming at a very expressive  combination of
the constraint, logical and functional programming paradigms.


This paper tackles foundational and practical issues concerning the
efficient use of constraints in $CFLP$ languages and systems. Both
the $CLP$ and the $CFLP$ schemes must be instantiated by a
parametrically given  {\em constraint domain}  $\cdom$ which
provides specific data values, constraints based on specific
primitive operations, and a dedicated constraint solver. Therefore,
there are different {\em instances}  $\clp{\cdom}$ of the $CLP$
scheme for various choices of $\cdom$, and analogously for $CFLP$,
whose instances $\cflp{\cdom}$ provide a declarative framework for
any chosen domain $\cdom$. Useful `pure' constraint domains include
the Herbrand domain $\herbrand$ which supplies equality and
disequality  constraints over symbolic terms; the domain $\rdom$
which supplies arithmetic constraints over real numbers; and the
domain $\fd$ which supplies arithmetic and finite domain constraints
over integers. Practical applications, however, often involve more
than one
`pure' domain, and sometimes problem solutions have  to be
artificially adapted to fit a particular choice of domain and
solver.

Combining decision procedures for theories is a well-known problem,
thoroughly investigated since the seminal paper of Nelson and Oppen
\cite{NO79}. In constraint programming, however, the emphasis is
placed in computing answers by the interaction of constraint solvers
with user given programs, rather than in deciding satisfiability of
formulas. The cooperative combination of constraint domains and
solvers has evolved during the last decade as a relevant research
issue that is raising an increasing interest in the $CLP$ community.
Here we mention
\cite{baader+:combination-cp95,benhamou:heterogeneous-cs-alp96,monfroy:phd96,M98,granvilliers+:coperative-solvers-short-intro-alpnewsletter2001,MIS99b,Hofstedt:phd-thesis-2001,DBLP:conf/advis/MonfroyC04,HP07}
as a limited selection of references illustrating various approaches
to the problem. An important idea emerging from the research in this
area is that of `hybrid' constraint domain, built as a combination
of simpler `pure' domains  and designed to support the cooperation
of its components, so that more declarative and efficient solutions
for practical problems can be promoted.

\subsection{Aims of this paper} \label{aims}


The first aim of this paper is to present a computational model for
the cooperation of constraint domains in the $CFLP$ context, where
sophisticated functional programming features such as higher-order
functions and lazy evaluation must collaborate with constraint
solving. Our computational model is based on the $CFLP$ scheme and
goal solving calculus recently proposed in \cite{LMR04,LRV07}, which
will be enriched with new mechanisms for modeling the intended
cooperation. Moreover, we rely on the domain cooperation techniques proposed in our previous papers
\cite{estevez+:prole06,DBLP:journals/entcs/MartinFHRV07,estevez+:ppdp08},
 where we have introduced so-called {\em bridges} as a key tool for
communicating constraints between different domains.

Bridges are constraints of the form {\tt X \#==$_{d_i, d_j}$Y} which
relate the values of two variables {\tt X :: d$_i$, Y :: d$_j$} of
different base types, requiring them to be equivalent. For instance,
{\tt X \#==$_{int, real}$ Y} (abbreviated as {\tt X \#== Y} in the
rest of the paper) constrains the real variable {\tt Y :: real} to
take an integral real value equivalent to that of the integer
variable {\tt X :: int}. Note that the two types {\tt int} and {\tt
real} are kept distinct and their respective values are not
confused.

Our cooperative  computation model keeps different stores for
constraints corres\-ponding to various domains and solvers. In
addition, there is a special store where the bridge constraints
which arise during the computation are placed. A bridge constraint
{\tt X \#== Y} available in the bridge store can be used to {\em
project} constraints involving the variable {\tt X} into constraints
involving the variable {\tt Y}, or vice versa. For instance,  the
$\rdom$ constraint {\tt RX <= 3.4} (based on the inequality
primitive {\tt <=} -- `less or equal' -- for the type {\tt real})
can be projected into the $\fd$ constraint {\tt X \#<= 3} (based on
the inequality primitive  {\tt \#<=} -- `less or equal' -- for the
type {\tt int}) in case that the bridge {\tt X \#== RX} is
available. Projected constraints are submitted to their
corresponding store, with the aim of improving the performance of
the corresponding solver. In this way, projections behave as an
important cooperation mechanism, enabling certain solvers to profit
from (the projected forms) of constraints originally intended for
other solvers.

We have borrowed the idea of constraint projection from the work
of P. Hofstedt  et al. \cite{hofstedt:tigher-cooperation-cl00,DBLP:conf/cp/Hofstedt00,Hofstedt:phd-thesis-2001,HP07}, adapting it
to our $CFLP$ scheme and adding bridge constraints as a novel
technique which makes projections more flexible and compatible
with type discipline. In order to formalize our computation model,
we present a construction of  {\em coordination domains} $\ccdom$
as a special kind of `hybrid' domains built as a combination of
various `pure' domains intended to cooperate. In addition to the
specific constraints supplied by its various components,
coordination constraints also supply bridge constraints. As
particular case of practical interest, we present a coordination
domain $\ccdom$ tailored to the cooperation of the three pure
domains $\herbrand$, $\rdom$ and $\fd$.

Building upon the notion of coordination domain,
we also present a formal goal solving calculus called $\cclnc{\ccdom}$
(standing for Cooperative Constraint Lazy Narrowing Calculus over $\ccdom$)
which is sound and complete with respect to the instance
$\cflp{\ccdom}$ of the generic $CFLP$ scheme.
$\cclnc{\ccdom}$ embodies computation rules for creating bridges,
invoking constraint solvers, and performing constraint projections
as well as other more ad hoc operations for communication
among different constraint stores.
Moreover, $\cclnc{\ccdom}$ uses {\em lazy narrowing}
(a combination of lazy evaluation and unification)
for processing calls to program defined functions,
ensuring that function calls are evaluated only as far as demanded by the resolution
of the constraints involved in the current goal.


A second objective of the paper is to describe the implementation of a $CFLP$ system
which supports
the cooperation of solvers via bridges and projections for
the Herbrand domain $\herbrand$ and the two
numeric domains $\rdom$ and $\fd$,
following the computational model provided by the $\cclnc{\ccdom}$ goal solving calculus.
The implementation follows the techniques summarized in our previous papers
\cite{ciclops'07,estevez+:esop08}.
It has  been developed on top of the $\toy$ system \cite{toyreport},
which is in turn implemented on top of SICStus Prolog \cite{SP}.
The $\toy$ system already supported non-cooperative $CFLP$ programming
using the $\fd$ and $\rdom$ solvers provided  by
SICStus along with Prolog code for the $\herbrand$ solver.
This former system has been extended, including a store for bridges and
implementing mechanisms for computing bridges and projections according
to the $\cclnc{\ccdom}$ computation model.

Last but not  least, another important aim of the paper is to provide some evidence
on the  practical use and performance of our implementation. To this  purpose, we
present some illustrative examples and a set of benchmarks tailored to test the
performance of $\cclnc{\ccdom}$ as implemented in $\toy$ in comparison with the closest
related system we are aware of, namely the META-S tool
\cite{DBLP:conf/ki/FrankHM03,DBLP:conf/flairs/FrankHM03,DBLP:conf/wlp/FrankHR05}
which implements Hofstedt's framework for solver cooperation \cite{HP07}.
The experimental results we have obtained are quite encouraging.


The present paper thoroughly revises, expands  and elaborates
our previous related  publications in many respects.
In fact,  \cite{estevez+:prole06} was a very preliminary work which focused on
presenting bridges and providing evidence for their usefulness.
Building upon these ideas, \cite{DBLP:journals/entcs/MartinFHRV07}
introduced coordination domains and a cooperative goal solving
calculus over an arbitrary coordination domain,
proving local soundness and completeness results,
while \cite{estevez+:ppdp08} further elaborated the cooperative goal solving calculus,
providing stronger soundness and completeness results and experimental data on an
implementation tailored to the cooperation of the domains $\herbrand$, $\fd$ and $\rdom$.
Significant novelties in this article include:
technical improvements in the formalization of domains;
a new notion of solver taking care of critical variables and well-typed solutions;
a new notion of domain-specific constraint to clarify the behaviour of coordination domains;
various elaborations in the cooperative goal solving transformations needed to deal with critical variables and domain-specific constraints;
a more detailed presentation of the implementation results previously reported in \cite{ciclops'07,estevez+:esop08,estevez+:ppdp08};
and quite extensive comparisons to other related approaches.

\subsection{Motivating Examples} \label{examples}

As a motivation for the rest of the paper, we present in this subsection a few simple examples,
intended to  illustrate the different cooperation mechanisms that are supported by
the  computation model $\cclnc{\ccdom}$, as well as the benefits resulting from the cooperation.


To start with, we present  a small program written in $\toy$ syntax,
which solves the problem of searching for a $2D$ point lying in the
intersection of a discrete grid and a continuous region. The program
includes type declarations, equations for defining functions and
clauses for defining predicates. Type declarations are similar to
those used in functional languages such as Haskell \cite{haskell}.
Function applications use {\em curried notation},  also typical of
Haskell and other higher-order functional languages. The equations
used to define functions must be understood as conditional rewrite
rules of the form $\sfrule$, whose condition $\Delta$ is a
conjunction of constraints. Predicates are viewed as Boolean
functions, and clauses are understood as an abbreviation of
conditional rewrite rules of the form $\sprule$, whose righthand
side is the Boolean constant $true$. Moreover, conditions consisting
of a Boolean expression $exp$ are understood as an abbreviation of
the {\em strict equality} constraint $exp$ {\tt ==} $true$, using
the strict equality operator {\tt ==} which is a primitive operation
supplied by the Herbrand domain $\herbrand$. The program's text is
as follows:\\

{\tt \% Discrete versus continuous points:

\underline{type} dPoint = (int, int)\\
\hspace*{0.375cm}\underline{type} cPoint = (real, real)\\

\% Sets and membership:  \hspace*{1.cm}

\underline{type} setOf A = A -> bool

isIn :: setOf A -> A -> bool

isIn Set Element = Set Element\\

\% Grids and regions as sets of points:

\underline{type} grid = setOf  dPoint\\
\hspace*{0.375cm}\underline{type} region = setOf  cPoint\\
}

{\tt \% Predicate for computing intersections of regions and grids:

bothIn :: region -> grid -> dPoint -> bool

bothIn Region Grid (X, Y) :- X \#== RX, Y \#== RY,

\hspace*{1.0cm}isIn Region (RX, RY), isIn Grid (X,Y), labeling [~] [X,Y] \\
}

{\tt \% Square grid (discrete):

square :: int -> grid

square N (X,Y) :- domain [X,Y] 0 N \\
}

{\tt \% Triangular region (continuous):

triangle :: cPoint -> real  -> real -> region

triangle (RX0,RY0)  B H (RX,RY) :-

\hspace*{2.cm} RY >= RY0 - H,

\hspace*{2.cm} B * RY - 2 * H * RX <= B * RY0 - 2 * H * RX0,

\hspace*{2.cm} B * RY + 2 * H * RX <= B * RY0 + 2 * H * RX0\\ \\
}

{\tt \% Diagonal segment (discrete):

diagonal :: int -> grid

diagonal N (X,Y) :- domain [X,Y] 0 N, X == Y  \\
}

{\tt \% Parabolic line (continuous):

parabola :: cPoint ->  region

parabola (RX0,RY0)  (RX,RY) :- RY - RY0 == (RX - RX0) * (RX - RX0) \\
}

\noindent Because of all the conventions explained above, the clause
for the {\tt bothIn} predicate
included in the program must be understood as an abbreviation of the rewrite rule \\

{\tt
\hspace*{0.5cm} bothIn Region Grid (X,Y) -> true <==

\hspace*{2.cm} X \#== RX, Y \#== RY,

\hspace*{2.cm} isIn Region (RX,RY) == true, isIn Grid (X,Y) == true,

\hspace*{2.cm} labeling [~] [X,Y]
} \\

\noindent whose condition includes two bridge constraints, two
strict equality constraints provided by the domain $\herbrand$, and
a last constraint using the {\tt labeling} primitive supplied by the
domain $\fd$. The other clauses and equations in the program can be
analo\-gously understood as conditional rewrite rules whose
conditions are constraints supported by some of the three domains
$\herbrand$, $\rdom$ or $\fd$.

Let us now discuss the intended meaning of the program.
The {\tt bothIn} predicate is intended to check if a given discrete point {\tt (X,Y)} belongs to the
intersection of the continuous region {\tt Region} and the discrete grid {\tt Grid}
given as parameters, and the constraints occurring as conditions are designed to
this purpose. More precisely, the two bridge constraints {\tt X \#== RX, Y \#== RY}
ensure that the discrete point {\tt (X,Y)} and the continuous point {\tt (RX,RY)} are equivalent;
the two strict equality constraints  {\tt isIn Region (RX, RY) == true, isIn Grid (X,Y) == true}
ensure membership to {\tt Region} and {\tt Grid}, respectively;
and finally the $\fd$ constraint {\tt labeling [~] [X,Y]}
ensures that the variables {\tt X} and {\tt Y} are bound to integer values.

\vspace*{0.75cm}
\begin{figure}[h]
\begin{center}
\vspace*{-0.4cm}
\includegraphics[scale=0.3,angle=0]{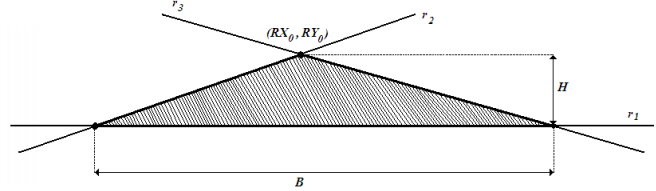}
\caption{Triangular region}\label{triangle}
\end{center}
\end{figure}

Note that both grids and regions are represented as sets,
represented themselves as Boolean functions. They can be  passed as
parameters because our programming framework supports higher-order
programming features. The program also defines two functions {\tt
square} and {\tt triangle}, intended to compute representations of
square grids and triangular regions, respectively. Let us discuss
them in turn. We first note that the type declaration for {\tt
triangle} can be written in the  equivalent form {\tt triangle ::
cPoint -> real  -> real -> (cPoint -> bool)}. A function call of the form
{\tt triangle (RX0,RY0)  B H} is intended to return a Boolean function
representing the region of all continuous $2D$ points lying within
the isosceles triangle with upper vertex {\tt (RX0,RY0)},
base {\tt B} and height {\tt H}. Applying this Boolean function
to the argument {\tt (RX,RY)} yields a function call written as
{\tt triangle (RX0,RY0)  B H (RX,RY)} and expected to
return {\tt true} in case that {\tt (RX,RY)} lies within the
intended isosceles triangle, whose three vertices are
{\tt (RX0,RY0)}, {\tt (RX0-B/2,RY0-H)} and {\tt (RX0+B/2,RY0-H)}.
The three sides of the triangle are mathematically characterized by the equations
 {\tt RY = RY0-H},
 {\tt B*RY-2*H*RX = B*RY0-2*H*RX0} and
 {\tt B*RY+2*H*RX = B*RY0+2*H*RX0}
 (corresponding to the lines $r_1$, $r_2$ and $r_3$ in Fig.
\ref{triangle}, respectively). Therefore, the conjunction of three
linear inequality $\rdom$ constraints occurring as conditions in the
clause for {\tt triangle} succeeds for those points {\tt (RX,RY)}
lying within the intended triangle.

Similarly, the type declaration for {\tt square} can be written in the  equivalent form
{\tt square :: int  -> (dPoint -> bool)},
and a function call of the form  {\tt square N} is intended to
return a Boolean function representing the grid of all discrete $2D$ points
with coordinates belonging to the interval of integers
between {\tt 0} and {\tt N}. Therefore, a function call of the form
{\tt square N (X,Y)} must return {\tt true} in case that {\tt (X,Y)} lies within
the intended grid, and for this reason the single $\fd$ constraint placed as condition in the
clause for {\tt square} has been chosen to impose the interval of integers
between {\tt 0} and {\tt N} as the domain of possible values for the variables
{\tt X} and {\tt Y}.

Finally, the last two functions {\tt diagonal} and {\tt parabola} are defined in
such a way that {\tt diagonal N}  returns a Boolean function
representing the diagonal of the grid represented by {\tt square N},
while {\tt parabola (RX0,RY0)} returns a Boolean function
representing the parabola whose equation is  {\tt RY-RY0 =
(RX-RX0)*(RX-RX0)}. The type declarations and clauses for these
functions can be understood similarly to the case of {\tt square}
and {\tt triangle}.

Different {\em goals} can be posed and solved using the small
program just described and the cooperative goal solving calculus
$\cclnc{\ccdom}$ as implemented in the $\toy$ system. For the sake
of discussing some of them, assume two fixed positive integer values
{\tt d} and {\tt n} such that {\tt n = 2*d}. Then {\tt (d,d)} is the
middle point of the grid {\tt (square n)}, which includes  {\tt
(n+1)$^2$} discrete points. The three following goals ask for points
in the intersection of this fixed square grid with three different
triangular regions:

\begin{itemize}
\item {\bf Goal 1}: {\tt bothIn (triangle (d, d+0.75)  n 0.5) (square n) (X,Y)}.\\
This goal  fails.
\item {\bf Goal 2}: {\tt bothIn (triangle (d, d+0.5)  2 1) (square n) (X,Y)}.\\
This goal computes one solution for {\tt (X,Y)}, corresponding to
the point  {\tt (d,d)}.
\item {\bf Goal 3}: {\tt bothIn (triangle (d, d+0.5)  (2*n) 1) (square n) (X,Y)}.\\
This goal computes {\tt n+1} solutions  for {\tt (X,Y)}, corresponding to
the points {\tt  (0,d), (1,d), \ldots, (n,d)}.
\end{itemize}

These three goals are illustrated in Fig. \ref{goals} for the
particular  case {\tt n = 4} and {\tt d = 2}, although the shapes
and positions of the three triangles with respect to the middle
point of the grid would be the same for any even positive integer
{\tt n = 2*d}. The expected solutions for each of the three goals
are clear from the figures.

\begin{figure}[h]
\begin{center}
\vspace*{0.2cm}
\hspace*{-0.275cm}\includegraphics[scale=0.345,angle=0]{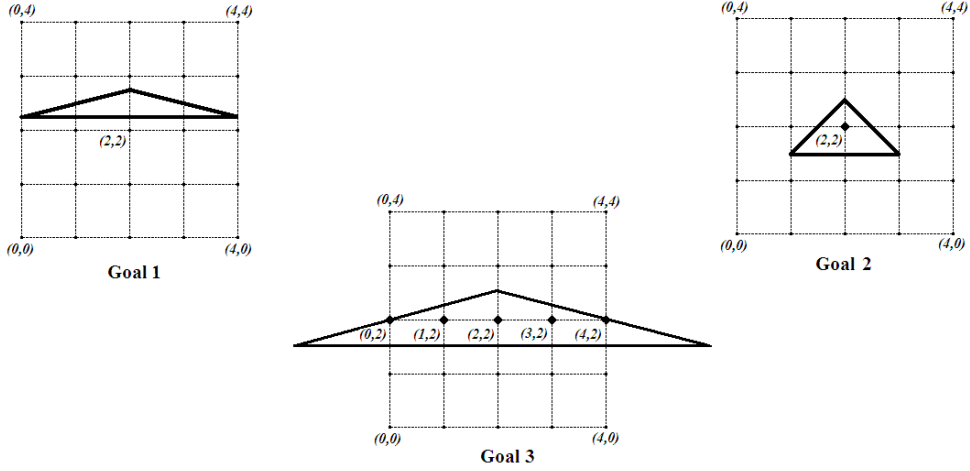}
\end{center}
\caption{Intersection of a fixed square grid with three different
triangular regions}\label{goals}
\end{figure}

 In the three cases, cooperation between
the $\rdom$ solver and the $\fd$ solver is crucial for the
efficiency of the computation. In the case of {\bf Goal 2},
cooperative goal solving implemented in $\toy$ according to the
$\cclnc{\ccdom}$ computation model uses the clauses in the program
and eventually reduces the problem of solving the goal to the
problem of solving a constraint system that, suitably simplified,
becomes:\\

\hspace*{1.5cm}
{\tt X \#== RX, Y \#== RY,}

\hspace*{1.5cm}
{\tt RY >= d-0.5, RY-RX <= 0.5, RY+RX <= n+0.5,}

\hspace*{1.5cm} {\tt domain [X,Y] 0 n, labeling [\,] [X,Y].}
\vspace*{0.25cm}

 \noindent
The $\toy$ system has the option to enable or  disable the computation of projections.
When projections are disabled, the two bridges do still work as constraints,
and the last $\fd$ constraint {\tt labeling [\,] [X,Y]} forces the enumeration of all
possible values for {\tt X} and {\tt Y} within their domains,
eventually finding the unique solution {\tt X = Y = d} after $\mathcal{O}$({\tt n}$^2$) steps.
When projections are enabled, the available bridges are used to project the $\rdom$ constraints
{\tt RY >= d-0.5, RY-RX <= 0.5, RY+RX <= n+0.5} into the $\fd$ constraints
{\tt Y \#>= d, Y\#-X \#<= 0, Y\#+X \#<= n}. Since {\tt n = 2*d}, the only possible solution of these inequalities  is {\tt X = Y = d}. Therefore, the $\fd$ solver drastically prunes the domains of {\tt X} 
and {\tt Y} to the singleton set {\tt \{d\}}, and solving the last labeling constraint leads to the unique solution with no effort. For a big value of  {\tt n = 2*d} the performance of the computation is greatly enhanced in comparison to the case where projections
are disabled, as  confirmed by the experimental results given in Subsection~\ref{subsect:TOY}.
The expected speed-up in execution time corresponds to the improvement
from the $\mathcal{O}({\tt n}^2)$ steps needed to execute the labeling constraint {\tt  labeling [\,] [X,Y]}
when the domains of both {\tt X} and {\tt Y} have size $\mathcal{O}({\tt n})$,
to the  $\mathcal{O}(1)$ steps needed to execute the same constraint when the domains
of both {\tt X} and {\tt Y} have been pruned to size $\mathcal{O}(1)$.
Similar speedups are observed  when solving
 {\bf Goal 1} (which finitely fails, and where the expected execution time also improves
 from  $\mathcal{O}({\tt n}^2)$ to $\mathcal{O}(1)$)
 and  {\bf Goal 3} (which has just {\tt n+1} solutions, and where the expected execution time reduces
 from  $\mathcal{O}({\tt n}^2)$ to $\mathcal{O}({\tt n})$).

\noindent The three goals just discussed mainly illustrate the
benefits obtained by the $\fd$ solver from the projection of $\rdom$
constraints. In fact, when $\toy$ solves these three goals according
to the cooperative computation model $\cclnc{\ccdom}$, the available
bridge constraints also allow to project the $\fd$ constraint {\tt
domain [X,Y] 0 n} into the conjunction of the $\rdom$ constraints
{\tt 0 <= RX, RX <= n, 0 <= RY, RY <= n}. These constraints,
however, are not helpful for optimizing the resolution of the
previously computed $\rdom$ constraints {\tt RY >= d-0.5, RY-RX <=
0.5, RY+RX <= n+0.5}.

In general, it seems easier for the $\fd$ solver to profit from the projection of $\rdom$ constraints than the other way round.
This is because the solution of many practical problems is arranged to finish with solving $\fd$ labeling constraints,
which means enumerating values for integer variables, and this process can greatly benefit from a
reduction of the variables' domains due to previous projections of $\rdom$ constraints.
However, the projection of $\fd$ constraints into $\rdom$ constraints can help to define the intended solutions
even if the performance of the $\rdom$ solver does not improve. For instance, assume that  the value chosen for {\tt n = 2*d} is big, and consider the goal

\begin{itemize}
\item {\bf Goal 4}: {\tt bothIn (triangle (d,d)  n d) (square 4) (X,Y)}.
\end{itemize}
\vspace*{-0.15cm}

\noindent whose resolution eventually reduces to the problem of
solving a constraint system that, suitably simplified, becomes:
\vspace*{0.15cm}

\hspace*{1.5cm}
{\tt X \#== RX, Y \#== RY,}

\hspace*{1.5cm}
{\tt RY >= 0, RY-RX <= 0, RY+RX <= n,}

\hspace*{1.5cm} {\tt domain [X,Y] 0 4, labeling [\,] [X,Y].}
\vspace*{0.1cm}

\noindent The solutions correspond to the points lying in the
 intersection of a big isosceles triangle and a tiny square grid.
 Projecting  {\tt RY >= 0, RY-RX <= 0, RY+RX <= n}
into $\fd$ constraints via the two bridges {\tt X \#== RX, Y \#== RY}
brings no significant gains to the $\rdom$ solver whose task is
anyhow trivial. The $\rdom$ constraints projected from
 {\tt domain [X,Y] 0 4} (i.e., {\tt 0 <= RX, RX <= 4, 0 <= RY, RY <= 4})
do not improve the performance of the $\rdom$
solver either, but they help to define the intended solutions. In
this example, the last labeling constraint eventually
enumerates the right solutions even if the projection of the domain
constraint to $\rdom$ does not take place, but this
projection would allow the $\rdom$ solver to compute suitable
constraints as solutions in case that the labeling constraint
were removed.
\vspace*{0.1cm}

There are also some cases where the performance of the $\rdom$
solver can benefit from the cooperation with the  $\fd$ domain.
Consider for instance the goal

\begin{itemize}
\item {\bf Goal 5}: {\tt bothIn (parabola (2,0)) (diagonal 4) (X,Y)}.
\end{itemize}
\vspace*{-0.15cm}

\noindent asking for points in the intersection of the discrete
diagonal segment of size 4 and a parabola with vertix {\tt (2,0)}
(see Fig. \ref{parabola}). Solving this goal eventually reduces to
solving a constraint system that, suitably simplified,
becomes:\vspace*{0.15cm}

\hspace*{1.5cm} {\tt X \#== RX, Y \#== RY,}

\hspace*{1.5cm} {\tt RY == (RX-2)*(RX-2),}

\hspace*{1.5cm} {\tt domain [X,Y] 0 4, X == Y, labeling [\,]
[X,Y].}\vspace*{0.1cm}

\begin{figure}[h]
\begin{center}
\hspace*{1.0cm}\includegraphics[scale=0.4,angle=0]{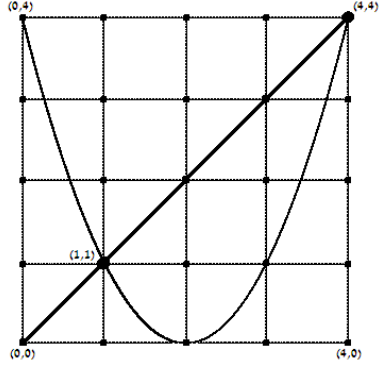}
\end{center}
\caption{Intersection of a parabolic line and a diagonal segment}\label{parabola}
\end{figure}

\noindent Cooperative goal solving as implemented in $\toy$
processes the constraints within the current goal in left-to right
order, performing projections whenever possible, and sometimes
delaying a constraint that cannot be managed by the available
solvers. In this case, the quadratic $\rdom$ constraint {\tt RY ==
(RX-2)*(RX-2)} is delayed because the $\rdom$ solver used by $\toy$
(inherited form SICStus Prolog) cannot solve non-linear constraints.
However, since this strict equality relates expressions of type {\tt
real}, it is accepted as a $\rdom$ constraint and projected via the
available bridges, producing  the $\fd$ constraint {\tt Y == (X-2)*(X-2)},
which is submitted to the $\fd$ solver.
Next, projecting the $\fd$ constraint {\tt domain [X,Y] 0 4} and
solving {\tt X == Y} causes the $\rdom$ constraints {\tt 0 <= RX, RX
<= 4, 0 <= RY, RY <= 4} to be submitted to the $\rdom$ solver, and
the variable {\tt X} to be substituted in place of {\tt Y} all over
the goal. The bridges  {\tt X \#== RX, Y \#== RY} become then {\tt X
\#== RX, X \#== RY}, and the labeling constraint becomes {\tt
labeling [\,] [X,X]}. An especial mechanism called {\em bridge
unification} infers from the two bridges {\tt X \#== RX, X \#== RY}
the strict equality constraint {\tt RX == RY}, which is solved by
substituting {\tt RX} for {\tt RY} all over the current goal. At
this point, the delayed $\rdom$ constraint becomes {\tt RX ==
(RX-2)*(RX-2)}. Finally, the $\fd$ constraint {\tt labeling [\,]
[X,X]} is solved by enumerating all the possible values for {\tt X}
allowed by its domain, and continuing a different alternative
computation with each of them. Due to the bridge {\tt X \#== RX},
each integer value {\tt v} assigned to {\tt X} by the labeling
process causes the variable {\tt RX} to be bound to the
integral real number {\tt rv}  equivalent to {\tt v} (in our computation model,
this is part of the behaviour of a solver in charge of solving
bridge constraints). The binding of {\tt RX} to {\tt rv} awakens the
delayed constraint {\tt RX == (RX-2)*(RX-2)}, which becomes the
linear (and even ground) constraint {\tt rv == (rv-2)*(rv-2)} and
succeeds if {\tt rv} is an integral solution of the delayed
quadratic equation. In this way, the two solutions of {\bf Goal 5}
are eventually computed, corresponding to the two points {\tt (X,Y)}
lying in the intersection of the parabolic line and the diagonal
segment: {\tt  (1,1)} and {\tt (4,4)}, as seen in Fig. \ref{parabola}.

All the computations described in this subsection can be actually executed in the $\toy$ system and also
formally represented in the cooperative goal solving calculus $\cclnc{\ccdom}$.
The formal representation of goal solving computations in $\cclnc{\ccdom}$ performs quite many detailed intermediate steps.
In particular, constraints are transformed into a flattened form (without nested calls to primitive functions) before
performing projections, and especial mechanisms for creating new bridges in some intermediate steps
are provided. Detailed explanations and examples are given in Section \ref{cooperative}.
\vspace*{-.3cm}

\subsection{Structure of the Paper} \label{structure}


To finish the introduction, we summarize the organization of the rest of the paper.
Section \ref{coordination} starts by presenting the main features of the $CFLP$ scheme, including
a mathematical formalization of constraint domains and solvers. The presentation follows
\cite{LRV07}, adding an explicit consideration of type discipline and an improved presentation of
constraint domains, solvers and their formal properties. The rest of the section is new with
respect to previous presentations of $CFLP$ schemes: it discusses bridge constraints and the
construction of coordination domains, concluding with a presentation of a particular coordination
domain $\ccdom$ tailored to the cooperation of the domains $\herbrand$, $\rdom$ and $\fd$.
In the subsequent sections, $\ccdom$ always refers to this particular coordination domain.

Section \ref{cooperative} presents our proposal of a computational
model for cooperative programming and goal solving in
$\cflp{\ccdom}$. Programs and goals are introduced, the cooperative
goal solving calculus $\cclnc{\ccdom}$ is discussed in detail, and
its main formal properties (namely {\em soundness} and {\em limited
completeness} w.r.t. the declarative semantics of $\cflp{\ccdom}$ provided by the
$CFLP$ scheme) are presented.

Section \ref{implementation} sketches the implementation of the
$\cclnc{\ccdom}$ computational model on top of the $\toy$ system
\cite{toyreport}, which is itself implemented on top of SICStus
Prolog \cite{SP}. The architectural components of the current $\toy$ system are described 
and the extensions of $\toy$ responsible for the treatment of bridges
and projections according to the formal model provided by the previous section
are briefly discussed.

Section \ref{performance} discusses the practical use of the $\toy$
system for solving problems involving the cooperation of the domains
$\herbrand$, $\rdom$ and $\fd$.
A significant set of benchmarks is analyzed in order to study
how the cooperation mechanisms affect to the performance of the system,
and a detailed comparison to the performance of the META-S tool
is also presented.

Section \ref{relatedwork} is devoted to a discussion of related
work, trying to give an overview of different approaches in the
area of cooperative constraint solving. Section \ref{conclusions}
summarizes the main results of the paper, points out to some
limitations of the current  $\toy$ implementation, and presents a
brief outline of planned  future work.

The results reported in this paper are supported by the experimental results presented in 
Section \ref{performance} and a number of mathematical proofs, most of which have been 
collected in the Appendices \ref{pSolvCdom} and \ref{PropertiesCalculus}. In the case of reasonings concerning type discipline, we have refrained from providing full details,
that would be technically tedious and distract from the main emphasis of the paper.
More detailed proofs could be worked out, if desired, by adapting the techniques from \cite{GHR01}.
\vspace*{-.3cm}
%
%

\section{Coordination of Constraint Domains in the $CFLP$ Scheme} \label{coordination}

The scheme presented in \cite{LRV07}  serves as a logical
and semantic framework for lazy Constraint Functional Logic
Programming (briefly $CFLP$) over a parametrically given
constraint domain. The aim of this
section  is to model the coordination of several constraint
domains with their respective solvers using instances
$\cflp{\ccdom}$ of the  $CFLP$ scheme, where  $\ccdom$ is a
so-called {\em coordination domain} built as a suitable
combination of the various domains intended to cooperate. We use an
enhanced version of the  $CFLP$ scheme, extending \cite{LRV07}
with an explicit treatment of a polymorphic type discipline in the
style of Hindley-Milner-Damas and an improved presentation of
constraint domains, solvers and their formal properties. In this
setting, we discuss the three
`pure' constraint domains $\herbrand$, $\rdom$ and $\fd$ along
with their solvers. Next, we present bridge constraints and  the
construction of coordination domains, concluding with the
construction of a particular coordination domain $\ccdom$ tailored
to the cooperation of the domains $\herbrand$, $\rdom$ and $\fd$,
which is the topic of the rest of the paper.

\subsection{Signatures and Types} \label{types}


We assume a  {\em universal signature} $\Omega = \langle TC,\ BT,\ DC,\ DF,\ PF \rangle$
consisting of five pairwise disjoint sets of symbols, where

\begin{itemize}
\item
$TC = \bigcup_{n \in {\mathbb N}} TC^n$ is a family of countable and mutually disjoint sets
of {\em type constructors}, indexed by arities.
\item
$BT$ is a set of  {\em base types}.
\item
$DC = \bigcup_{n \in {\mathbb N}} DC^n$ is a family of countable and mutually disjoint sets
of {\em data constructors}, indexed by arities.
\item
$DF = \bigcup_{n \in {\mathbb N}} DF^n$ is a family of countable and mutually disjoint sets
of {\em defined function symbols}, indexed by arities.
\item
$PF = \bigcup_{n \in {\mathbb N}} PF^n$ is a family of countable and
mutually disjoint sets of {\em primi\-tive function symbols},
indexed by arities.
\end{itemize}

The idea is that base types and primitive function symbols are related to specific constraint domains,
while type constructors, data constructors and defined function symbols are related to user given programs. For each choice of a specific family of base types $SBT \subseteq BT$
and a specific family of primitive function symbols $SPF \subseteq PF$, we will say that
$\Sigma = \langle TC,\ SBT,\ DC,\ DF,\ SPF \rangle$ is a {\em domain specific signature}.
Note that any domain specific signature $\Sigma$ inherits all the type constructors, data constructors and defined function symbols from the universal signature $\Omega$, since different programs over a given constraint domain of signature $\Sigma$ might use them. All symbols belonging to the family
$DC \cup DF \cup SPF$ are collectively called {\em function symbols}.


All along the paper we will work with a static type discipline based
on the Hindley-Milner-Damas type system
\cite{hindley69,Mil78,damas+:types}. A detailed study of polymorphic
type discipline in the context of Functional Logic Programming
(without constraints) can be found in \cite{GHR01}. In the sequel we
assume a countably
infinite set  $\tvar$ of {\em type variables}. {\em Types} $\tau$
$\in$ $Type_\Sigma$ have the syntax
 $\tau ::= A \mid d \mid (c_t~\tau_{1}\ldots\tau_{n}) \mid (\tau_1, \dots, \tau_n) \mid (\tau_1 \to \tau_0)$,
where $A \in \tvar$, $d \in SBT$ and $c_t \in TC^n$.
By convention, parenthesis are omitted when there is no ambiguity,
$c_t~\overline{\tau}_{n}$ abbrevia\-tes $c_t~\tau_{1} \ldots\tau_{n}$, and
``$\to$'' associates to the right,
$\overline{\tau}_{n}\to \tau$ abbrevia\-tes $\tau_{1} \to \cdots \to \tau_{n} \to \tau$.
Types $c_t~\overline{\tau}_{n}$, $(\tau_1, \dots, \tau_n)$ and $\tau_1 \to \tau_0$
are used to represent constructed values, tuples and functions, respectively.
A type without any occurrence of ``$\to$'' is called a {\em datatype}.


{\em Type substitutions} $\sigma_t, \theta_t \in TSub_\Sigma$ are mappings
from $\tvar$ into $Type_\Sigma$, extended to mappings from
$Type_\Sigma$ into $Type_\Sigma$ in the natural way. By convention,
we write $\tau \sigma_t$ ins\-tead of $\sigma_t(\tau)$ for any type $\tau$.
Whenever $\tau' = \tau \sigma_t$ for some $\sigma_t$, we say that $\tau'$ is an
{\em instance} of $\tau$ (or also that $\tau$ is more general than $\tau'$)
 and we write $\tau \preceq \tau'$.

The set of type variables occurring in $\tau$ is written $tvar(\tau)$.
A type $\tau$  is called {\em monomorphic} iff $tvar(\tau) = \emptyset$, and {\em polymorphic} otherwise.
A polymorphic type $\tau$  must be understood as representing all its possible monomorphic instances $\tau'$.


Function symbols in any signature $\Sigma$ are required to come along with a so-called
{\em principal type declaration}, which indicates its most general type.
More precisely,

\begin{itemize}
\item
Each  $n$-ary $c \in DC^n$ must have attached a principal type declaration of the form
$c\,::\,\overline{\tau}_n \to c_t\, \overline{A}_k$, where $n, k \geq 0$,
$A_{1},\ldots,A_{k}$ are pairwise different type variables, $c_t \in TC^k$,
$\tau_{1},\ldots,\tau_{n}$ are datatypes, and
$\bigcup_{i = 1}^{n} tvar(\tau_{i}) \subseteq \{A_{1},\ldots, A_{k}\}$
(so-called \emph{transparency property}).
\item
Each $n$-ary $f \in DF^n$ must have attached a principal type declaration of the form
$f\,::\,\overline{\tau}_n \to \tau$, where $\tau_{i}\, (1 \leq i \leq n)$ and $\tau$  are arbitrary types.
\item
Each $n$-ary $p \in SPF^n$ must have attached a principal type declaration of the form
$p\,::\,\overline{\tau}_n \to \tau$, where $\tau_1$, $\ldots$, $\tau_n$, $\tau$ are datatypes
and $\tau$ is not a type variable.
\end{itemize}


For the sake of semantic considerations, we assume a special data constructor $(\bot\,::\, A)\, \in\, DC^0$,
intended to represent an {\em undefined value} that belongs to any type.
The type and data constructors needed to work with Boolean values and lists are also assumed
to be present in the universal signature $\Omega$.
We also assume that $SPF^2$ includes the polymorphic primitive function symbol 
 {\tt ==\,::\,A -> A -> bool}, that will be written  in infix notation and used to express {\em strict  equality constraints} in those domains where it is available.


In concrete programming languages such as $\mathcal{TOY}$ \cite{toyreport} and Curry 
\cite{curryreport}, data constructors and their principal types are introduced by datatype declarations,
the principal types of defined functions can be either declared or inferred by the compiler,
the principal types of primitive functions are predefined and known to the users,
and $\bot$ does not textually occur in programs.


\begin{example}[Signatures and Types] \label{typesExample}
In order to illustrate the main notions concerning signatures and types, let us consider
the signature $\Sigma$ underlying the program presented in Subsection \ref{examples}.
There we find:

\begin{itemize}
\item
Two base types {\tt int} and {\tt real} for the integer and real numeric values, respectively.
\item
A nullary type constructor {\tt bool} for the type of Boolean values,
and a  unary type constructor {\tt list} for the type of polymorphic lists.
The concrete syntax for  {\tt list A}  is {\tt [A]}.
\item
{\tt [A]} is a datatype, since it has no occurrences of the type constructor {\tt ->}.
Moreover, it is polymorphic, since it includes a type variable.
Among the instances of {\tt [A]} we can find {\tt [int]} (for lists of integers) and
{\tt [int -> int]}  (for lists of functions of type {\tt int -> int}).
Note that an instance of a datatype must not be a datatype.
\item
Two nullary data constructors
{\tt false, true\,::\,bool} (for Boolean values);
a nullary data constructor {\tt nil\,::\,[A]} (for the empty list);
and a binary data constructor {\tt cons\,::\,A -> [A] -> [A]} (for nonempty  lists).
The concrete syntax for {\tt nil}
(resp. {\tt cons}) is {\tt [\,]} (resp. {\tt :}),
where {\tt :} is intended to be used as an infix operator.
\item
The principal types of the constructors in the previous item can be derived from the datatype declarations \\
\hspace*{0.5cm}  {\tt data bool = false | true} \\
\hspace*{0.5cm}  {\tt data [A] = [\,] | (A : [A])} \\
which are predefined and do not need to be included within programs.
\item
In the program presented in Subsection \ref{examples} there are also {\em type alias} declarations, such as \\
\hspace*{0.5cm}  {\tt type dPoint = (int,int)} \\
\hspace*{0.5cm}  {\tt type setOf A = A -> bool} \\
\hspace*{0.5cm}  {\tt type region = setOf dPoint} \\
Such declarations are just a practical convenience for naming certain types.
They cannot involve recursion, and the names of type alias so introduced are not
considered to belong to the signature.
\item
Defined function symbols of various arities, as e.g. {\tt isIn}, {\tt square} $\in DF^2$.
These two function symbols are binary because the rewrite rules given for them within
the program expect two formal parameters at their left hand sides. In general, rewrite rules
included in programs  for defining the behaviour of symbols $f \in DF^n$ are expected to
have $n$ formal parameters at their left hand sides. In some cases, this $n$ may not
identically correspond to the number of arrows observed in the principal type of $f$.
For instance, although {\tt square} $\in DF^2$, the principal type is
{\tt square\,::\,int -> grid}. The apparent contradiction disappears by noting that {\tt grid} is
declared as a type alias for {\tt (int,int) -> bool}. Since the type constructor {\tt ->} associates
to the right, we have in fact {\tt square\,::\,int -> (int,int) -> bool}.
\item
Primitive function symbols of various arities, as e.g. the binary
primitives {\tt \#==}, {\tt labeling}, {\tt +} and {\tt <=}, and the
ternary primitive {\tt domain}. The concrete syntax requires {\tt
\#==},  {\tt +} and {\tt <=} to be used in infix notation. Each
primitive has a predefined principal type. For instance, {\tt
\#==\,::\,int -> real -> bool}, {\tt +\,::\,real -> real -> real}
and {\tt domain\,::\,[int] -> int -> int -> bool}. These
declarations do not need to be included within programs.
\end{itemize}
\end{example}
\vspace*{-.4cm}

\subsection{Expressions and Substitutions} \label{expressions}


For any domain of specific signature $\Sigma$, constraint programming will use expressions
which may have occurrences of certain values of base type. Therefore, in order to
define the syntax of expressions we assume a $SBT$-indexed family
$\mathcal{B} = \{\mathcal{B}_d\}_{d \in SBT}$, where each $\mathcal{B}_d$ is
a non-empty set whose elements are understood as {\em base values} of type $d$.
In the sequel, we will use letters $u, v, \ldots$ to indicate base values.
By an abuse of notation, we will also write $u \in \mathcal{B}$ instead of
$u \in \bigcup_{d \in SBT} \mathcal{B}_d$.


Moreover, we also assume a countable infinite set $\var$ of {\em data variables}
(disjoint from $\mathcal {T\!V}\!\!ar$ and  $\Sigma$), and we define
 {\em applicative $\Sigma$-expressions} $e \in Exp_{\Sigma}({\mathcal{B}})$ over $\mathcal{B}$
with the syntax
$e::= X \mid u \mid h \mid (e_1, \ldots, e_n) \mid (e\,e_1)$, where
$X\in\mathcal V\!\!ar$,  $u \in \mathcal{B}$ and  $h\in DC \cup DF \cup SPF$.

Expressions $ (e_1, \ldots, e_n)$ represent ordered $n$-tuples,
while expressions $(e\, e_{1})$ -- not to be confused with ordered
pairs $(e, e_{1})$ -- stand for the {\em application} of the
function represented by $e$ to the argument represented by $e_{1}$.
Following a usual convention, we assume that application associates
to the left, and we use the notation $(e\, \tpp{e}{n})$ to
abbreviate $(e\, e_{1} \ldots e_{n})$. More generally, parenthesis
can be omitted when there is no ambiguity. Applicative syntax is
common in higher-order functional languages. The usual first-order
syntax for expressions can be translated to applicative syntax by
means of so-called {\em curried notation}. For instance, $f(X,g(Y))$
becomes $(f\, X\, (g\, Y))$.


Expressions without repeated variable occurrences are called {\em linear},
variable-free expressions are called {\em ground} and expressions without any occurrence of $\bot$
are called {\em total}.  Some particular expressions are intended to represent data values that
do not need to be evaluated. Such expressions are  called
$\Sigma$-patterns $t \in Pat_{\Sigma}({\mathcal{B}})$ over $\mathcal{B}$
and have the syntax
$t ::= X \mid u \mid (t_1, \ldots, t_n) \mid c\, \overline{t}_m \mid f\,\overline{t}_m \mid p\, \overline{t}_m$, where
$X \in \mathcal V\!\!ar$, $u \in \mathcal{B}$,
$c \in DC^n$ for some $m \leq n$,
$f \in DF^n$ for some $n > m$, and
$p \in SPF^n$ for some $n > m$.
The restrictions concerning arities in the last three cases are motivated by the idea that an expression
of the form $h\,\overline{t}_n$ (where $h \in DF^n \cup SPF^n$) is potentially evaluable and therefore not to be viewed as representing data.

The set of all ground patterns over  $\mathcal{B}$ is noted
$GPat_{\Sigma}({\mathcal{B}})$. Sometimes we will write
$\mathcal{U}_{\Sigma}({\mathcal{B}})$ in place of
$GPat_{\Sigma}({\mathcal{B}})$, viewing this set as the {\em
universe of values} over $\mathcal{B}$. The following classification
of expressions is also useful: $(X\, \overline{e}_m)$ (with $X \in
\mathcal V\!\!ar$ and $m \geq 0$) is called a {\em flexible}
expression; while $u \in \mathcal{B}$ and all expressions of the
form $(h~\overline{e}_m)$ (with $h \in DC \cup DF \cup SPF$) are
called {\em rigid}. Moreover, a rigid expression $(h\,
\overline{e}_m)$ is called {\em passive} iff $h \in DF^n \cup SPF^n$
and $m < n$, and {\em active} otherwise. Tuples $(e_1, \ldots, e_n)$
are also considered as passive expressions. The idea is that any
passive expression has the outermost appearance of a pattern,
although it might not be a pattern in case that any of its inner
subexpressions is active.


As illustrated by the program presented in Subsection \ref{examples}, tuples are useful for programming and therefore
the tuple syntax is supported by many programming languages, including $\mathcal{TOY}$.
On the other hand, tuples can be treated as a particular case of constructed values, just by
assuming data constructors $tup_n \in DC^n$ in the universal signature and viewing any tuple
$(e_1, \ldots, e_n)$ as syntactic sugar for  $tup_n\, e_1 \ldots e_n$.
For this reason, in the rest of the paper we will omit the explicit mention
to tuples, although we will continue to use them in examples.


As usual in programming languages that adopt a static type
discipline, all expressions occurring in programs are expected to be
well-typed. Deriving or checking the types of expressions relies on
two kinds of information: Firstly, the principal types of symbols
belonging to the signature, that we assume to be attached to the
signature itself; and secondly, the types of variables occurring in
the expression. In order to represent this second kind of
information, we will use {\em type environments} $\Gamma =
\{X_1\;::\,\tau_1, \ldots, X_n\;::\,\tau_n\}$, representing the
assumption that variable $X_i$ has type $\tau_i$ for all $1 \leq i
\leq n$. Following well-known ideas stemming from the work of
Hindley, Milner and Damas
 \cite{hindley69,Mil78,damas+:types},
it is possible to define type inference rules for deriving {\em type
judgements} of the form $\Sigma,\, \Gamma \vdash_{WT} e :: \tau$
meaning that the assertion  $e :: \tau$ (in words, ``$e$ has type
$\tau$") can be deduced from the type assumptions for symbols resp.
variables given in $\Sigma$ resp. $\Gamma$. The reader is referred
to \cite{GHR01} for a presentation of type inference rules well
suited to Functional Logic Languages without constraints. Adding the
treatment of constraints would be a relatively straightforward task.
An expression $e$ is called {\em well-typed} iff there is some type
environment $\Gamma$ such that $\Sigma,\, \Gamma \vdash_{WT} e ::
\tau$ can be derived for at least one type $\tau$. Although this
$\tau$ is not unique in general, it can be proved that a {\em most
general type} $\tau$ (called the {\em principal type} of $e$ and
unique up to renaming of type variables) can be derived for any
well-typed expression $e$. In practice, principal types of
well-typed expressions can be automatically inferred by compilers.

We will write $\Sigma,\, \Gamma \vdash_{WT} \overline{e}_n :: \overline{\tau}_n$
to indicate that $\Sigma,\, \Gamma \vdash_{WT} e_i :: \tau_i$ can be derived for all $1 \leq i \leq n$,
and $\Sigma,\, \Gamma \vdash_{WT} a :: \tau :: b$
to indicate that both $\Sigma,\, \Gamma \vdash_{WT} a :: \tau$
and  $\Sigma,\, \Gamma \vdash_{WT} b :: \tau$ hold.
An expression $e$ is called {\em well-typed} iff
$\Sigma,\, \Gamma \vdash_{WT} e :: \tau$ can be derived for some type $\tau$ using the underlying signature $\Sigma$ and
some suitable type environment $\Gamma$.
Sometimes we will write simply $e :: \tau$, meaning that
$\Sigma,\, \Gamma \vdash_{WT} e :: \tau$ can be derived using the underlying $\Sigma$ and some proper choice of $\Gamma$
(which can be just $\emptyset$ if $e$ is ground).


For the sake of semantic considerations, it is useful to define an {\em information ordering}
$\sqsubseteq$ over $Exp_{\Sigma}({\mathcal{B}})$, such that $e \sqsubseteq e'$ is
intended to mean that the information provided by $e'$ is greater or equal than the
information provided by  $e$. Mathematically, $\sqsubseteq$ is defined as
the least partial ordering over $Exp_{\Sigma}({\mathcal{B}})$ such that
$\bot \sqsubseteq e'$ for all $e' \in Exp_{\Sigma}({\mathcal{B}})$ and
$(e\, e_1) \sqsubseteq (e'\, e'_1)$ whenever $e \sqsubseteq e'$ and $e_1 \sqsubseteq e'_1$.
For later use, we accept without proof the following lemma.
It is similar to the {\em Typing Monotonicity Lemma} in \cite{GHR01}
and it says that the type of any expression is
also valid for its semantic approximations.
It can be proved thanks to the fact  that
the undefined value $\bot$ belongs to all the types.\\

\begin{lemma} [Type Preservation Lemma] \label{tpl}
Assume that $\Sigma,\, \Gamma \vdash_{WT} e' :: \tau$ and $e \sqsubseteq e'$ hold.
Then $\Sigma,\, \Gamma \vdash_{WT} e :: \tau$ is also true.
\end{lemma}


As part of the definition of signatures $\Sigma$ we have required a transparency
property for the principal types of data constructors. Due to transparency, the types
of the variables occurring in a data term $t$ can be deduced from the type of $t$. It
is useful to isolate those patterns that have a similar property. To this purpose, we
adapt some definitions from \cite{GHR01}. A type which can be written as $\tpp{\tau}{m} \to \tau$
is  called {\em $m$-transparent} iff  $tvar(\tpp{\tau}{m}) \subseteq tvar(\tau)$ and {\em $m$-opaque} otherwise.
Also, defined function symbols $f$ and primitive function symbols $p$ are called
$m$-transparent iff their principal types are  $m$-transparent and $m$-opaque otherwise.
Note that a data constructor $c$ is always $m$-transparent for  all $m\, \leq\, ar(c)$.

Then,   {\em transparent patterns} can be defined as those having the syntax
$t ::= X \mid u \mid c\, \overline{t}_m \mid f\,\overline{t}_m \mid p\, \overline{t}_m$, with
$X \in \mathcal V\!\!ar$, $u \in \mathcal{B}$,
$c \in DC^n$ for some $m \leq n$,
$f \in DF^n$ for some $n > m$, and
$p \in SPF^n$ for some $n > m$,
where the subpatterns $t_{i}$ in $(c~\tpp{t}{m})$, $(f~\tpp{t}{m})$ and $(p\, \overline{t}_m)$
must be recursively  transparent, and the principal types of both the defined function symbol $f$ in  $(f~\tpp{t}{m})$
and the primitive function symbol $p$ in $(p\, \overline{t}_m)$
must be $m$-transparent.

For instance, assume a defined function symbol with principal type declaration {\tt snd :: A -> B -> B}.
Then {\tt snd} is  $1$-opaque and the  pattern {\tt (snd X)} is also opaque.
In fact, the principal type  {\tt B -> B}  of {\tt (snd X)}
 reveals no information on the type of {\tt X},
and different instances of {\tt (snd X)} keep the principal type {\tt B -> B},
independently of the type of the expression substituted for {\tt X}.
Such a behaviour is not possible for transparent patterns,
due to the {\em Transparency Lemma}  stated without proof below. Similar results were proved
in \cite{GHR01} in a slightly different context.

\begin{lemma} [Transparency Lemma] \label{trl}
\begin{enumerate}
\item
Assume a transparent pattern  $t$ and two type environments $\Gamma_{1}$, $\Gamma_{2}$
such that  $\Sigma,\, \Gamma_1 \vdash_{WT} t :: \tau$ and $\Sigma,\, \Gamma_2 \vdash_{WT} t :: \tau$,
for a common type $\tau$. 

Then $\Gamma_{1}(X) = \Gamma_{2}(X)$ holds for every $X \in var(t)$.
\item
Assume that $\Sigma,\, \Gamma \vdash_{WT} h~\tp{a}_{m}~::~\tau~::~h~\tp{b}_{m}$
holds for some $m$-transparent $h \in DC \cup DF \cup PF$ and some
common type $\tau$. 

Then, there exist types $\tau_{i}$ such that $\Sigma,\, \Gamma \vdash_{WT} a_{i}~::~\tau_{i}~::~b_{i}$
holds for all $1 \leq i \leq m$.
\end{enumerate}
\end{lemma}


 {\em Substitutions} $\sigma, \theta \in Sub_{\Sigma}({\mathcal{B}})$ over $\mathcal{B}$
 are mappings from $\var$ to $Pat_{\Sigma}({\mathcal{B}})$, extended to mappings
 from $Exp_{\Sigma}({\mathcal{B}})$ to $Exp_{\Sigma}({\mathcal{B}})$ in the natural way.
 For given $ e \in Exp_{\Sigma}({\mathcal{B}})$ and $\sigma \in Sub_{\Sigma}({\mathcal{B}})$,
 we will usually write $e \sigma$ ins\-tead of $\sigma(e)$.
 Whenever $e' = e\sigma$ for some substitution $\sigma$, we say that $e'$ is an
{\em instance} of $e$ (or also that $e$ is more general than $e'$) and we write $e \preceq e'$.

 We  write $\varepsilon$ for the identity substitution and $\sigma\theta$ for the {\em composition}
of $\sigma$ and $\theta$, such that $e(\sigma\theta) = (e\sigma)\theta$ for any expression $e$.
A substitution $\sigma$ such that $\sigma\sigma = \sigma$ is called {\em idempotent}.
The {\em domain} $vdom(\sigma)$ and the {\em variable range} $vran(\sigma)$ of a
substitution are defined as usual:
$vdom(\sigma) = \{X \in \mathcal V\!\!ar \mid X\sigma \neq X\}$ and
$vran(\sigma) = \bigcup_{X \in vdom(\sigma)}var(X\sigma)$.

A substitution $\sigma$ is called {\em finite} iff $vdom(\sigma)$ is
a finite set, and {\em ground} iff $X\sigma$ is a ground pattern for
all $X \in vdom(\sigma)$. In the sequel, we will assume that the
substitutions we work with are finite, unless otherwise said. We
adopt the usual notation $\sigma = \{X_1 \mapsto t_1, \ldots, X_n
\mapsto t_n\}$, whenever $vdom(\sigma) = \{X_1, \ldots, X_n\}$ and
$X_i\sigma = t_i$ for all $1 \leq i \leq n$. In particular,
$\varepsilon = \{\,\} = \emptyset$. We also write $\sigma[X \mapsto
t]$ for the substitution $\sigma'$ such that $X\sigma' = t$ and
$Y\sigma' = Y\sigma$ for any variable $Y \in \mathcal V\!\!ar
\setminus \{X\}$.

For any set of variables $\varx \subseteq \var$ we define the {\em restriction}
$\sigma \restrict_{\varx}$ as the substitution $\sigma'$ such that
$vdom(\sigma') = \varx$ and $\sigma'(X) = \sigma(X)$ for all $X \in \varx$.
We use the notation $\sigma =_{\varx} \theta$ to indicate that
$\sigma \restrict_{\varx} = \theta \restrict_{\varx}$, and we abbreviate
$\sigma =_{\mathcal V\!\!ar \setminus \varx} \theta$ as $\sigma =_{\setminus \varx} \theta$.

Given two substitutions $\sigma$ and $\theta$, we define the {\em application} of $\theta$ to $\sigma$
as the substitution $\sigma \star \theta =_{def} \sigma\theta \restrict vdom(\sigma)$.
In other words,  for any $X \in \var$, $X (\sigma \star \theta) = X \sigma \theta$ if $X \in vdom(\sigma)$
and $X (\sigma \star \theta) = X$ otherwise.

We consider two different ways of comparing given
substitutions $\sigma,\, \sigma' \in Sub_{\Sigma}({\mathcal{B}})$:
\begin{itemize}
\item
$\sigma$ is said to be more general than $\sigma'$ over $\varx \subseteq \mathcal V\!\!ar$
(in symbols, $\sigma \preceq_{\varx} \sigma'$) iff
$\sigma\theta =_{\varx} \sigma'$ for some $\theta \in Sub_{\Sigma}({\mathcal{B}})$.
We abbreviate $\sigma \preceq_{\mathcal V\!\!ar} \sigma'$ as $\sigma \preceq  \sigma'$
and $\sigma \preceq_{\mathcal V\!\!ar \setminus \varx} \sigma'$ as $\sigma \preceq_{\setminus \varx} \sigma'$.
\item
$\sigma$ is said to bear less information than $\sigma'$ over $\varx \subseteq \mathcal V\!\!ar$
(in symbols, $\sigma \leqinfo_{\varx} \sigma'$) iff
$\sigma(X) \leqinfo \sigma'(X)$ for all $X \in \varx$.
We abbreviate $\sigma \leqinfo_{\mathcal V\!\!ar} \sigma'$ as $\sigma \leqinfo  \sigma'$
and $\sigma \leqinfo_{\mathcal V\!\!ar \setminus \varx} \sigma'$ as $\sigma \leqinfo_{\setminus \varx} \sigma'$.
\end{itemize}


\begin{example}[Well-typed Expressions] \label{expExample}

Let us consider the specific signature $\Sigma$ and the family of base values  $\mathcal{B}$
underlying the program presented in Subsection \ref{examples}. There we find:
\begin{itemize}
\item
The sets of base values $\mathcal{B}_{int} = \mathbb{Z}$ and $\mathcal{B}_{real} = \mathbb{R}$.
\item
Well-typed expressions such as {\tt square 4 (2,3) :: bool},  {\tt RX-RY :: real}, 
{\tt (RY-RX <= RY0-RX0) :: bool}.
\item
Well-typed patterns such as {\tt 3 :: int}, {\tt 3.01 :: real}, {\tt [X,Y] :: [int]}, \\
{\tt square 4 :: dPoint -> bool}. Note that {\tt [X,Y]} abbreviates {\tt (X:(Y:[\,]))}, as usual in functional languages that use an infix list constructor.
\item
Finally, note that  {\tt $\bot$ $\sqsubseteq$ (0 : $\bot)$ $\sqsubseteq$ (0 : (1 : $\bot$)) $\sqsubseteq$ \ldots } illustrates the behaviour of the information ordering $\sqsubseteq$ when restricted to the comparison of patterns belonging to the universe $\mathcal{U}_{\Sigma}({\mathcal{B}})$.  The list patterns of type {\tt [int]} used in this example are not  allowed to occur textually in programs because of the occurrences of the undefined value $\bot$, but they are meaningful as semantic representations of partially computed lists of integers.
\end{itemize}
\end{example}
\vspace*{-.4cm}

\subsection{Domains, Constraints and Solutions} \label{dom}


Intuitively, a {\em constraint domain} provides data values and constraints oriented to some particular
application domain. Different approaches have been proposed for formalizing the notion
of constraint domain, using mathematical notions borrowed from algebra, logic and category theory;
see e.g. \cite{jaffar+:clp-popl87,Saraswat,JM94,JMM+98}. The following
definition is an elaboration of the domain notion given in  \cite{LRV07}:

\begin{definition} [Constraint Domain] \label{dcdom}
A constraint domain of specific signature $\Sigma$ (shortly,
$\Sigma$-domain) is a structure $\cdom =
\langle\mathcal{B}^{\cdom},\, \{p^{\cdom}\}_{p \in SPF}\rangle$,
where $\mathcal{B}^{\cdom} = \{\mathcal{B}_d^{\cdom}\}_{d \in
SBT}$ is a $SBT$-indexed family of sets of base values and the
{\em interpretation} $p^{\cdom}$ of each  primitive function
symbol $p :: \overline{\tau}_n \to \tau$ in $SPF^n$ is required to
be a set of $(n+1)$-tuples $p^{\cdom} \subseteq
\mathcal{U}_{\Sigma}({\mathcal{B}^{\cdom}})^{n+1}$. In the sequel,
we abbreviate $\mathcal{U}_{\Sigma}({\mathcal{B}^{\cdom}})$ as
$\mathcal{U}_{\cdom}$ (called the {\em universe of values} of $\cdom$),
and we write $p^{\cdom} \overline{t}_n \to
t$ to indicate $(\overline{t}_n,t) \in p^{\cdom}$. The intended
meaning of ``$p^{\cdom} \overline{t}_n \to t$" is that the
primitive function $p^{\cdom}$ with given arguments
$\overline{t}_n$ can return a result $t$. Moreover,
the interpretations of primitive symbols are required to satisfy four
conditions:
\begin{enumerate}
\item {\bf Polarity}: For all $p \in SPF$, ``$p^{\cdom} \overline{t}_n \to t$" behaves
monotonically w.r.t. the arguments $\overline{t}_n$ and antimonotonically
w.r.t. the result $t$.

Formally: For all $\overline{t}_n, \overline{t'}_n, t, t' \in  \mathcal{U}_{\cdom}$ such that
$p^{\cdom} \overline{t}_n \to t$, $\overline{t}_n \sqsubseteq \overline{t'}_n$ and $t \sqsupseteq t'$,
$p^{\cdom} \overline{t'}_n \to t'$ also holds.
\item {\bf Radicality}: For all $p \in SPF$, as soon as the arguments given to $p^{\cdom}$ have enough information to return  a result other than $\bot$, the same arguments suffice already  for returning a total result.

Formally:
For all $\overline{t}_n, t \in \mathcal{U}_{\cdom}$, if $p^{\cdom} \overline{t}_n \to t$ then $t=\bot$ or else
there is some total $t' \in \mathcal{U}_{\cdom}$ such that $p^{\cdom} \overline{t}_n \to t'$ and $t' \sqsupseteq t$.
\item {\bf Well-typedness}: For all $p \in SPF$, the behaviour of $p^{\cdom}$ is
well-typed w.r.t. any monomorphic instance of $p$'s principal type.

Formally:
For any monomorphic type instance
$(\overline{\tau'}_n \to \tau') \succeq (\overline{\tau}_n \to \tau)$
and for all  $\overline{t}_n, t \in \mathcal{U}_{\cdom}$
such that $\Sigma\, \vdash_{WT} \overline{t}_n :: \overline{\tau'}_n$
and $p^{\cdom} \overline{t}_n \to t$, the type judgement
$\Sigma\, \vdash_{WT} t :: \tau'$ also holds.
\item {\bf Strict Equality}: The primitive {\tt ==} (in case that it belongs to $SPF$)
is interpreted as {\em strict equality} over $\mathcal{U}_{\cdom}$,
so that for all  $t_1, t_2, t \in \mathcal{U}_{\cdom}$,
one has $t_1${\tt ==}$^{\cdom} t_2 \to t$ iff some of the three following cases holds:
\begin{enumerate}
\item[(a)]
$t_1$ and $t_2$ are one and  the same total pattern,
and $true \geqinfo t$.
\item[(b)]
$t_1$ and $t_2$ have no common upper bound in  $\mathcal{U}_{\cdom}$ w.r.t. the information ordering $\leqinfo$,
and $false \geqinfo t$.
\item[(c)]
$t = \bot$.
\end{enumerate}
With this definition, it is easy to check that  {\tt ==}$^{\cdom}$ satisfies the polarity, radicality and well-typedness conditions.
\end{enumerate}
\end{definition}

In Subsection \ref{pdom} we will introduce the notion of solver,
 and we will see that the three domains $\herbrand$, $\rdom$ and $\fd$ mentioned in the
introduction can be formalized according to the previous definition.
In the rest of this subsection we discuss how  to work with constraints over a given domain.


For any given domain $\cdom$ of signature $\Sigma$, the set
$\mathcal{U}_{\cdom} = \mathcal{U}_{\Sigma}({\mathcal{B}^{\cdom}}) = GPat_{\Sigma}({\mathcal{B}}^{\cdom})$ is called the {\em universe of values} of the domain $\cdom$.
We will also  write $Exp_{\cdom}$, $Pat_{\cdom}$ and $Sub_{\cdom}$
in place of $Exp_{\Sigma}({\mathcal{B}}^{\cdom})$, $Pat_{\Sigma}({\mathcal{B}}^{\cdom})$ and $Sub_{\Sigma}({\mathcal{B}}^{\cdom})$, respectively.
Note that requirement 4. in Definition \ref{dcdom} imposes a fixed interpretation of {\tt ==} as the
strict equality operation {\tt ==}$^{\cdom}$ over  $\mathcal{U}_{\cdom}$, for every domain $\cdom$
whose  specific signature includes this primitive. It is easy to check that  the polarity, radicality and
well-typedness requirements are satisfied by strict equality. The following definition will be useful:


\begin{definition} [Conservative Extension of a given Domain]\label{defCextension}
Given two domains $\cdom$, $\cdom'$ with respective signatures $\Sigma$ and $\Sigma'$,
$\cdom'$ is called a {\em conservative extension} of $\cdom$ iff the following conditions hold:
\begin{enumerate}
\item
$\Sigma \subseteq \Sigma'$, i.e. $SBT \subseteq SBT'$ and $SPF \subseteq SPF'$.
\item
For all $d \in SBT$, one has $\mathcal{B}_d^{\cdom'} = \mathcal{B}_d^{\cdom}$.
\item
For all $p \in SPF^n$ other than {\tt ==}\,Êand for every $\overline{t}_n, t \in  \mathcal{U}_{\cdom}$,
one has $p^{\cdom'}\, \overline{t}_n \to t$ iff $p^{\cdom}\, \overline{t}_n \to t$.
\end{enumerate}
\end{definition}


\noindent As usual in constraint programming, we define {\em
constraints} over a given domain $\cdom$ as logical formulas built
from atomic constraints by means of  conjunction $\wedge$ and
existential quantification $\exists$. More precisely,  constraints
$\delta \in Con_{\cdom}$ over the constraint domain $\cdom$ have the
syntax $\delta ::= \alpha \mid (\delta_1 \wedge \delta_2) \mid
\exists X \delta$, where $\alpha$ is any atomic constraint over
$\cdom$ and $X \in \var$ is any variable. We allow two kinds of {\em
atomic constraints} $\alpha$ over $\cdom$: a)  $\lozenge$ and
$\blacklozenge$, standing for truth (success) and falsity (failure),
respectively; and b) atomic constraints of the form $p\,
\overline{e}_n\, \to!\, t$ with $p \in SPF^{n}$, where
$\overline{e}_n \in Exp_{\cdom}$, $t \in Pat_{\cdom}$, and $t$ is
required to be total (i.e., without any occurrences of $\bot$).
 The intended meaning of  $p\, \overline{e}_n\, \to!\, t$ constrains the value returned by the call
 $p\, \overline{e}_n$ to be a total pattern matching the form of $t$.


 By convention, constraints of the form $p\, \overline{e}_n \to!\, true$
are abbreviated as $p\, \overline{e}_n$. Sometimes constraints of
the form $p\, \overline{e}_n \to!\, false$ are abbreviated as $p'\,
\overline{e}_n$, using some symbol $p'$
 to suggest the `negation' of $p$. In particular, {\em strict equality
constraints} $e_1$ {\tt ==} $e_2$ and  {\em strict disequality
constraints} $e_1$ {\tt /=} $e_2$ are understood as abbreviations
of $e_1$ {\tt ==} $e_2 \to!\, true$ and $e_1$ {\tt ==} $e_2\,
\to!\, false$, respectively. The next definition introduces some useful
notations for different kinds of constraints.


\begin{definition} [Notations for various kinds of constraints]\label{defSpecial}
Given two domains $\cdom$, $\cdom'$ with respective signatures $\Sigma$ and $\Sigma'$,
such that $\cdom'$ is a {\em conservative extension} of $\cdom$. Let $SPF \subseteq SPF'$ be
the sets of specific primitive function symbols of $\cdom$ and $\cdom'$, respectively.
We define:
\begin{enumerate}
\item
$ACon_{\cdom} \subseteq Con_{\cdom}$ is the set of all {\em atomic constraints} over $\cdom$.
\item
$APCon_{\cdom} \subseteq ACon_{\cdom}$ is the set of all {\em atomic primitive constraints} over $\cdom$.
By definition, $\alpha \in APCon_{\cdom}$ iff $\alpha$ has the form $\lozenge$, $\blacklozenge$ or
$p\, \overline{t}_n\, \to!\, t$, where $\overline{t}_n \in Pat_{\cdom}$ are patterns.
\item
$PCon_{\cdom} \subseteq Con_{\cdom}$ is the set of all {\em primitive constraints} $\pi$ over $\cdom$.
By definition, a constraint $\pi \in Con_{\cdom}$ is called primitive iff all the atomic parts of $\pi$ are primitive.
Note that $APCon_{\cdom} =  ACon_{\cdom}\, \cap\, PCon_{\cdom}$.
\item
$Con_{\cdom'} \upharpoonright SPF$ is the set of all $SPF$-{\em
restricted} constraints over $\cdom'$. By defi\-nition, a constraint
$\delta  \in Con_{\cdom'}$ is called $SPF$-restricted iff all the
atomic parts of $\delta$ have the form $\lozenge$, $\blacklozenge$
or $p\, \overline{e}_n\, \to!\, t$, where $p \in SPF^n$. The subsets
$APCon_{\cdom'} \upharpoonright SPF \subseteq ACon_{\cdom'}
\upharpoonright SPF \subseteq Con_{\cdom'} \upharpoonright SPF$ are
defined in the natu\-ral way. In particular, $APCon_{\cdom'}
\upharpoonright SPF$ is the set of all the $SPF$-restricted atomic
primitive constraints over $\cdom'$, which have the form $\lozenge$
or $\blacklozenge$ or $p\, \overline{t}_n\, \to!\, t$, with $p \in
SPF^{n}$, $\overline{t}_n,\, t \in Pat_{\cdom'}$ and $t$  total.
\end{enumerate}
\end{definition}


A particular occurrence of a  variable $X$ within a constraint
$\delta$  is called {\em free} iff it is not affected by any
quantification, and {\em bound} otherwise. In the sequel, we will
write $var(\delta)$ (resp. $fvar(\delta)$) for the set of all
variables having some occurrence (resp. free occurrence) in the
constraint $\delta$. The notations $var(\Delta)$ and $fvar(\Delta)$
for a set of constraints  $\Delta \subseteq Con_{\cdom}$ have a
similar meaning.


The type inference rules mentioned in Subsection \ref{expressions} can be naturally extended to
derive also type judgments of the form  $\Sigma,\, \Gamma \vdash_{WT} \delta$, meaning that the
constraint $\delta$ is well-typed w.r.t.  the type assumptions for symbols resp.
variables given in $\Sigma$ resp. $\Gamma$. Sometimes
we will simply claim that  $\delta$ is {\em well-typed} to indicate that $\Sigma,\, \Gamma \vdash_{WT}
\delta$ can be derived using the underlying signature $\Sigma$ and
some suitable type environment $\Gamma$ (which can be just $\emptyset$ if $\delta$ has no free variables).


The set of {\em valuations} $Val_{\cdom}$ over the domain $\cdom$ consists of all ground substitutions
$\eta$ such that $vran(\eta) \subseteq \mathcal{U}_{\cdom}$. Those valuations which satisfy a given constraint are called \emph{solutions}. For those constraints $\delta$ that include subexpressions of the form
$f\, \overline{e}_n$ for some $f \in DF^n$,  the solutions of $\delta$ depend on the behaviour of $f$, which is not included in the domain $\cdom$, but must be deduced from some user given program, as we will see in Section \ref{cooperative}. However, the solutions of primitive constraints  depend only on the domain $\cdom$. More precisely:

\begin{definition} [Solutions of Primitive Constraints]\label{defPrimSol}
\begin{enumerate}
\item
The {\em set of solutions} of a primitive constraint
$\pi \in PCon_{\cdom}$ is a subset $Sol_{\cdom}(\pi) \subseteq Val_{\cdom}$ defined by recursion on the syntactic structure
of $\pi$ as follows:
\begin{itemize}
\item
$Sol_{\cdom}(\lozenge) = Val_{\cdom}$; $Sol_{\cdom}(\blacklozenge) = \emptyset$.
\item
$Sol_{\cdom}(p\, \overline{t}_n \to!\, t) = \{\eta \in Val_{\cdom}
\mid (p\, \overline{t}_n\, \to! t)\eta\, \textnormal{ground},\,
p^{\cdom} \overline{t}_n\eta \to t\eta,\,  t\eta\,
\textnormal{total}\}$.
\item
$Sol_{\cdom}(\pi_1 \wedge \pi_2) = Sol_{\cdom}(\pi_1) \cap Sol_{\cdom}(\pi_2)$.
\item
$Sol_{\cdom}(\exists X \pi) = \{\eta \in Val_{\cdom} \mid
\textnormal{exists}\, \eta' \in Sol_{\cdom}(\pi)\, \textnormal{s.t.}\, \eta'  =_{\setminus \{X\}} \eta\}.$
\end{itemize}
\item
Any set $\Pi \subseteq PCon_{\cdom}$ is interpreted  as a conjunction, and therefore
$Sol_{\cdom}(\Pi) = \bigcap_{\pi \in \Pi} Sol_{\cdom}(\pi)$.
\item
The {\em set of well-typed solutions} of a primitive constraint
$\pi \in PCon_{\cdom}$ is a subset $WTSol_{\cdom}(\pi) \subseteq Sol_{\cdom}(\pi)$ consisting
of all $\eta \in Sol_{\cdom}(\pi)$ such that $\pi \eta$ is well-typed.
\item
Finally, for any $\Pi \subseteq PCon_{\cdom}$ we define
$WTSol_{\cdom}(\Pi) = \bigcap_{\pi \in \Pi} WTSol_{\cdom}(\pi)$.
\end{enumerate}
\end{definition}


Note that any solution $\eta \in Sol_{\cdom}(\pi)$ must verify $vdom(\eta) \supseteq fvar(\pi)$.
For later use, we accept the two following technical lemmata.
The first one can be easily proved by induction on the syntactic structure of $\Pi$
and the second one is a simple consequence of the polarity properties of primitive functions.
The notation $(WT)Sol$ used in both lemmata is intended to indicate that they are valid both
for plain solutions and for well-typed solutions.

\begin{lemma} [Substitution Lemma] \label{sl}
For any given $\Pi \subseteq PCon_{\cdom}$, $\sigma \in Sub_{\cdom}$ and $\eta \in Val_{\cdom}$,
the equivalence $\eta \in (WT)Sol_{\cdom}(\Pi\sigma) \Leftrightarrow 
\sigma\eta \in (WT)Sol_{\cdom}(\Pi)$ is valid.
\end{lemma}

\begin{lemma} [Monotonicity Lemma] \label{ml}
For any given $\Pi \subseteq PCon_{\cdom}$ and $\eta, \eta' \in Val_{\cdom}$
such that $\eta \sqsubseteq \eta'$ and $\eta \in (WT)Sol_{\cdom}(\Pi)$, one also has 
$\eta' \in (WT)Sol_{\cdom}(\Pi)$.
\end{lemma}


A given solution $\eta \in Sol_{\cdom}(\Pi)$ can bind some variables $X$ to the undefined
value $\bot$. Intuitively, this will happen whenever the value of $X$ is not needed for checking the
satisfaction of the constraints in $\Pi$. Formally, a variable $X$ is  {\em demanded} by a set of constraints $\Pi \subseteq PCon_{\cdom}$ iff $\eta(X) \neq \bot$ for all $\eta \in \sol{\cdom}{\Pi}$.
We write $dvar_{\cdom}(\Pi)$ to denote the set of all $X \in fvar(\Pi)$ such that $X$ is demanded by 
$\Pi$.

In practice, $CFLP$ programming requires effective procedures for recognizing
`obvious' occurrences of demanded variables in the case that $\Pi$
is a set of atomic primitive constraints. We assume that for any
practical constraint domain $\cdom$ and any primitive atomic
constraint $\pi \in APCon_{\cdom}$ there is an effective way of
computing a subset $odvar_{\cdom}(\pi) \subseteq
dvar_{\cdom}(\pi)$. Variables $X \in odvar_{\cdom}(\pi)$ will be
said to be  {\em obviously demanded} by $\pi$. We extend the
notion to finite constraint sets $\Pi \subseteq  APCon_{\cdom}$ by
defining the set $odvar_{\cdom}(\Pi)$ of all variables {\em
obviously demanded}  by $\Pi$ as $\bigcup_{\pi \in \Pi}
odvar_{\cdom}(\pi)$. In this way, it is clear that
$odvar_{\cdom}(\Pi) \subseteq dvar_{\cdom}(\Pi)$ holds for any
$\Pi \subseteq APCon_{\cdom}$; i.e., obviously demanded variables
are always demanded. The inclusion is strict in general.


In particular, for any constraint domain $\cdom$ whose specific
signature includes the strict equality primitive {\tt  ==} and any
primitive atomic constraint of the form  $\pi$  =  ($t_1${\tt
==}$t_2 \to!\, t)$, $odvar_{\cdom}(\pi)$ is defined by a case
distinction, as follows:
\begin{itemize}
\item
$odvar_{\cdom}(t_1${\tt ==}$t_2 \to!\, R) = \{R\}$, if $R \in \var$.
\item
$odvar_{\cdom}(X${\tt ==}$Y) = \{X,Y\}$, if $X,Y \in \var$.
\item
$odvar_{\cdom}(X${\tt ==}$t) =  odvar_{\cdom}(t${\tt ==}$X) =\{X\}$, if $X \in \var$ and $t \notin \var$.
\item
$odvar_{\cdom}(t_1${\tt ==}$t_2) = \emptyset$, otherwise.
\item
$odvar_{\cdom}(X${\tt /=}$Y) = \{X,Y\}$, if $X,Y \in \var$, $X$ and $Y$ not identical.
\item
$odvar_{\cdom}(X${\tt /=}$t) =  odvar_{\cdom}(t${\tt /=}$X) =\{X\}$, if $X \in \var$ and $t \notin \var$.
\item
$odvar_{\cdom}(t_1${\tt /=}$t_2) = \emptyset$, otherwise.
\end{itemize}

The inclusion $odvar_{\cdom}(\pi) \subseteq dvar_{\cdom}(\pi)$
is easy to check, by considering the behaviour of the interpreted strict equality
operation {\tt ==}$^{\cdom}$.
The method for computing $odvar_{\cdom}(\pi)$ for atomic primitive
constraints based on primitive functions other than equality must be
given as part of a practical presentation of the corresponding
domain $\cdom$. In the sequel, we will call {\em critical} to those
variables occurring in $\Pi$ which are not obviously demanded, and
we will write $cvar_{\cdom}(\Pi) = var(\Pi) \setminus
odvar_{\cdom}(\Pi)$ for the set of all critical variables. As we
will see in Section \ref{cooperative}, goal solving methods for
$CFLP$ programming rely on the effective recognition of critical
variables. Therefore,  the proper behaviour of goal solving depends
on well-defined methods for the computation of obviously demanded
variables.


In the rest of the paper we will often use {\em constraint stores} of the form $S = \Pi\, \Box\, \sigma$,
where $\Pi \subseteq APCon_{\cdom}$ and $\sigma$ is an idempotent substitution such that
$vdom(\sigma) \cap var(\Pi) = \emptyset$. We will need to work with solutions of
constraint stores, possibly affected by an existential prefix. This notion is defined as follows:

\begin{definition} [Solutions of Constraint Stores]\label{defStoSol}
\begin{enumerate}
\item
$Sol_{\cdom}(\exists \overline{Y} (\Pi\, \Box\, \sigma)) =
\{\eta \in Val_{\cdom} \mid  \textnormal{exists}\, \eta' \in Sol_{\cdom}(\Pi\, \Box\, \sigma),\,
\textnormal{s.t.}\, \eta'  =_{\setminus \overline{Y}} \eta\}$.
\item
$Sol_{\cdom}(\Pi\, \Box\, \sigma) = Sol_{\cdom}(\Pi) \cap Sol(\sigma)$.
\item
$Sol(\sigma) = \{\eta \in Val_{\cdom} \mid \eta = \sigma\eta\}$ \\
(Note that $\eta = \sigma\eta$ holds iff $X\eta = X\sigma\eta\, \, \textnormal{for all}\, X \in vdom(\sigma)$).
\item
$WTSol_{\cdom}(\exists \overline{Y} (\Pi\, \Box\, \sigma)) =
\{\eta \in Val_{\cdom} \mid  \textnormal{ex.}\, \eta' \in WTSol_{\cdom}(\Pi\, \Box\, \sigma),\,
\textnormal{s.t.}\, \eta'  =_{\setminus \overline{Y}} \eta\}$.
\item
$WTSol_{\cdom}(\Pi\, \Box\, \sigma)
= \{\eta \in Sol_{\cdom}(\Pi\, \Box\, \sigma) \mid (\Pi\, \Box\, \sigma)\star\eta\, \textnormal{is well-typed}\}$,
where $(\Pi\, \Box\, \sigma)\star\eta =_{def} \Pi\eta\, \Box\, (\sigma\star\eta)$.
\end{enumerate}
\end{definition}

\begin{example}[Constraints and Their Solutions] \label{pConExample}
Let us now illustrate different notions concerning constraints by
referring again to the motivating example from Subsection
\ref{examples}. The domain $\ccdom$ underlying this example is a
`hybrid' domain supporting the cooperation of three `pure' domains
named $\herbrand$, $\rdom$ and $\fd$, as we will see in
Subsections \ref{pdom} and \ref{cdomains}. For the moment, note
that $\ccdom$ allows to work with four different kinds of
constraints, namely bridge constraints and the specific
constraints supplied by $\herbrand$, $\rdom$ and $\fd$, as
explained in Section \ref{introduction}.
\begin{enumerate}
\item
Concerning well-typed constraints, we note that the small program
in this example is well-typed. Therefore, all the constraints
occurring there are also well-typed. For instance:
\begin{itemize}
\item
{\tt domain [X,Y] 0 N} is well-typed (w.r.t. any type environment
which includes the type assumptions {\tt X :: int, Y :: int, N ::
int}).
\item
\begin{sloppypar}
{\tt RY+RX <= RY0+RX0} is also well-typed (w.r.t. any type
environment which includes the type assumptions {\tt RY :: real,
RX :: real, RY0 :: real, RX0 :: real}).
\end{sloppypar}
\end{itemize}
Of course, the signature underlying the example allows to write
constraints such as {\tt domain [X,Y] true 3.2}, which cannot be
well-typed in any type environment. Due to static type discipline,
the compiler will reject programs including ill-typed constraints.
\item
Concerning constraint solutions, note that computing by means of the cooperative goal solving calculus
presented in Section \ref{cooperative} eventually triggers the computation of solutions for primitive constraints.
As already discussed in Subsection \ref{examples}, solving  {\bf Goal  2} eventually leads to
the following set $\Pi$ of primitive constraints (understood as logical conjunction):
\vspace*{.1cm}

\hspace*{1.5cm}
{\tt X \#== RX, Y \#== RY,}

\hspace*{1.5cm}
{\tt RY >= d-0.5, RY-RX <= 0.5, RY+RX <= n+0.5,}

\hspace*{1.5cm}
{\tt domain [X,Y] 0 n, labeling [\,] [X,Y].}
\vspace*{.1cm}

$\Pi$ happens to be the union of three sets of primitive constraints corresponding to the three lines above:
A set of two {\em bridge constraints} $\Pi_M$, a set of three {\em real arithmetical constraints} $\Pi_R$,
and a set of two {\em finite domain constraints} $\Pi_F$. Therefore,
$Sol_{\ccdom}(\Pi) = Sol_{\ccdom}(\Pi_M) \cap Sol_{\ccdom}(\Pi_R) \cap Sol_{\ccdom}(\Pi_F)$.
As we have seen in Subsection \ref{examples}, the only possibility for $\eta \in \sol{\ccdom}{\Pi}$
is {\tt $\eta$(X) = $\eta$(Y) = d}, and the computation proceeds with the help of  {\em constraint solvers}
and {\em projections}, among other mechanisms.
\item
Concerning obviously demanded variables, let us remark that all the variables occurring in the constraint set $\Pi$ shown in the previous item are obviously demanded. This will become clear
from the discussion of the domains $\herbrand$, $\rdom$ and $\fd$ in Subsection \ref{pdom}.
\item
Concerning critical variables, note that a variable may be critical either because it is demanded but not obviously demanded, or else because it is not demanded at all. For instance, variables {\tt A} and {\tt B} are demanded but not obviously demanded by the strict equality constraint {\tt (A,2) == (1,B)}. Therefore, they are critical variables. To illustrate the case of  critical but not demanded variables,
consider the  primitive  constraint $\pi$ = {\tt L /= X:Xs}.
Due to the definition of `obvious demand' for strict disequality constraints,
variable {\tt L}  is obviously demanded by $\pi$, while  {\tt X} and {\tt Xs} are not obviously demanded, and therefore critical.
Moreover, it can be argued that neither  {\tt X} nor {\tt Xs} is demanded by $\pi$.
Variable {\tt X} is not demanded because there  exist solutions $\eta \in \sol{\cdom}{\pi}$ such that
$\eta({\tt X}) = \bot$ (either with $\eta({\tt L}) = [\,]$ or else with $\eta({\tt L}) = {\tt t:ts}$ such that
$\eta({\tt Xs})$ is different from {\tt ts}).
Variable {\tt Xs} is not demanded because of similar reasons.
\end{enumerate}
\end{example}
\vspace*{-.3cm}

\subsection{Pure Domains and  their Solvers}\label{pdom}

In order to be helpful for programming purposes, constraint domains must
provide so-called  {\em constraint solvers},  which process the
constraints arising in the course of a computation. For some theoretical purposes, it suffices to model
a solver as a function which maps any given constraint to one of the three different values
$true$, $false$  or $unknown$; see e.g. \cite{JMM+98}.
In practice, however, solvers are expected to have the ability  of reducing
primitive constraints to so-called {\em solved forms}, which are simpler and can be shown as
computed answers to the users. As discussed in the introduction
(see in particular Subsection \ref{examples}),
the constraint domain underlying many  practical problems may involve heterogeneous primitives
related to different base types. In such cases, it may be not realistic to expect that a single solver
for the whole domain is directly available.


In the sequel, we will make a pragmatic distinction between {\em
pure constraint domains} which are given `in one piece' and come
equipped with a solver, and {\em hybrid constraint domains} which
are built as a combination of simpler domains and must rely on the
solvers of their components. In the rest of this subsection we give
a mathe\-matical formalization of the notion of solver tailored to
the needs of the $CFLP$ scheme, followed by a presentation of
$\herbrand$, $\rdom$  and $\fd$ as pure domains equipped with
solvers. In the case of $\rdom$  and $\fd$, we limit ourselves to
describe their most basic primitives, although other useful
facilities are available in the $\toy$ implementation. A proposal
for the construction of so-called {\em coordination domains} as a
particular kind of hybrid domains will be presented in Subsection
\ref{cdomains}. \vspace*{-.3cm}

\subsubsection{Constraint Solvers} \label{csolvers}


For any pure constraint domain $\cdom$ we postulate a {\em constraint solver} which can reduce any given finite set $\Pi$ of atomic primitive constraints to an equivalent simpler form, while taking proper  care of critical variables occurring in $\Pi$.
Since the value of a critical variable $X$ may be needed by some solutions of $\Pi$ and irrelevant
for some other solutions, we require that solvers have the ability to compute a distinction of cases
discriminating such situations.


\begin{definition} [Formal Requirements for Solvers]\label{defSolver}
A  constraint solver for the domain $\cdom$  is modeled as
a function $solve^{\cdom}$ which can be applied to pairs of the form $(\Pi,\varx)$,
where $\Pi \subseteq APCon_{\cdom}$ is a finite set of atomic primitive constraints and
$\varx \subseteq cvar_{\cdom}(\Pi)$ is a finite set including some of the critical variables in $\Pi$,
where the two extreme cases $\varx = \emptyset$ and $\varx = cvar_{\cdom}(\Pi)$ are allowed.
By convention, we may abbreviate $solve^{\cdom}(\Pi,\emptyset)$ as $solve^{\cdom}(\Pi)$.
We require that any {\em solver invocation} $solve^{\cdom}(\Pi,\varx)$ returns a
finite disjunction $\bigvee_{j=1}^{k}\exists \overline{Y}_j (\Pi_j\, \Box\, \sigma_j)$
of existentially quantified constraint stores, fulfilling the following conditions:
\begin{enumerate}
\item
{\bf Fresh Local Variables:}
For all $1 \leq j \leq k$: $(\Pi_j\, \Box\, \sigma_j)$ is a store, $\overline{Y}_j = var(\Pi_j\, \Box\, \sigma_j)\setminus var(\Pi)$ are fresh local variables and $vdom(\sigma_j) \cup vran(\sigma_j) \subseteq var(\Pi) \cup  \overline{Y}_j$.
\item
{\bf Solved Forms:}
For all $1 \leq j \leq k$: $\Pi_j\, \Box\, \sigma_j$ is in solved form w.r.t. $\varx$.
By definition, this means that
$solve^{\cdom}(\Pi_j,\varx) = \Pi_j\, \Box\,\varepsilon$.
\item
{\bf Safe Bindings:}
For all $1 \leq j \leq k$ and for all $X \in \varx \cap vdom(\sigma_j)$: $\sigma_j(X)$ is a constant.
\item
{\bf Discrimination:}
Each computed $\varx$-solved form
$\Pi_j\, \Box\, \sigma_j$ $(1 \leq j \leq k)$ must satisfy:
Either $\varx\, \cap\, odvar_{\cdom}(\Pi_j) \neq \emptyset$
or else $\varx \cap var(\Pi_j) = \emptyset$
(i.e., either some critical variable in $\varx$ becomes obviously demanded,
or else all critical variables in $\varx$  disappear).
\item
{\bf Soundness:}
$Sol_{\cdom}(\Pi) \supseteq \bigcup_{j=1}^{k} Sol_{\cdom}(\exists\overline{Y}_j (\Pi_j\, \Box\, \sigma_j))$.
\item
{\bf Completeness:}
$WTSol_{\cdom}(\Pi) \subseteq \bigcup_{j=1}^{k} WTSol_{\cdom}(\exists\overline{Y}_j (\Pi_j\, \Box\, \sigma_j))$.
\end{enumerate}
Moreover, $solve^{\cdom}$  is called an  {\em extensible solver} iff
the solver invocation $solve^{\cdom}(\Pi,\varx)$ is defined
and satisfies the conditions listed in this definition
not just for $\Pi \subseteq APCon_{\cdom}$ and $\varx \subseteq cvar_{\cdom}(\Pi)$,
 but more  generally for $\Pi \subseteq APCon_{\cdom'} \upharpoonright SPF$
 and $\varx \subseteq cvar_{\cdom'}(\Pi)$, where $\cdom'$ is any conservative extension of $\cdom$.
The idea is that an extensible solver can deal with constraints involving
the primitives in $\cdom$ and values described by
patterns over arbitrary conservative extensions of $\cdom$.
\end{definition}


The presentation of goal solving in Section \ref{cooperative} will discuss the proper way of choosing a set $\varx$ of critical variables for each particular solver invocation. The idea is that $\varx$ should include all critical variables which are waiting to be bound to the result of evaluating some expression  at some other place within the goal. This idea also motivates the {\em safe bindings} condition.

Operationally, the alternatives within the disjunctions returned by solver invocations are usually explored in some sequential order with the help of a  backtracking mechanism.
Assuming that $solve^{\cdom}(\Pi,\varx) = \bigvee_{j=1}^{k}\exists \overline{Y}_j (\Pi_j\, \Box\, \sigma_j)$,
we will sometimes use the following notations:
 \begin{itemize}
\item
$\Pi \vdash\!\!\vdash_{solve^{\cdom}_{\varx}} \exists \overline{Y'} (\Pi'\, \Box\, \sigma')$
to indicate that $\exists \overline{Y'} (\Pi'\, \Box\, \sigma')$ is
$\exists \overline{Y}_j (\Pi_j\, \Box\, \sigma_j)$ for some $1 \leq j \leq k$.
In this case we will speak of a {\em successful solver invocation}.
\item
$\Pi \vdash\!\!\vdash_{solve^{\cdom}_{\varx}} \blacksquare$
to indicate that $k = 0$.  In this case we will speak of a {\em failed solver invocation}, yielding the obviously unsatisfiable store $\blacksquare = \blacklozenge\, \Box\, \varepsilon$.
\end{itemize}


As defined above, a constraint store $\Pi\, \Box\, \sigma$ is said
to be in {\em solved form} w.r.t. a set of critical variables
$\varx$ (or simply in solved form if $\varx = \emptyset$) iff
$solve^{\cdom}(\Pi,\varx) =\Pi\, \Box\,\varepsilon$. In practice,
solved forms can be recognized by syntactical criteria, and a
solver invocation $solve^{\cdom}(\Pi,\varx)$ is performed only in
the case that $\Pi\, \Box\, \sigma$ is not yet solved w.r.t.
$\varx$. Whenever a solver is invoked, the {\em soundness} condition requires that
no new spurious solution (whether well-typed or not) is introduced,
while the {\em completeness} condition requires that no
{\em well-typed} solution is lost.
In practice, any solver can be  expected to be sound, but  completeness
may hold only for some choices of the constraint set $\Pi$ to be solved.
Demanding completeness for arbitrary (rather than well-typed) solutions
would be still less realistic.
The solvers of interest for this paper suffer some limitations regarding completeness,  as explained
in Subsections \ref{hdom},  \ref{rdom} and  \ref{fdom} below.


From a user's viewpoint, a solver can behave as a {\em black-box} or
as a {\em glass-box}. Black-box solvers can just be invoked to
compute disjunctions of solved forms, but users cannot observe their
inner workings, in contrast to the case of glass-box solvers. Users
can define glass-box solvers by means of appropriate tools, such as
{\em Constraint Handling Rules} \cite{fruehwirth98:chr}. In this
paper we propose to use {\em store transformation systems} as a
convenient abstract technique for specifying the behaviour of
glass-box solvers. A store transformation system (briefly $sts$)
over the constraint domain $\cdom$ is specified as a set of store
transformation rules (briefly $str$s) {\bf RL} that describe
different ways to transform a given store $\Pi\, \Box\, \sigma$
w.r.t. a
given set $\varx$ of critical variables. The notions and notations
defined below are useful for working with $sts$s. Some of them refer
to a selected set of $str$s noted as $\mathcal{RS}$.
\begin{itemize}
\item
$\Pi\, \Box\, \sigma \vdash\!\!\vdash_{\cdom,\, \varx}\, \Pi'\, \Box\, \sigma'$
indicates that the store $\Pi\, \Box\, \sigma$ can be transformed into $\Pi'\, \Box\, \sigma'$ in one step,
using one of the available $str$s. This notation can be also used to indicate a failing
transformation step, writing the inconsistent store $\blacksquare = \blacklozenge\, \Box\, \varepsilon$
in place of $\Pi'\, \Box\, \sigma'$.
\item
$\Pi\, \Box\, \sigma \vdash\!\!\vdash^{*}_{\cdom,\, \varx}\, \Pi'\, \Box\, \sigma'$
indicates that $\Pi\, \Box\, \sigma$ can be transformed into $\Pi'\, \Box\, \sigma'$ in finitely many steps.
\item
The store $\Pi\, \Box\, \sigma$ is called $\mathcal{RS}$-{\em irreducible} iff there is no $str$
${\bf RL} \in \mathcal{RS}$ that can be applied to transform $\Pi\, \Box\, \sigma$.
Note that  this is trivially true if $\mathcal{RS}$ is the empty set.
If $\mathcal{RS}$ is the set of all the available $str$s, the
store $\Pi\, \Box\, \sigma$ is called simply irreducible
(or also a $\varx$-solved form).
\item
$\Pi\, \Box\, \sigma \vdash\!\!\vdash^{*}_{\cdom,\, \varx}!\, \Pi'\, \Box\, \sigma'$
indicates that $\Pi\, \Box\, \sigma \vdash\!\!\vdash^{*}_{\cdom,\, \varx}\, \Pi'\, \Box\, \sigma'$
holds, and moreover, the final store $\Pi'\, \Box\, \sigma'$ is irreducible.
\end{itemize}


Assume a given $sts$ over $\cdom$ such that
for any finite $\Pi \subseteq APCon_{\cdom}$ and any $\varx \subseteq cvar_{\cdom}(\Pi)$,
the set $\mathcal{SF}_{\cdom}(\Pi,\varx) = \{\Pi'\, \Box\, \sigma' \mid \Pi\, \Box\, \varepsilon
\vdash\!\!\vdash^{*}_{\cdom,\, \varx}!\,\, \Pi'\, \Box\, \sigma'\}$ is finite.
Then the  solver defined by the $sts$ can be specified to behave as follows:
$$solve^{\cdom}(\Pi,\varx)=\bigvee \{\exists \overline{Y'}(\Pi' \Box\sigma')\mid\Pi' \Box \sigma'
\in\mathcal{SF}_{\cdom}(\Pi,\varx),\overline{Y'}=var(\Pi' \Box \sigma')\setminus var(\Pi)\}$$
Once $solve^{\cdom}$ has been so defined, the notation
$\Pi \vdash\!\!\vdash_{solve^{\cdom}_{\varx}} \exists \overline{Y'} (\Pi'\, \Box\, \sigma')$
actually happens to mean that
$\Pi\, \Box\, \varepsilon$  $ \vdash\!\!\vdash^{*}_{\cdom,\, \varx}!\,\, \Pi'\, \Box\, \sigma'$ and
$\overline{Y'} = var(\Pi'\, \Box\, \sigma') \setminus var(\Pi)$.
Therefore,  the symbols $\vdash\!\!\vdash_{solve^{\cdom}_{\varx}}$
and $\vdash\!\!\vdash^{*}_{\cdom,\, \varx}!$ should not be confused, but have related meanings.
 The following definition specifies different properties of store transformation systems that are useful
 to check that the corresponding solvers satisfy  the conditions stated in Definition \ref{defSolver}.

 \begin{definition} [Properties of Store Transformation Systems]\label{defpsts}
 Assume a store transformation system over $\cdom$ whose transition relation is
 $\sts{\cdom}{\varx}$, and a selected set $\mathcal{RS}$ of $str$s.
 Then the $sts$ is said to satisfy:
\begin{enumerate}
\item
The {\bf  Fresh Local Variables Property} iff
$\Pi\, \Box\, \sigma \vdash\!\!\vdash_{\cdom,\, \varx}\, \Pi'\, \Box\, \sigma'$
implies that $\Pi'\, \Box\, \sigma'$ is a store,
$\overline{Y'}  = var(\Pi'\, \Box\, \sigma') \setminus var(\Pi\, \Box\, \sigma)$ are fresh local variables,
and $\sigma' = \sigma \sigma_1$ for some substitution $\sigma_1$
(responsible for the variable bindings created at this step) such that
$vdom(\sigma_1) \cup vran(\sigma_1) \subseteq var(\Pi) \cup  \overline{Y'}$.
\item The {\bf Safe Bindings Property}  iff
$\Pi\, \Box\, \sigma \vdash\!\!\vdash_{\cdom,\, \varx}\, \Pi'\, \Box\, \sigma'$
implies that $\sigma_1(X)$ is a constant for all $X \in \varx \cap vdom(\sigma_1)$,
where $\sigma' = \sigma \sigma_1$ as in the previous item.
\item
The {\bf Finitely Branching Property} iff for any fixed $\Pi\, \Box\, \sigma$
there are finitely many $\Pi'\, \Box\, \sigma'$ such that 
$\Pi\, \Box\, \sigma \vdash\!\!\vdash_{\cdom,\, \varx}\, \Pi'\, \Box\, \sigma'$.
\item
The {\bf Termination Property} iff there is no infinite sequence 
$\{\Pi_i\, \Box\, \sigma_i \mid i \in \mathbb{N}\}$ such that 
$\Pi_i\, \Box\, \sigma_i \vdash\!\!\vdash_{\cdom,\, \varx}\, \Pi_{i+1}\, \Box\, \sigma_{i+1}$
for all $i \in \mathbb{N}$.
\item
The {\bf Local Soundness Property} iff
for any $\cdom$-store $\Pi\, \Box\, \sigma$, the union
$$\bigcup \{Sol_{\cdom}(\exists\overline{Y'} (\Pi' \, \Box\, \sigma')) \mid
\Pi\, \Box\, \sigma \vdash\!\!\vdash_{\cdom,\, \varx}\, \Pi'\, \Box\, \sigma',
\overline{Y'}  = var (\Pi'\, \Box\, \sigma') \setminus var(\Pi\, \Box\, \sigma)\}$$
is a subset of $Sol_{\cdom}(\Pi\, \Box\, \sigma)$.
\item
The {\bf Local Completeness Property} for $\mathcal{RS}$-free steps iff
for any $\cdom$-store $\Pi\, \Box\, \sigma$
which is $\mathcal{RS}$-irreducible but not in $\varx$-solved form,
$WTSol_{\cdom}(\Pi\, \Box\, \sigma)$ is a subset of the union
$$\bigcup \{WTSol_{\cdom}(\exists\overline{Y'} (\Pi' \, \Box\, \sigma'))\mid\Pi \Box \sigma \vdash\!\!\vdash_{\cdom,\, \varx} \Pi'  \Box \sigma',
\overline{Y'}  = var (\Pi'\Box \sigma') \setminus var(\Pi \Box \sigma)\}$$
If $\mathcal{RS}$ is the empty set
(in which case all the stores are trivially $\mathcal{RS}$-irreducible)
this property is called simply \emph{local completeness}.
\end{enumerate}
In the case of an extensible solver, the six conditions
listed in this definition must be checked for any conservative extension $\cdom'$ of $\cdom$
and any set $\Pi$ of $SPF$-restricted atomic primitive constraints over $\cdom'$.
\end{definition}


Assume a solver $solve^{\cdom}$ defined by means of a given
$sts$ with transition relation $\sts{\cdom}{\varx}$ and a
selected set $\mathcal{RS}$ of $str$s. If the $sts$ is terminating,
the following recursive definition makes sense:
A given store $\Pi\, \Box\, \sigma$ is {\em hereditarily} $\mathcal{RS}$-{\em irreducible}
iff $\Pi\, \Box\, \sigma$ is $\mathcal{RS}$-{\em irreducible} and all the stores
$\Pi'\, \Box\, \sigma'$ such that
$\Pi\, \Box\, \sigma \vdash\!\!\vdash_{\cdom,\, \varx}\, \Pi'\, \Box\, \sigma'$ (if any)
are also {\em hereditarily} $\mathcal{RS}$-{\em irreducible}.
A solver invocation $solve^{\cdom}(\Pi,\varx)$
is called $\mathcal{RS}$-{\em free} iff the store $\Pi\, \Box\, \varepsilon$ is
hereditarily $\mathcal{RS}$-irreducible. This notion occurs in the
following technical lemma (proved in Appendix \ref{pSolvCdom}), which can be applied  to
ensure that $solve^{\cdom}$ satisfies the requirements for solvers listed in
Definition \ref{defSolver}.


\begin{lemma}[Solvers defined by means of Store Transformation Systems]\label{psts}
Any finitely branching and terminating $\cdom$-store transformation system verifies:
\begin{enumerate}
\item
$\mathcal{SF}_{\cdom}(\Pi,\varx)$ is always finite, and hence $solve^{\cdom}$ is well defined
and trivially satisfies the solved forms property.
\item
$solve^{\cdom}$ has the fresh local variables resp. safe bindings property if the
store transformation system has the corresponding property.
\item
$solve^{\cdom}$ is sound if the store transformation system is locally sound.
\item
$solve^{\cdom}$ is complete for $\mathcal{RS}$-free invocations
if the store transformation system is locally complete for $\mathcal{RS}$-free steps.
In the case that $\mathcal{RS}$ is empty, this amounts to say that
$solve^{\cdom}$ is complete if the store transformation system is locally complete.
\end{enumerate}
\end{lemma}

Note that this  lemma can be used for proving global properties of
extensible solvers, provided that the $sts$ can work with constraint
stores $\Pi\, \Box\, \sigma$, where $\Pi$ is a finite set of
$SPF$-restricted atomic primitive constraints over some arbitrary
conservative extension $\cdom'$ of $\cdom$, and the local properties
required by the lemma hold for any such $\cdom'$.


In the rest of this paper, we will work with the three pure domains $\mathcal{H}$, $\mathcal{FD}$
and $\mathcal{R}$ introduced in the following sections. We will rely on black-box solvers for
$\mathcal{R}$ and $\mathcal{FD}$ provided by SICStus Prolog and we will define an extensible glass-box solver for $\mathcal{H}$ using the store transformation technique just explained.
\vspace*{-.3cm}

\subsubsection{The Pure Constraint Domain $\mathcal{H}$} \label{hdom}


The {\em Herbrand domain} $\mathcal{H}$ supports computations with symbolic equality and
disequality constraints over values of any type. Formally, it is defined as follows:
\begin{itemize}
\item
Specific signature  $\Sigma = \langle TC,\ SBT,\ DC,\ DF,\ SPF \rangle$, where $SBT$ is empty
and $SPF$  includes just the strict equality operator  {\tt  == :: A ->  A ->  bool}.
\item
Interpretation {\tt ==}$^{\herbrand}$, defined as for any domain whose specific signature
includes {\tt ==}.
\end{itemize}


Recall Definition \ref{defCextension} and note that a conservative
extension of $\herbrand$ is any domain $\cdom$ whose specific
signature includes the primitive {\tt ==}. Such a $\cdom$ will be
called a {\em domain with equality} in the sequel. The
$\{\textnormal{{\tt ==}}\}$-restricted constraints over a given
domain with equality are also called {\em extended} Herbrand
constraints. As already explained in Subsection \ref{dom}, atomic
Herbrand constraints have the form $e_1$ {\tt ==} $e_2\, \to!\, t$,
{\em strict equality constraints} $e_1$ {\tt ==} $e_2$ abbreviate
$e_1$ {\tt ==} $e_2 \to!$ {\tt  true}, and  {\em strict disequality
constraints} abbreviate $e_1$ {\tt ==} $e_2\, \to!$ {\tt false}.


Obviously demanded variables (and thus critical variables) for
primitive extended Herbrand constraints are computed as explained in
Subsection \ref{dom}. An extensible Herbrand solver must be able to
solve any finite set $\Pi \subseteq APCon_{\cdom} \upharpoonright
\{\textnormal{{\tt ==}}\}$ of atomic primitive extended Herbrand
constraints, w.r.t. any $\varx \subseteq cvar_{\cdom}(\Pi)$ of
cri\-tical variables. Roughly speaking, the solver proceeds by symbolic
decomposition and binding propagation transformations. More precisely,
we define an extensible glass-box solver for $\herbrand$ by means of the
store transformation technique explained in Subsection \ref{csolvers},
using the transformation rules for $\herbrand$ stores shown in
Table \ref{htable}. Each of these rules has the form
$\pi, \Pi\, \Box\, \sigma \vdash\!\!\vdash_{\herbrand,\, \varx}\, \Pi'\, \Box\, \sigma'$
and indicates the transformation of any store $\pi, \Pi\, \Box\,
\sigma$, which includes the atomic constraint $\pi$ plus other
constraints  $\Pi$; no sequential ordering is intended. We say that
$\pi$ is the {\em selected atomic constraint} for this
transformation step. The notation $\overline{t_m \textnormal{{\tt
==}} s_m}$ in transformation {\bf H3} abbrevia\-tes $t_1$ {\tt ==}
$s_1,\, \ldots,\,t_m$ {\tt ==} $s_m$ and will be used at some other
places. All the $sts$s make sense for arbitrary extended Herbrand
constraints, which ensures extensibility of the $\herbrand$-solver.
Note that transformations {\bf H3} and {\bf H7} involve
decompositions. An application of {\bf H3} or {\bf H7} is called
{\em opaque} iff $h$ is $m$-opaque in the sense explained in
Subsection \ref{expressions}, in which case the new constraints
resulting from the decomposition may become ill-typed. Note also
that an application of transformation {\bf H13} may obviously lose
solutions. An invocation $solve^{\herbrand}(\Pi,\varx)$  of the
$\herbrand$-solver  is called {\em safe} iff it has been computed
without any opaque application of the store transformation rules
{\bf H3} and {\bf H7} and without any  application of the store
transformation rule {\bf H13}. More formally,
$solve^{\herbrand}(\Pi,\varx)$ is a safe invocation of the
$\herbrand$-solver iff it is $\mathcal{URS}$-free, where
$\mathcal{URS}$ is the set $\{{\bf OH3}, {\bf OH7}, {\bf H13}\}$
consisting of {\bf H13} and the unsafe instances  {\bf OH3} and {\bf
OH7} corresponding to opaque applications of {\bf H3} and {\bf H7},
respectively.

 \vspace*{-.3cm}
\begin{table}[h*]
\begin{center}
\begin{tabular}{p{11.cm}}
\hline
\begin{itemize}
\item[{\bf H1}] $(t$ {\tt ==} $s)$ $\to!$ $R,~\Pi$ $\Box$ $\sigma$ $\red_{\herbrand, \varx}$
$(t$ {\tt ==} $s,~\Pi)\sigma_1$  $\Box$ $\sigma\sigma_1$ ~~where
$\sigma_1 = \{R \mapsto true\}$. \vspace*{.1cm}
\item[{\bf H2}]
$(t$ {\tt ==} $s)$ $\to!$ $R,~\Pi$ $\Box$ $\sigma$ $\red_{\herbrand,
\varx}$ $(t$ {\tt /=} $s,~\Pi)\sigma_1$  $\Box$ $\sigma\sigma_1$ ~~
where $\sigma_1 = \{R \mapsto false\}$. \vspace*{.1cm}
\item[{\bf H3}] $h\, \tpp{t}{m}$ {\tt ==} $h\, \tpp{s}{m},~\Pi$ $\Box$ $\sigma$ $\red_{\herbrand, \varx}$
$\overline{t_m \textnormal{{\tt ==}} s_m},~\Pi$ $\Box$ $\sigma$
\vspace*{.1cm}
\item[{\bf H4}] $t$ {\tt  ==} $X,~\Pi$ $\Box$ $\sigma$ $\red_{\herbrand, \varx}$
$X$ {\tt  ==} $t,~\Pi$ $\Box$ $\sigma$ ~~ if $t$ is not a variable.
\vspace*{.1cm}
\item[{\bf H5}] $X$ {\tt ==} $t,~\Pi$ $\Box$ $\sigma$ $\red_{\herbrand, \varx}$
$tot(t),~\Pi\sigma_1$ $\Box$ $\sigma\sigma_1$ ~~ if $X \notin
\varx$, $X \notin var(t)$, $X \neq t$,

where $\sigma_1$ $=$ $\{X$ $\mapsto$ $t\}$, $tot(t)$ abbreviates
$\bigwedge_{Y\in var(t)}(Y\textnormal{{\tt ==}}Y)$. \vspace*{.1cm}
\item[{\bf H6}] $X$ {\tt ==} $t,~\Pi$ $\Box$ $\sigma$ $\red_{\herbrand, \varx}$ $\blacksquare$ ~~
if $X \in var(t)$, $X \neq t$. \vspace*{.1cm}
\item[{\bf H7}] $h\, \tpp{t}{m}$ {\tt  /=} $h\, \tpp{s}{m},~\Pi$ $\Box$ $\sigma$ $\red_{\herbrand, \varx}$
$(t_i$ {\tt  /=} $s_i,~\Pi$ $\Box$ $\sigma)$~~ for each $1\leq i
\leq m$. \vspace*{.1cm}
\item[{\bf H8}] $h\, \tpp{t}{n}$ {\tt  /=} $h'\, \tpp{s}{m},~\Pi$ $\Box$ $\sigma$ $\red_{\herbrand, \varx}$
$\Pi$ $\Box$ $\sigma$ ~~if $h \neq h'$ or $n \neq m$.
\vspace*{.1cm}
\item[{\bf H9}] $t$ {\tt  /=} $t,~\Pi$ $\Box$ $\sigma$ $\red_{\herbrand, \varx}$
$\blacksquare$ ~~ if $t \in \mathcal {V}\!\!ar \cup DC \cup DF \cup
SPF$. \vspace*{.1cm}
\item[{\bf H10}] $t$ {\tt  /=} $X,~\Pi$ $\Box$ $\sigma$ $\red_{\herbrand, \varx}$
$X$ {\tt  /=} $t,~\Pi$ $\Box$ $\sigma$ ~~ if $t$ is not a variable.
\vspace*{.1cm}
\item[{\bf H11}] $X${\tt  /=} $c\, \tpp{t}{n},\Pi \Box \sigma\red_{\herbrand, \varx}$
$(Z_i${\tt /=}$t_i, \Pi)\sigma_1\Box\sigma\sigma_1$  if
$X$$\notin$$\varx$, $c$$\in$$DC^n$ and $\varx$ $\cap$ $var(c\,
\tpp{t}{n})$$\neq$$\emptyset$

where $1$$\leq$$i$$\leq$$n$ (non-deterministic choice),
$\sigma_1$$=$$\{X$$\mapsto$ $c\, \tpp{Z}{n}\}$, $\tpp{Z}{n}$ fresh
variables.
\item[{\bf H12}] $X$ {\tt  /=} $c\, \tpp{t}{n},~\Pi$ $\Box$ $\sigma$ $\red_{\herbrand, \varx}$
$\Pi\sigma_1$ $\Box$ $\sigma\sigma_1$ ~if $X \notin \varx$, $c \in DC^n$ and
$\varx \cap var(c\, \tpp{t}{n}) \neq \emptyset$

where  $\sigma_1 = \{X \mapsto d\, \tpp{Z}{m}\}$, $c \in DC^n$, $d
\in DC^m$, $d \neq c$, $d$ belongs to the same datatype as $c$,
$\tpp{Z}{m}$ fresh variables. \vspace*{.1cm}
\item[{\bf H13}] $X$ {\tt  /=} $h\, \tpp{t}{m},~\Pi$ $\Box$ $\sigma$ $\red_{\herbrand, \varx}$
$\blacksquare$ ~~ if $X \notin \varx$, $\varx \cap var(h\,
\tpp{t}{m}) \neq \emptyset$ and $h \notin DC^m$.
\end{itemize}\\
\hline
\end{tabular}
\caption{Store Transformations for $solve^\mathcal{H}$}\label{htable}
\end{center}
\end{table}
\vspace*{-.4cm}


The idea of using equality and disequality constraints in Logic
Programming stems from Colmerauer
\cite{Col84,colmerauer:prolog-iii-cacm90}. The problem of sol\-ving
these constraints, as well as related decision problems for theories
involving equations and disequations, has been widely investigated
in works such as \cite{LMM87,Mah88,CL89,Com91,Fer92,BB94}, among
others. These papers assume the classical algebraic semantics for
the equality relation, and propose methods for solving so-called
{\em unification and disunification problems} bearing some analogies
to the transformation rules shown in Table \ref{htable}. However,
there are also some differences, because strict equality in $CFLP$
is designed to work with lazy and possibly non-deterministic
functions, whose behaviour does not correspond to the semantics of
equality  in classical algebra and equational logic, as argued in
\cite{Rod01}. Note in particular transformation {\bf H5}, which introduces constraints of 
the form $Y$ {\tt ==} $Y$ in $\herbrand$-solved forms.
These are called {\em totality constraints}, because a valuation $\eta$ is a solution of 
$Y$ {\tt ==} $Y$ iff $\eta(Y)$ is a total pattern.
An approach to disequality constraints close to our
semantic framework can be found in \cite{arenas94combining}, but no
formalization of a Herbrand solver is provided.


The following theorem ensures that the $sts$ for $\herbrand$-stores
can be accepted as a correct specification of an extensible glass-box solver
for the domain $\herbrand$, which is complete for safe solver invocations.

\begin{theorem}[Formal Properties of $solve^{\herbrand}$] \label{hsolver}
The $sts$ with transition relation $\sts{\herbrand}{\varx}$ is finitely branching and terminating,
and therefore
$$solve^{\herbrand}(\Pi,\varx) =
\bigvee \{\exists \overline{Y'} (\Pi'\, \Box\, \sigma') \mid \Pi'\,
\Box\, \sigma' \in \mathcal{SF}_{\herbrand}(\Pi,\varx),\,
\overline{Y'} = var(\Pi'\, \Box\, \sigma') \setminus var(\Pi)\}$$ is
well defined for any domain with equality $\cdom$, any finite $\Pi
\subseteq APCon(\cdom) \upharpoonright \{\textnormal{{\tt ==}}\}$
and any $\varx \subseteq cvar_{\cdom}(\Pi)$. Moreover, for any
arbitrary choice of a domain $\cdom$ with equality,
$solve^{\herbrand}$ satisfies all the requirements for solvers
enumerated in  Definition \ref{defSolver}, except that the {\em
completeness} property may fail for some choices of the constraint
set $\Pi \subseteq APCon(\cdom) \upharpoonright \{\textnormal{{\tt
==}}\}$ to be solved, and is guaranteed to hold only if  the solver
invocation $solve^{\herbrand}(\Pi,\varx)$ is safe
(i.e., $\{{\bf OH3}, {\bf OH7}, {\bf H13}\}$-free).
\end{theorem}
\vspace*{-.1cm}

The proof of the previous theorem is rather technical and can be
found in Appendix \ref{pSolvCdom}. At this point, we just make a few
remarks related to the discrimination and completeness  properties,
that may help to understand some differences between our
$\herbrand$-solver and more classical methods for solving
unification and disunification problems. On the one hand,
transformations {\bf H11} and {\bf H12} are designed to ensure the
discrimination property while preserving completeness w.r.t.
well-typed solutions. On the other hand, transformation  {\bf H13}
trivially ensures discrimination, but it sacrifices completeness
because it fails without making sure that no well-typed solutions
exist. This  corresponds to situations unlikely to occur in practice
and such that no practical way of preserving completeness is at
hand. The other two failing transformations given in Table
\ref{htable} (namely   {\bf H6} and {\bf H9}) respect completeness,
because they are applied to unsatisfiable stores. Finally, the other
cases where completeness may be lost correspond to unsafe
decomposition steps performed with the opaque instances {\bf OH3}
and {\bf OH7} of the $str$s {\bf H3} and {\bf H7}. Due to the
termination property of the $\herbrand$-$sts$, it is decidable
wether a  given $\herbrand$-store $\Pi\, \Box\, \sigma$ is
hereditarily $\mathcal{URS}$-irreducible, in which case no opaque
decompositions will occur when solving the store. However,
computations in the cooperative goal solving calculus presented in
Section \ref{cooperative} can sometimes give rise to
$\herbrand$-stores whose resolution involves opaque decomposition
steps. Due to theoretical results proved in \cite{GHR01}, the
eventual occurrence of opaque decomposition steps during  goal
solving is an undecidable problem.  In case that opaque
decompositions occur, they should be signaled as warnings to the
user.

\begin{example}[Behaviour of $solve^{\herbrand}$] \label{HsolverExample}
In order to illustrate the behaviour of  $solve^{\herbrand}$,
consider the disequality constraint {\tt L /= X:Xs} discussed in item 4. of  Example \ref{pConExample}.
Remember that variable {\tt L} is obviously demanded, while variables {\tt X} and {\tt Xs} are both critical.
Therefore, there are four possible choices for the set $\varx$ of critical variables
to be used within the solver invocation, namely: $\emptyset$, $\{{\tt X}\}$, $\{{\tt Xs}\}$ and $\{{\tt X,\, Xs}\}$.
Let us discuss these cases one by one.
\vspace*{-.4cm}
\begin{itemize}
\item
Choosing $\varx = \emptyset$ means that the solver is not asked to discriminate w.r.t. any critical variable. In this case,
$solve^{\herbrand}$({\tt L/=X:Xs},$\emptyset)$ returns
{\tt L/=X:Xs} $\Box\, \varepsilon$, showing that {\tt L/=X:Xs} is seen as a solved form w.r.t. the
 empty set of critical variables.
\item
Choosing $\varx = \{{\tt X}\}$ asks the solver to discriminate w.r.t. the critical variable {\tt X}.
$solve^{\herbrand}$({\tt L/=X:Xs},\{{\tt X}\}) returns a disjunction of alternatives
$$(\lozenge\, \Box \{{\tt L}\mapsto[\,]\})\lor (\textnormal{{\tt X'/= X}}\, \Box \{{\tt L}\mapsto{\tt X':Xs'}\})
\lor (\textnormal{{\tt Xs'/= Xs}}\, \Box \{{\tt L}\mapsto{\tt
X':Xs'}\})$$ whose members correspond to the three different stores
$\Pi'\, \Box\, \sigma'$ such that the step {\tt L/=X:Xs} $\Box\,
\varepsilon \vdash\!\!\vdash_{\herbrand,\, \{{\tt X}\}}\, \Pi'\,
\Box\, \sigma'$ can be performed with transformation {\bf H12}.
Since these stores are solved w.r.t. $\{{\tt X}\}$, no further
transformations are required. Note that {\tt X} does not occur in
the first and third alternatives, while it has become obviously
demanded in the second one. In this way, the discrimination property
required for solvers is fulfilled.
\item
For each of the two choices $\varx = \{{\tt Xs}\}$ and $\varx =
\{{\tt X,\,Xs}\}$, it is easy to check that the solver invocation
$solve^{\herbrand}(${\tt L/=X:Xs}$,\varx)$ returns the same
disjunction of three alternatives as in the previous item, and the
discrimination property is also fulfilled w.r.t. the chosen set
$\varx$ in both cases.
\end{itemize}

\end{example}
\vspace*{-.4cm}

\subsubsection{The Pure Constraint Domain $\mathcal{R}$} \label{rdom}


The $\rdom$ domain supporting  computation with arithmetic constraints
over real numbers is a familiar idea, used in the well-known
instance $\clp{\rdom}$ of the $CLP$ scheme \cite{JMSY92}. In the
context of our $CFLP$ framework, a convenient formal definition of
the domain $\rdom$ is as follows:
\begin{itemize}
\item
Specific signature  $\Sigma = \langle TC,\ SBT,\ DC,\ DF,\ SPF \rangle$,
where $SBT$ =  \{{\tt real}\} includes just one base type whose values represent real numbers,
and $SPF$ includes the following binary primitive symbols, all of them intended to be
used in infix notation:
\begin{itemize}
\item
The strict equality operator {\tt == :: A ->  A ->  bool}.
\item
The arithmetical operators {\tt +, -, *, / :: real -> real -> real}.
\item
The inequality operator {\tt <= :: real -> real -> bool}.
\end{itemize}
\item
Set of base values $\mathcal{B}_{{\tt real}}^{\rdom} = \mathbb{R}$.
\item
Interpretation {\tt ==}$^{\rdom}$, defined as for any domain whose specific signature includes {\tt ==}.
\item
Interpretation {\tt +}$^{\rdom}$, defined so that for all $t_1,\, t_2,\, t \in  \mathcal{U}_{\rdom}$: \\
$t_1$ {\tt +}$^{\rdom} t_2 \to t$ is defined to hold iff some of the following cases holds:\\
Either $t_1$, $t_2$ and $t$ are real numbers, $t$  being equal to the addition of $t_1$ and $t_2$,
or else $t = \bot$. The interpretations of {\tt -, *} and {\tt /} are defined analogously.
\item
Interpretation {\tt <=}$^{\rdom}$, defined so that
for all $t_1,\, t_2,\, t \in  \mathcal{U}_{\rdom}$: \\
$t_1$ {\tt <=}$^{\rdom} t_2 \to t$ is defined to hold iff some of the following cases holds:\\
Either $t_1$, $t_2$ are real numbers such that $t_1$ is less than or equal to $t_2$, and $t$ = {\tt true};
or else $t_1$, $t_2$ are real numbers such that $t_1$ is greater than $t_2$, and $t$ = {\tt false};
or else $t = \bot$.
\end{itemize}


Atomic $\rdom$-constraints have the form $e_1 \odot e_2\, \to!\, t$,
where $\odot$ is the strict equa\-lity operator or the inequality
operator or an arithmetical operator. An atomic $\rdom$-constraint
is called {\em proper} iff $\odot$ is not the strict equality
operator, and an extended Herbrand constraint otherwise. As already
explained in previous sections, {\em strict equality constraints}
$e_1$ {\tt ==} $e_2$ and  {\em strict disequality constraints} $e_1$
{\tt /=} $e_2$ can be understood as abbreviations of extended
Herbrand constraints. Moreover, various kinds of {\em inequality
constraints} can also be defined as abbreviations, as follows:

\begin{itemize}
\item $e_1$ {\tt <} $e_2$ =$_{def}$ $e_2$ {\tt <=} $e_1$ $\to!$ {\tt false}
\qquad  $e_1$ {\tt <=} $e_2$ =$_{def}$ $e_1$ {\tt <=} $e_2$ $\to!$ {\tt true}
\item $e_1$ {\tt >} $e_2$ =$_{def}$ $e_1$ {\tt <=} $e_2$ $\to!$ {\tt false}
\qquad  $e_1$ {\tt >=} $e_2$ =$_{def}$ $e_2$ {\tt <=} $e_1$ $\to!$ {\tt true}
\end{itemize}


Concerning the solver $solve^{\rdom}$, we expect that it is able to deal
with $\rdom$-specific constraint sets $\Pi \subseteq APCon_{\rdom}$
consisting of atomic primitive constraints $\pi$ of the two following kinds:

\begin{itemize}
\item
Proper $\rdom$-constraints
$t_1 \odot t_2\, \to!\, t$, where $\odot$ is either the inequality operator or an arithmetical operator.
\item
$\rdom$-specific Herbrand constraints having the form
 $t_1$ {\tt ==} $t_2$ or  $t_1$ {\tt /=} $t_2$,
where each of the two patterns $t_1$ and $t_2$ is either a real constant value or a variable
whose type is known to be {\tt real} prior to the solver invocation.
\end{itemize}

For any finite  $\rdom$-specific $\Pi \subseteq APCon_{\rdom}$, it is clear that $dvar_{\rdom}(\Pi) = var(\Pi)$. Therefore, it is safe to define $odvar_{\rdom}(\Pi) = var(\Pi)$ and thus $cvar_{\rdom}(\Pi) = \emptyset$. Consequently, invocations to $solve^\rdom$
can be assumed to be always of the form $solve^{\rdom}(\Pi,\emptyset)$ (abbreviated as
$solve^{\rdom}(\Pi)$), and the discrimination requirements for critical variables become trivial.
Assuming that $solve^{\rdom}$ is used under the restrictions described above and implemented as a black-box solver on top of SICStus Prolog, we are confident that the postulate stated below is a reasonable one. In particular,
we  assume that SICStus Prolog solves $\rdom$-specific Herbrand constraints in a way
compatible with the behaviour of the extensible $\herbrand$-solver described in the previous subsection.

\begin{postulate}[Assumptions on the $\rdom$ Solver]  \label{rsolver}
$solve^{\rdom}$ satisfies five of the six properties required for solvers in Definition \ref{defSolver}
(namely {\em Fresh Local Variables}, {\em Solved Forms}, {\em Safe Bindings}, {\em Discrimination} and
{\em Soundness}), although the {\em Completeness} property may fail for some choices of the
$\rdom$-specific $\Pi \subseteq APCon_{\rdom}$ to be solved.
Moreover, whenever $\Pi \subseteq APCon_{\rdom}$ is $\rdom$-specific and
$\Pi \vdash\!\!\vdash_{solve^{\rdom}} \exists \overline{Y'} (\Pi'\, \Box\, \sigma')$,
the constraint set $\Pi'$ is also $\rdom$-specific, and for all $X \in vdom(\sigma')$: Either $\sigma'(X)$ is a real
value, or else $X$ and  $\sigma'(X)$ belong to $var(\Pi)$.
\end{postulate}

\begin{example}[Behaviour of the $\rdom$ Solver]\label{RsolverExample}
Let us now illustrate the behaviour of $solve^\rdom$ by considering the set of primitive atomic constraints
{\tt RY >= d-0.5, RY-RX <= 0.5, RY+RX <= n+0.5} occurring in item 2. of Example \ref{pConExample}.
The solver invocation $solve^{\rdom}(\Pi_R)$  returns one single alternative
$\Pi'_R\, \Box\, \varepsilon$ with $\Pi'_R = \Pi_R\, \cup$ $\{${\tt RY <= d+0.5}$\}$.
In this case, the new constraint {\tt RY <= d+0.5} has been inferred by adding the two former constraints
{\tt RY-RX <= 0.5} and {\tt  RY+RX <= n+0.5} and taking into account that
{\tt n = 2*d}. In other cases, the $\rdom$ solver can perform
other inferences by means of arithmetical reasoning valid in the
mathematical theory of the real number field. In general, solved
forms computed by solvers help to make more explicit the
requirements on variable values already `hidden' in the
constraints prior to solving (as the upper bound {\tt RY <= d+0.5}
in this example).
\end{example}
\vspace*{-.4cm}

\subsubsection{The Pure Constraint Domain $\mathcal{FD}$} \label{fdom}

The idea of a $\fd$ domain supporting  computation with arithmetic
and finite domain constraints over the integers is a familiar one
within the $CLP$ community, see e.g. \cite{HSD92,HSD98}. In the
context of our $CFLP$ framework, a convenient formal definition of
the domain $\fd$ is as follows:
\begin{itemize}
\item
Specific signature  $\Sigma = \langle TC,\ SBT,\ DC,\ DF,\ SPF \rangle$,
where $SBT$ =  \{{\tt int}\} includes just one base type whose values represent integer numbers,
and $SPF$ includes the following primitive symbols:
\begin{itemize}
\item
The strict equality operator {\tt == :: A ->  A ->  bool}.
\item
The arithmetical operators {\tt \#+, \#-, \#*, \#/ :: int -> int ->
int}.
\item The following primitive symbols related to computation with finite domains:
\begin{itemize}
\item {\tt domain:: [int] -> int -> int -> bool}
\item {\tt belongs:: int -> [int] -> bool}
\item {\tt labeling:: [labelType] -> [int] ->bool},
where {\tt labelType} is an enumerated datatype used to represent labeling strategies.
\end{itemize}
\item
The inequality operator {\tt \#<= :: int -> int -> bool}.
\end{itemize}
\item
Set of base values $\mathcal{B}_{{\tt int}}^{\fd} = \mathbb{Z}$.
\item
Interpretation {\tt ==}$^{\fd}$, defined as  for any domain whose specific signature includes {\tt ==}.
\item
Interpretation {\tt \#+}$^{\fd}$, defined so that for all $t_1,\, t_2,\, t \in  \mathcal{U}_{\fd}$: \\
$t_1$ {\tt \#+}$^{\fd} t_2 \to t$ is defined to hold iff some of the following cases holds:\\
Either $t_1$, $t_2$ and $t$ are integer numbers, $t$  being equal to the addition of $t_1$ and $t_2$,
or else $t = \bot$. The interpretations of {\tt \#-, \#*} and {\tt \#/} are defined analogously.
\item
Interpretation  {\tt domain}$^{\fd}$, defined so that for all $t_1,\, t_2,\, t_3,\, t \in  \mathcal{U}_{\fd}$: \\
{\tt domain}$^{\fd}$ $t_1\, t_2\, t_3\, \to\, t$ is defined to hold iff some of the following cases holds:\\
Either $t_2$ and $t_3$ are integer numbers $a$ and $b$ such that $a \leq b$,
$t_1$ is a non empty finite list of integers belonging to the interval $a..b$ and $t$ = {\tt true};
or else $t_2$ and $t_3$ are integer numbers $a$ and $b$ such that $a \leq b$,
$t_1$ is a non empty list of integers some of which does not belong to the interval $a..b$ and
$t$ = {\tt false}; or else $t_2$ and $t_3$ are integer numbers $a$ and $b$ such that $a > b$ and
$t$ = {\tt false}; or else $t = \bot$.
\item
Interpretation  {\tt belongs}$^{\fd}$, defined so that for all $t_1,\, t_2,\, t \in  \mathcal{U}_{\fd}$: \\
{\tt belongs}$^{\fd}$ $t_1\, t_2\, \to\, t$ is defined to hold iff some of the following cases holds:\\
Either $t_1$ is an integer, $t_2$ is a finite list of integers including $t_1$ as element, and $t$ = {\tt true};
or else $t_1$ is an integer, $t_2$ is a finite list of integers not including $t_1$ as element, and
$t$ = {\tt false}; or else $t = \bot$.
\item
Interpretation  {\tt labeling}$^{\fd}$, defined so that for all $t_1,\, t_2,\, t \in  \mathcal{U}_{\fd}$: \\
{\tt labeling}$^{\fd}$ $t_1\, t_2\, \to\, t$ is defined to hold iff some of the following cases holds:\\
Either $t_1$ is a defined value of type {\tt labelType}, $t_2$ is a finite list of integers, and $t$ = {\tt true};
or else $t = \bot$.
\item
Interpretation {\tt \#<=}$^{\fd}$, defined so that for all $t_1,\, t_2,\, t \in  \mathcal{U}_{\fd}$: \\
$t_1$ {\tt \#<=}$^{\fd} t_2 \to t$ is defined to hold iff some of the following cases holds:\\
Either $t_1$, $t_2$ are integer numbers such that $t_1$ is less than or equal to $t_2$, and $t$ = {\tt true}; or else $t_1$, $t_2$ are integer numbers such that $t_1$ is greater than $t_2$, and
$t$ = {\tt false}; or else $t = \bot$.
\end{itemize}


Atomic $\fd$-constraints include those of the form $e_1 \odot e_2
\to!\, t$, where $\odot$ is either the strict equality operator or
the inequality operator or an arithmetical operator, as well as {\em
domain constraints} {\tt domain} $e_1\, e_2\, e_3 \to!\, t$, {\em
membership constraints} {\tt belongs} $e_1\, e_2 \to!\, t$ and {\em
labeling constraints} {\tt labeling} $e_1\, e_2 \to!\, t$. Atomic
$\fd$-constraints are called extended Herbrand if they have the form
$e_1\, \textnormal{{\tt ==}}\, e_2\, \to!\, t$, and {\em proper}
$\fd$-constraints otherwise. As already explained in previous
sections, {\em strict equality constraints} $e_1$ {\tt ==} $e_2$ and
{\em strict disequality constraints} $e_1$ {\tt /=} $e_2$ can be
understood as abbreviations of extended Herbrand constraints.
Moreover, various kinds of {\em inequali\-ty constraints} can also
be defined as abbreviations, as follows:

\begin{itemize}
\item $e_1$ {\tt \#<} $e_2$ =$_{def}$ $e_2$ {\tt \#<=} $e_1$ $\to!$ {\tt false}
\qquad $e_1$ {\tt \#<=} $e_2$ =$_{def}$ $e_1$ {\tt \#<=} $e_2$
$\to!$ {\tt true}
\item $e_1$ {\tt \#>} $e_2$ =$_{def}$ $e_1$ {\tt \#<=} $e_2$ $\to!$ {\tt false}
\qquad $e_1$ {\tt \#>=} $e_2$ =$_{def}$ $e_2$ {\tt \#<=} $e_1$ $\to!$ {\tt true}
\end{itemize}


Concerning the solver $solve^{\fd}$, we expect that it is able to deal
with $\fd$-specific constraint sets $\Pi \subseteq APCon_{\fd}$
consisting of atomic primitive constraints $\pi$ of the two following kinds:
\begin{itemize}
\item
Proper $\fd$ atomic primitive constraints
(which may be  $t_1 \odot t_2\, \to!\, t$, where $\odot$ is either an integer arithmetical primitive or an inequality primitive,
or primitive domain, membership and labeling constraints).
\item
$\fd$-specific Herbrand constraints having the form
 $t_1$ {\tt ==} $t_2$ or  $t_1$ {\tt /=} $t_2$,
where each of the two patterns $t_1$ and $t_2$ is either an integer constant value or a variable
whose type is known to be {\tt int} prior to the solver invocation.
\end{itemize}

For any finite  $\fd$-specific $\Pi \subseteq APCon_{\fd}$, it is clear that $dvar_{\fd}(\Pi) = var(\Pi)$.
Therefore, it is safe to define $odvar_{\fd}(\Pi) = var(\Pi)$ and thus $cvar_{\fd}(\Pi) = \emptyset$.
Consequently, invocations to $solve^\fd$ can be assumed to be always of the form $solve^{\fd}(\Pi,\emptyset)$
(abbreviated as $solve^{\fd}(\Pi)$), and the discrimination requirements for critical variables become trivial.
Assuming that $solve^{\fd}$ is used under the restrictions described above and implemented as a black-box solver on top of SICStus Prolog, we are confident that the postulate stated below is a reasonable one. In particular,
we  assume that SICStus Prolog solves $\fd$-specific Herbrand constraints in a way
compatible with the behaviour of the extensible $\herbrand$-solver described in the previous subsection.

\begin{postulate}[Assumptions on the $\fd$ Solver]  \label{fsolver}
$solve^{\fd}$ satisfies five of the six properties required for solvers in Definition \ref{defSolver}
(namely {\em Fresh Local Variables}, {\em Solved Forms}, {\em Safe Bindings}, {\em Discrimination} and
{\em Soundness}), although the {\em Completeness} property may fail for some choices of the
$\fd$-specific $\Pi \subseteq APCon_{\fd}$ to be solved.
Moreover, whenever $\Pi \subseteq APCon_{\fd}$ is $\fd$-specific and
$\Pi \vdash\!\!\vdash_{solve^{\fd}} \exists \overline{Y'} (\Pi'\, \Box\, \sigma')$,
the constraint set $\Pi'$ is also $\fd$-specific,  and for all $X \in vdom(\sigma')$: Either $\sigma'(X)$ is an integer
value, or else $X$ and  $\sigma'(X)$ belong to $var(\Pi)$.
\end{postulate}

In particular, labeling constraints are solved by a systematic enumeration of possible values for certain integer variables. Therefore, $solve^\fd$ is unable to solve a labeling constraint $\pi$ unless the current constraint store includes domain or membership constraints for all the variables occurring in $\pi$. The next example shows a typical situation.

\begin{example}[Behaviour of the $\fd$ Solver]\label{FDsolverExample}
In order to illustrate the behaviour of $solve^{\fd}$, let us
consider the set of primi\-tive atomic constraints $\Pi_F$ =
$\{${\tt domain [X,Y] 0 n, labeling [\,] [X,Y]}$\}$ occurring also
in item 2. of Example \ref{pConExample}. The solver invocation
$solve^{\fd}(\Pi_F)$ must solve the conjunction of a domain
constraint and a labeling constraint, both involving the integer
variables {\tt X} and {\tt Y}. The solver proceeds by enumerating
all the possible values  of both variables {\tt X} and {\tt Y}
within their respective domains (determined in this case by the
domain constraint {\tt domain [X,Y] 0 n}) and returns a disjunction
of {\tt (n+1)}$^2$ alternatives, each of which describes one single
solution:
$$(\lozenge\, \Box \{{\tt X} \mapsto 0,\, {\tt Y} \mapsto 0\}) \lor \cdots \lor (\lozenge\, \Box \{{\tt X} \mapsto n,\, {\tt Y} \mapsto n\})$$
\end{example}

In general, solving labeling constraints can give rise to very
expensive enumera\-tions of solutions, unless the finite domains of
the integer variables involved have been pruned by some precedent
computation. As already discussed in Subsection \ref{examples}, the
efficiency of solving the constraint system occurring in item 2. of
Example \ref{pConExample} can be greatly improved by  cooperation
among the the domains $\herbrand$, $\rdom$ and $\fd$. We propose to
use the {\em coordination domains} defined in the next subsection as
a vehicle for domain cooperation in $CFLP$ programming.
\vspace*{-.3cm}

\subsection{Coordination Domains} \label{cdomains}


{\em Coordination domains} $\mathcal{C}$ are a kind of `hybrid'  domains built from various component domains $\cdom_i$, intended to cooperate. The construction of coordination domains also involves  a so-called {\em mediatorial domain} $\mathcal{M}$, whose purpose is to supply {\em bridge constraints} for communication among the component domains. In practice, the component domains will be chosen as pure domains equipped with solvers, and the communication provided by the mediatorial domain will also benefit the solvers.


Mathematically, the construction of coordination domains relies on a
joinabi\-lity condition. Two given constraint domains $\cdom_1$ and
$\cdom_2$ with specific signatures $\Sigma_1 = \langle TC,\ SBT_1,\
DC,\ DF,\ SPF_1 \rangle$ and $\Sigma_2 = \langle TC,\ SBT_2,\ DC,\
DF,\ SPF_2 \rangle$, respectively, are called {\em joinable} iff the
two following conditions hold:
\begin{itemize}
\item
$SPF_1 \cap SPF_2 \subseteq \{${\tt ==}$\}$;
i.e., the only primitive function symbol $p$ allowed to belong both to  $SPF_1$ and  $SPF_2$
is the strict equality operator {\tt ==}.
\item
For every common base type $d \in SBT_1 \cap SBT_2$, one has $\mathcal{B}_d^{\cdom_1} = \mathcal{B}_d^{\cdom_2}$.
\end{itemize}

The {\em amalgamated sum} $\mathcal{S}Ê= \cdom_1 \oplus \cdom_2$ of two joinable domains $\cdom_1$ and $\cdom_2$
is defined as a new domain with signature
$\Sigma = \langle TC,\ SBT_1 \cup SBT_2,\ DC,\ DF,\ SPF_1 \cup SPF_2 \rangle$,
constructed as follows:
\begin{itemize}
\item
For $i = 1,2$ and for all $d \in SBT_i$: $\mathcal{B}_d^{\mathcal{S}} = \mathcal{B}_d^{\cdom_i}$. \\
(no conflict will arise for those $d \in SBT_1 \cap SBT_2$, because of joinability)
\item
For $i = 1,2$, for all $p \in SPF_i$, $p$ other than {\tt ==}, and for all $\overline{t}_n,\, t \in
\mathcal{U}_{\mathcal{S}}$: \\
$p^{\mathcal{S}} \overline{t}_n \to t$ is defined to hold iff one of
the two following cases holds: \\
Either $t = \bot$ or else there exist
$\overline{t'}_n,\, t' \in \mathcal{U}_{\cdom_i}$ such that
$\overline{t'}_n \sqsubseteq \overline{t}_n$, $t' \sqsupseteq t$ and
$p^{\cdom_i} \overline{t'}_n \to t'$.
\end{itemize}
Note that the value universe $\mathcal{U}_{\mathcal{S}}$ underlying an amalgamated sum
$\mathcal{S}Ê= \cdom_1 \oplus \cdom_2$ is a superset of $\mathcal{U}_{\cdom_i}$ for $i =1,2$.
The interpretation of {\tt ==} in $\mathcal{S}$ will behave as defined for any constraint domain, see
Subsection \ref{dom}. For primitive functions $p \in SPF_i$ other than {\tt ==}, the definition of
$p^{\mathcal{S}}$ is designed to obtain an extension of $p^{\cdom_i}$ which satisfies the
technical conditions required by Definition \ref{dcdom}.


The amalgamated sum $\cdom_1 \oplus \cdots \oplus \cdom_n$ of $n$ pairwise joinable domains
$\cdom_i$ $(1 \leq i \leq n)$ can be defined analogously. The following definition and theorem
guarantee the expected behaviour of amalgamated sums as conservative extensions of their components. The proof of the theorem is given in Appendix \ref{pSolvCdom}.

\begin{definition}[Domain specific constraints and truncation operator]\label{dsc}
Assume $\mathcal{S} = \cdom_1 \oplus \cdots \oplus \cdom_n$ of signature $\Sigma$,
constructed as the amalgamated sum of $n$  pairwise joinable domains $\cdom_i$ of signatures 
$\Sigma_i$. Let any $1 \leq i \leq n$ be arbitrarily fixed.
\begin{enumerate}
\item
A set $\Pi \subseteq APCon_{\cdom_i}$ is called $\cdom_i$-{\em specific} iff every valuation $\eta \in Val_{\sdom}$
such that $\eta \in Sol_{\sdom}(\Pi)$ satisfies $\eta(X) \in \uni{\cdom_i}$ for all $X \in var(\Pi)$.
Note that the $\rdom$-specific and $\fd$-specific sets of constraints previously introduced in subsections
\ref{rdom} and \ref{fdom} are also specific in the sense just defined.
\item
Consider the information ordering $\sqsubseteq$ over $\uni{\sdom}$.
The $\cdom_i$-{\em truncation} of a given $\sdom$ value $t \in \uni{\sdom}$
is defined as the $\sqsubseteq$-greatest $\cdom_i$ value $\trunc{t}{\cdom_i} \in \uni{\cdom_i}$ which satisfies
$\trunc{t}{\cdom_i} \sqsubseteq t$, so that any other $\cdom_i$ value $\hat{t} \in \uni{\cdom_i}$
such that $\hat{t} \sqsubseteq t$ must satisfy $\hat{t} \sqsubseteq \trunc{t}{\cdom_i}$.
An effective construction of $\trunc{t}{\cdom_i}$ from $t$ can be obtained by substituting $\bot$
 in place of any subpattern of $t$ which has any of the two following forms: a basic value $u$ which does not belong
 to $\uni{\cdom_i}$, or a partial application of a primitive function which does not belong to $\cdom_i$ specific signature.
 Note that $\trunc{t}{\cdom_i} = t$ if and only if $t \in \uni{\cdom_i}$.
\item
The $\cdom_i$-{\em truncation} of a given $\sdom$-valuation $\eta \in Val_{\sdom}$
 is the $\cdom_i$ valuation $\trunc{\eta}{\cdom_i}$ defined by the condition
 $\trunc{\eta}{\cdom_i}(X) = \trunc{\eta(X)}{\cdom_i}$, for all $X \in \var$.
 Note that $\trunc{\eta}{\cdom_i} = \eta$ if and only if $\eta \in Val_{\cdom_i}$.
\end{enumerate}
\end{definition}

\begin{theorem}[Properties of Amalgamated Sums]\label{sumProperties}
For any $\mathcal{S} = \cdom_1 \oplus \cdots \oplus \cdom_n$ of signature $\Sigma$ constructed as the amalgamated sum of $n$  pairwise joinable domains $\cdom_i$ of signatures $\Sigma_i$ $(1 \leq i \leq n)$:
\begin{enumerate}
\item
$\mathcal{S}$ is well-defined as a constraint domain; i.e., the
interpretations of primitive  function symbols in $\mathcal{S}$
satisfy the four conditions listed in Definition \ref{dcdom}
from Subsection \ref{dom}.
\item
$\mathcal{S}$ is a conservative extension of  $\cdom_i$ for all $(1 \leq i \leq n)$;
i.e., for all $1 \leq i \leq n$, for any $p \in SPF_i^m$ other than {\tt ==},
and for every $\overline{t}_m, t \in  \mathcal{U}_{\cdom_i}$,
one has $p^{\cdom_i}\, \overline{t}_m \to t$ iff $p^{\sdom}\, \overline{t}_m \to t$.
\item
For all $1 \leq i \leq n$, for any set of primitive constraints $\Pi
\subseteq  APCon_{\cdom_i}$ and for every valuation  $\eta  \in
Val_{\cdom_i}$, one has $\eta \in (WT)Sol_{\cdom_i}(\Pi)$ iff $\eta \in (WT)Sol_{\mathcal{S}}(\Pi)$.
\item
For all $1 \leq i \leq n$, for any set of $\cdom_i$-specific primitive constraints
$\Pi \subseteq  APCon_{\cdom_i}$ and for every valuation
$\eta  \in Val_{\sdom}$, one has $\eta \in (WT)Sol_{\sdom}(\Pi)$ iff $\trunc{\eta}{\cdom_i} \in (WT)Sol_{\cdom_i}(\Pi)$.
\end{enumerate}
\end{theorem}


Note that amalgamated sums of the form $\herbrand \oplus \cdom$ are
always possible, and give rise to compound domains that  can profit
from the extensible Herbrand solver. However, in order to construct
more interesting sums tailored to the communication among several
pure domains, so-called {\em mediatorial domains} are needed. Given
$n$ pairwise joinable domains $\cdom_i$ with specific signatures
$\Sigma_i = \langle TC,\ SBT_i,\ DC,\ DF,\ SPF_i \rangle$ ($1 \leq i
\leq n$), a  {\em mediatorial domain} for  the communication among
$\cdom_1, \ldots, \cdom_n$ is defined as any domain $\mathcal{M}$
with specific signature $\Sigma_0 = \langle TC,\ SBT_0,\ DC,\ DF,\
SPF_0 \rangle$ constructed in such a way that the following
conditions are satisfied:
\begin{itemize}
\item
$SBT_0 \subseteq \bigcup_{i = 1}^n SBT_i$, and $SPF_0 \cap SPF_i = \emptyset$ for all
$1 \leq i \leq n$.
\item
For each $d \in SBT_0$ and for any $1 \leq i \leq n$ such that $d \in SBT_i$,
$\mathcal{B}_d^{\mathcal{M}} = \mathcal{B}_d^{\cdom_i}$.\\
(no confusion can arise, since the domains $\cdom_i$ are pairwise
joinable).
\item
Each $p \in SPF_0$ is a so-called {\em equivalence primitive} {\tt \#==}$_{d_i, d_j}$
with declared principal type $d_i \to d_j \to bool$, for some $1 \leq i, j \leq n$ and some
$d_i \in SBT_i$, $d_j \in SBT_j$.
\item
Moreover, each equivalence primitive {\tt \#==}$_{d_i, d_j}$ is used in infix notation
and there is an  injective partial mapping
$inj_{d_i,d_j} : \mathcal{B}^{\cdom_i}_{d_i} \to \mathcal{B}^{\cdom_j}_{d_j}$ used to define the interpretation of {\tt \#==}$_{d_i, d_j}$ in $\mathcal{M}$ as follows: For all
$s, t, r \in \mathcal{U}_{\mathcal{M}}$, $s$ {\tt \#==}$_{d_i, d_j}^{\mathcal{M}}\, t\, \to r$  iff some
of the three cases listed bellow holds:
\vspace*{-.2cm}
\begin{enumerate}
\item
$s \in dom(inj_{d_i,d_j})$, $t \in \mathcal{B}^{\cdom_j}_{d_j}$, $t = inj_{d_i,d_j}(s)$ and $true \sqsupseteq r$.
\item
$s \in dom(inj_{d_i,d_j})$, $t \in \mathcal{B}^{\cdom_j}_{d_j}$, $t \neq inj_{d_i,d_j}(s)$ and $false \sqsupseteq r$.
\item
$r = \bot$.
\end{enumerate}
\end{itemize}
\vspace*{-.2cm}


Equivalence primitives {\tt \#==}$_{d_i, d_j}$ allow to write well-typed atomic 
{\em mediatorial constraints} of the form $a$ {\tt \#==}$_{d_i, d_j}\, b\, \to!\, c$, using expressions 
$a :: d_i$, $b :: d_j$ and $c :: bool$. Constraints of the form $a$ {\tt \#==}$_{d_i, d_j}\, b\, \to!$ {\tt  true} resp. $a$ {\tt \#==}$_{d_i,d_j}\, b\, \to!$ {\tt false} are abbreviated as 
$a$ {\tt \#==}$_{d_i,d_j}\, b$ resp. $a$ {\tt \#/==}$_{d_i, d_j}\, b$ and
called {\em bridges} and {\em antibridges}, respectively. The
usefulness of bridges for cooperative goal solving in $CFLP$ has
been motivated in the introduction and will elaborated when
presenting the cooperative goal solving calculus $\cclnc{\ccdom}$ in
Section \ref{cooperative}. Antibridges and mediatorial constraints
$a$ {\tt \#==}$_{d_i, d_j}\, b\, \to!\, R$, where $R$ is a variable,
can also occur in $\cclnc{\ccdom}$ computations, but  they are not
so directly related to domain cooperation as bridges.


Each particular choice of injective partial mappings $inj_{d_i,d_j}$
and their corresponding equivalence primitives {\tt \#==}$_{d_i,
d_j}$ gives rise to the construction of a particular mediatorial
domain $\mathcal{M}$, suitable for communication among the
$\cdom_i$. Moreover, it is clear by construction that the $n+1$
domains $\mathcal{M},\, \cdom_1,\, \ldots,\, \cdom_n$ are pairwise
joinable, and it is possible to build the amalgamated sum $\ccdom =
\mathcal{M} \oplus  \cdom_1 \oplus \cdots \oplus \cdom_n$. This
`hybrid' domain supports the communication among the domains
$\cdom_i$ via bridge constraints provided by $\mathcal{M}$.
Therefore, $\mathcal{M}$ is called a {\em coordination domain} for
$\cdom_1, \ldots, \cdom_n$.

In practice, it is advisable to include the Herbrand domain $\herbrand$ as one of the component domains $\cdom_i$  when building a coordination domain $\ccdom$. In application programs over such a coordination domain, the $\herbrand$ solver is typically helpful for solving symbolic equality and disequality constraints over user defined datatypes, while the solvers of other component domains $\cdom_i$ whose specific signatures include the primitive
{\tt ==} may be helpful for computing with equalities and disequalities related to $\cdom_i$'s specific base types.
\vspace*{-.3cm}

\subsection{The Coordination Domain $\ccdom = \mathcal{M} \oplus \herbrand \oplus \fd \oplus \rdom$} \label{ourcdom}

In this subsection, we explain the construction of a coordination
domain for coope\-ration among the three pure domains $\herbrand$,
$\fd$ and $\rdom$.

First,  we define a mediatorial domain $\mathcal{M}$ suitable to
this purpose. It is built with specific signature $\Sigma_0 =
\langle TC,\, SBT_0,\, DC,\, DF,\, SPF_0 \rangle$, where $SBT_0 =
\{int,real\}$ and $SPF_0 = \{${\tt \#==}$_{int, real}\}$. The
equivalence primitive {\tt \#==}$_{int, real}$ is interpreted with
respect to the total injective mapping $inj_{int,real} :: \mathbb{Z}
\to \mathbb{R}$, which maps each integer value into the equivalent
real value. In the sequel, we will write {\tt \#==} in place of {\tt
\#==}$_{int, real}$ when referring to this equivalence primitive. We
will use the same abbreviation for writing mediatorial constraints.

Next, we use this mediatorial domain for building $\ccdom =
\mathcal{M} \oplus \herbrand \oplus \fd \oplus \rdom$. In the rest
of the paper, $\ccdom$ will always stand for this particular
coordination domain, whose usefulness has been motivated in Section
\ref{introduction} and will become more evident in Section
\ref{cooperative}. Note that bridges {\tt X \#== RX} and antibridges
{\tt X \#/== RX} can be useful just as constraints; in particular,
{\tt X \#== RX} acts as an {\em integrality constraint} over the
value of the real variable {\tt RX}. More importantly, in Section
\ref{cooperative} the mediatorial domain $\ccdom$ will serve as a
basis for useful cooperation facilities, including the projection of
$\rdom$ constraints to the $\fd$ solver (and vice versa) using
bridges, the specialization of  $\herbrand$ constraints to become
$\rdom$-specific  or $\fd$-specific in some computational contexts, and some
other special mechanisms designed for processing the mediatorial
constraints occurring in computations.

In particular, computation rules for simplifying mediatorial
constraints will be needed. Although $\mathcal{M}$ is not a  'pure'
domain, simplifying $\mathcal{M}$-constraints is most conveniently
thought of as the task of a $\mathcal{M}$-solver.
This solver is expected to deal with $\mdom$-specific constraint sets $\Pi \subseteq APCon_{\mdom}$
consisting of atomic primitive constraints $\pi$ of the form
$t$ {\tt \#==}$ s\, \to!\, b$, where $b$ is either a variable or a boolean constant and
each of the two patterns $t$ and $s$ is either a variable or a numeric value of the proper type
({\tt int} for $t$ and {\tt real} for $s$).
For any finite set $\Pi \subseteq APCon_{\mathcal{M}}$ of such $\mdom$-specific constraints,
it is clear that $dvar_{\mathcal{M}}(\Pi) = var(\Pi)$.
Therefore, it is safe to define $odvar_{\mathcal{M}}(\Pi) = var(\Pi)$ and thus
$cvar_{\mathcal{M}}(\Pi) = \emptyset$. We define a glass-box solver
$solve^{\mathcal{M}}$  by means of the store transformation
technique explained in Subsection \ref{csolvers}, using the $str$s
for $\mathcal{M}$-stores shown in Table \ref{mtable}. Due to the
absence of critical variables, one-step transformations of
$\mathcal{M}$-stores do not depend on a parametrically given set
$\varx$ of critical variables and have the form $\pi, \Pi\, \Box\,
\sigma \vdash\!\!\vdash_{\mathcal{M}}\, \Pi'\, \Box\, \sigma'$,
indicating the transformation of any store $\pi, \Pi\, \Box\,
\sigma$, which includes the atomic constraint $\pi$ plus other
constraints  $\Pi$; no sequential ordering is intended. We say that
$\pi$ is the {\em selected atomic constraint} for this
transformation step.

 \vspace*{-.3cm}
\begin{table}[h*]
\begin{center}
\begin{tabular}{p{11.2cm}}
\hline \begin{itemize} \item[{\bf M1}] $(t $ {\tt \#==} $s)$ $\to!$
$B,~\Pi$ $\Box$ $\sigma$ $\red_{\mathcal{M}}$ $(t$ {\tt \#==}$
~s,~\Pi)\sigma_1$ $\Box$ $\sigma\sigma_1$

if $t \in \mathcal V\!\!ar \cup \mathbb{Z}$, $s \in \mathcal V\!\!ar
\cup \mathbb{R}$, $B \in \mathcal V\!\!ar$, where $\sigma_1= \{B
\mapsto true\}$. \vspace*{.2cm}

\item[{\bf M2}] $(t $ {\tt \#==} $s)$ $\to!$ $B,~\Pi$ $\Box$ $\sigma$ $\red_{\mathcal{M}}$
$(t$ {\tt \#/==}$ ~s,~\Pi)\sigma_1$ $\Box$ $\sigma\sigma_1$

if $t \in \mathcal V\!\!ar \cup \mathbb{Z}$, $s \in \mathcal V\!\!ar
\cup \mathbb{R}$, $B \in \mathcal V\!\!ar$, where $\sigma_1= \{B
\mapsto false\}$. \vspace*{.2cm}

\item[{\bf M3}] $X$ {\tt \#==}$ u',~\Pi$ $\Box$ $\sigma$ $\red_{\mathcal{M}}$
$\Pi\sigma_1$ $\Box$ $\sigma\sigma_1$

if $u' \in \mathbb{R}$, and there is $u \in \mathbb{Z}$ such that
$u$ {\tt \#==}$^{\mathcal{M}}~u' \to true$, where $\sigma_1 = \{X
\mapsto u\}$.
\vspace*{.05cm}

\item[{\bf M4}] $X$ {\tt \#==} $ u',~\Pi$ $\Box$ $\sigma$ $\red_{\mathcal{M}}$ $\blacksquare$

if $u' \in \mathbb{R}$, and there is no $u \in \mathbb{Z}$ such that
$u$ {\tt \#==}$^{\mathcal{M}}~u' \to true$. \vspace*{.2cm}

\item[{\bf M5}] $u$ {\tt \#==} $RX,~\Pi$ $\Box$ $\sigma$ $\red_{\mathcal{M}}$
$\Pi\sigma_1$ $\Box$ $\sigma\sigma_1$

if $u \in \mathbb{Z}$ and $u'$$\in$$\mathbb{R}$ is so chosen that
$u$ {\tt \#==}$^{\mathcal{M}}~u'$ $\to$ $true$, where $\sigma_1 =
\{RX \mapsto u'\}$.
\vspace*{.1cm}

\item[{\bf M6}] $u$ {\tt \#==} $u',~\Pi$ $\Box$ $\sigma$ $\red_{\mathcal{M}}$
$\Pi$ $\Box$ $\sigma$ ~~ if $u \in \mathbb{Z}$, $u' \in \mathbb{R}$,
and $u$ {\tt \#==}$^{\mathcal{M}}~u' \to true$. \vspace*{.2cm}

\item[{\bf M7}] $u$ {\tt \#==} $u',~\Pi$ $\Box$ $\sigma$ $\red_{\mathcal{M}}$
$\blacksquare$ ~~~~~~~~  if $u \in \mathbb{Z}$, $u' \in \mathbb{R}$, and
$u$ {\tt \#==}$^{\mathcal{M}}~u' \to false$.
\vspace*{.2cm}

\item[{\bf M8}] $u$ {\tt \#/==} $u',~\Pi$ $\Box$ $\sigma$ $\red_{\mathcal{M}}$
$\Pi$ $\Box$ $\sigma$ ~~if $u \in \mathbb{Z}$, $u' \in \mathbb{R}$, and
$u$ {\tt \#==}$^{\mathcal{M}}~u' \to false$
\vspace*{.2cm}

\item[{\bf M9}] $u$ {\tt \#/==} $u',~M$ $\Box$ $\sigma$ $\red_{\mathcal{M}}$
$\blacksquare$~~~~~~~~if $u \in \mathbb{Z}$, $u' \in \mathbb{R}$, and
$u$ {\tt \#==}$^{\mathcal{M}}~u' \to true$.
\end{itemize}\\
\hline
\end{tabular}
\caption{ Store Transformations for $solve^\mathcal{M}$}
\label{mtable}
\end{center}
\end{table}
 \vspace*{-.4cm}


The following theorem ensures that the $sts$ for $\mdom$-stores
can be accepted as a correct specification of a glass-box solver for the domain $\mdom$.

\begin{theorem}[Formal Properties of the $\mdom$ Solver] \label{msolver}
The $sts$ with transition relation $\ests{\mdom}$ is finitely branching and terminating, and therefore
$$solve^{\mdom}(\Pi) =
\bigvee \{\exists \overline{Y'} (\Pi'\, \Box\, \sigma') \mid \Pi'\, \Box\, \sigma' \in
\mathcal{SF}_{\mdom}(\Pi),\, \overline{Y'} = var(\Pi'\, \Box\, \sigma') \setminus var(\Pi)\}$$
is well defined for any finite $\Pi \subseteq APCon_{\mdom}$ of $\mdom$-specific constraints.
The solver $solve^{\mdom}$ satisfies all the requirements for solvers enumerated in  Definition \ref{defSolver}.
Moreover, whenever $\Pi \subseteq APCon_{\mdom}$ is $\mdom$-specific and
$\Pi \vdash\!\!\vdash_{solve^{\mdom}} \exists \overline{Y'} (\Pi'\, \Box\, \sigma')$,
the constraint set $\Pi'$ is also $\mdom$-specific and $\sigma'(X)$ is either a boolean value, an  integer value or a real value
for all $X \in vdom(\sigma')$.
\end{theorem}

The proof is omitted, because it is completely similar to that of
Theorem \ref{hsolver} but much easier. In fact, the $sts$ for $\mdom$-stores involves
no decompositions. Actually, this $sts$ is finitely branching, terminating, locally sound
and locally complete. Therefore, Lemma \ref{psts} can be applied.


The framework for cooperative programming and the cooperative goal solving calculus $\cclnc{\ccdom}$
presented in Section \ref{cooperative} essentially rely on the coordination domain $\ccdom$ just discussed,
and the instance $\cflp{\mathcal{C}}$ of the $CFLP$ scheme \cite{LRV07}
provides a declarative semantics for proving the soundness and completeness of $\cclnc{\ccdom}$.
As  we will see, some cooperative goal solving rules in $\cclnc{\ccdom}$ rely on the  identification
of certain atomic primitive Herbrand constraints $\pi$ as $\fd$-specific or $\rdom$-specific, respectively,
on the basis of the mediatorial constraints available in a given $\mdom$-store $M$.
The notations  $M \vdash \pi$ {\tt in} $\fd$  and $M \vdash \pi$ {\tt in} $\rdom$ defined below serve to
this purpose.

\begin{definition} [Inference of domain-specific extended Herbrand constraints] \label{defSpecific}
Assume a mediatorial store $M$ and a well-typed atomic extended
Herbrand constraint $\pi$ having the form $t_1$ {\tt ==}
$t_2$ or  $t_1$ {\tt /=} $t_2$, where each of the two patterns $t_1$
and $t_2$ is either a numeric constant $v$ or a variable $V$. Then
we define:
\begin{enumerate}
\item
$M \vdash \pi$ {\tt in} $\fd$ (read as `M allows to infer that $\pi$
is $\fd$-specific') iff some of the three following conditions
holds:
\begin{enumerate}
\item
$t_1$ or $t_2$ is an integer constant.
\item
$t_1$ or $t_2$ is a variable that occurs as the left argument of the
opera\-tor {\tt \#==} within some mediatorial constraint belonging
to $M$.
\item
$t_1$ or $t_2$ is a variable that has been recognized to have type {\tt int}  by
some implementation dependent device.
\end{enumerate}
\item
$M \vdash \pi$ {\tt in} $\rdom$ (read as `M allows to infer that
$\pi$ is $\rdom$-specific') iff some of the three following
conditions holds:
\begin{enumerate}
\item
$t_1$ or $t_2$ is a real constant.
\item
$t_1$ or $t_2$ is a variable that occurs as the right argument of the operator
{\tt \#==} within some mediatorial constraint belonging to $M$.
\item
$t_1$ or $t_2$ is a variable that has been recognized to have type {\tt real} by
some implementation dependent device.
\end{enumerate}
\end{enumerate}
\end{definition}

\vspace*{-.3cm}

%
%

\section{Cooperative Programming and Goal Solving in $CFLP(\mathcal{C})$} \label{cooperative}

This section presents our cooperative computation model for goal solving.
After introducing programs and goals in the first subsection,
the subsequent subsections deal with goal solving rules,
illustrative computation examples, and results concerning the formal properties of
the computation model.

Our goal solving rules work by transforming initial goals into final
goals in solved form which serve as computed answers, as in the
previously published {\em Constrained Lazy Narrowing Calculus}
$CLNC(\mathcal{D})$ \cite{LMR04}, which works over any
parametrically given domain $\cdom$ equipped with a solver. We have
substantially extended $CLNC(\mathcal{D})$ with various mechanisms
for domain cooperation via bridges, projections and some more ad hoc
operations. The result is a  {\em Cooperative Constrained Lazy
Narrowing Calculus} $\cclnc{\ccdom}$ which is sound and complete
(with some limitations) w.r.t. the instance $\cflp{\ccdom}$ of the
generic $CFLP$ scheme \cite{LRV07}. For the sake of typographic
simplicity, we have restricted our presentation of $\cclnc{\ccdom}$
to the case $\ccdom$ $=$ $\mathcal{M}$ $\oplus$ $\herbrand$ $\oplus$
$\fd$ $\oplus$ $\rdom$, although it could be easily extended to
other coordination domains, as sketched in our previous paper
\cite{DBLP:journals/entcs/MartinFHRV07}.
\vspace*{-.3cm}

\subsection{Programs and Goals} \label{programa}

$\cflp{\mathcal{C}}$-{\em programs} are sets of constrained
rewrite rules that define the beha\-vior of possibly higher-order
and/or non-deterministic lazy functions over $\mathcal{C}$, called
{\em program rules}. More precisely, a program rule $Rl$ for a
defined function symbol $f$ $\in$ $DF_{\Sigma}^n$ with principal
type $f$ $::$ $\overline{\tau}_n$ $\to$ $\tau$ has the form $f\,
\overline{t}_n$ $\to$ $r$ $\Leftarrow$ $\Delta$, where
$\overline{t}_n$ is a linear sequence of patterns, $r$ is an
expression, and $\Delta$ is a finite conjunction
$\delta_1,\ldots,\delta_m$ of atomic constraints $\delta_i \in ACon_{\ccdom}$.
Each program rule $Rl$ is required to be well-typed, i.e.,
there must exist some type environment $\Gamma$
for the variables occurring in $Rl$ such that
$\Sigma,\, \Gamma\, \vdash_{WT}\, t_i :: \tau_i$  for all $1 \leq i \leq n$,
$\Sigma,\, \Gamma\, \vdash_{WT}\, r :: \tau$ and
$\Sigma,\, \Gamma\, \vdash_{WT}\, \delta_i$  for all $1 \leq i \leq m$.

The left-linearity requirement is quite common in functional
and functional logic programming. As in constraint logic
programming, the conditional part of a program rule needs no
explicit occurrences of existential quantifiers. A program rule $Rl$
is said to include {\em free occurrences of higher-order logic variables} iff there is some variable
$X$ which does not occur  in the left-hand side of $Rl$ but has some
occurrence in a context of the form $X\, \tpp{e}{m}$ (with $m > 0$) somewhere else
in $Rl$. A program $\prog$ includes free occurrences of higher-order logic variables
iff some of the program rules in $\prog$ does.

As in functional languages such as Haskell \cite{haskell}, our
programs rules can deal with higher-order functions and are not expected to be always terminating.
Moreover, in contrast to Haskell and most other functional languages,
we do not require program rules to be confluent. Therefore,
some program defined functions can be {\em non-deterministic} and return several values for
a fixed  choice of  arguments in some cases.
As a concrete example of typed $\cflp{\mathcal{C}}$-program
written in the concrete syntax of the $\toy$ system, we
refer to the program rules presented in Subsection \ref{examples}.

Programs are used to solve {\em goals} using a cooperative goal solving calculus which will be described in 
subsections \ref{lnr}, \ref{dcr} and \ref{csr} below. Goals over the coordination domain $\ccdom$ have the general form
$G$ $\equiv$ $\exists \overline{U}.$ $P$ $\Box$ $C$ $\Box$ $M$ $\Box$ $H$ $\Box$ $F$ $\Box$ $R$,
where the symbol $\Box$ is interpreted as conjunction and:
\vspace*{-.2cm}

\begin{itemize}
\item $\overline{U}$ is a finite set of so-called {\em existential variables}, intended to represent local
varia\-bles in the computation.
\item $P$ is a set of so-called {\em productions} of the form $e_1 \to t_1, \ldots, e_m \to t_m$,
where $e_i \in Exp_{\ccdom}$ and $t_i \in Pat_\ccdom$ for all $1\leq i \leq m$. 
The set of {\em produced variables} of $G$ is defined as
the set $pvar(P)$ of variables occurring in $t_1 \ldots t_m$.
During goal solving, productions are used to compute values for the produced variables
insofar as demanded, using the goal solving rules for constrained lazy narrowing presented in Subsection \ref{lnr}. 
We consider a {\em production relation} between variables, such that $X \gg_{P} Y$ holds  iff
$X, Y \in pvar(P)$ and  there is some
 $1 \leq i \leq m$ such that $X \in var(e_i)$ and $Y \in var(t_i)$.
\item $C$ is  the so-called {\em constraint pool},  a finite set of constraints to be solved,
possibly including active occurren\-ces of defined functions symbols.
\item $M = \Pi_M\, \Box \, \sigma_M$ is the {\em mediatorial store}, including a finite set of atomic
primi\-tive constraints $\Pi_M \subseteq APCon_{\mathcal{M}}$ and a
substitution $\sigma_M$. We will write $B_M \subseteq \Pi_M$ for the
set of all $\pi \in \Pi_M$ which are {\em bridges} $t$ {\tt \#==}
$s$, where each of the two patterns $t$ and $s$ may be either a
variable or a numeric constant.
\item $H = \Pi_H\, \Box \, \sigma_H$ is the {\em Herbrand store}, including a
finite set of atomic primitive constraints $\Pi_H \subseteq APCon_{\herbrand}$
and a substitution $\sigma_H$.
\item $F = \Pi_F\, \Box \, \sigma_F$ is the  {\em finite domain store}, including a
finite set of atomic primi\-tive constraints $\Pi_F \subseteq
APCon_{\fd}$ and a substitution $\sigma_F$.
\item $R = \Pi_R\, \Box \, \sigma_R$ is the  {\em real arithmetic store}, including a
finite set of atomic primitive constraints $\Pi_R \subseteq APCon_{\rdom}$
and a substitution $\sigma_R$.
\end{itemize}

A goal $G$ is said to have {\em free occurrences of higher-order logic variables}
iff there is some variable $X$ occurring in $G$ in some context of the form
$X\, \tpp{e}{m}$, with $m > 0$.
Two special kinds of goals are particularly interesting.
{\em Initial goals} just consist of a well-typed constraint pool $C$.
More precisely, the existential prefix $\overline{U}$,
productions in $P$, and stores $M$, $H$, $F$ and $R$ are empty.
{\em Solved goals} (also called {\em solved forms}) have empty $P$ and $C$ parts
and cannot be transformed by any goal solving rule.
Therefore, they  are used to represent {\em computed answers}.
We will also write $\blacksquare$ to denote an {\em inconsistent goal}.

\begin{example}[Initial and Solved Goals]\label{example_goals}

Consider the initial goals {\bf Goal 1},  {\bf Goal 2} and  {\bf Goal 3} presented in $\toy$
syntax in Subsection  \ref{examples}, for the choice {\tt d = 2}, {\tt n = 4}. When written with
the abstract syntax for general $\cflp{\mathcal{C}}$-goals they become 
\vspace*{.2cm}

\noindent $~~~~~~~~~1) ~~\Box$ bothIn (triangle (2, 2.75) 4 0.5) (square 4) (X,Y) $\Box$$\Box$$\Box$$\Box$ \\
$~~~~~~~~~~2) ~~\Box$ bothIn (triangle (2, 2.5) 2 1) (square 4) (X,Y) ~~~~$\Box$$\Box$$\Box$$\Box$ \\
$~~~~~~~~~~3) ~~\Box$ bothIn (triangle (2, 2.5) 8 1) (square 4) (X,Y) ~~~~$\Box$$\Box$$\Box$$\Box$ 
\vspace*{.2cm}

The expected solutions for these goals have been explained in Subsection  \ref{examples}.
A general notion of solution for goals will be defined in Subsection \ref{SC}.
The resolution of these example goals  in our cooperative goal solving calculus $\cclnc{\ccdom}$ 
will be discussed in detail in Subsection \ref{exampletriangle}.
The respective solved forms obtained as computed answers
(restricted to the variables in the initial goal) will be: 
\vspace*{.2cm}

\noindent $~~~~~~~~~1)~$ $\blacksquare$\\
$~~~~~~~~~~2)~$ $\Box$$\Box$$\Box$$\Box$ ($\true$ $\Box$ $\{X
\mapsto 2, Y \mapsto 2\}$) $\Box$\\
$~~~~~~~~~~3)~$ $\Box$$\Box$$\Box$$\Box$ ($\true$ $\Box$ $\{X \mapsto 0, Y \mapsto 2\}$) $\Box$ \\
$~~~~~~~~~~~~~~$ $\Box$$\Box$$\Box$$\Box$ ($\true$ $\Box$ $\{X \mapsto 1, Y \mapsto 2\}$) $\Box$ \\
$~~~~~~~~~~~~~~$ $\Box$$\Box$$\Box$$\Box$ ($\true$ $\Box$ $\{X \mapsto 2, Y \mapsto 2\}$) $\Box$ \\
$~~~~~~~~~~~~~~$ $\Box$$\Box$$\Box$$\Box$ ($\true$ $\Box$ $\{X \mapsto 3, Y \mapsto 2\}$) $\Box$ \\
$~~~~~~~~~~~~~~$ $\Box$$\Box$$\Box$$\Box$ ($\true$ $\Box$ $\{X \mapsto 4, Y \mapsto 2\}$) $\Box$
\vspace*{.2cm}
\end{example}
\vspace*{-.2cm}

The goal  solving rules of the $\cclnc{\ccdom}$ calculus presented
in the rest of this section has been designed as an extension of an
existing goal solving calculus for the $CFLP$ scheme \cite{LMR04},
adding the new features needed to support solver coordination via
bridge constraints. In contrast to previous related works such as
\cite{LLR93,AEH94,AEH00,vado+:ppdp03,Vad05,vado:ictac07}, we have
omitted the use of so-called {\em definitional trees} to ensure an
optimal selection of needed narrowing steps. This feature could be
easily added to $\cclnc{\ccdom}$ following the ideas from
\cite{Vad05}, but  we have decided not do so in order to avoid
technical complications which do not contribute to a better
understanding of domain coope\-ration. Moreover, the design of
$\cclnc{\ccdom}$ is tailored to programs and goals having no free
occurrences of higher-order logic variables. As shown in
\cite{GHR01}, goal solving rules for dealing with free higher-order
logic variables give rise to ill-typed solutions very often. If
desired, they could be easily added to our present setting.

Let us finish this subsection with a brief discussion of some technical issues
needed in the sequel. The set $odvar(G)$ of {\em obviously demanded variables}
in a given goal $G$ is defined as the least subset of $var(G)$ which
satisfies the two following conditions:
\begin{enumerate}
\item $odvar(G)$ includes $odvar_{\mdom}(\Pi_M)$, $odvar_{\herbrand}(\Pi_H)$, $odvar_{\fd}(\Pi_F)$ and $odvar_{\rdom}(\Pi_R)$
which are defined as explained in Subsections  \ref{dom} and \ref{pdom}.
\item $X \in odvar(G)$ for any production $(X \tpp{a}{k}\rightarrow t)\in P$
such that $k>0$ and either $t \notin \var$ or else $t \in odvar(G)$.
\end{enumerate}

Note that $odvar(G)$ boils down to  $odvar_{\mdom}(\Pi_M) \cup odvar_{\herbrand}(\Pi_H) \cup
odvar_{\fd}(\Pi_F) \cup odvar_{\rdom}(\Pi_R)$
in the case that $G$ has no free occurrences of higher-order variables.
Productions $e \to$ $X$ such that $e$ is an active expression and $X \notin odvar(G)$ is a
not obviously demanded variable are called {\em suspensions},
and play an important role during  goal solving.

Certain properties are trivially satisfied by initial goals and kept  invariant through the application of goal transformations. Such {\em goal invariant properties} include those formalized in previous works as e.g. \cite{LMR04}: Each produced variable is produced only once, all the produced variables must be existential, the transitive closure $\gg_{P}^{+}$ of the relation between produced variables must be irreflexive, and no produced variable occurs in  the answer substitutions. Other goal invariants are more specific of our current cooperative setting based on the coordination domain $\ccdom$:

\begin{itemize}
\item
The domains of the substitutions $\sigma_M$, $\sigma_H$, $\sigma_F$ and
$\sigma_R$ are pairwise disjoint.
\item
For any store $S$ in $G$, the application of $\sigma_S$ causes no modification to the goal.
\item
For any $X \in vdom(\sigma_M)$, $\sigma_M(X)$ is either a boolean value, an integer value or a real value.
\item
For any $X \in vdom(\sigma_F)$, $\sigma_F(X)$ is either an integer value or a variable occurring in $\Pi_F$.
\item
For any $X \in vdom(\sigma_R)$, $\sigma_R(X)$ is either a real value or a variable occurring in $\Pi_R$.
\end{itemize}

These properties remain invariant through goal transformations because of Theorem \ref{msolver}
and  Postulates \ref{fsolver} and \ref{rsolver}, and also because the bindings computed by each particular solver are properly propagated to the rest  of the goal.
In particular, whenever a variable binding $\{X \mapsto t\}$ arises in one of the stores during goal solving, 
the propagation of this binding to the goal {\em applies} the binding everywhere, but {\em places} it  only within the substitution of this particular store,
so that the first item above is ensured.

At this point we must introduce some auxiliary notations in order to make this idea more precise.
Let $\cdom$ stand for any of the four domains $\mdom$, $\herbrand$, $\fd$ or $\rdom$
and consider the store  $S = \Pi_S\,\Box\,\sigma_S$ corresponding to $\cdom$.
We will note as $(P\,\Box\,C\,\Box\,M \,\Box\,H\,\Box\,F\,\Box\,R)@_{\cdom}\sigma'$
the result of {\em applying} $\sigma'$ to $P\,\Box\,C\,\Box\,M \,\Box\,H\,\Box\,F\,\Box\,R$
and {\em composing} $\sigma_S$ with $\sigma'$. More formally, in the particular case that $\cdom$ is chosen as $\fd$, we define $(P\,\Box\,C\,\Box\,M \,\Box\,H\,\Box\,F\,\Box\,R)@_{\fd}\sigma'$ as 
$P\sigma'\,\Box\,C\sigma'\,\Box\,M\star\sigma' \,\Box\,H\star\sigma'\,\Box\,F@\sigma'\,\Box\,R\star\sigma'$,
where $F@\sigma'$ is defined as $\Pi_F\sigma'\,\Box\,\sigma_F\sigma'$
and $S\star\sigma'$ is defined as $\Pi_S\sigma'\,\Box\,\sigma_S\star\sigma'$ for $S$ being $M$, $H$ or $R$. 
Recall that the {\em application} of $\sigma'$ to $\sigma_S$ has been defined as $\sigma_S \star \sigma' = \sigma_S\sigma' \restrict vdom(\sigma_S)$ in Subsection \ref{expressions}, and note that $\sigma_S \star \sigma'$ retains the same domain as $\sigma_S$.

The notations explained in the previous paragraph will be used for presenting several goal transformation rules in the next subsections. The formal definition for the other three possible choices of $\cdom$ is completely analogous.
In the rest of the paper, we will restrict our attention to so-called {\em admissible goals} $G$ that arise from initial goals through the iterated application of goal transformation rules and enjoy the goal invariant properties just described.
\vspace*{-.2cm}

\subsection{Constrained Lazy Narrowing Rules}\label{lnr}


The core of our cooperative goal solving calculus $\clnc{\ccdom}$
consists of the goal sol\-ving rules displayed in Table
\ref{table3}. Roughly speaking, these rules model the behaviour of
constrained lazy narrowing ignoring domain cooperation and solver
invocation. They have been adapted from \cite{LMR04} and can be
classified as follows: The first four rules encode unification
transformations similar to those found in the $\herbrand$ $sts$ (see
Subsection \ref{hdom}) and other related formalisms; rule {\bf EL}
discharges unneeded suspensions, rule {\bf DF} (presented in two
cases in order to optimize the $k = 0$ case) applies program rules
to deal with  calls to program defined functions; rule {\bf PC}
transforms demanded calls to primitive functions into atomic
constraints that are placed in the pool; and rule {\bf FC}, working
in interplay with {\bf PC}, transforms the atomic constraints in the
pool into a flattened form consisting  of a conjunction of atomic
primitive constraints with new existential variables.


\begin{table}[h!]
\begin{center}
\begin{tabular}{p{12.5cm}}
\hline\\
{\scriptsize {\bf DC DeComposition}}\vspace*{0.15cm}

{\scriptsize $\exists \overline{U}.$ $h~\overline{e_m}$ $\to$
$h~\overline{t_m},~P$ $\Box$ $C$ $\Box$ $M$ $\Box$ $H$ $\Box$ $F$
$\Box$ $R$ $\vdash\!\!\vdash_{\bf DC}$ $\exists \overline{U}.$
$\overline{e_m \to t_m},~P$ $\Box$ $C$ $\Box$ $M$ $\Box$ $H$ $\Box$
$F$ $\Box$ $R$}\\\\

{\scriptsize {\bf CF Conflict Failure}}\vspace*{0.15cm}

{\scriptsize $\exists \overline{U}.$ $e$ $\to$ $t,P$ $\Box$ $C$
$\Box$ $M$ $\Box$ $H$ $\Box$ $F$ $\Box$ $R$ $\vdash\!\!\vdash_{\bf
CF}$ $\blacksquare$}\vspace*{0.1cm}

{\scriptsize If $e$ is rigid and passive, $t$ $\notin$ $\mathcal
{V}\!\!ar$, $e$ and $t$ have conflicting roots}.\\\\

{\scriptsize {\bf SP Simple Production}}\vspace*{0.15cm}

{\scriptsize $\exists \overline{U}.$ $s$ $\to$ $t,~P$ $\Box$ $C$
$\Box$ $M$ $\Box$ $H$ $\Box$ $F$ $\Box$ $R$ $\vdash\!\!\vdash_{\bf
SP}$ $\exists \overline{U}'.$ $(P$ $\Box$ $C$ $\Box$ $M$ $\Box$ $H$
$\Box$ $F$ $\Box$ $R)@_{\herbrand}\sigma'$}\vspace*{0.1cm}

{\scriptsize  If $s = X \in \var$, $t \notin \var$, $\sigma' = \{X \mapsto t\}$ and  $\overline{U'} = \overline{U}$
or else $s \in Pat_{\mathcal{C}}$, $t = X \in \var$, $\sigma' = \{X \mapsto s\}$ and  $\overline{U}'=\overline{U}\backslash \{X\}$}.\\\\

{\scriptsize {\bf IM IMitation}\vspace*{0.15cm}

{\scriptsize $\exists X,\overline{U}.$ $h~\overline{e_m}$$\to$$X,~P$
$\Box$ $C$ $\Box$ $M$ $\Box$ $H$ $\Box$ $F$ $\Box$ $R$
$\vdash\!\!\vdash_{\bf IM}$} $\exists \overline{X_m},\overline{U}.$
$(\overline{{e_m}{\to}{X_m}},~P$ $\Box$ $C$ $\Box$ $M$ $\Box$ $H$
$\Box$ $F$ $\Box$ $R)\sigma'$}\vspace*{0.1cm}

{\scriptsize If $h~\overline{e_m}$ $\notin$ $Pat_{\mathcal{C}}$ is
passive, $X \in odvar(G)$ and $\sigma' = \{X \mapsto h~\overline{X_m}\}$.}\\\\

{\scriptsize {\bf EL ELimination}}\vspace*{0.15cm}

{\scriptsize $\exists X,\overline{U}.$ $e \to X,~P$ $\Box$ $C$
$\Box$ $M$ $\Box$ $H$ $\Box$ $F$ $\Box$ $R$ $\vdash\!\!\vdash_{\bf
EL}$ $\exists \overline{U}.$ $P$ $\Box$ $C$ $\Box$ $M$ $\Box$ $H$
$\Box$ $F$ $\Box$ $R$}\vspace*{0.1cm}

{\scriptsize If $X$ does not occur in the rest of the goal.}\\\\

{\scriptsize {\bf DF Defined Function}}\vspace*{0.15cm}

{\scriptsize $\exists \overline{U}.$ $f~\overline{e_n}$ $\to$ $t,~P$
$\Box$ $C$ $\Box$ $M$ $\Box$ $H$ $\Box$ $F$ $\Box$ $R$
$\vdash\!\!\vdash_{\bf DF_{\emph{f}}}$}

{\scriptsize
$~~~~~~~~~~~~~~~~~~~~~~~~~~~~~~~~~~~~~~~~~~~~~~~~~~~~~~~~~~~~~~\exists
\overline{Y},\overline{U}.$ $\overline{e_n \to t_n},$ $r \to t,~P$
$\Box$ $C',~C$ $\Box$ $M$ $\Box$ $H$ $\Box$ $F$ $\Box$
$R$}\vspace*{0.1cm}

{\scriptsize If $f$ $\in$ $DF^n$, $t$ $\notin$ $\mathcal {V}\!\!ar$
or $t \in odvar(G)$ and $Rl : f~\overline{t_n} \to r$ $\Leftarrow$ $C'$ is a fresh variant of
a rule in $\mathcal{P}$, with $\overline{Y} = var(Rl)$ new variables.} \\\\

{\scriptsize $\exists \overline{U}.$
$f~\overline{e_n}\overline{a_k}$ $\to$ $t,~P$ $\Box$ $C$ $\Box$ $M$
$\Box$ $H$ $\Box$ $F$ $\Box$ $R$ $\vdash\!\!\vdash_{\bf
DF_{\emph{f}}}$}

{\scriptsize ~~~~~~~~~~~~~~~~~~~~~~~~~~~~~~~~~~~~~~~~~~~$\exists
X,\overline{Y},\overline{U}.$ $\overline{e_n \to t_n},$ $r \to X,$
$X~\overline{a_k} \to t,~P$ $\Box$ $C',~C$ $\Box$ $M$ $\Box$ $H$
$\Box$ $F$ $\Box$ $R$}\vspace*{0.1cm}

{\scriptsize If $f$ $\in$ $DF^n$ $(k$ $>$ $0),$ $t$ $\notin$
$\mathcal {V}\!\!ar$ or $t \in odvar(G)$ and
$Rl : f~\overline{t_n} \to r$ $\Leftarrow$ $C'$ is a fresh
variant of a rule in $\mathcal{P}$,
with $\overline{Y} = var(Rl)$ and $X$ new variables.}\\\\

{\scriptsize {\bf PC Place Constraint}}\vspace*{0.15cm}

{\scriptsize $\exists \overline{U}.$ $p~\overline{e_n}$ $\to$ $t,~P$
$\Box$ $C$ $\Box$ $M$ $\Box$ $H$ $\Box$ $F$ $\Box$ $R$
$\vdash\!\!\vdash_{\bf PC}$ $\exists \overline{U}.$ $~P$ $\Box$
$p~\overline{e_n}$ $\to$! $t,~C$ $\Box$ $M$ $\Box$ $H$ $\Box$ $F$
$\Box$ $R$}\vspace*{0.1cm}

{\scriptsize If $p \in PF^n$ and $t \notin \var$ or $t \in odvar(G)$.}\\\\

{\scriptsize {\bf FC Flatten Constraint}}\vspace*{0.15cm}

{\scriptsize $\exists \overline{U}.$ $P$ $\Box$ $p~\overline{e_n}$
$\to!$ $t,~C$ $\Box$ $M$ $\Box$ $H$ $\Box$ $F$ $\Box$ $R$
$\vdash\!\!\vdash_{\bf FC}$} \\
{\scriptsize~~~~~~~~~~~~~~~~~~~~~~~~~~~~~~~~~~~~~~~~~~~~~~~~$\exists
\overline{V}_m, \overline{U}.$ $\overline{a_m \to V_m},$ $P$ $\Box$
$p~\overline{t}_n$ $\to!$ $t,~C$ $\Box$ $M$ $\Box$ $H$ $\Box$ $F$
$\Box$ $R$}\vspace*{0.1cm}

{\scriptsize If $p$ $\in$ $PF^n$, some $e_i$ $\notin$
$Pat_{\mathcal{C}}$, $\overline{a_m}$ ($m$ $\leq$ $n$) are those
$e_i$ which are not patterns, $\overline{V_m}$ are new variables,
and $p~\overline{t_n}$ is obtained from $p~\overline{e_n}$ by
replacing each $e_i$ which is not a pattern by $V_i$.}\\\\
\hline
\end{tabular}
\end{center}
\caption{Rules for Constrained Lazy Narrowing}\label{table3}
\end{table}


The behaviour of the main rules in Table \ref{table3} will be illustrated in Subsection \ref{exampletriangle}.
Example \ref{exCFlat}  below focuses on the transformation rules {\bf PC} and {\bf FC}.
Their iterated application flattens the atomic $\rdom$-constraint  {\tt (RX + 2*RY)*RZ <= 3.5} into a conjunction of four
atomic primitive $\rdom$-constraints involving three new existential variables,
that are placed in the constraint pool. Note that \cite{LMR04} and other previous related calculi
also include rules that can be used to achieve constraint flattening, but the resulting atomic primitive constraints are
placed in a constraint store. In our present setting, they are kept in the pool in order that the
domain cooperation rules described in the next subsection can process them.

\begin{example}[Constraint Flattening]\label{exCFlat}
\vspace*{.2cm}

{\scriptsize
\noindent $~~~~~~~~~\Box$ $\underline{(RX+2*RY)*RZ \textnormal{\em ~{\tt <= }} 3.5}$ $\Box$$\Box$$\Box$$\Box$ $\vdash\!\!\vdash_{\bf FC}$\\

\noindent $~~~~~~~~~\exists RA.$ $\underline{(RX+2*RY)*RZ~\to~RA}$
$\Box$ $RA$ {\tt<=} $3.5$
$\Box$$\Box$$\Box$$\Box$ $\vdash\!\!\vdash_{\bf PC}$\\

\noindent $~~~~~~~~~\exists RA.$ $\Box$
$\underline{(RX+2*RY)*RZ~\to!~RA},$ $RA$ {\tt <=} $3.5$
$\Box$$\Box$$\Box$$\Box$ $\vdash\!\!\vdash_{\bf FC}$\\

\noindent $~~~~~~~~~\exists RB,$ $RA.$$\underline{RX+2*RY\to RB}$
$\Box$ $RB*RZ$$\to!$$RA,RA$ {\tt <=} $3.5$ $\Box$$\Box$$\Box$$\Box$ $\vdash\!\!\vdash_{\bf PC}$\\

\noindent $~~~~~~~~~\exists RB, RA.$ $\Box$
$\underline{RX+2*RY~\to!~RB},$ $RB*RZ$
$\to!$ $RA,$ $RA$ {\tt <=} $3.5$ $\Box$$\Box$$\Box$$\Box$ $\vdash\!\!\vdash_{\bf FC}$\\

\noindent $~~~~~~~~~\exists RC, RB, RA.$ $\underline{2*RY \to RC}$ $\Box$ $RX+RC$$\to!$$RB,$ $RB*RZ$$\to!$$RA,$ $RA$ {\tt <=} $3.5$ $\Box$$\Box$$\Box$$\Box$ $\vdash\!\!\vdash_{\bf PC}$\\

\noindent $~~~~~~~~~\exists RC, RB, RA.$ $\Box$ $2*RY \to !RC,$
$RX+RC$$\to!$$RB,$ $RB*RZ$$\to!$$RA,$ $RA$ {\tt <=} $3.5$
$\Box$$\Box$$\Box$$\Box$}

\end{example}


Note that   suspensions $e \to X$ can be discharged by rule {\bf EL} in case that $X$
does not occur in the rest of the goal. Otherwise, they must wait until $X$ gets bound to a non-variable 
pattern or becomes obviously demanded, and then they can be processed by using either rule {\bf DF} 
or rule {\bf PC}, according to the syntactic  form of $e$. Moreover, all the substitutions produced by the transformations bind variables $X$ to patterns $t$, standing for computed values that are shared by all the occurrences of $t$ in the current goal. In this way, the goal transformation rules encode a lazy narrowing strategy.

\subsection{Domain Cooperation Rules} \label{dcr}

This subsection presents the goal transformation rules in $\cclnc{\ccdom}$
which take care of domain cooperation. The core of the subsection deals
with bridges and projections. A few more ad hoc cooperation rules are
presented at the end of the subsection.


Given a goal $G$ whose pool $C$ includes an atomic primitive constraint $\pi \in APCon_{\fd}$
and whose mediatorial store $M$  includes a set of bridges $B_M$,  we will consider three
possible goal transformations intended to convey useful information from $\pi$
to the $\rdom$-solver:
\begin{itemize}
\item
To compute new bridges $bridges^{\mathcal{FD} \to \mathcal{R}}(\pi,B_M)$ to add to $M$,
by means of a {\em bridge generation} operation  $bridges^{\mathcal{FD} \to \mathcal{R}}$
defined to this  purpose.
\item
To compute projected $\rdom$-constraints $proj^{\mathcal{FD} \to \mathcal{R}}(\pi,B_M)$ to be added to $R$,
by means of a {\em projection} operation  $proj^{\mathcal{FD} \to \mathcal{R}}$
defined to this  purpose.
\item
To place $\pi$ into the $\fd$ store $F$.
\end{itemize}

Similar goal transformations  based on two operations $bridges^{\mathcal{R} \to \mathcal{FD}}$
and $proj^{\mathcal{R} \to \mathcal{FD}}$ can be used to convey useful information from a
primitive atomic constraint $\pi \in PCon_{\rdom}$ to the $\fd$-solver.
Rules {\bf SB}, {\bf PP} and {\bf SC} in  Table \ref{table4} formalize these transformations,
while tables \ref{table2} and \ref{table5} give an effective
specification of the bridge generation and projection operations.

\begin{table}[h]
\begin{center}
\begin{tabular}{p{12.5cm}} \hline\\
{\scriptsize {\bf SB Set Bridges}}\vspace*{0.4mm}

{\scriptsize $~~~~~~\exists \overline{U}.$ $P$ $\Box$ $\pi, ~C$
$\Box$ $M$ $\Box$ $H$ $\Box$ $F$ $\Box$ $R$ $\vdash\!\!\vdash_{\bf
SB}$} {\scriptsize$\exists\overline{V}',\overline{U}.$ $P$ $\Box$
$\pi, ~C$ $\Box$ $M'$ $\Box$ $H$ $\Box$ $F$ $\Box$
$R$}\\\vspace*{0.4mm}
{\scriptsize If $\pi$ is a primitive atomic constraint and either
(i) or (ii) holds, where}
{\scriptsize \begin{enumerate}
\item [(i)] $\pi$ is a proper $\mathcal{FD}$-constraint or else an extended $\herbrand$-constraint such that
$M \vdash \pi$ {\tt in} $\fd$, and $M' = B', M$, where $\exists
\overline{V'}\, B' = bridges^{\mathcal{FD} \to \mathcal{R}}( \pi,
B_M) \neq \emptyset$.
\item [(ii)] $\pi$ is a proper $\mathcal{R}$-constraint or else an extended $\herbrand$-constraint such that
$M \vdash \pi$ {\tt in} $\rdom$, and $M' = B', M$, where $\exists
\overline{V'}\, B' = bridges^{\mathcal{R} \to \mathcal{FD}}(\pi,
B_M) \neq \emptyset$.
\end{enumerate}}\\

{\scriptsize {\bf PP Propagate Projections}}\vspace*{0.4mm}

{\scriptsize $~~~~~~\exists \overline{U}.$ $P$ $\Box$ $\pi, ~C$
$\Box$ $M$ $\Box$ $H$ $\Box$ $F$ $\Box$ $R$ $\vdash\!\!\vdash_{\bf
PP}$} {\scriptsize$\exists\overline{V'},\, \overline{U}.$ $P$ $\Box$
$\pi,~C$ $\Box$ $M$ $\Box$ $H$ $\Box$ $F'$ $\Box$
$R'$}\vspace*{0.4mm}

{\scriptsize If $\pi$ is a primitive atomic constraint and either
(i) or (ii) holds, where}

{\scriptsize \begin{enumerate}

\item [(i)] $\pi$ is a proper $\mathcal{FD}$-constraint or else an extended $\herbrand$-constraint such that
$M \vdash \pi$ {\tt in} $\fd$, $\exists \overline{V'}\, \Pi' =
proj^{\mathcal{FD} \to \mathcal{R}}(\pi,B_M) \neq \emptyset$, $F' =
F$, and $R' = \Pi',R$, or else,
\item [(ii)] $\pi$ is a proper $\mathcal{R}$-constraint or else an extended $\herbrand$-constraint such that
$M \vdash \pi$ {\tt in} $\rdom$,  $\exists \overline{V'}\, \Pi' =
proj^{\mathcal{R} \to \mathcal{FD}}(\pi,B_M) \neq \emptyset$, $F' =
\Pi', F$, and $R' = R$.

\end{enumerate}}\\

{\scriptsize {\bf SC Submit Constraints}}\vspace*{0.4mm}

{\scriptsize $~~~~~~\exists \overline{U}.$ $P$ $\Box$ $\pi, ~C$
$\Box$ $M$ $\Box$ $H$ $\Box$ $F$ $\Box$ $R$ $\vdash\!\!\vdash_{\bf
SC}$ $\exists \overline{U}.$ $P$ $\Box$ $C$ $\Box$ $M'$ $\Box$ $H'$
$\Box$ $F'$ $\Box$ $R'$}\\\vspace*{0.4mm}
{\scriptsize If $\pi$ is a primitive atomic constraint and one of
the following cases applies:}

{\scriptsize \begin{enumerate}

\item [(i)] $\pi$ is a $\mdom$-constraint, $M' = \pi, M$, $H' = H$, $F' = F$, and $R' = R$.
\item [(ii)] $\pi$ is an extended $\herbrand$-constraint such that neither $M \vdash \pi$ {\tt in} $\fd$
nor $M \vdash \pi$ {\tt in} $\rdom$, $M' = M$, $H' = \pi, H$, $F' =
F$, and $R' = R$.
\item [(iii)] $\pi$ is a proper $\mathcal{FD}$-constraint or else an extended $\herbrand$-constraint such that
$M \vdash \pi$ {\tt in} $\fd$,  $M' = M$, $H' = H$, $F' = \pi, F$,
and $R' = R$.
\item [(iv)] $\pi$ is a proper $\mathcal{R}$-constraint or else an extended $\herbrand$-constraint such that
$M \vdash \pi$ {\tt in} $\rdom$,  $M' = M$, $H' = H$, $F' = F$, and
$R' = \pi, R$.
\end{enumerate}}\\
\hline
\end{tabular}
\end{center}
\caption{Rules for Bridges and Projections}\label{table4}
\end{table}

The formulation of  {\bf SB}, {\bf PP} and {\bf SC} in Table \ref{table4} relies on the identification
of certain atomic primitive Herbrand constraints $\pi$ as $\fd$-specific or $\rdom$-specific,
as indicated by the notations $M \vdash \pi$ {\tt in} $\fd$
and $M \vdash \pi$ {\tt in} $\rdom$, previously explained in Subsection \ref{ourcdom}.
The notation $\Pi, S$ is used at several places to indicate the new store
obtained by adding the set of constraints $\Pi$ to the constraints
within store $S$. The notation $\pi, S$ (where $\pi$ is a single constraint) must be understood similarly.
In practice, {\bf SB}, {\bf PP} and {\bf SC}  are best applied in this  order.
Note that {\bf PP} places the projected constraints in their corresponding stores,
while constraints in the pool that are not useful anymore for computing additional bridges or projections
will be eventually placed into their stores by means of transformation {\bf SC}.

The functions $bridges^{\cdom \to \cdom'}$ and $proj^{\cdom \to \cdom'}$
are specified in Table \ref{table2} for the case $\cdom = \fd,\, \cdom' = \rdom$ and
in Table \ref{table5} for the case $\cdom = \rdom,\, \cdom' = \fd$.
Note that the primitive {\tt \#/} is not  considered in Table \ref{table2} because
integer division constraints cannot be projected into  real division constraints.
The notations $\lceil a \rceil$ (resp. $\lfloor a \rfloor$) used in Table \ref{table5}
 stand for the least integer upper bound (resp. the greatest integer lower bound) of $a \in \mathbb{R}$.
 Constraints $t_1$ {\tt >} $t_2$, $t_1$ {\tt >=} $t_2$
 are not explicitly considered in Table \ref{table5};
 they are  treated  as $t_2$ {\tt <} $t_1$, $t_2$ {\tt <=} $t_1$, respectively.
 In tables \ref{table2} and \ref{table5}, the existential quantification of the
 new variables $\overline{V'}$ is left implicit, and  results displayed as an empty set of
constraints must be read as an empty (and thus trivially true) conjunction.

\vspace*{-.3cm}
\begin{table}[h]
\begin{center}
\begin{tabular}{p{3.cm} p{4.cm} p{4.2cm}}
\hline {\small ~~~~~~~~~~$\pi$} & {\small
~~~~~$bridges^{\mathcal{FD} \to \mathcal{R}}(\pi, B)$} &
{\small ~~~~~~~~$proj^{\mathcal{FD} \to \mathcal{R}}(\pi, B)$}\\
\hline {\scriptsize {\tt domain} [$X_1, \ldots, X_n]$ $a$ $b$} &
{\scriptsize\{$X_i$ {\tt \#==} $RX_i~|$ $1 \leq i \leq n$,  $X_i$
has no bridge in $B$ and $RX_i$ new\}} &{\scriptsize\{$a$ {\tt <=}
$RX_i$, $RX_i$ {\tt <=} $b$ ~ $|$~ $1 \leq i \leq n$ and
$(X_i$ {\tt \#==} $RX_i)$ $\in$ $B$\}}\\
\hline {\scriptsize {\tt belongs} $X$ [$a_1, \ldots, a_n$] }  &
{\scriptsize \{$X$ {\tt \#==} $ RX~|$ $X$ has no bridge in $B$ and
$RX$ new\} }& {\scriptsize  \{$min(a_1,\dots, a_n)$ {\tt <=} $RX$,
$RX$ {\tt <=} $max(a_1,\dots, a_n)$ $|$ $1 \leq i \leq n$ and $(X$ {\tt \#==} $RX)$ $\in$ $B$\} }\\
\hline
 {\scriptsize $t_1$ {\tt \#<} $t_2$} \\ {\scriptsize(resp. {\tt \#<=}, {\tt \#>}, {\tt \#=>})}
& {\scriptsize \{$X_i$ {\tt \#==} $RX_i$~ $|$ ~$1 \leq i \leq 2$,
$t_i$ is a variable $X_i$ with no bridge in $B$, and $RX_i$ new\} }&
{\scriptsize \{$t^{\mathcal{R}}_1$ {\tt <} $t^{\mathcal{R}}_2$ $|$
For $1 \leq i \leq 2$: Either  $t_i$ is an integer constant $n$ and
$t^{\mathcal{R}}_i$ is the integral real $n$, or else $t_i$ is a
variable $X_i$ with $(X_i$ {\tt \#==} $RX_i)$ $\in$ $B$,
and $t^{\mathcal{R}}_i$ is $RX_i$\}}\\
\hline
 {\scriptsize $t_1$ {\tt ==} $t_2$} & {\scriptsize \{$X$ {\tt
\#==} $RX$ $|$ either $t_1$ is an integer constant and $t_2$ is a
variable $X$ with no bridges in $B$ (or vice versa) and $RX$ is
new\}} &
 {\scriptsize \{$t^{\mathcal{R}}_1$ {\tt ==} $t^{\mathcal{R}}_2$ $|$ For
$1 \leq i \leq 2$:  $t^{\mathcal{R}}_i$ is determined as  in the {\tt \#<}  case\} }\\
\hline
 {\scriptsize $t_1$ {\tt /=} $t_2$} & {\scriptsize \{$X$ {\tt
\#==} $RX$ $|$ either $t_1$ is an integer constant and $t_2$ is a
variable $X$ with no bridges in $B$ (or vice versa) and $RX$ is
new\}} &
 {\scriptsize \{$t^{\mathcal{R}}_1$ {\tt /=} $t^{\mathcal{R}}_2$ $|$ For
$1 \leq i \leq 2$:  $t^{\mathcal{R}}_i$ is determined as  in the {\tt \#<}  case\} }\\
\hline
 {\scriptsize  $t_1$ {\tt \#+} $t_2$ $\to!$ $t_3$} \\ {\scriptsize(resp. {\tt \#-}, {\tt \#*})
}& {\scriptsize  \{$X_i$ {\tt \#==} $RX_i$~$|$ $1 \leq i \leq 3$,
$t_i$ is a variable $X_i$ with no bridge in $B$ and $RX_i$ new\} }&
{\scriptsize \{$t^{\mathcal{R}}_1$ {\tt +} $t^{\mathcal{R}}_2$
$\to!$ $t^{\mathcal{R}}_3$ $|$ For $1 \leq i \leq 3$:
$t^{\mathcal{R}}_i$ is determined as  in the {\tt \#<}  case\}}\\
 \hline
\end{tabular}
\end{center}
\caption{Computing Bridges and Projections from $\mathcal{FD}$ to
$\mathcal{R}$}\label{table2}
\end{table}
\vspace*{-.3cm}

The next result states some basic properties of  $bridges^{\cdom \to \cdom'}$ and
$proj^{\cdom \to \cdom'}$. The easy proof is omitted.

\begin{proposition} [Properties of Bridges and Projections between $\fd$ and $\rdom$] \label{propBP}
Let $\cdom$ and $\cdom'$ be chosen as $\fd$ and $\rdom$, or vice
versa. Then:
\begin{enumerate}
\item
$bridges^{\cdom \to \cdom'}(\pi,B)$ and $proj^{\cdom \to \cdom'}(\pi,B)$
make sense for any atomic primitive constraint $\pi$
which is either $\cdom$-proper or extended Herbrand and $\cdom$-specific,
and for any finite set $B$ of bridges.
\item
$bridges^{\cdom \to \cdom'}(\pi,B)$ returns a possibly empty finite set $B'$ of new bridges
involving new variables $\overline{V'}$.
In particular, $bridges^{\cdom \to \cdom'}(\pi,B) = \emptyset$
is assumed whenever Tables \ref{table2} and \ref{table5}
do not include any row covering $\pi$.
The {\em completeness condition}
$WTSol_{\mathcal{C}}(\pi \wedge B) \subseteq
WTSol_{\mathcal{C}}(\exists \overline{V'}(\pi \wedge B \wedge B'))$
holds, where $B$ and $B'$ are interpreted as conjunctions.
Note that the {\em correctness condition}
$Sol_{\mathcal{C}}(\pi \wedge B) \supseteq
Sol_{\mathcal{C}}(\exists\overline{V'}(\pi \wedge B \wedge B'))$
also holds trivially.
\item
$proj^{\cdom \to \cdom'}(\pi,B)$ returns a finite set $\Pi' \subseteq APCon_{\cdom'}$
of atomic primitive $\cdom'$-constraints
involving new variables $\overline{V'}$.
In particular, $proj^{\cdom \to \cdom'}(\pi,B) = \emptyset$
is assumed whenever Tables \ref{table2} and \ref{table5}
do not include any row covering $\pi$.
The {\em completeness condition}
$WTSol_{\mathcal{C}}(\pi \wedge B) \subseteq
WTSol_{\mathcal{C}}(\exists \overline{V'}(\pi \wedge B \wedge \Pi'))$
holds, where $B$ and $\Pi'$ are interpreted as conjunctions.
Note that the {\em correctness condition}
$Sol_{\mathcal{C}}(\pi \wedge B) \supseteq
Sol_{\mathcal{C}}(\exists \overline{V'}(\pi \wedge B \wedge \Pi'))$
also holds trivially.
\end{enumerate}
\end{proposition}

\vspace*{-.3cm}
\begin{table}[h]
\begin{center}
\begin{tabular}{p{2.0cm} p{4.0cm} p{5.0cm}}
\hline {\scriptsize ~~~~~~$\pi$} & {\scriptsize
~~~~~~~~$bridges^{\mathcal{R} \to \mathcal{FD}}(\pi, B)$} &
{\scriptsize ~~~~~~~~~~~~~~~$proj^{\mathcal{R} \to \mathcal{FD}}(\pi, B)$}\\
\hline {\scriptsize $RX$ {\tt <} $RY$} & {\scriptsize
~~~$\emptyset$~~({\em no bridges are created})}
&{\scriptsize \{$X$ {\tt \#<} $Y$ $|$ $(X$ {\tt \#==} $RX)$,$(Y$ {\tt \#==} $RY)\in B$\}}\\
\hline {\scriptsize $RX$ {\tt <} $a$} & {\scriptsize
~~~$\emptyset$~~({\em no bridges are created})}
&  {\scriptsize\{$X$ {\tt \#<} $\lceil a \rceil$ $|$ $a$ $\in$ $\mathbb{R}$, $(X$ {\tt \#==} $RX)\in B$\}}\\
\hline {\scriptsize $a$ {\tt <} $RY$} & {\scriptsize
~~~$\emptyset$~~({\em no bridges are created})}
&  {\scriptsize\{$\lfloor a \rfloor$ {\tt \#<} $Y$ $|$ $a$ $\in$ $\mathbb{R}$, $(Y$ {\tt \#==} $RY) \in B$\}}\\
\hline {\scriptsize $RX$ {\tt <=} $RY$} & {\scriptsize
~~~$\emptyset$~~({\em no bridges are created})}
& {\scriptsize \{$X$ {\tt \#<=} $Y$$|$$(X$ {\tt \#==} $RX)$,$(Y$ {\tt \#==} $RY)\in B$\}}\\
\hline {\scriptsize$RX$ {\tt <=} $a$} & {\scriptsize
~~~$\emptyset$~~({\em no bridges are created})}
& {\scriptsize \{$X$ {\tt \#<=} $\lfloor a \rfloor$ $|$ $a$ $\in$ $\mathbb{R}$, $(X$ {\tt \#==} $RX)\in B$\}}\\
\hline {\scriptsize $a$ {\tt <=} $RY$} & {\scriptsize
~~~$\emptyset$~~({\em no bridges are created})}
&  {\scriptsize\{$\lceil a \rceil$ {\tt \#<=} $Y$ $|$ $a$ $\in$ $\mathbb{R}$, $(Y$ {\tt \#==} $RY) \in B$\}}\\
\hline {\scriptsize $t_1$ {\tt ==} $t_2$} & {\scriptsize \{$X$ {\tt
\#==} $RX$ $|$ either $t_1$ is an integral real constant and $t_2$
is a variable $RX$ with no bridges in $B$ (or vice versa) and $X$ is
new\}} &
 {\scriptsize \{$t^{\mathcal{FD}}_1$ {\tt ==} $t^{\mathcal{FD}}_2$ $|$ For
$1 \leq i \leq 2$: Either $t_i$ is an integral real constant $n$ and
$t^{\mathcal{FD}}_i$ is the integer $n$, or else  $t_i$ is a
variable $RX_i$ with
$(X_i$ {\tt \#==} $RX_i)\in B$, and $t^{\mathcal{FD}}_i$ is $X_i$\} }\\
\hline {\scriptsize $t_1$ {\tt /=} $t_2$} & {\scriptsize
~~~$\emptyset$~~({\em no bridges are created})} &
 {\scriptsize \{$t^{\mathcal{FD}}_1$ {\tt /=} $t^{\mathcal{FD}}_2$ $|$ For
$1 \leq i \leq 2$: Either $t_i$ is an integral real constant $n$ and
$t^{\mathcal{FD}}_i$ is the integer $n$, or else  $t_i$ is a
variable $RX_i$ with
$(X_i$ {\tt \#==} $RX_i)\in B$, and $t^{\mathcal{FD}}_i$ is $X_i$\} }\\
\hline {\scriptsize $t_1$ {\tt +} $ t_2$ $\to!$ $t_3$  (resp. {\tt
-}, {\tt *})} & {\scriptsize\{$X$ {\tt \#==} $RX$ $|$ $t_3$ is a
variable $RX$ with no bridge in $B$, $X$ new,  for $1 \leq i \leq
2$, $t_i$ is either an integral real constant or  a variable $RX_i$
with bridge $(X_i $ {\tt \#==} $ RX_i) \in B$\}} &
 {\scriptsize\{$t^{\mathcal{FD}}_1$ {\tt \#+} $t^{\mathcal{FD}}_2$$\to!$ $t^{\mathcal{FD}}_3$ $|$
For $1 \leq i \leq 3$: $t^{\mathcal{FD}}_i$ is determined as in the previous case\}}\\
\hline
 {\scriptsize$t_1$ {\tt /} $t_2$ $\to!$ $t_3$} & {\scriptsize ~~~$\emptyset$~~({\em no bridges are
created})} &
  {\scriptsize\{$t^{\mathcal{FD}}_2$ {\tt \#*}
$t^{\mathcal{FD}}_3$ $\to!$ $t^{\mathcal{FD}}_1$ $|$
For $1 \leq i \leq 3$ is determined as in the previous case\}}\\
 \hline
\end{tabular}
\end{center}
\caption{Computing Bridges and Projections from $\mathcal{R}$ to
$\mathcal{FD}$}\label{table5}
\end{table}
\vspace*{-.3cm}


Example \ref{exCBP} below illustrates the operation of the goal
transformation rules from Table \ref{table4} for computing bridges
and projections with the help of the functions speci\-fied in Tables
\ref{table2} and \ref{table5}.

\begin{example} [Computation of Bridges and Projections] \label{exCBP}
\vspace*{.2cm}

{\scriptsize
\noindent $\Box$ $\underline{(RX+2*RY)*RZ\textnormal{~\em {\tt <=}
3.5}}$ $\Box$ $X$ {\tt \#==} $RX,$ $Y$ {\tt \#==} $RY,$ $Z$ {\tt
\#==} $RZ$
$\Box$$\Box$$\Box$ $\vdash\!\!\vdash_{\bf FC^{3},PC^{3}}$\\
$\exists RC,RB,RA.$ $\Box$ $\underline{2*RY~\to!~RC},$
$\underline{RX+RC~\to!~RB},$ $\underline{RB*RZ~\to!~RA},$ $RA$ {\tt
<=} $3.5$ $\Box$\\
$~~~~~~~~~~~~~~~~~~~~~X$ {\tt \#==} $RX,$ $Y$ {\tt \#==} $RY,$ $Z$
{\tt \#==} $RZ$ $\Box$$\Box$$\Box$ $\vdash\!\!\vdash_{\bf SB^{3}}$\\
$\exists C,B,A,RC,RB,RA.$ $\Box$ $\underline{2*RY~\to!~RC},$
$\underline{RX+RC~\to!~RB},$ $\underline{RB*RZ~\to!~RA},$
$\underline{RA\textnormal{~\em {\tt <= }} 3.5}$ $\Box$\\
$~~~~~~~~~~~~~~~~~~~~~C$ {\tt \#==} $RC,$ $B$ {\tt \#==} $RB,$ $A$
{\tt \#==} $RA,$ $X$ {\tt \#==} $RX,$ $Y$ {\tt \#==} $RY,$ $Z$ {\tt
\#==} $RZ$ $\Box$$\Box$$\Box$ $\vdash\!\!\vdash_{\bf PP^{4}}$\\
$\exists C,B,A,RC,RB,RA.$ $\Box$ $\underline{2*RY~\to!~RC},$
$\underline{RX+RC~\to!~RB},$ $\underline{RB*RZ~\to!~RA},$
$\underline{RA\textnormal{~\em {\tt <= }} 3.5}$ $\Box$\\
$~~~~~~~~~~~~~~~~~~~~~C$ {\tt \#==} $RC,$ $B$ {\tt \#==} $RB,$ $A$
{\tt \#==} $RA,$ $X$ {\tt \#==} $RX,$ $Y$ {\tt \#==}
$RY,$ $Z$ {\tt \#==} $RZ$ $\Box$$\Box$\\
$~~~~~~~~~~~~~~~~~~~~~2$ {\tt \#*} $Y$ $\to!$ $C,$ $X$ {\tt \#+} $C$
$\to!$ $B,$ $B$ {\tt \#*} $Z$ $\to!$ $A,$ $A$ {\tt \#<=} $3$ $\Box$ $\vdash\!\!\vdash_{\bf SC^{4}}$\\
$\exists C,B,A,RC,RB,RA.$ $\Box$$\Box$ $C$ {\tt \#==} $RC,$ $B$ {\tt
\#==} $RB,$ $A$ {\tt \#==} $RA,$ $X$ {\tt \#==} $RX,$ $Y$ {\tt \#==}
$RY,$ $Z$ {\tt \#==} $RZ$ $\Box$$\Box$\\
$~~~~~~~~~~~~~~~~~~~~~2$ {\tt \#*} $Y$ $\to!$ $C,$ $X$ {\tt \#+} $C$
$\to!$ $B,$ $B$ {\tt \#*} $Z$ $\to!$ $A,$ $A$ {\tt \#<=} $3$ $\Box$\\
$~~~~~~~~~~~~~~~~~~~~~2$ {\tt *} $RY \to!~RC,$ $RX$ {\tt +} $RC
\to!~RB,$ $RB$ {\tt *} $RZ \to!~RA,$ $RA$ {\tt <=}
$3.5$}\vspace*{0.25cm}
\end{example}

Note that the initial goal in this current example  is an extension
of the initial goal in Example \ref{exCFlat}. The first six steps of
the current computation are similar to those in Example
\ref{exCFlat}, taking care of flattening the $\rdom$-constraint {\tt
(RX+2*RY)*RZ <= 3.5}. The subsequent steps use the transformation
rules from Table \ref{table4} until no further bridges and
projections can be computed and no constraints remain in the
constraint pool.\vspace*{0.25cm}

We have borrowed the projection idea from Hofstedt's work, see e.g.
\cite{Hofstedt:phd-thesis-2001,HP07}, but our proposal of using
bridges to compute projections is a novelty. In Hofstedt's approach,
projecting constraints from one domain into another depends on
common variables present in both stores. In our approach,
well-typing requirements generally prevent one and the same variable
to occur in constraints from different domains. In order to improve
the opportunities for computing projections, our cooperative goal
solving calculus $\cclnc{\ccdom}$ provides the goal transformation
rule {\bf SB} for creating new bridges during the computations. Some
other differences between $\cclnc{\ccdom}$ and the cooperative
computation model proposed by Hofstedt et al. are as follows:
\begin{itemize}
\item
All the projections presented in this paper return just one
$\exists\overline{V'}\, \Pi'$. In Hofs\-tedt's terminology, such
projections are called {\em weak}, while projections returning
disjunctions $\bigvee_{k=1}^{l}\exists\overline{V'}_k \Pi'_k$ with
$l > 1$ are called {\em strong}.  
The use of strong projections is illustrated in
\cite{HP07} by means of a problem dealing with the computation of resistors 
that have a certain capacity. 
The strong projection used in this example is a finite disjunction 
of conjunctions of the form $X == x \land Y == y$ for various numeric values $x$ and $y$.
Solving this disjunction gives rise to an enumeration of solutions.
In \cite{DBLP:journals/entcs/MartinFHRV07} we have presented
a solution of the resistors problem where an equivalent enumeration of solutions 
can be computed by the $\fd$-solver via backtracking,  without building any strong projection.
This is possible in our framework due to the presence of labeling constraints,
that are not used in the resistor example as presented in \cite{HP07}.
Therefore, strong projections are not necessary for this 
particular example of cooperation  between $\mathcal{FD}$ and $\mathcal{R}$.
Theoretically, strong projections could be useful in other problems, and
rule {\bf PP} in our $\cclnc{\ccdom}$ calculus could be very straightforwardly adapted to work 
with strong projections. However, we decided not to do so because we are not aware of 
any useful extension to extend tables \ref{table2} and \ref{table5} for computing strong projections.
We could find no formulation of practical  procedures for computing projections
in \cite{HP07} and related works, where all projections used in examples are presented in an {\em ad hoc} manner.
\item
Currently, our $\cclnc{\ccdom}$ calculus projects $\mathcal{FD}$
(resp.  $\mathcal{R}$) constraints from the pool $C$ into the
$\mathcal{R}$ store $R$ (resp. $\mathcal{FD}$ store $F$). Hofstedt's
proposal also allows to compute projections from constraints placed into the stores.
In our previous paper \cite{DBLP:journals/entcs/MartinFHRV07},
we have sketched a cooperative goal solving calculus
  where an arbitrary coordination domain was assumed and
 projections could act over the constraints within constraint stores.
 In fact, the resistor problem mentioned in the previous item was solved
 in  \cite{DBLP:journals/entcs/MartinFHRV07} by making essential use of projections 
 that acted over constraints within the $\mathcal{FD}$ and $\mathcal{R}$ stores.
 In the current paper, goal solving is restricted to the coordination domain
 $\ccdom$ $=$ $\mathcal{M}$ $\oplus$ $\herbrand$ $\oplus$ $\fd$ $\oplus$ $\rdom$
 and  projections can be applied only to the constraints placed in the  constraint pool.
 These two limitations correspond to the state of the current $\toy$ implementation.
 In particular, projections acting over stored constraints are not yet handled because
 the current $\toy$ system has no convenient mechanisms for processing the constraint 
 stores handled by the underlying SICStus Prolog.  
\item
Goal solving in $\cclnc{\ccdom}$ enjoys the soundness and completeness properties
presented in Subsection \ref{SC}. In our opinion, these are
more elaborate than the soundness and completeness results
provided in Hofstedt's work.
\end{itemize}

\vspace*{-.3cm}
\begin{table}[h]
\begin{center}
\begin{tabular}{p{11.5cm}} \hline\\

{\scriptsize

\indent {\bf IE Infer Equalities}} \vspace*{0.4mm}

{\scriptsize $~~\exists \overline{U}.$ $P \Box C$ $\Box$ $X$
{\tt \#==} $RX$, $X'$ {\tt \#==} $RX$, $M$ $\Box$ $H$ $\Box$ $F
\Box R$ $\vdash\!\!\vdash_{\bf UB}$}\\

{\scriptsize
~~~~~~~~~~~~~~~~~~~~~~~~~~~~~~~~~$\exists
\overline{U}.$ $P$ $\Box$ $C$ $\Box$ $X$ {\tt \#==} $RX$, $M$
$\Box$ $H$ $\Box$ $X$ {\tt ==} $X', F \Box R$}.\\

{\scriptsize $~~\exists \overline{U}.$ $P \Box C$ $\Box$ $X$
{\tt \#==} $RX$, $X$ {\tt \#==} $RX'$, $M$ $\Box$ $H$ $\Box$ $F
\Box R$ $\vdash\!\!\vdash_{\bf UB}$}\\

{\scriptsize
~~~~~~~~~~~~~~~~~~~~~~~~~~~~~~~~~$\exists
\overline{U}.$ $P$ $\Box$ $C$ $\Box$ $X$ {\tt \#==} $RX$, $M$
$\Box$ $H$ $\Box$ $F \Box RX ${\tt ==}$ RX', R$}.\\ \\

{\scriptsize {\bf ID Infer Disequalities}} \vspace*{0.4mm}

{\scriptsize
$\exists \overline{U}.$ $P \Box C$ $\Box$
$X${\tt \#/==}$u'$, $M$ $\Box$ $H$ $\Box$ $F \Box R$ $\vdash\!\!\vdash_{\bf ID}$
$\exists \overline{U}.$ $P$ $\Box$ $C$ $\Box$ $M$ $\Box$ $H$ $\Box$
$X${\tt /=}$u, F \Box R$}
\vspace*{0.4mm}

{\scriptsize if $u \in \mathbb{Z}$, $u' \in \mathbb{R}$ and $u$ {\tt
\#==}$^{\mathcal{M}}~u' \to true$}. \vspace*{0.4mm}

{\scriptsize $\exists \overline{U}.$ $P \Box C$ $\Box$ $u${\tt
\#/==}$RX$, $M$ $\Box$ $H$ $\Box$ $F \Box R$ $\vdash\!\!\vdash_{\bf
ID}$ $\exists\overline{U}.$ $P$ $\Box$ $C$ $\Box$ $M$ $\Box$ $H$
$\Box$ $F$ $\Box$ $RX${\tt /=}$u', R$} \vspace*{0.4mm}

{\scriptsize
if  $u \in \mathbb{Z}$, $u' \in \mathbb{R}$ and
$u$ {\tt \#==}$^{\mathcal{M}}~u' \to true$}.\\ \\
\hline
\end{tabular}
\end{center}
\caption{Rules for Inferring $\herbrand$-constraints from
$\mathcal{M}$-constraints}\label{table7}
\end{table}
\vspace*{-.3cm}


To finish this subsection, we present the goal transformation rules
in Table \ref{table7}, which can be used to infer
$\herbrand$-constraints from the $\mdom$-constraints placed in the
store $M$. The inferred $\herbrand$-constraints happen to be
$\fd$-specific or $\rdom$-specific, according to the case, and can
be placed in the corresponding store. Therefore, the rules in this
group model domain cooperation mechanisms other than bridges and
projections.
\vspace*{-.3cm}

\subsection{Constraint Solving Rules}\label{csr}

The presentation of $\cclnc{\ccdom}$ finishes with the constraint solving rules displayed in
Table \ref{Stable}.  Rule  {\bf SF} models the detection of failure by a solver, and the other
rules describe the possible transformation of a goal by a solver's invocation.
Each time a new constraint from the pool is placed into its store by means
of transformation {\bf SC}, it is pragmatically convenient to invoke the corresponding solver
by means of the rules in this table. The solvers for the four domains $\mathcal{M}$, $\mathcal{H}$,
$\mathcal{FD}$ and $\mathcal{R}$ involved in the coordination domain $\ccdom$ are considered.
The availability of the $\mdom$-solver means that solving mediatorial constraints
contributes to the cooperative goal solving process, in addition to the role of bridges
for guiding projections.

\vspace*{-.3cm}
\begin{table}[h]
\begin{center}
\begin{tabular}{p{12.5cm}}
\hline\\

{\scriptsize {\bf MS $\mathcal{M}$-Constraint Solver (\em {glass-box})}}\\\vspace*{0.1mm}

{\scriptsize $~~~~~~\exists \overline{U}.$ $P$ $\Box$ $C$ $\Box$ $M$
$\Box$ $H$ $\Box$ $F$ $\Box$ $R$ $\vdash\!\!\vdash_{\bf MS}$
$\exists \overline{Y'}, \overline{U}.$ $(P$ $\Box$ $C$ $\Box$
$(\Pi'$ $\Box$ $\sigma_M)$ $\Box$ $H$ $\Box$ $F$ $\Box$
$R)@_{\mdom}\sigma'$}\\\vspace*{0.1mm}

{\scriptsize If  $pvar(P) \cap var(\Pi_M) = \emptyset$,
$(\Pi_M\, \Box\, \sigma_M)$ is not solved,
$\Pi_M \vdash\!\!\vdash_{solve^{\mdom}} \exists \overline{Y'} (\Pi'\, \Box\, \sigma')$.} \\\\

{\scriptsize {\bf HS $\mathcal{H}$-Constraint Solver ({\em glass-box})}}\\\vspace*{0.1mm}

{\scriptsize $~~~~~~\exists \overline{U}.$ $P$ $\Box$ $C$ $\Box$ $M$
$\Box$ $H$ $\Box$ $F$ $\Box$ $R$ $\vdash\!\!\vdash_{\bf HS}$
$\exists \overline{Y'}, \overline{U}.$ $(P$ $\Box$ $C$ $\Box$ $M$
$\Box$ $(\Pi'$ $\Box$ $\sigma_H)$ $\Box$ $F$ $\Box$
$R)@_{\herbrand}\sigma'$}\\\vspace*{0.1mm}

{\scriptsize If $pvar(P) \cap odvar_{\herbrand}(\Pi_H) = \emptyset$,
$\varx =_{def} pvar(P) \cap var(\Pi_H)$,
$(\Pi_H\, \Box\, \sigma_H)$ is not $\chi$-solved,} \\
{\scriptsize $\Pi_H \vdash\!\!\vdash_{solve^{\herbrand}_{\varx}} \exists \overline{Y'} (\Pi'\, \Box\, \sigma')$.}\\\\

{\scriptsize {\bf FS $\mathcal{FD}$-Constraint Solver ({\em
black-box})}}\\\vspace*{0.1mm}
{\scriptsize $~~~~~~\exists \overline{U}.$ $P$ $\Box$ $C$ $\Box$ $M$
$\Box$ $H$ $\Box$ $F$ $\Box$ $R$ $\vdash\!\!\vdash_{\bf FS}$
$\exists \overline{Y'}, \overline{U}.$ $(P$ $\Box$ $C$ $\Box$ $M$
$\Box$ $H$ $\Box$ $(\Pi'$ $\Box$ $\sigma_F)$ $\Box$
$R)@_{\fd}\sigma'$}\\\vspace*{0.1mm}

{\scriptsize If $pvar(P) \cap var(\Pi_F) = \emptyset$,
$(\Pi_F\, \Box\, \sigma_F)$ is not solved,
$\Pi_F \vdash\!\!\vdash_{solve^{\fd}} \exists \overline{Y'} (\Pi'\, \Box\, \sigma')$.}\\\\

{\scriptsize {\bf RS $\mathcal{R}$-Constraint Solver ({\em
black-box})}}\\\vspace*{0.1mm}

{\scriptsize $~~~~~~\exists \overline{U}.$ $P$ $\Box$ $C$ $\Box$ $M$
$\Box$ $H$ $\Box$ $F$ $\Box$ $R$ $\vdash\!\!\vdash_{\bf RS}$
$\exists \overline{Y'}, \overline{U}.$ $(P$ $\Box$ $C$ $\Box$ $M$
$\Box$ $H$ $\Box$ $F$ $\Box$ $(\Pi'$ $\Box$
$\sigma_R))@_{\rdom}\sigma'$}\\\vspace*{0.1mm}

{\scriptsize If $pvar(P) \cap var(\Pi_R) = \emptyset$,
$(\Pi_R\, \Box\, \sigma_R)$ is not solved,
$\Pi_R \vdash\!\!\vdash_{solve^{\rdom}} \exists \overline{Y'} (\Pi'\, \Box\, \sigma')$.}\\\\
{\scriptsize {\bf SF Solving Failure}}\\\vspace*{0.1mm}
{\scriptsize $~~~~~~\exists \overline{U}.$ $P$ $\Box$ $C$ $\Box$ $M$
$\Box$ $H$ $\Box$ $F$ $\Box$ $R$ $\vdash\!\!\vdash_{\bf SF}$
$\blacksquare $}\\\vspace*{0.1mm}

{\scriptsize If $S$ is the $\mathcal{D}$-store ($\mathcal{D}$ being $\mathcal{M}$, $\mathcal{H}$, $\mathcal{FD}$ or
$\mathcal{R}$), $pvar(P) \cap odvar_{\cdom}(\Pi_S) = \emptyset$, $\varx =_{def} pvar(P) \cap var(\Pi_S)$,
$(\Pi_S\, \Box\, \sigma_S)$ is not $\chi$-solved and
$\Pi_S \vdash\!\!\vdash_{solve^{\cdom}_{\varx}} \blacksquare$. Note that $\varx \neq \emptyset$ is possible only in the case
$\mathcal{D} = \mathcal{H}$.}\\\\
\hline
\end{tabular}
\end{center}
\caption{Rules for $\mathcal{M}$, $\herbrand$, $\mathcal{FD}$ and
$\mathcal{R}$ Constraint Solving}\label{Stable}
\end{table}
\vspace*{-.2cm}

Let $\cdom$ be any of the four domains, and let $\Pi$ be the set of constraints included in the
$\cdom$ store in a given goal $G$ with productions $P$. As explained in Subsection  \ref{csolvers}, each invocation 
$solve^{\mathcal{D}}(\Pi,\varx)$ depends on a set  of critical variables
$\varx \subseteq cvar_{\cdom}(\Pi)$ which must be properly chosen. On the other hand, the goal invariants explained 
in Subsection \ref{programa} require that no produced variable is bound to a
non-linear pattern, and the {\em safe binding}  condition satisfied by any solver ensures that a solver invocation never 
binds any variable {\tt X} $\in \varx$,  except to a constant.

Because of these reasons, the rules in Table \ref{Stable} allow a solver invocation $solve^{\mathcal{D}}(\Pi,\varx)$
only if the following two conditions are satisfied:
\begin{enumerate}
\item[(a)] $pvar(P) \cap odvar_{\cdom}(\Pi) = \emptyset$.\\
Motivation: If this condition does not hold, for any choice of $\varx \subseteq cvar_{\cdom}(\Pi)$ there is some
variable {\tt X} $\in pvar(P) \setminus \varx$, and the solver invocation could bind {\tt X}  to a non-linear pattern.
\item[(b)] $\varx = pvar(P) \cap var(\Pi)$.\\
Motivation: Because of condition (a), this $\varx$ is a subset of $cvar_{\cdom}(\Pi)$, and the safe binding
condition of solvers ensures that the invocation $solve^{\mathcal{D}}(\Pi,\varx)$ will bind no produced variable, except to a constant.
\end{enumerate}

When $\cdom$ is not $\herbrand$, we know from Section \ref{coordination} that all the variables in
$\Pi$ can be assumed to be obviously demanded. Then $odvar_{\cdom}(\Pi) = var(\Pi)$,
condition (a) becomes $pvar(P) \cap var(\Pi) = \emptyset$, (b) becomes
$\varx = \emptyset$, and $solve^{\mathcal{D}}(\Pi,\emptyset)$ can be abbreviated
as $solve^{\mathcal{D}}(\Pi)$. The rules related to $\mathcal{M}$, $\mathcal{FD}$
and  $\mathcal{R}$ in Table \ref{Stable} assume the simplified form of condition (a), (b).
The notations $\Pi \vdash\!\!\vdash_{solve^{\cdom}_{\varx}} \exists \overline{Y'} (\Pi'\, \Box\, \sigma')$
and $\Pi \vdash\!\!\vdash_{solve^{\cdom}_{\varx}} \blacksquare$
introduced in Subsection \ref{csolvers} are used to indicate the
non-deterministic choice of an alternative returned by a successful $\cdom$-solver invocation
and a failed $\cdom$-solver invocation, respectively.
Note also the use of the notation $( \ldots )@_{\cdom}\sigma'$
explained near the end of Subsection \ref{programa}.

At this point, we can precise the notion of {\em solved goal} as follows:
a goal $G$ is solved iff it has the form
$\exists \overline{U}.$ $\Box$ $\Box$ $M$ $\Box$ $H$ $\Box$ $F$ $\Box$ $R$
(with empty $P$ and $C$) and the $\clnc{\ccdom}$-transformations in Tables \ref{table7} and \ref{Stable}
cannot be applied to $G$. The $\clnc{\ccdom}$-transformations in Tables \ref{table3} and \ref{table4}
are obviously not applicable to solved goals, since they refer to $P$ and $C$.
\vspace*{-.2cm}

\subsection{One Example of Cooperative Goal Solving} \label{exampletriangle}

In order to illustrate the overall behavior of our cooperative goal solving
calculus, we present a $\cclnc{\ccdom}$ computation solving the goal  {\bf Goal 2}
discussed in Subsection \ref{examples}.
The reader is referred to Figure \ref{goals} for a graphical representation of the
problem and to Subsection \ref{programa} for a formulation of the goal and the
expected solution in the particular case {\tt d = 2}, {\tt n = 4}.
However, the solution is the same for any choice
of  positive integer values {\tt d} and {\tt n} such that {\tt n = 2*d},
and here we will discuss the general case.

The $\cclnc{\ccdom}$ calculus leaves ample room for choosing a particular goal transformation
at each step, so that many different computations are possible in principle.
However, the $\toy$ implementation follows a particular strategy.
The part $P\, \Box\, C$ of the current goal is treated as a sequence and processed from left to right,
with the only exception of suspensions $e \to X$ that are delayed until they can be
safely eliminated by means of rule {\bf EL} or the goal is so transformed that they cease to be suspensions.
As long as the current goal is not in solved form, a subgoal is selected and processing according to a strategy
which can be roughly described as follows:
\begin{enumerate}
\item
If $P$ includes some production which can be handled by the
constrained lazy narrowing rules in Table \ref{table3}, the leftmost
such production is selected and processed. Note that the selected
production must be either a suspension $e \to X$ that can be
discharged by rule {\bf EL}, or else a production that is not a
suspension. The applications of rule {\bf DF} are performed in an
optimized way by using definitional trees \cite{Vad05,vado:ictac07}.
\item
If $P$ is empty or consists only of productions $e \to X$ that cannot be processed by means of the  constrained 
lazy narrowing rules in Table \ref{table3}, and moreover some of the stores $M$, $H$, $F$ or $R$ is not in solved 
form and its constraints include no produced variables, then the solvers for such stores are invoked, choosing the 
set $\varx$ of critical variables as explained in Table \ref{Stable}.
\item
If neither of the two previous items applies and $C$ is not empty, the leftmost atomic constraint $\delta$ in $C$ is selected. 
In case it is not primitive, the flattening rule {\bf FC} from Table \ref{table3} is applied.
Otherwise, $\delta$ is a primitive atomic constraint $\pi$, and exactly one of the following cases applies:
\begin{enumerate}
\item
If $\pi$ is a proper $\fd$-constraint or else an extended $\herbrand$-constraint such that $M \vdash \pi$ {\tt in} $\fd$,
then $\pi$ is processed by means of the rules {\bf SB}, {\bf PP} and {\bf SC} from Table \ref{table4}.
This generates bridges and projected constraints $\pi'$, if possible, and submits $\pi$ to the store $F$.
Then, the rules from Table \ref{Stable} are used for invoking the $\fd$-solver
(in case that the constraints in $F$  include no produced variables)
and  the $\rdom$-solver (in case that the constraints in $R$  include no produced variables).
\item
If $\pi$ is a proper $\rdom$-constraint or else an extended $\herbrand$-constraint such that $M \vdash \pi$ {\tt in} $\rdom$,
then $\pi$ is processed by means of the rules {\bf SB}, {\bf PP} and {\bf SC} from Table \ref{table4}.
This generates bridges and projected constraints $\pi'$, if possible, and submits $\pi$ to the store $R$.
Then, the rules from Table \ref{Stable} are used for invoking the $\rdom$-solver
(in case that the constraints in $R$  include no produced variables)
and  the $\fd$-solver (in case that the constraints in $F$  include no produced variables).
\item
If $\pi$ is an extended $\herbrand$-constraint such that neither  $M \vdash \pi$ {\tt in} $\fd$ nor $M \vdash \pi$ {\tt in} $\rdom$,
then $\pi$ is submitted to the store $H$ by means of rule {\bf SC}, and the $\herbrand$-solver is invoked in case that
the constraints in $H$ include no obviously demanded produced variables.
\item
If $\pi$ is a $\mdom$-constraint, then $\pi$ is submitted to the store $M$ by means of rule {\bf SC},
the rules of Table \ref{table7} are applied if possible,
and the $\mdom$-solver is invoked in case that
the constraints in $M$ include no produced variables.
\end{enumerate}
\end{enumerate}

The series of goals $G_0$ up to $G_{12}$ displayed below correspond to the initial goal, the
final solved goal and a selection of intermediate goals in a computation which roughly models the
strategy of the $\toy$ implementation, working with the projection functionality activated.
In the initial goal,  {\tt d} and {\tt n}  are arbitrary positive integers such that {\tt n = 2*d} and {\tt d' = d+0.5}. \\

{\scriptsize

\noindent
$G_0:$ $\Box$ $\underline{bothIn\,(triangle\,(d,d')\,2\,1)\,(square\,n)\,(X,Y)\,{\tt == true}}$ $\Box$$\Box$$\Box$$\Box$  $\vdash\!\!\vdash_{\bf FC}$\\

\noindent
$G_1:$ $\exists \overline{U_1}.$ $bothIn\,(triangle\,(d,d')\,2\,1)\,(square\,n)\,(X,Y)\, \to A$ $\Box$
$\underline{A {\tt == true}}$ $\Box$$\Box$$\Box$$\Box$  $\vdash\!\!\vdash_{\bf SC(ii)}$\\

\noindent
$G_2:$ $\exists \overline{U_2}.$ $\underline{bothIn\,(triangle\,(d,d')\,2\,1)\,(square\,n)\,(X,Y)}\, \to A$
$\Box$ $\Box$$\Box$ $A$ {\tt == true}$\Box$$\Box$
$\vdash\!\!\vdash_{\bf DF_{\textnormal{\em bothIn}}}$\\ 

\noindent
$G_3:$ $\exists \overline{U_3}.$$\underline{triangle\,(d,d')\,2\,1 \to R}$, $\underline{square\, n \to G}$,
$\underline{(X,Y) \to (X',Y')}$, $\underline{{\tt true} \to A}$ $\Box$\\
\hspace*{2.cm}$X'$ {\tt \#==} $RX$, $Y'$ {\tt \#==} $RY$, $isIn\,R\, (RX,RY)$ {\tt == true}, $isIn\, G\,(X',Y')$ {\tt == true},\\
\hspace*{2.cm}$labeling\,[\,]\,[X',Y']$ $\Box$$\Box$ \underline{$A$ {\tt == true}} $\Box$$\Box$
$\vdash\!\!\vdash^{*}_{\bf SP^2, DC, SP^3,HS}$\\

\noindent
$G_4:$ $\exists \overline{U_4}.$  $\Box$\underline{$X$ {\tt \#==} $RX$}, \underline{$Y$ {\tt \#==} $RY$},
$isIn\,(triangle\,(d,d')\,2\,1)\, (RX,RY)$ {\tt == true}, \\
\hspace*{2.cm}$isIn\,(square\, n)\,(X,Y)$ {\tt == true}, $labeling\,[\,]\,[X,Y]$ $\Box$$\Box$ $\sigma_H$ $\Box$$\Box$  $\vdash\!\!\vdash^{*}_{\bf SC(i)^2,MS}$\\

\noindent
$G_5:$ $\exists \overline{U_5}.$  $\Box$ $\underline{isIn\,(triangle\,(d,d')\,2\,1)\, (RX,RY) {\tt == true}}, \underline{isIn\,(square\, n)\,(X,Y) {\tt == true}}$,\\
\hspace*{2.cm}$labeling\,[\,]\,[X,Y]$ $\Box$$X$ {\tt \#==} $RX$, $Y$ {\tt \#==} $RY$$\Box$
$\sigma_H$ $\Box$$\Box$  $\vdash\!\!\vdash^{*}_{\bf CLN}$\\

\noindent
$G_6:$ $\exists \overline{U_6}.$  $\Box$ \underline{$RY>=d'-1$}, $2*RY-2*1*RX<=2*d'-2*1*d$,\\
\hspace*{2.cm}$2*RY+2*1*RX<=2*d'+2*1*d$, $domain\,[X,Y]\,0\,n$, \\
\hspace*{2.cm}$labeling\,[\,]\,[X,Y]$ $\Box$ $X$ {\tt \#==} $RX$, $Y$ {\tt \#==} $RY$$\Box$ $\sigma'_H$ $\Box$$\Box$  $\vdash\!\!\vdash^{*}_{\bf FC,PC}$\\

\noindent
$G_7:$ $\exists \overline{U_7}.$  $\Box$ \underline{$d'-1\, \to! RA$}, $RY>=RA$,
$2*RY-2*1*RX<=2*d'-2*1*d$,\\
\hspace*{2.cm}$2*RY+2*1*RX<=2*d'+2*1*d$, $domain\,[X,Y]\,0\,n$, \\
\hspace*{2.cm}$labeling\,[\,]\,[X,Y]$ $\Box$ $X$ {\tt \#==} $RX$, $Y$ {\tt \#==} $RY$$\Box$ $\sigma'_H$ $\Box$$\Box$  $\vdash\!\!\vdash^{*}_{\bf SC(iv), RS}$\\

\noindent
$G_8:$ $\exists \overline{U_8}.$  $\Box$   \underline{$RY>=d''$}, $2*RY-2*1*RX<=2*d'-2*1*d$,\\
\hspace*{2.cm}$2*RY+2*1*RX<=2*d'+2*1*d$, $domain\,[X,Y]\,0\,n$, \\
\hspace*{2.cm}$labeling\,[\,]\,[X,Y]$ $\Box$ $X$ {\tt \#==} $RX$, $Y$ {\tt \#==} $RY$$\Box$ $\sigma'_H$ $\Box$$\Box$ $S_R$
$\vdash\!\!\vdash^{*}_{\bf BP, CS}$\\

\noindent
$G_9:$ $\exists \overline{U_9}.$  $\Box$ \underline{$2*RY-2*1*RX<=2*d'-2*1*d$},\\
\hspace*{1.cm}\underline{$2*RY+2*1*RX<=2*d'+2*1*d$}, $domain\,[X,Y]\,0\,n$, $labeling\,[\,]\,[X,Y]$ $\Box$\\
\hspace*{1.cm} $X$ {\tt \#==} $RX$, $Y$ {\tt \#==} $RY$ $\Box$ $\sigma'_H$ $\Box$ $Y\#>=d$ $\Box$ $RY>=d'', S_R$ $\vdash\!\!\vdash^{*}_{\bf FR, BP}$\\

\noindent
$G_{10}:$ $\exists \overline{U_{10}}.$  $\Box$ $domain\,[X,Y]\,0\,n$, $labeling\,[\,]\,[X,Y]$ $\Box$ \\
\hspace*{1.cm}\underline{$X$ {\tt \#==} $RX$, $Y$ {\tt \#==} $RY$ $B$ {\tt \#==} $RB$, $C$ {\tt \#==} $RC$, $S'_M$}
$\Box$  $\sigma'_H$ $\Box$ \\
\hspace*{1.cm}\underline{$Y\#>=d, 2\#*Y\#-2\#*X \to! B, B \#<= 1, 2\#*Y\#+2\#*X \to! C, C  \#<= n', S'_F$} $\Box$ \\
\hspace*{1.cm}\underline{$RY>=d'', 2*RY-2*RX \to! RB, RB <= 1, 2*RY+2*RX \to! RC, RC <= n', S'_R$} $\vdash\!\!\vdash^{*}_{\bf CS}$\\

\noindent
$G_{11}:$ $\exists \overline{U_{11}}.$  $\Box$ \underline{$domain\,[d,d]\,0\,n$},
\underline{$labeling\,[\,]\,[d,d]$} $\Box$
$S''_M$ $\Box$  $\sigma'_H$ $\Box$ $S''_F$ $\Box$ $S''_R$ $\vdash\!\!\vdash^{*}_{\bf SC(iii), FS, SC(iii), FS}$\\

\noindent
$G_{12}:$ $\exists \overline{U_{12}}.$  $\Box$ $\Box$ $S''_M$ $\Box$  $\sigma'_H$ $\Box$ $S''_F$ $\Box$ $S''_R$ \\

}
\vspace*{.2cm}

The local existential variables $\exists \overline{U_i}$ of each goal $G_i$ are not explicitly displayed,
and the notation $G_{i-1}\, \vdash\!\!\vdash^{*}_{\bf \mathcal{RS}}\, G_i$ is used to indicate the transformation
of $G_{i-1}$ into $G_i$ using the goal solving rules indicated by ${\bf \mathcal{RS}}$.
At some steps, ${\bf \mathcal{RS}}$ indicates a particular sequence of individual rules, named as explained
in the previous subsections. In other cases, namely for $i = 6$ and $9 \leq i \leq 11$,
${\bf \mathcal{RS}}$ indicates sets of goal transformation rules, named according to the following conventions:

\begin{itemize}
\item
{\bf CLN} names the set of constrained lazy narrowing rules presented in Table \ref{table3}.
\item
{\bf FR} names the set consisting of the two rules {\bf FC} and {\bf PC} displayed at the end
of Table \ref{table3}, used for constraint flattening.
\item
{\bf BP} names the set of rules for bridges and projections presented in Table \ref{table4}.
\item
{\bf CS} names the set of constraint solving rules presented in Table \ref{Stable}.
\end{itemize}

We finish with some comments on the computation steps:

\begin{itemize}
\item
Transition from $G_0$ to $G_1$:
The only constraint in $C$ is flattened, giving rise to one suspension and one flat constraint in the new goal.
The produced variable {\tt A} is not obviously demanded because the constraint {\tt A == true}
is not yet placed in the $\herbrand$-store.
\item
Transition from $G_1$ to $G_2$: The only suspension is delayed,
and the only constraint in the pool is processed by submitting  it to the
$\herbrand$-store. However, the $\herbrand$-solver cannot be invoked at this point,
because $A$ has become an obviously demanded variable that is also produced.
\item
Transition from $G_2$ to $G_3$: The former suspension has become a production
which is processed by applying the program rule defining the function {\tt bothIn},
which introduces new productions in $P$ and new constraints in $C$.
\item
Transition from $G_3$ to $G_4$: The four productions in $P$ are processed by
binding propagations and decompositions (rules {\bf SP} and {\bf DC}),
until $P$ becomes empty. Then the  $\herbrand$-solver can be invoked.
At this point,  the $\herbrand$-store just contains a substitution $\sigma_H$ resulting from
the previous binding steps.
\item
Transition from $G_4$ to $G_5$: $P$ is empty, and the two first constraints in $C$ are bridges.
They are submitted to the $\mdom$-store and the $\mdom$-solver is invoked,
which has no effect in this case.
\item
Transition from $G_5$ to $G_6$: There are no productions, and the two first constraints in the pool are processed
by steps similar to those used in the transition going from $G_0$ to $G_4$.
Upon completing this process, the new pool includes a number of new constraints coming
from the conditions in the program rules defining the functions {\tt isIn}, {\tt triangle} and {\tt square},
and the substitution stored in $H$ has changed.
At this point, $P$ is empty again and the constraints in $C$ plus the bridges in $M$ amount
to a system equivalent to the one used in Subsection \ref{examples} for an informal discussion of the resolution
of {\bf Goal 2}.
\item
Transition from $G_6$ to $G_7$ and from $G_7$ to $G_8$: There are no productions, and flattening the first constraint in $C$
gives rise to the primitive constraint {\tt d'-1\, $\to!$  RA}. This is submitted to the $\rdom$-store and the
$\rdom$-solver is invoked, which computes {\tt d''} as the numeric value of {\tt d'-1} and propagates
 the variable binding {\tt RA $\mapsto$ d''} to the whole goal,
possibly causing some other internal changes in the $\rdom$-store.
\item
Transition from $G_8$ to $G_9$: There are no productions, and the first constraint in $C$ is now {\tt RY >= d''}.
Since {\tt d'' = d'-1 = d+0.5-1 = d-0.5}, we have {\tt $\lceil$d''$\rceil$ = d}.
Therefore, projecting {\tt RY >= d''} with the help of the available bridges (including  {\tt Y \#== RY}) allows to compute
{\tt Y \#>= d} as a projected $\fd$-constraint. Both {\tt RY >= d''} and {\tt Y \#>= d} are submitted to their
respective stores and the two solvers are invoked, having no effect in this case.
\item
Transition from $G_9$ to $G_{10}$: There are no productions, and the
two first atomic constraints in the pool of $G_9$ (two
$\rdom$-constraints  $\delta_1$ and $\delta_2$) are processed by
steps similar to those used in the transition going from $G_6$ to
$G_9$, except that the solver invocations are delayed to the
transition from $G_{10}$ to $G_{11}$ and commented in the next item.
(Actually, the $\toy$ implementation would invoke the
solvers two times: The first time when processing $\delta_1$ and the
second time when processing $\delta_2$. Here we explain the overall
effect of the two invocations for the sake of simplicity.) Upon
completing this process, $G_{10}$ stays as follows: $P$ is empty,
$C$ includes the two other constraints which were there in $G_9$,
and the stores $M$, $F$ and $R$ have changed because of new bridges
and projections. In fact, the constraints within the stores $F$ and
$R$ in $G_{10}$ would be equivalent but not identical to the ones
shown in this presentation, due to intermediate flattening steps
that we have not shown explicitly. In parti\-cular, the
$\rdom$-constraint {\tt 2*RY-2*RX $\to!$ RB} and its
$\fd$-projection {\tt 2\#*Y\#-2\#*X $\to!$ B} would really not occur
in this form, but a conjunctions of primitive constraints obtained
by flattening them  would occur at their place.
\item
Transition from $G_{10}$ to $G_{11}$: At this point, the $\fd$-solver is able to infer that the constraints
in the $\fd$ store imply one single solution for  the variables {\tt X} and {\tt Y}, namely {\tt \{X $\mapsto$ d, Y $\mapsto$ d\}}.
Therefore, the $\fd$-solver propagates these bindings to the whole goal, affecting in particular to the bridges in $M$.
Then, the $\mdom$-solver propagates the corresponding bindings {\tt \{RX $\mapsto$ rd, RY $\mapsto$ rd\}}.
({\tt rd} being the representation of {\tt d} as an integral real number),
and the $\rdom$-solver succeeds.
\item
Transition from $G_{11}$ to $G_{12}$: The two constraints in $C$ have now become trivial.
Submitting them to their stores and invoking the respective solvers leads to a solved goal,
whose restriction to the variables in the initial goal  is the computed answer
$\Box$$\Box$$\Box$$\Box$ ($\true$ $\Box$ {\tt \{X $\mapsto$ d, Y $\mapsto$ d\}}) $\Box$.
Note that no labeling whatsoever has been performed, independently of the size of {\tt n}.
\end{itemize}
\vspace*{-.2cm}

\subsection{Properties of the Cooperative Goal Solving Calculus $CCLNC(\mathcal{C})$} \label{SC}

This final subsection presents the main semantic results of the
paper, namely {\em soundness} and {\em limited  completeness} of the coopera\-tive goal solving calculus
$\cclnc{\ccdom}$  w.r.t. the declarative semantics of
$\cflp{\ccdom}$ given in  \cite{LRV07}. To start with, we define the notion of solution for a
given goal.

\begin{definition} [Solutions of Goals and their Witnesses]\label{defGoalSol}
\begin{enumerate}
\item
Let $G$ $\equiv$ $\exists \overline{U}.$ $P$ $\Box$ $C$ $\Box$ $M$ $\Box$ $H$ $\Box$ $F$ $\Box$ $R$
be an admissible goal for  a given $\cflp{\mathcal{C}}$-program $\prog$.
The {\em set of solutions}   $Sol_{\prog}(G)$ of $G$ w.r.t. $\prog$ includes
all those $\mu \in Val_{\mathcal{C}}$ such that there is some
$\mu' \in Val_{\mathcal{C}}$ verifying
$\mu'$ $=_{\backslash \overline{U}}$ $\mu$
and $\mu' \in Sol_{\prog}(P\, \Box\, C \, \Box \,M\,\Box\,H\,\Box\,F\,\Box\,R)$,
which holds iff the following two conditions are satisfied:
\begin{enumerate}
\item
$\mu' \in Sol_{\prog}(P\, \Box\, C)$. By definition, this means
$\prog \vdash_{\crwl{\mathcal{C}}}(P\,\Box\,C)\mu'$, which is
equivalent to $\prog \vdash_{\crwl{\mathcal{C}}}P\mu'$ and $\prog
\vdash_{\crwl{\mathcal{C}}}C\mu'$. This notation refers to the
existence of proofs in the instance $\crwl{\mathcal{C}}$ of  the
{\em Constrained Rewriting Logic} $CRWL$,
whose inference rules can be found in
\cite{LRV07}.
\item
$\mu' \in Sol_{\mathcal{C}}(M\,\Box\,H\,\Box\,F\,\Box\,R)$, which is equivalent to
$\mu' \in Sol_{\mathcal{C}}(M) \cap Sol_{\mathcal{C}}(H) \cap Sol_{\mathcal{C}}(F) \cap Sol_{\mathcal{C}}(R)$.
\end{enumerate}
\item
If $\mathcal{M}$ is a multiset having as its members the $CRWL(\mathcal{C})$-proofs mentioned in item 1.(a) above,
we will say that $\mathcal{M}$ is a {\em witness} for the fact that $\mu \in Sol_{\prog}(G)$,
and we will write  $\mathcal{M} : \mu \in Sol_{\prog}(G)$.
\item
A  solution $\mu \in Sol_{\prog}(G)$
is called {\em well-typed} iff the valuation $\mu'$
$=_{\backslash \overline{U}}$ $\mu$ mentioned in item 1. can be so chosen
that $(P\, \Box\, C \, \Box \,M\,\Box\,H\,\Box\,F\,\Box\,R)\mu'$ is well-typed,
which is noted as $\mu' \in WTSol_{\prog}(P\, \Box\, C \, \Box \,M\,\Box\,H\,\Box\,F\,\Box\,R)$.
The set of all well-typed solutions of $G$ w.r.t. $\prog$ is written as $WTSol_{\prog}(G)$.
In case that $\mathcal{M}$ is a witness for $\mu \in Sol_{\prog}(G)$, we also say that
$\mathcal{M}$ is a witness for $\mu \in WTSol_{\prog}(G)$ and we write $\mathcal{M} : \mu \in WTSol_{\prog}(G)$.\\
\end{enumerate}

In case that $G$ is a solved goal $S$, we write $Sol_{\ccdom}(S)$ (resp. $WTSol_{\ccdom}(S)$)
in place of $Sol_{\prog}(S)$ (resp. $WTSol_{\prog}(S)$).
\end{definition}

Concerning item 1.(b) in the previous definition, note that the equivalence
$\eta \in Sol_{\mathcal{C}}(M) \cap Sol_{\mathcal{C}}(H) \cap
Sol_{\mathcal{C}}(F) \cap Sol_{\mathcal{C}}(R)$ $ \Leftrightarrow$
$\eta \in Sol_{\mathcal{M}}(M) \cap Sol_{\mathcal{H}}(H) \cap
Sol_{\mathcal{F}}(F) \cap Sol_{\mathcal{R}}(R)$  does not make sense in general,
because a given valuation $\eta \in Val_{\ccdom}$ is not always a $\cdom$ valuation
when $\cdom$ is chosen as one of the four components of $\ccdom$.
However, Theorem \ref{sumProperties} from Subsection \ref{cdomains} allows
to reason with solutions known for $\ccdom$ in terms of solutions known for the four
components, as we will see in the mathematical proofs of Appendix \ref{PropertiesCalculus}.

Before presenting our soundness and completeness results for $\cclnc{\ccdom}$
let us comment  on some limitations concerning completeness:
\begin{itemize}
\item
As already said in Subsection \ref{programa}, the design of $\cclnc{\ccdom}$ is tailored to
programs and goals having no free occurrences of higher-order logic variables.
Therefore, the completeness results of this subsection are limited to this kind of programs and goals.
\item
The completeness of $\cclnc{\ccdom}$ is obviously conditioned  by the completeness of the solvers
invoked by the goal transformation rules in Table \ref{Stable}.
On the other hand, the completeness requirement for solvers in Definition \ref{defSolver}
is limited to well-typed solutions.
Therefore, the completeness results of this subsection refer only to well-typed solutions of the initial goal.
\item
As  discussed in Subsections \ref{hdom}, \ref{rdom} and \ref{fdom}, certain  invocations of constraint  solvers
can be incomplete even w.r.t.  well-typed solutions. Therefore, the completeness results of this subsection are also limited
by the assumption that no incomplete solver invocations occur during goal solving.
\item
Finally, the goal transformation rule {\bf DC} from Table \ref{table3} can give rise to {\em opaque decompositions}.
Similarly  to the opaque decompositions caused by  the transformation rules  {\bf H3} and {\bf H7}
for $\herbrand$-stores (see Subsection \ref{hdom}),  the opaque decompositions caused by {\bf DC} can lose well-typed solutions.
In what follows, we will say that an application of the goal transformation rule {\bf DC}
is {\em transparent} iff the expression $h\, \tpp{e}{m}$ involved in the rule application is such that
$h$ is $m$-transparent (or equivalently, $h$ is not $m$-opaque).
Only transparent  applications of the  rule {\bf DC}  can be trusted to preserve well-typed solutions.
For this reason, the completeness results of this subsection are limited by the assumption that
no opaque applications of  rule {\bf DC} occur during goal solving.
Unfortunately,  the eventual occurrence of opaque decomposition steps during  goal solving
(be they due to rule {\bf DC} from Table \ref{table3} or to the $sts$s {\bf H3} and {\bf H7} of the $\herbrand$-solver)
is an undecidable problem, because of theoretical results proved in \cite{GHR01}.
\end{itemize}


In the sequel we will use the notation $G \red_{{\bf RL},\gamma, \prog}G'$ to indicate that the
admissible goal $G$ for the $\cflp{\mathcal{C}}$-program $\prog$ is transformed into the new goal $G'$
by an application of the {\em selected rule} {\bf RL} applied to the {\em selected part} $\gamma$ of $G$.
It is important  to note that the  selected part $\gamma$ of $G$ must have the form expected by the selected rule {\bf RL}.
More precisely, $\gamma$ must be selected as one of the stores in case that {\bf RL} is some the
transformations in Table \ref{Stable}, as a pair of bridges in case that {\bf RL} is the transformation
{\bf IE} from Table \ref{table7}, and as an atom in any other case.
We will use also the notation $G \red_{{\bf RL},\gamma, \prog}^{+} G'$ to indicate the existence of some
computation of the form $G \red_{{\bf RL},\gamma, \prog} G_1 \red_{\prog}^{*}G'$
transforming $G$ in $G'$ in $n$ steps for some $n \geq 1$.

 We are now in a position to present the main results of this subsection.
First, we state a theorem which guarantees {\em local}
soundness and completeness for the {\em one-step} transformation of
a given goal. A proof is given in Appendix \ref{PropertiesCalculus}.


\begin{theorem}[Local Soundness and Limited Local Completeness]\label{localSC}
Assume a given program $\prog$ and an admissible goal $G$ for $\prog$ which is not in solved form.
Choose any rule {\bf RL} applicable to $G$ and select a part $\gamma$  of $G$ suitable for applying {\bf RL}.
Then there are finitely many possible transformations $G \red_{{\bf RL},\gamma, \prog}G'_j$  ($1 \leq j \leq k$), and moreover:
\begin{enumerate}
\item {\bf Local Soundness:} $Sol_{\prog}(G) \supseteq \bigcup_{j=1}^{k} Sol_{\prog}(G'_j)$.
\item {\bf Limited Local Completeness:} $WTSol_{\prog}(G) \subseteq \bigcup_{j=1}^{k} WTSol_{\prog}(G'_j)$,
provided that the application of {\bf RL}  to the selected part $\gamma$ of $G$ is {\em safe} in the following sense:
it is neither an opaque application of {\bf DC}
nor an application of a rule from Table \ref{Stable} involving an incomplete solver invocation.
\end{enumerate}
\end{theorem}


A global soundness result for $\cclnc{\ccdom}$  follows easily from the first item of Theorem \ref{localSC}.
In particular, it ensures that the solved forms obtained as computed
answers for an initial goal using the rules of the cooperative goal
solving calculus are indeed semantically valid answers of $G$.

\begin{theorem}[Soundness Theorem]\label{globalS}
Assume a $\cflp{\mathcal{C}}$-program $\prog$,
an admissible goal $G$ for $\prog$,
and a solved goal $S$  such that $G$ $\red_{\prog}^{*}$ $S$.
Then, $Sol_{\mathcal{C}}(S) \subseteq Sol_{\prog}(G)$.
\end{theorem}

\vspace*{-.4cm}
\begin{proof}

As an obvious consequence of Theorem \ref{localSC} (item 1.), one gets
$Sol_{\prog}(G')$ $\subseteq$ $Sol_{\prog}(G)$ for any $G'$ such
that $G$ $\red_{\prog}$ $G'$. From this, an easy induction shows
that $Sol_{\prog}(S)$ $\subseteq$ $Sol_{\prog}(G)$ holds for each
solved form $S$ such that $G$ $\red_{\prog}^{*}$ $S$. Since
$Sol_{\prog}(S)$ $=$ $Sol_{\mathcal{C}}(S)$, the soundness result is proved. \hfill
\end{proof}


Note that the local completeness part (item 2.) of Theorem \ref{localSC}
also implies that failing goals have no solution; i.e., from a
failing transformation step $G\, \red_{{\bf RL},\prog}\, \blacksquare$
we can conclude $WTSol_{\prog}(G) = \emptyset$,
provided that the application of {\bf RL} is safe.
However, a global completeness result for $\cclnc{\ccdom}$
does not immediately follow from item 2. of Theorem \ref{localSC}.
For an arbitrarily given $\mu \in WTSol_{\prog}(G)$,
completeness needs to ensure a terminating $\cclnc{\ccdom}$ computation
ending up with a solved form $S$ such that $\mu \in WTSol_{\ccdom}(S)$.
According to Definition \ref{defGoalSol}, $\mu \in WTSol_{\prog}(G)$ implies
the existence of a witness $\mathcal{M} : \mu \in WTSol_{\prog}(G)$.
In Appendix \ref{PropertiesCalculus} we have defined a {\em well-founded progress
ordering} $\vartriangleright$ between pairs $(G,\mathcal{M})$ formed
by a goal $G$ and a witness, and we have proved the following result:

\begin{lemma}[Progress Lemma]\label{progress}
Consider  an admissible goal $G$  for a $CFLP(\mathcal{C})$-program
$\prog$ and  a witness $\mathcal{M} : \mu \in WTSol_{\prog}(G)$.
Assume that neither $\prog$ nor $G$ have free occurrences of higher-order variables,
and that $G$ is not in solved form. Then:
\begin{enumerate}
\item There is some  {\bf RL} applicable to $G$ which is not a failing rule.
\item Assume any choice of  a rule {\bf RL} (not a failure rule) and a part $\gamma$ of $G$,
such that {\bf RL} can be applied  to $\gamma$ in a safe manner.
Then there is some finite computation  $G \red_{{\bf RL},\gamma, \prog}^{+} G'$ such that:
\begin{itemize}
\item
$\mu \in WTSol_{\prog}(G')$.
\item
There is a witness $\mathcal{M'} : \mu \in WTSol_{\prog}(G')$ verifying
$(G,\mathcal{M}) \vartriangleright (G',\mathcal{M}')$.
\end{itemize}
\end{enumerate}
\end{lemma}


Using the former lemma, we can prove the following completeness result:

\begin{theorem}[Limited Completeness Theorem]\label{globalC}
Let an admissible goal $G$ for a program $\prog$ and  a well-typed solution $\mu \in WTSol_{\prog}(G)$ be given.
Assume that neither $\prog$ nor $G$ have free occurrences of higher-order variables.
Then, unless prevented by some unsafe rule application,
one can find a $CCLNC(\mathcal{C})$-computation $G \red_{\prog}^{*} S$ ending with
a goal in solved form $S$ such that $\mu \in WTSol_{\ccdom}(S)$.
\end{theorem}

\vspace*{-.4cm}

\begin{proof}
The thesis of the theorem can be rephrased  by writing $\mu \in WTSol_{\prog}(S)$
 in place of the equivalent condition $\mu \in WTSol_{\ccdom}(S)$.
The hypothesis allow us to choose a witness $\mathcal{M} : \mu \in WTSol_{\prog}(G)$.
In order to prove the rephrased thesis we reason by induction on the well-founded ordering
$\vartriangleright$ (see e.g. \cite{BN98} for an explanation of this proof technique).
In case that $G$ is a solved goal, the rephrased thesis holds trivially with $S$ taken as $G$ itself.
In case that $G$ is not solved,  we apply the Progress Lemma \ref{progress} to $\prog$ and
$\mathcal{M} : \mu \in WTSol_{\prog}(G)$ and we obtain a rule {\bf RL}  and a part $\gamma$ of $G$
such that {\bf RL} can be applied to $\gamma$. Assuming that this rule application is a safe one,
Lemma \ref{progress} also provides a finite computation $G \red_{{\bf RL},\gamma, \prog}^{+} G'$ such that
there is a witness $\mathcal{M'} : \mu \in WTSol_{\prog}(G')$ fulfilling $(G,\mathcal{M}) \vartriangleright (G',\mathcal{M}')$.
Since neither $\prog$ nor $G$ have free occurrences of higher-order variables,
the same must be true for $G'$.
By well-founded induction hypothesis we can then conclude that,
unless  prevented by some unsafe goal transformation step, one can find a computation
$G' \red_{\prog}^{*} S$ ending with a goal in solved form $S$ such that $\mu \in WTSol_{\prog}(S)$.
The desired computation is then $G \red_{{\bf RL},\gamma, \prog}^{+} G' \red_{\prog}^{*} S$. \hfill
\end{proof}

%
%

\section{Implementation}\label{implementation}

This section sketches the implementation of the
$\cclnc{\ccdom}$ computational model on top of the $\toy$ system.
The current implementation has evolved from older versions that supported the
domains $\herbrand$ and $\rdom$, but not yet $\fd$ and its cooperation with $\herbrand$ and $\rdom$.  
We describe the architectural components of the current $\toy$ system and we briefly discuss
the implementation of the main cooperation mechanisms provided by $\cclnc{\ccdom}$, namely bridges and projections.
The reader is referred to \cite{toyreport,EFS06,ciclops'07,estevez+:esop08} for more details.
 
Instead of using an abstract machine for running byte-code or
intermediate code, $\toy$ programs are compiled to and executed in
Prolog, as in other related systems \cite{antoy00compiling}. 
The compilation generates Prolog code that implements
goal solving by constrained lazy narrowing guided by {\em definitional trees},
a well known device for ensuring an optimal behaviour  of  lazy narrowing
\cite{LLR93,AEH94,AEH00,vado+:ppdp03,Vad05,vado:ictac07}.
$\toy$ relies on an efficient Prolog system, SICStus Prolog \cite{SP}, which provides many
libraries, including constraint solvers for the domains $\fd$ and $\rdom$.

$\toy$ is distributed ({\tt http://toy.sourceforge.net}) as a free open-source 
Sourceforge project and runs on several platforms.
Installation is quite simple. Console and windows executables are
provided, no further software is required. 
In addition, $\toy$ can be used inside ACIDE \cite{ACIDE}, an emerging
multiplatform and configurable integrated development environment
(alpha development status).
\vspace*{-0.2cm}

\subsection{Architectural Components of the Cooperation Schema}\label{architecture}

Fig.  \ref{fig:Arq_toy} shows the architectural components of the
cooperation schema in $\toy$. As explained in Subsection
\ref{ourcdom}, the three pure constraint
domains $\herbrand$, $\rdom$ and $\fd$ are combined with a
mediatorial domain $\mathcal{M}$ to yield the coordination domain
$\ccdom = \mathcal{M} \oplus \herbrand \oplus \fd \oplus \rdom$
which supports our cooperation model. Therefore, these four domains
are supported by the implementation. Moreover, the set of primitives
supported by the domains $\rdom$ and $\fd$ in the $\toy$
implementation is wider than the simplified description given in
subsections \ref{rdom} and \ref{fdom}.

\begin{figure}[htbp]
\begin{center}
\includegraphics[scale=0.7,angle=0]{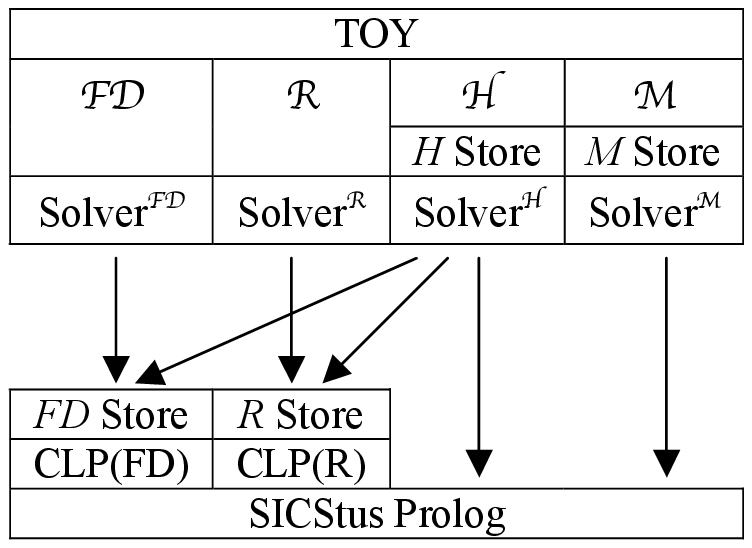}
\end{center}
\caption{Architectural Components of the Cooperation Schema in $\toy$}\label{fig:Arq_toy}
\end{figure}

The solvers and constraint stores for the domains $\mathcal{FD}$ and
$\mathcal{R}$ are provided by the SICStus Prolog constraint libraries. The impedance
mismatch problem among the host language constraint primitives and
these solvers is tackled by glue code (see Subsection \ref{coop}).
Proper $\mathcal{FD}-$ and $\mathcal{R}-$constraints,
as well as Herbrand constraints specific to $\fd$ and $\rdom$
(see Subsections \ref{fdom} and \ref{rdom}) are posted
to the respective stores and handled by the respective SICStus Prolog solvers.
On the other hand, the stores and solvers for  the domains ${\cal H}$ and ${\cal M}$
are built into the code of the $\toy$ implementation, rather than being provided by the
underlying SICStus Prolog system.
\vspace*{-0.3cm}

\subsection{Implementing Domain Cooperation} \label{coop}

This subsection explains the implementation of the fundamental
mechanisms for domain cooperation: Bridges and projections.
The constraints provided by the mediatorial domain $\mdom$
and their semantics have been explained in Subsections
\ref{cdomains} and \ref{ourcdom}.
{\em Mediatorial constraints} have the general form
$a$ {\tt \#==}$\, b\, \to!\, c$, with $a :: ${\tt int}, $b :: $ {\tt real} and $c :: $ {\tt bool},
while  {\em bridges} $a$ {\tt \#==} $b$ and {\em antibridges} $a$ {\tt \#/==} $b$
abbreviate $a$ {\tt \#==}$\, b\, \to!$ {\tt  true}
and $a$ {\tt \#==}$\, b\, \to!$ {\tt false}, respectively.

In order to deal with $\herbrand$ and $\mdom$ constraints, the $\toy$ system uses a so-called
{\em mixed store} which keeps a representation of the $\herbrand$ and $\mathcal{M}$ stores as one single Prolog structure.
It includes encodings of $\herbrand$-constraints in solved form
(i.e., totality constraints {\tt X == X} and disequality constraints {\tt X /= t}),
as well as encodings of  bridges and antibridges.
The implementation of the $\herbrand$ and $\mathcal{M}$ solvers in $\toy$ is plugged into the Prolog code of various predicates which control the transformation of the mixed store (passed as argument)  by means of two auxiliary arguments  {\tt Cin} and {\tt Cout}.

In the next three subsections we discuss the implementation of mediatorial constraints and projections.
We will show and comment selected fragments of Prolog code, involving various predicates with auxiliary arguments
 {\tt Cin} and {\tt Cout}, as explained above.
 Regarding projections, the $\toy$ implementation has been designed to support two modes of use:
A {\em `disabled projections'} mode which allows to solve mediatorial
constraints, but computes no projections; and a
{\em `enabled projections'} mode which also computes  projections.
For each particular problem, the user can analyze the trade-off
between communication flow and performance gain and decide the best
option to execute a goal in the context of a given program.
\vspace*{-0.2cm}

\subsubsection{The Equivalence Primitive  {\tt \#==}}\label{sec: implementacion bridge}

The  equivalence primitive  {\tt \#== :: int $\to$ real $\to$ bool}
used for building media\-torial constraints is implemented as a
Prolog predicate (also named  {\tt \#==}) with five arguments, whose
explanation follows. Arguments  {\tt L} and {\tt R} stand for the
left (integer) and right (real) parameters of the primitive {\tt
\#==}. Argument  {\tt Out} stands for its result. Finally, arguments
{\tt Cin} and {\tt Cout} stand for the state of the  mixed store
before and after performing a call to the primitive  {\tt \#==},
respectively. Fig. \ref{bridge_impl} shows the Prolog code for the
predicate {\tt \#==}, and the comments below explain why this code
implements the $\mathcal{M}$ solver described in Table \ref{mtable}
of Subsection \ref{ourcdom} and the special cooperation rules {\bf IE}
and {\bf ID} of the $\cclnc{\ccdom}$ calculus specified in Table
\ref{table7} from Subsection \ref{dcr}.

Lines {\tt (2)} and {\tt (3)} compute the head normal forms (hnfs)  of {\tt L} and {\tt R}
into  {\tt HL} and {\tt HR},  respectively. This process may generate
new Herbrand constraints that will  be added to the mixed store.
The value of {\tt HL} resp. {\tt HR} will be either a variable or a number,
ensuring that no suspensions will occur in the Prolog code from line {\tt (4)} on.
This code is intended to process the constraint {\tt HL \#== HR $\to!$ Out}
according to the behaviour of the $\mdom$-solver specified in in Table \ref{mtable}, Subsection \ref{ourcdom}.
Due to rules {\bf M1}  and {\bf M2} in Table \ref{mtable},
the constraint is handled as a bridge  {\tt HL \#== HR} when {\tt Out} equals {\tt true},
and as an antibridge  {\tt HL \#/== HR} when {\tt Out} equals {\tt false}.
 For this reason, one can say that the  {\tt \#==} primitive accepts {\em reification}.
 Indeed, in Fig. \ref{bridge_impl} we find that a bridge  {\tt HL \#== HR} is posted to the mixed store
 if the value for {\tt Out} can be unified with {\tt true} (line {\tt (6)}),
 whereas an antibridge  {\tt HL \#/== HR} is posted if the value for {\tt Out} can be unified with 
 {\tt false} (line {\tt (10)}).

\vspace*{-0.1cm}
\begin{figure}[htb]
{\small
\begin{verbatim}
(1) #==(L, R, Out, Cin, Cout):-
(2)    hnf(L, HL, Cin, Cout1),
(3)    hnf(R, HR, Cout1, Cout2),
(4)    tolerance(Epsilon),
(5)    ( (Out=true,
(6)          Cout3 = ['#=='(HL,HR)|Cout2],
(7)          freeze(HL, {HL - Epsilon =< HR, HR =< HL + Epsilon} ),
(8)          freeze(HR, (HL is integer(round(HR)))));
(9)      (Out=false,
(10)         Cout3 = ['#/=='(HL,HR)|Cout2],
(11)         freeze(HL, (F is float(HL), {HR =\= F})),
(12)         freeze(HR, (0.0 is float_fractional_part(HR) ->
(13)                  (I is integer(HR), HL #\= I); true)))),
(14)   cleanBridgeStore(Cout3,Cout).
\end{verbatim}
}
\caption{Implementation of Mediatorial Constraints ({\tt \#== / 2})}\label{bridge_impl}
\end{figure}
\vspace*{-0.2cm}

Solving  both bridges and antibridges is accomplished by using the
concurrent predicate {\tt freeze}, which suspends the evaluation of
its second argument until the first one becomes ground. Solving a
bridge   {\tt HL \#== HR} amounts to impose the equiva\-lence of its
two arguments (variables or constants), which are of different type
(integer and real), so that type casting is needed. Variable  {\tt
HL} is assigned to the integer version of {\tt HR} (line {\tt (8)})
via the Prolog functions {\tt round} and {\tt integer}, implementing
rule {\bf M3} in Table \ref{mtable}. Similarly, line {\tt (7)} is
roughly intended to assign the float version of {\tt HL} to {\tt HR}
in order to implement rule {\bf M5} in Table \ref{mtable}. However,
due to the imprecise nature of real solvers, occasionally {\tt HR}'s
value  will be an approximation to an integer value. Therefore, line
{\tt (7)} actually constrains the real variable {\tt HR} to take a
value between {\tt HL - Epsilon} and {\tt HL + Epsilon}, where {\tt
Epsilon} (line {\tt (4)}) is a user-defined parameter (zero by
default) which introduces a tolerance and avoids undesirable
failures due to inexact computations of integer values.
Lines  {\tt (7)} and {\tt (8)} also cover the implementation of rule  {\bf M6}
in Table \ref{mtable}. On the other hand, solving an antibridge
{\tt HL \#/== HR}  amounts to impose that both arguments are not equivalent.
Therefore, as soon as {\tt HL} or {\tt HR}  becomes bound to one
numeric value,  a disequality constraint between the (suitably type-casted)
value of the bound variable and its mate argument is posted to the proper
SICStus Prolog solver (lines {\tt 11-13}). The code in these lines implements
rule {\bf M8} in Table \ref{mtable} and rule {\bf ID} in Table \ref{table7}.

Moreover,  the failure rules in Table \ref{mtable}  (namely {\bf M4}, {\bf M7} and {\bf M9})
are also implemented by the frozen goals in  lines {\tt (7)-(8)} and {\tt (11)-(13)} of Fig. \ref{bridge_impl}.
Indeed,  whenever  {\tt HL} and  {\tt HR} become bound, the corresponding frozen goal is triggered and the equivalence
(resp. non-equivalence) is checked, which may yield to success or failure,
thus implementing rules {\bf M7} and {\bf M9};
and wherever {\tt HR} becomes bound to a non integral real value,
the frozen goal in line {\tt (8)} yields failure, thus implementing rule {\bf M4}.
Finally, line {\tt (14)}) invokes a predicate that simplifies the mixed store by
implementing the effect of rule {\bf IE} in  Table \ref{table7} applied as much as possible to all the
available (encodings of) bridges between variables.
\vspace*{-0.3cm}

\subsubsection{Projection: $\mathcal{FD}$ to $\mathcal{R}$}

If the user has enabled projections with the command {\tt /proj},
the $\toy$ system can process  a given atomic primitive $\fd$-constraint
by computing bridges and projected $\rdom$-constraints as explained in Subsection \ref{dcr}.
The Prolog implementation has a different piece of code (Prolog clause) for each
${\cal FD}$ primitive which can be used to build projectable constraints.
The information included in Table \ref{table2} for computing bridges and projections from
different kinds of $\fd$-constraints, as well as the effect of the goal transformation rules
in Table \ref{table4},  is  plugged into these pieces of Prolog code.
The code excerpt below shows the basic  behaviour of the implementation for the case of $\fd$-constraints
built with the inequality primitive  {\tt \#<}, without considering optimizations:

{\small
\begin{verbatim}
(1) #<(L, R, Out, Cin, Cout):-
(2)     hnf(L, HL, Cin, Cout1), hnf(R, HR, Cout1, Cout2),
(3)     ((Out=true,  HL #<  HR); (Out=false, HL #>= HR)),
(4)     (proj_active ->
(5)        (searchVarsR(HL, Cout2, Cout3, RHL),
(6)         searchVarsR(HR, Cout3, Cout, RHR),
(7)        ((Out==true,  { RHL < RHR });
(8)         (Out==false, { RHL >= RHR })));
(9)       Cout=Cout2).
\end{verbatim}
}

Following a technique similar to that explained for {\tt \#==} above, the primitive  {\tt \#<} is
implemented by a Prolog predicate with five arguments  (line {\tt (1)}). Its two input
arguments ({\tt L} and {\tt R}) are reduced to hnf (line {\tt (2)}), and a primitive
constraint is posted to the SICStus ${\cal FD}$-solver, depending on the Boolean
result ({\tt Out}) returned by  {\tt \#<} (line {\tt (3)}). Moreover, if projection is active (indicated by the dynamic predicate {\tt proj\_active} in line {\tt (4)}), then, the predicate {\tt searchVarsR}
(lines  {\tt (5-6)}) inspects  the mixed store looking  for bridges relating the $\fd$ variable {\tt HL}
and {\tt HR}  to the $\rdom$ variables {\tt RHL} and {\tt RHR}, respectively. In case that some
of these variables is bound to a numeric variable, the relation to the mate variable just means
that their numeric values are equivalent. Predicate  {\tt searchVarsR} also creates new bridges if necessary, according to the specifications  in Table \ref{table2}, and returns the modified state of the mixed store in its third argument. Finally, the projected constraints computed as specified in
Table \ref{table2} (in this case, a single constraint,  which is either {\tt RHL < RHR } or {\tt RHL >= RHR} depending on the value of {\tt Out}) are sent to the SICStus ${\cal R}$-solver.
\vspace*{-0.2cm}

\subsubsection{Projection: $\mathcal{R}$ to $\mathcal{FD}$}

If the user has enabled projections, the $\toy$ system can also
process  a given atomic primitive $\rdom$-constraint by computing
bridges and projected $\fd$-constraints as explained in Subsection
\ref{dcr}. The Prolog implementation is similar to that discussed in
the previous subsection, with a different piece of code (Prolog
clause) for each $\rdom$ primitive which can be used to build
projectable constraints, and encoding the information from Table
\ref{table5}. A comparison between
Tables \ref{table2} and \ref{table5} shows that there are less
opportunities for building bridges from $\rdom$ to $\fd$ than the
other way round, but more opportunities for building projections.
The code excerpt below shows the basic  behaviour of the
implementation for the case of $\rdom$-constraints built with the
inequality primitive  {\tt >}, ignoring optimizations:
\vspace*{0.2cm}

{\small
\begin{verbatim}
(1) >(L, R, Out, Cin, Cout):-
(2)  hnf(L, HL, Cin, Cout1), hnf(R, HR, Cout1, Cout),
(3)  (Out = true, {HR > HL} ; Out = false, {HL =< HR}),
(4)   (proj_active ->
(5)    (searchVarsFD(HL, Cout, BL, FDHL),
(6)     searchVarsFD(HR, Cout, BR, FDHR),
(7)     ((BL == true,  BR == true,  Out == true,  FDHL #>  FDHR);
(8)     (BL == true,  BR == true,  Out == false, FDHL #=< FDHR);
(9)     (BL == true,  BR == false, Out == true,  FDHL #>  FDHR);
(10)     (BL == true,  BR == false, Out == false, FDHL #=< FDHR);
(11)     (BL == false, BR == true,  Out == true,  FDHL #> FDHR);
(12)     (BL == false, BR == true,  Out == false, FDHL #=<  FDHR);
(13)      true); true).
\end{verbatim}
}

As in the previous subsection, the primitive  {\tt >} is
implemented by a Prolog predicate with five arguments  (line {\tt (1)}). Its two input
arguments ({\tt L} and {\tt R}) are reduced to hnf (line {\tt (2)}), and a primitive
constraint is posted to the SICStus $\rdom$-solver, depending
on the Boolean result ({\tt Out}) returned by  {\tt >} (line {\tt (3)}).
Moreover, if projection is active (line {\tt (4)}),
then predicate {\tt searchVarsFD} (lines  {\tt (5-6)}) inspects  the mixed store looking  for
bridges relating the $\rdom$-variables {\tt HL} and {\tt HR} to $\fd$-variables.
As shown in Table \ref{table5}, no new bridges can be created during this process.
Therefore, in contrast to the predicate {\tt searchVarsR} presented in the previous subsection,
the third argument of predicate {\tt searchVarsFD} does not represent a modified state of the mixed store.
Instead, it is a Boolean value that indicates whether a bridge has been found or not.
More precisely, in line {\tt (5)} there are two possibilities:
Either {\tt BL} is {\tt true} and {\tt HL} is a non-bound $\rdom$-variable related to the $\fd$-variable
{\tt FDHL} by means of some bridge in the mixed store {\tt Cout};
or else  {\tt BL} is {\tt false}, {\tt HL} is bound to a real value $u$, and {\tt FDHL} is
computed as $\lceil u \rceil$.
Analogously, in line {\tt (6)} there are two possibilities:
Either {\tt BR} is {\tt true} and {\tt HR} is a non-bound $\rdom$-variable related to the $\fd$-variable
{\tt FDHR} by means of some bridge in the mixed store {\tt Cout};
or else  {\tt BR} is {\tt false}, {\tt HR} is bound to a real value $u$, and {\tt FDHR} is
computed as $\lfloor u \rfloor$.
Finally, lines {\tt (7-12)} perform a distinction of cases corresponding to all the
possiblities for projecting the constraint {\tt HL > HR} $\to!$ {\tt Out} according to
Table \ref{table5} and the various values of {\tt BL}, {\tt BR} and {\tt Out}.
In each case, the projected $\fd$-constraint is posted to the SICStus $\fd$-solver.

As a concrete example, when solving the
conjunctive goal {\tt X \#== RX, RX > 4.3},
line {\tt (11)} in the Prolog code for {\tt >} just explained will
eventually work for solving the right subgoal. In this case, viewing
{\tt RX}  as  {\tt HL} and {\tt 4.3} as {\tt HR}, the value computed for  {\tt BL} will be {\tt true}
because the bridge {\tt X \#== RX} will be available in the mixed store, and
{\tt FDHL} will be {\tt X}. On the other hand, the value computed for
{\tt BR} will be {\tt false}, and the value of {\tt FDHR} will be computed as
$\lfloor${\tt 4.3}$\rfloor$, i.e. {\tt 4}. Applying the proper case in Table   \ref{table5},
the projected constraint  {\tt X \#> 4} will be posted to the SICStus $\fd$-solver.
\vspace*{-0.5cm}

%
%

\section{Performance Results}
\label{performance}

In this section we study the performance of  the systems  $\toy$
\cite{toyreport,estevez+:esop08} and META-S
\cite{DBLP:conf/ki/FrankHM03,DBLP:conf/flairs/FrankHM03,DBLP:conf/wlp/FrankHR05},
i.e., the closest related approach we are aware of, when solving
various problems requiring domain cooperation. After presenting a
set of benchmarks in the first subsections, the three following
subsections deal with  an analysis of the benchmarks in each of the
two systems and an a comparison between both.

\subsection{The Benchmarks}\label{subsect:The Benchmarks}

We have selected a reasonably wide set of benchmarks which allows to
analyze what happens when the set of constraints involved in the
formulation of a programming problem is solved differently depending
on the combination of domains that are involved in their solving.
A concise description of the benchmarks is presented below.
\vspace*{-0.1cm}

\begin{itemize}
\item {\bf Donald (donald)}:  A cryptoarithmethic problem with 10 ${\cal FD}$ variables, one linear
equation, and one {\em alldifferent} constraint. It consists of
solving the equation $\text{DONALD} + \text{GERALD} =
\text{ROBERT}$.

\item {\bf Send More Money (smm)}: Another cryptoarithmethic problem with 8 ${\cal FD}$
variables ranging over [0,9], one linear equation, 2 disequations
and one {\em alldifferent}  constraint. It consists of solving the
equation $\text{SEND} + \text{MORE} = \text{MONEY}$.

\item {\bf Non-Linear Crypto-Arithmetic (nl-csp)}: A problem with 9 ${\cal FD}$ variables
and non-linear equations.

\item {\bf Wrong-Wright (wwr)}: Another  cryptoarithmethic problem with
8 ${\cal FD}$ variables ranging over [1,9], one linear equation, and
one {\em alldifferent} constraint. It consists of solving the
equation $\text{WRONG} + \text{WRONG} = \text{RIGHT}$.

\item {\bf 3 $\times$ 3 Magic Square (mag.sq.)}: A problem that involves 9 ${\cal FD}$ variables and 7 linear
equations.

\item {\bf Equation 10 (eq.10)}: A system of 10 linear equations with
7 ${\cal FD}$ variables ranging over [0,10].

\item {\bf Equation 20 (eq.20)}: A system of 20 linear equations with
7 ${\cal FD}$ variables ranging over [0,10].

\item {\bf Knapsack (knapsack)}: A classical knapsack problem taken from
\cite{hooker2000}. We considered two versions: One as a constraint satisfaction
problem (labeled as {\bf csp}) and another one as an optimization one (labeled as {\bf opt}).

\item {\bf Electrical Circuit (circuit)}: A problem taken from \cite{hofstedt:tigher-cooperation-cl00},
in which one has an electric circuit with some connected resistors
(i.e., ${\cal R}$ variables) and a set of capacitors (i.e., ${\cal
FD}$ variables). The goal consists of knowing which capacitor has
to be used so that the voltage reaches the 99\% of the final
voltage between a given time range.

\item {\bf bothIn (goal2)}: The problem of solving the goal presented as
{\bf Goal 2} in Subsection~\ref{examples} for several values of {\tt n}.
Instances {\tt goal2(n)} of this benchmark correspond to solving
an instance of {\bf Goal 2}  for the corresponding {\tt n}.

\item {\bf bothIn (goal3)}: The problem of solving the goal presented as
{\bf Goal 3} in Subsection~\ref{examples} for several values of {\tt
n}. Instances {\tt goal3(n)} of this benchmark correspond to solving
an instance of {\bf Goal 3}  for the corresponding  {\tt n}.

\item {\bf Distribution (distrib)}. An optimized distribution problem
involving the cooperation of the domains $\rdom$ and $\fd$.
The problem deals with a communication network  where
{\tt NR} continuous and {\tt ND} discrete suppliers of raw material
have an attached cost to be minimized (see Appendix 8 in \cite{toyreport}).
The global optimum is computed.
The various instances  {\tt distrib(ND,NR)} of this benchmark
 correspond to different choices of values for {\tt ND} and {\tt NR}.
\end{itemize}
\vspace*{-0.1cm}

All the benchmarks were coded using ${\cal FD}$ variables and most
of them demand the solving of (non-)linear equations. Only the last
four of them strictly require cooperation between $\fd$ and $\rdom$
and  cannot be solved by using just one of these domains. However,
the rest of the benchmarks are also useful to evaluate the overhead
introduced when the different solvers are enabled. The formulation
of benchmarks {\bf nl-csp}, {\bf mag.sq}, {\bf circuit} and {\bf
smm} was taken from the distribution of META-S. Full details and
code of the benchmarks (written in both $\toy$ and META-S) are  available 
at {\tt http://www.lcc.uma.es/$\sim$afdez/cflpfdr/}.

All the benchmarking process was done using the same Linux machine
(under the professional version of Suse Linux 9.3) with an Intel
Pentium M processor running at 1.70GHz and with a RAM memory of 1 GB. 
In the rest of this section, performance numbers, in milliseconds,
have been computed as the average result of ten runs for each
benchmark.
In all tables, the best result obtained for each benchmark among
those computed under the various configurations has been highlighted
in {\it boldface}. \vspace*{-0.3cm}

\subsection{Benchmark Analysis in $\toy$} \label{subsect:TOY}

In this section we briefly present empirical  support for two claims:
a) that the activation of the cooperation mechanism
between ${\cal FD}$ and ${\cal R}$ does not penalize the execution time
in problems which can be solved  by using  the domain ${\cal FD}$;
and b)  that the cooperation mechanism using projections helps to speed-up the execution time
in problems where both the domain $\fd$ and the domain $\rdom$ are needed.

Tables \ref{toyfd}-\ref{toy all} show the performance of each
benchmark for several {\em configurations} of the $\toy$ system, as
explained below. The first column in each table displays the name of the benchmark to
be solved, and the next columns corresponds to  different {\em
activation modes}  of the $\toy$ system, namely:
\begin{itemize}
\item
${\cal TOY(FD)}$, an activation mode where the $\fd$ solver (but not
the $\rdom$ solver) is enabled. Actually, this corresponds to an
older version of the $\toy$ system which did not provide
simultaneous support for $\rdom$-constraints.
\item
${\cal TOY(FD + R)}$, an activation mode where both the $\fd$ solver
and the $\rdom$ solver are enabled, but the projection mechanism is
disabled.
\item
${\cal TOY(FD + R)}_p$, an activation mode where the $\fd$ solver,
the $\rdom$ solver and the projection mechanism are all
enabled.%
\end{itemize}

\vspace*{-0.3cm}
\begin{table}[h]
\begin{center}
\begin{tabular}{l|rr|rr|rr|}
\hline
                    \multicolumn{7}{c}{${\cal FD}$ Constraint Solving}                                                                                           \\
\hline
         & \multicolumn{2}{c|}{${\cal TOY(FD)}$}  & \multicolumn{2}{c|}{${\cal TOY(FD+R)}$}  &
         \multicolumn{2}{c|}{${\cal TOY(FD+R)}_{p}$}  \\
\hline
BENCHMARK           & {na\"{i}ve}     & ff            & {na\"{i}ve}          & {ff}               & {na\"{i}ve}                & {ff}            \\
\hline
donald              & 1078                  & 195           & 1040                 & {\bf 188}          & 7476                       &   678           \\
smm                  & 16                    & 15            & {\bf 14}             & 16                 & 47                         &   49            \\
nl-csp                & {\bf 15}              & 20            & {\bf 15}             & 18                 & 39                         &   86            \\
wwr                    & {\bf 18}              & 19            & {\bf 18}             & 19                 & 58                         &   52            \\
maq.sq.             & 92                    & 91            & 89                   & 89                 & {\bf 87}                   &   91            \\
eq.10                 & {\bf 74}              & 90            & {\bf 74}             & 81                 & 284                        &   261           \\
eq.20                 & 138                   & 134           & 139                  & {\bf 131}          & 431                        &   421           \\
knapsack (csp)      & {\bf 5}               & {\bf 5}       & {\bf 5}              & {\bf 5}            & {\bf 5}                    &   {\bf 5}       \\
knapsack (opt)      & 40                    & {\bf 15}      & 35                   & {\bf 15}           & 70                         &   40            \\
\hline
\end{tabular}
\end{center}
\caption{Solving ${\cal FD}$ Benchmarks in $\toy$ (First Solution Search). 
Overload evaluation.}\label{toyfd}
\end{table}
\vspace*{-0.3cm}

The heading `${\cal FD}$ Constraint Solving' in Table~\ref{toyfd}
indicates that all the benchmarks have been formulated in such a way
that all the constraints needed to solve them  are submitted to the
${\cal FD}$ solver and the $\rdom$ solver is not invoked. Note that
although  the activation mode ${\cal TOY(FD)}$ is sufficient to
execute all  the benchmarks presented in this  table, the benchmarks
have been also executed in the modes ${\cal TOY(FD + R)}$ and ${\cal
TOY(FD + R)}_p$ with the aim of analyzing the overhead caused by the
activation of these more complex modes when solving problems that do
not need them.

The heading `${\cal FD \sim R}$ Constraint Solving' in Tables
\ref{toyfdrdom}-\ref{toy all} indicates that the formulations of the
benchmarks require both the $\fd$ solver and the $\rdom$ solver to
be enabled; more precisely, although the benchmarks shown in Table
\ref{toyfdrdom} admit a natural formulation that can be totally
solved by the $\fd$ solver, we have used an alternative formulation
in which the (non-)linear constraints were submitted to the $\rdom$
solver, whereas the rest of the constraints were sent to the $\fd$
solver; also solving the benchmarks shown in Tables \ref{toyfdrdom2}
and \ref{toy all} strictly requires cooperation between $\fd$ and
$\rdom$. These tables only consider the two activation modes of the
$\toy$ system which make sense for such benchmarks, namely ${\cal
TOY(FD + R)}$ and ${\cal TOY(FD + R)}_p$.

Tables~\ref{toyfd}-\ref{toy all} also include two columns
corresponding to two different {\em labeling strategies}:  {\em
na\"{i}ve}, in which $\fd$ variables are labeled in a prefix order
(i.e., the leftmost variable is selected);
 and {\em first fail (ff)}, in which the $\fd$ variable with the smallest domain is chosen
 first for enumerating. Combined with the distinct  activation modes,
 this yields a number of configurations (i.e., six in Table~\ref{toyfd} and four
 in the rest).

\begin{table}[h]
\begin{center}
\begin{tabular}{l|rr|rr|} \hline
                     \multicolumn{5}{c}{${\cal FD \sim R}$ Constraint Solving}                                                  \\
                    \hline
                    & \multicolumn{2}{c|}{${\cal TOY(FD+R)}$}  & \multicolumn{2}{c|}{${\cal TOY(FD+R)}_{p}$}  \\ \hline
BENCHMARK           &{na\"{i}ve}           & {ff}               &
{na\"{i}ve}                & {ff}             \\ \hline
donald              & 304970               & 288700             & 8305                       &   {\bf 727}          \\
smm                 & 22528                & 22627              & 41                         &   {\bf 40}            \\
nl-csp              & 411                  & 383                & {\bf 44}                   &   87                  \\
wwr                 & 411                  & 420                & {\bf 54}                   &   58                  \\
maq.sq.             & 166                  & 168                & {\bf 158}                  &   163                 \\
eq.10               & {\bf 266}            & 271                & 290                        &   269                 \\
eq.20               & 402                  & 408                & 433                        &   {\bf 397}           \\
knapsack (csp)      & {\bf 5}              & {\bf 5}            & {\bf 5}                    &   {\bf 5}             \\
knapsack (opt)      & 16                   & 15                 & {\bf 11}                   &   14                  \\
\hline
\end{tabular}
\end{center}
\caption{Solving $\cal FD \sim R$ Benchmarks in $\toy$ (First Solution Search). Evaluation of the constraint projection mechanism.}\label{toyfdrdom}
\end{table}
\vspace*{-0.3cm}

\begin{sloppypar}
Inspection of Table~\ref{toyfd} reveals that the performance of all
the benchmarks does not get worse when moving from ${\cal TOY(FD)}$
to ${\cal TOY(FD + R)}$ and ${\cal TOY(FD + R)}_p$, and it even
improves in some cases. For those benchmarks that are most naturally
coded in the domain ${\cal FD}$ (as, for instance, {\bf smm}, {\bf
wwr} and {\bf mag.sq}) the best results are not those obtained in
${\cal TOY(FD + R)}_p$, but even in such cases the appreciable
overload is not a great one.
\end{sloppypar}


Inspection of Tables~\ref{toyfdrdom} and~\ref{toyfdrdom2}  reveals
that the projection mechanism causes a significant speed-up of the solving process in most cases.
Note that  this mechanism behaves specially well in solving the {\bf goal2(n)}  and
{\bf goal3(n)} benchmarks, where the running time is stabilized in the range between 11 ms 
and 17 ms when projections are enabled. Significant speed-ups (i.e., at least two or more magnitude orders) are also detected in {\bf donald} and {\bf smm} benchmarks as well as in
the different {\bf distrib} benchmark instances.

\begin{table}[h]
\begin{center}
\begin{tabular}{l|rr|rr|} \hline
                     \multicolumn{5}{c}{${\cal FD \sim R}$ Constraint Solving}                                                  \\
                    \hline
                    & \multicolumn{2}{c|}{${\cal TOY(FD+R)}$}  & \multicolumn{2}{c|}{${\cal TOY(FD+R)}_{p}$}  \\ \hline
BENCHMARK           &{na\"{i}ve}           & {ff}               &
{na\"{i}ve}                & {ff}             \\ \hline
circuit             & 14                   & {\bf 13}           & 14                         &   20                  \\
distrib (2,5.0)     & 662                  & 506                & {\bf 144}                  &   504    \\
distrib (3,3.0)     & 1486                 & 810                & {\bf 132}                  &   814    \\
distrib (3,4.0)     & 2098                 & 1290               & {\bf 156}                  &   1178   \\
distrib (4,5.0)     & 20444                & 12670              & {\bf 240}                  &   12744  \\
distrib (5,2.0)     & 29108                & 5162               & {\bf 198}                  &   7340   \\
distrib (5,5.0)     & 141734               & 85856              & {\bf 272}                  &   86497  \\
distrib (5,10.0)    & 568665               & 464230             & {\bf 474}                  &   462980 \\
goal2 (100)        & 25                   & 28                 & {\bf 14}                   &   {\bf 14}            \\
goal2 (200)        & 40                   & 44                 & {\bf 13}                   &       15            \\
goal2 (400)        & 70                   & 72                 & {\bf 12}                   &       13            \\
goal2 (800)        & 131                  & 135                & {\bf 12}                   &       15            \\
goal2 (10000)      & 704                  & 713                & {\bf 14}                   &       16            \\
goal2 (20000)      & 1271                 & 1270               & {\bf 12}                   &       16            \\
goal2 (40000)      & 2325                 & 2333               & {\bf 11}                   &       16            \\
goal2 (80000)      & 4452                 & 4472               & {\bf 13}                   &       16            \\
goal2 (200000)     & 10725                & 10781              & {\bf 13}                   &       15            \\
goal3 (100)        & 18                 & 20               & {\bf 15}                   & 16            \\
goal3 (200)        & 26                 & 28               & {\bf 13}                   & {\bf 13}      \\
goal3 (400)        & 41                 & 44               & {\bf 15}                   & 16            \\
goal3 (800)       & 75                 & 77               & {\bf 16}                   & 17            \\
goal3 (5000)     & 354                & 360              & {\bf 14}                   & 16          \\
\hline
\end{tabular}
\end{center}
\caption{Solving $\cal FD \sim R$ Benchmarks in $\toy$ (First Solution Search). Evaluation of the constraint projection mechanism
on benchmarks necessarily demanding solver cooperation.}\label{toyfdrdom2}
\end{table}


Finally, Table  \ref{toy all}  presents the results corresponding to computing all the results
 for  the last five benchmarks in Table \ref{toyfdrdom2}. The execution times are naturally
 higher than those shown in Table \ref{toyfdrdom2}, where only first solutions were computed.
 However, the significant speed-up caused by the activation of projections remains clearly
 observable.

\begin{table}[htb]
\begin{center}
\begin{tabular}{l|rr|rr|} \hline
                     \multicolumn{5}{c}{${\cal FD \sim R}$ Constraint Solving}                                                      \\
                    \hline
                    & \multicolumn{2}{c|}{${\cal TOY(FD+R)}$}  & \multicolumn{2}{c|}{${\cal TOY(FD+R)}_{p}$}  \\ \hline
BENCHMARK           &{na\"{i}ve}           & {ff}               &
{na\"{i}ve}                & {ff}             \\ \hline
goal3 (100)      &   673               &   625               & 265            &  {\bf 242} \\
goal3 (200)      &   1867             &  1844              & {\bf 329}    &  352      \\
goal3 (400)      &   6527             &  6573              & 583            &  {\bf 579}       \\
goal3 (800)      &   24460           &  24727           & {\bf 976}     &  994       \\
goal3 (5000)    &   911880    &  920670     & {\bf 5365}   &  6135      \\
\hline
\end{tabular}
\end{center}
\caption{Solving {\bf goal3(n)} Benchmarks in $\toy$ (All Solutions Search).}\label{toy all}
\end{table}

\subsection{Benchmark Analysis in META-S}\label{subsect:META-S}

In this subsection, we present the results of executing benchmarks
in META-S, a flexible meta-solver framework that implements the ideas proposed
in \cite{Hofstedt:phd-thesis-2001,HP07} for the dynamic integration
of external stand-alone solvers to enable the collaborative
processing of constraints. As already mentioned in Sections
\ref{introduction} and \ref{cooperative}, the cooperative framework
underlying META-S bears some analogies with the approach described in this paper. Both
META-S and $\toy$ provide means for different numeric constraints
domains to cooperate. $\toy$ supports cooperation between the
domains $\herbrand$,  $\fd$ and $\rdom$, while META-S connects
several kind of solvers, such as:
\begin{itemize}
\begin{sloppypar}
\item A $\fd$ solver (for floats, strings, and rationals) that was implemented in
Common Lisp using as reference a library of routines for solving
binary constraint satisfaction problems provided by Peter van Beek
and available from {\tt http://www.ai.uwaterloo.ca/$\sim$vanbeek/software/\-csplib.tar.gz}.
\end{sloppypar}

\item A solver for linear arithmetic, i.e., the
constraint solver L{\small IN}A{\small R} described in
\cite{Krzikalla97}. This solver is based on the Simplex algorithm
and was implemented in the language C. It handles linear equations,
inequalities, and disequations over rational numbers.

\item
\begin{sloppypar}
An interval arithmetic solver, that uses the sound math library
(available at {\tt http://interval.sourceforge.net/interval/index.html}), an ANSI C
library implemented on the basis of the solver for interval
arithmetic of Timothy J. Hickey from Brandeis University (available
from {\tt http://www.cs.brandeis.edu/$\sim$tim/}).
\end{sloppypar}
\end{itemize}

The interested reader is referred to \cite{frankmai:td2002} for more details
on the META-S solvers.
There are also some other significant differences between both systems.
META-S is implemented in Common Lisp whereas $\mathcal{TOY}$ is
implemented in Prolog. In contrast to $\toy$, META-S does not
support different activation modes (corresponding to ${\cal
TOY(FD)}$, ${\cal TOY(FD + R)}$ and ${\cal TOY(FD + R)}_p$ in
$\toy$), neither explicit labeling strategies, nor  facilities for
optimization. On the other hand, META-S supports the choice of
different constraint solving strategies
\cite{DBLP:conf/ershov/FrankHPR06}, which is not the case in $\toy$.
More details regarding the comparison between $\toy$ and META-S can
be found in Subsection \ref{subsect:versus}.


We have investigated  the performance of META-S in solving the
benchmarks already considered for $\mathcal{TOY}$ in the previous
section and the performance results are shown in
Tables~\ref{META-S}-\ref{META-S all}. The organization of rows and
columns is also similar to the $\toy$ tables (but considering the
two different strategies explained below). The occurrences of the
symbol `$-$'  indicate that the corresponding benchmark (namely, the
knapsack optimization and the distribution problem) could not be
executed because the META-S system provides no optimization
facilities; the term `error' corresponds with a failure returned by
the system, that was not able to solve the goal.
We have used the version 1.0 of META-S (kindly provided by its
implementors on our request) compiled using SuSE Linux version 9.3
(professional version), based on CMU Common Lisp 18d.

\vspace*{-0.3cm}
\begin{center}
\begin{table}[htb]
\begin{tabular}{l|rr|rr|} \hline
                     \multicolumn{5}{c}{META-S}                                   \\ \hline
                    & \multicolumn{2}{c|}{eager}    & \multicolumn{2}{c|}{heuristic}  \\ \hline
BENCHMARK           & standard        & ordered     & standard        & ordered       \\ \hline
donald              & 268510          & 469370      & {\bf 5290}      & 6140          \\
smm                 & 950             & 620         & 590             & {\bf 580}     \\
nl-csp              & 344800          & 1230        & 302314          & {\bf 970}     \\
wwr                 & 10930           & 650         & {\bf 620}       & {\bf 620}     \\
maq.sq.             & 1160            & 1220        & {\bf 520}       & 540           \\
eq.10               & {\bf 60}        & {\bf 60}    & 70              & 70            \\
eq.20               & {\bf 60}        & {\bf 60}    & 70              & 70            \\
knapsack (csp)      & {\bf 60}        & {\bf 60}    & 70              & 70            \\
knapsack (opt)      & -               & -           & -               & -             \\
distrib (2,5.0)     & -               & -           & -               & -              \\
distrib (3,3.0)     & -               & -           & -               & -              \\
distrib (3,4.0)     & -               & -           & -               & -              \\
distrib (4,5.0)     & -               & -           & -               & -              \\
distrib (5,2.0)     & -               & -           & -               & -              \\
distrib (5,5.0)     & -               & -           & -               & -              \\
distrib (5,10.0)    & -               & -           & -               & -              \\
circuit             & {\bf 70}        & {\bf 70}    & {\bf 70}        & {\bf 70}      \\
goal2 (100)   & {\bf 330}       & {\bf 330}   & {\bf 330}       & {\bf 330}    \\
goal2 (200)   & {\bf 730}       & 740         & 740             & 740          \\
goal2 (400)   & {\bf 2340}      & {\bf 2340}  & {\bf 2340}      & 2350         \\
goal2 (800)   & 8550            & {\bf 8540}  & 8560            & 8560         \\
goal3 (100)  & {\bf 410}      & {\bf 410}    & 460             &  460           \\
goal3 (200)  & {\bf 900}      & {\bf 900}    & 1080            &  1080          \\
goal3 (400)  & {\bf 2870}     & 2880         & 3520            &  3540          \\
goal3 (800)  & {\bf 10630}    & 10720        & 13140           &  13370         \\
\hline
\end{tabular}
\caption{Solving the Benchmarks in META-S (First Solution Search).}\label{META-S}
\end{table}
\end{center}
\vspace*{-0.5cm}


For the META-S benchmarks we have utilized the combination of the
${\cal FD}$ solver (usually for rationals) and an arithmetic solver
which was found analogous to the $\fd$ plus $\rdom$ combination used
in the corresponding $\toy$ benchmark. In fact, for META-S, we have
selected the linear arithmetic solver since the interval arithmetic
solver yielded poorer results in all cases. In addition, we have
considered the best problem formulation (in terms of the target
solver for each constraint) that yielded the best running time.
Moreover, we have executed each META-S benchmark  under four
different constraint solving strategies:
\begin{itemize}
\item {\em Standard eager}, in which all constraint information is propagated
as early as possible.
\item {\em Ordered  eager}, working as the previous one complemented with
user-given information for determining the order of projection operations.
\item {\em Standard heuristic}, working as the standard eager strategy complemented
with an heuristic for giving priority to those variable bindings more likely to lead to
failure.
\item {\em Ordered heuristic}, working as the previous one complemented with
user-given information for determining the order of projection operations.
\end{itemize}

In certain form, na\"{i}ve and ff labeling in ${\cal TOY}$ are
similar, respectively, to eager and heuristic strategies in META-S.
For the sake of  a fair comparison, whenever possible we have
encoded the META-S benchmarks using exactly the same problem
formulations as well as the same constraints that were used in the
corresponding $\mathcal{TOY}$ benchmarks.
Benchmarks were coded using the functional logic language FCLL of
META-S. Also, we took care that the variable orders were identical
for the different resolution/labeling strategies in both systems.

\vspace*{-0.3cm}
\begin{center}
\begin{table}[htb]
\begin{tabular}{l|rr|rr|} \hline
                     \multicolumn{5}{c}{META-S}
                     \\ \hline
                    & \multicolumn{2}{c|}{eager}    & \multicolumn{2}{c|}{heuristic}  \\ \hline
BENCHMARK           & standard        & ordered     & standard &
ordered       \\ \hline
goal3 (100)  &      8930      &       8880   & {\bf 6940}      &  {\bf 6940}    \\
goal3 (200)  & {\bf 60700}    &       60870  & 47190           &  {\bf 46880}   \\
goal3 (400)  & {\bf 453330}   &      459980  & {\bf 346930}    &   348900       \\
goal3 (800) & error          &   error     & error                &   error        \\
\hline
\end{tabular}
\caption{Solving {\bf goal3(n)} Benchmarks in META-S (All Solutions Search).} \label{META-S all}
\end{table}
\end{center}
\vspace*{-.5cm}

Note that the META-S benchmarks shown in Table \ref{META-S} (resp.
Table \ref{META-S all}) correspond to the $\toy$ benchmarks in
Tables \ref{toyfd}-\ref{toyfdrdom2} (resp. Table \ref{toy all}), all
of which refer to first solution search (resp. all solutions search).
\vspace*{-0.3cm}

\subsection{$\toy$ versus META-S}\label{subsect:versus}

The tables displayed in this subsection are intended to compare the performance of  ${\cal TOY}$ and META-S. 
Table \ref{META-S vs TOY} compares the behaviour of both systems when computing the first solution of various benchmarks, 
while the results in  Table \ref{META-S vs TOY all} correspond to the computation of all the solutions for a few instances of the 
benchmark {\bf goal3(n)}. More precisely, the execution times and META-S/$\toy$ rates displayed in Table \ref{META-S vs TOY} 
correspond to the best results for each benchmark under those obtained for the various configurations in Tables 
\ref{toyfdrdom}-\ref{toyfdrdom2} and \ref{META-S}, respectively;  while Table \ref{META-S vs TOY all} has been built from 
the information displayed in Tables \ref{toy all} and \ref{META-S all} in a similar way.

\vspace*{-0.3cm}
\begin{center}
\begin{table}[htb]
\begin{tabular}{lrrr} \hline
SYSTEM            &  ${\cal TOY}$ &  META-S   & META-S/${\cal TOY}$ \\ \hline
donald                & 188                 & 5290          & 28.13               \\
smm                    & 14                   & 580             & 41.42               \\
nl-csp                  & 15                   & 970             & 64.66               \\
wwr                     & 18                   & 620              & 34.44               \\
maq.sq.              & 87                   & 520              & 5.97                \\
eq.10                  & 74                   & 60                & 0.81                \\
eq.20                  & 131                 & 60                & 0.45                \\
knapsack (csp) & 5                      & 60                & 12                  \\
knapsack (opt)  & 11                   & -                    & -            \\
circuit                 & 13                     & 70                & 5.38                \\
goal2 (100)       & 14                     & 330             & 23.57                \\
goal2  (200)      & 13                     & 730             & 56.15                \\
goal2  (400)      & 12                     & 2340           & 195.00                 \\
goal2  (800)      & 12                     & 8540           & 711.66                 \\
goal3 (100)  & 15                          & 410             & 27.33               \\
goal3 (200)  & 13                         & 900              & 69.23               \\
goal3 (400)  & 15                         & 2870            & 191.33                 \\
goal3 (800) & 16                          & 10630          & 664.375                 \\
\hline
\end{tabular}
\caption{Solving Benchmarks in ${\mathcal TOY}$ vs. META-S (First Solution Search)} \label{META-S vs TOY}
\end{table}
\end{center}
\vspace*{-0.5cm}

The analogies and differences between the domain cooperation
mechanisms supported by $\toy$ and META-S have been discussed at the
end of  Subsection \ref{dcr}.  In both cases, projections play a key
role, and the information displayed on Tables \ref{META-S vs TOY}
and \ref{META-S vs TOY all} allows mainly to draw  conclusions on
the computational performance of both systems. META-S seems to
behave particularly well in the solving of linear equations,
especially when the problem requires no global constraints
(such as an {\em alldifferent} constraint used in benchmarks  {\bf eq10} and {\bf eq20}).
The reason maybe two-fold: First, that the linear arithmetic solver
of META-S performs better than its $\fd$ solver, and, second, that flattening
a nested constraint in ${\cal TOY}$ generates as many flat
constraints as the number of operators it includes.

However, in general, $\toy$ shows an improvement of about one order
of magnitude with respect to the META-S system, for the benchmarks
used in our comparison. As an extreme case, the computation time for
obtaining the first solution of  the benchmark {\bf goal3(800)}
increases more than three orders of magnitude with respect to ${\cal
TOY}$, and computing all the solutions for this benchmark in META-S
does not succeed. In certain form, the experimental results suggest
that our proposal is not only promising but also interesting in its
current state.

\vspace*{-.3cm}
\begin{center}
\begin{table}[htb]
\begin{tabular}{lrrr} \hline
SYSTEM       &    ${\cal TOY}$ &    META-S   & META-S/${\cal TOY}$ \\ \hline
goal3 (100)  &  242         & 6940       & 28.67            \\
goal3 (200)  &  329         & 46880     & 142.49             \\
goal3 (400)  &  579         & 346930   & 599.18                \\
goal3 (800)  &  976         & error        &  -                 \\
\hline
\end{tabular}
\caption{Solving {\bf goal3(n)}  in   ${\mathcal TOY}$ vs. META-S (All Solutions Search)} \label{META-S vs TOY all}
\end{table}
\end{center}
\vspace*{-.5cm}

In any case, the `superior' performance of $\toy$ with respect to
META-S has to be interpreted carefully. One reason for $\toy$'s
advantage may be that the numerical solvers connected in the current
version of META-S have been implemented just
to experiment with the concepts of the underlying theoretical framework
described in \cite{Hofstedt:phd-thesis-2001,HP07}, without much concern for optimization,
while $\toy$ relies on the optimized solvers provided by SICStus Prolog. 
Another advantage of $\toy$ is the availability of global constraints such as  
{\em alldifferent}, that  are lacking in META-S. 
Admittedly, a better comparison of the performance results in both systems 
would be obtained by comparing independently the integrated solvers in each of the
systems, and then normalizing the global results for the systems; or
alternatively, by  connecting the same solvers to both systems.
This would be possible if all the integrated solvers were effectively
black-boxes that can be unplugged from the systems. Unfortunately,
this is not the case, as the solvers attached to $\toy$ are used as
provided by SICStus Prolog and they were not internally adjusted to
work in a cooperation system, whereas the solvers used in META-S
were implemented with regard to their integration into the
implementation of META-S as a system with cooperating components.

In favor of META-S, we mention that the cooperation model proposed
in META-S seems to be more flexible than the cooperation model
currently implemented in $\toy$, and provides facilities not yet
available in $\toy$. For instance, META-S allows to integrate and/or
redefine evaluation strategies \cite{DBLP:conf/flairs/FrankHM03}
whereas $\toy$ relies on a fixed strategy for goal solving and
constraint evaluation. Also, the projection mechanism currently
implemented in $\toy$ is less powerful than in META-S, because
projections cannot be applied to the constraints inside the constraint
stores. Finally, META-S enables the integration of different host
languages \cite{DBLP:conf/wlp/FrankHR05}, whereas
the $\cclnc{\ccdom}$ goal solving calculus implemented  in $\toy$
is intended for declarative languages fitting the  $CFLP$ scheme.

%
%

\section{Related Work}
\label{relatedwork}

In this section, in addition to already mentioned related works, we
extend the discussion to other proposals developed in the area of
cooperative constraint solving. Of course, the issues of
communication and cooperation are relevant to many aspects of
computation. Here, we discuss a selection of the literature
concerning proposals for communication and cooperation in constraint
and declarative programming. Existing cooperative systems are very
diverse and range from domain combinations to a mix of distinct
techniques for solving constraints over the same domain. Moreover,
the cooperating systems may be very different in nature: Some of
them perform complete constraint solving whereas others just
execute basic forms of propagation. In general, depending on the
nature of the cooperation, we catalogue cooperative constraint
solving in four non-disjoint categories:

\begin{enumerate}
\item Cooperation of (built-in) domains coexisting in the same
system.

\item Interchange of information between different solvers/domains via special
constructs.

\item Interoperability or communication between independent
solvers.

\item Combination or integration of entities with distinct nature (i.e., methods and/or solvers based on different
algorithms, or languages with different resolution mechanisms).

\end{enumerate}

In the following four subsections we discuss some of the relevant work done in each
of these categories, as well as their relation to our own approach.
\vspace*{-0.2cm}

\subsection{Cooperation of (Built-In) Domains Coexisting in the Same System} \label{bd}

There are a number of constraint systems that provide support for
the interaction between built-in and predefined domains. In these
systems, a solver is viewed as a device that
transforms the original set of constraints to an equivalent
reduced set. As examples, we can cite the following systems:

\begin{itemize}
\item CLP(BNR) \cite{benhamou:applying-jlp97}, Prolog III \cite{colmerauer:prolog-iii-cacm90}
     and Prolog IV \cite{n'dong:prologIV-jfplc97}
     allow solver cooperation, mainly limited to booleans,
     reals and naturals (as well as term structures such as lists and trees).

\item The language NCL \cite{zhou:the-language-NCL-jlp2000}
      provides an integrated constraint framework that strongly combines boolean logic,
      integer constraints and set reasoning.
      Currently, NCL also integrates efficient CP domain cutting techniques and OR
      algorithms.
\end{itemize}

Most existing systems of this kind have two main problems:
firstly, the cooperation is restricted to a limited set
of computation domains supported by the system;
and secondly, the cooperation mechanism is very dependant on the involved computation
domains and thus presents difficulties to be generalized to other
computation domains.

Our computational model for the cooperation of the domains $\herbrand$, $\fd$ and $\rdom$
and its current  $\toy$ implementation can be  cata\-logued in this category, and insofar it shares
the two limitations just mentioned.
However, our approach is more general because it is based on a generic scheme for $CFLP$ programming over
a parametrically given coordination domain $\ccdom$.
The cooperative goal solving calculus $\clnc{\ccdom}$ presented in
Section \ref{cooperative} refers to the particular coordination domain
$\ccdom = \mdom \oplus \herbrand \oplus \fd \oplus \rdom$,
but it  can be easily extended to other coordination domains,
as sketched in our previous paper
\cite{DBLP:journals/entcs/MartinFHRV07}.
\vspace*{-0.2cm}

\subsection{Interchange of Information between Computation Domains and/or Solvers via Special Constructs} \label{sc}

Another cooperation technique consists of providing special
built-in constructs designed to propagate information
among different computation domains that coexist in the same system.
For example, this is  the case with the reified constraints that
enable a communication between arithmetic computations  and a Boolean domain.

Within this type  of cooperation  we can cite {\em Conjunto}
\cite{gervet:clp(sets)-constraints97}, a constraint
language for propagating interval constraints defined over finite sets of integers.
This language provides so-called {\em graduated constraints} which map
sets onto arithmetic terms, thus allowing a one-way cooperative channel
from the set domain to the integer domain. Graduated constraints can be used
in a number of applications as, for instance, to handle optimization problems
by applying a cost function to quantifiable terms (i.e., arithmetic terms which are associated to set terms).

Also, a generic framework for defining and solving interval
constraints on any set of domains (finite or infinite) with a
lattice structure is formulated in \cite{DBLP:journals/toplas/FernandezH04,DBLP:journals/jucs/FernandezH06}.
This approach also belongs to the cooperation category described in Subsection \ref{bd}.
It enables the construction of new (compound) constraint solvers from existing
solvers using lattice combinators, so that different solvers (possibly on distinct domains) can communicate
and cooperate in solving a problem.
The $\mathit{clp({\cal L})}$ language presented in \cite{DBLP:journals/toplas/FernandezH04} is
a  prototype implementation of this framework and allows information to be transmitted among different
(possibly user-defined) computation domains.

Our proposal in this paper can be also considered to fit into the special
constructs category, by viewing  bridge constraints as channels that
enable the propagation of information between different computation domains.
\vspace*{-0.2cm}

\subsection{Interoperability}\label{sect:interoperafbility}

A number of recent publications deal with approaches to solver cooperation requiring {\em interoperability},
understood as the behaviour of some  coordinating system that  supports communication  between  several autonomous systems.
In such settings,  cooperation relies on suitable interfaces,
which have to be specified and implemented according to the specific formats required by the various domains and solvers.

For instance,
 \cite{goualard:componet-programing-ercim01}
proposes a C++ constraint solving library called aLiX for
communicating between different solvers, possibly written in
different languages. Two of the main aims of aLiX are to permit the
transparent communication  of solvers and ensure {\em type safety},
that is to say, the capacity to prevent {\em a priori} the
connection of a solver  that does not conform to
 the input format of the interface with another solver. The current version of
aLiX is not mature yet,  although its interoperability
 approach offers
interesting possibilities. One of the main  shortcomings of the
current aLiX version is that a component for solving continuous
constraints is not yet integrated into the system, and thus real
constraints  cannot be processed.

In the same spirit,  many constraint systems provide both a linear and a non-linear solver for the real domain.
As the linear solver is more efficient of the two, it should be used whenever the constraints are linear,
and there is a need for communication between the two real solvers.
As an example,
\cite{Monfroy+:implmenting-nn-linear-cons-coop-researchreport95}
describes a client/server architecture to enable communication
between the component solvers. This consists of managers for
the system and the solvers that must be defined on the same
computational domain (the real numbers, for example) but with
different classes of admissible constraints (i.e., linear and
non-linear constraints). The constraint logic programming (CLP)
system {\em CoSAc} is an implementation of this architecture. A
built-in platform permits the integration and connection of the
components. The exchange of information is managed by means of
pipes  and the exchanged data are character strings. One
of the main drawbacks of this system is the lack of type safety.
Moreover, the cooperation happens at a fixed level that prevents the
communication of solvers in a transparent way, since the solvers
cannot obtain additional information from the structure  of the
internal constraint store.
As already discussed at the end of Subsection \ref{dcr}, the current $\toy$ implementation
of our cooperative computation model suffers from a similar limitation,
preventing the constraints already placed into the $\fd$ and $\rdom$ stores to be projected.
This issue should be addressed in future improvements of our system.

As CoSAc does not permit solver combination, Monfroy designed a
domain independent environment for {\em solver collaboration}, and
he used this concept in order to unify solver cooperation and
combination. Basically, solver cooperation means the use of several
solvers with data exchange between them, whereas solver combination
is understood as the construction of new solvers  from other previously
defined solvers. In his PhD thesis \cite{monfroy:phd96}, Monfroy developed the system BALI
(Binding Architecture for Solver Integration)
that facilitates the integration of heterogeneous solvers as well as the specification of
solver cooperation via a number of cooperation primitives. Monfroy's approach assumes that
all the solvers work over a common store, while our own proposal
requires communication among different stores.
Monfroy also designed SoleX \cite{MR99},
a domain-independent scheme for constraint solver extension. This
schema consists of a set of rules for transforming constraints that
cannot be managed by a solver into constraints that can be treated by that solver,
thus extending  the range of  solvable constraints. Unfortunately, as commented
in \cite{monfroy:phd96} (page 195), SoleX and BALI were not
integrated. Such an integration could lead to a
framework including both solver collaboration and solver extension.

The interoperability category also includes a line of research dealing with
the development of coordination languages, aiming at the specification of cooperation between solvers.
There exist several proposals whose main goal is to study the use of
control languages to specify elementary constraint solvers as well
as the collaboration of solvers in a uniform and flexible way.
For instance, \cite{DBLP:conf/sac/ArbabM98} proposes to use the
coordination language MANIFOLD for improving the constraint solver collaboration language of BALI.
More recent works such as \cite{DBLP:conf/advis/MonfroyC04,DBLP:conf/aimsa/CastroM04}
aim at providing means of designing strategies that are
more complex than simple master­-slave approaches.
Basically, Castro and Monfroy propose an asynchronous language composed of
interaction components that control external agents (in particular solvers) by managing the data flow.
A software framework for constructing distributed constraint solvers,
implemented  in the coordination language MANIFOLD,
has been described in \cite{DBLP:conf/sac/Zoeteweij03}.
A different point of view regarding solver cooperation is analyzed in \cite{PM03},
where a paradigm to enable the user to separate computation
strategies from the search phases is presented.

Also it is worth mentioning the project COCONUT 
\footnote{See {\tt http://www.mat.univie.ac.at/~neum/glopt/coconut/}}
whose goal was to integrate techniques from mathematical programming, constraint
programming and interval analysis (and thus it can also be
catalogued in the category of cooperation via techniques combination
as described in Section~\ref{da}). A modular solver environment,
that can be extended with open-source and commercial solvers, was
provided for nonlinear continuous global optimization. This
framework was also designed for distributed computing and has a
strategy engine that can be programmed using a specific interpreted
language based on Python.

Mircea Marin has developed in his PhD thesis \cite{Mar00} a $CFLP$ scheme that
combines Monfroy's approach to solver cooperation \cite{monfroy:phd96} with a
higher-order lazy narrowing calculus somewhat similar to
\cite{LMR04} and the goal solving calculus presented in Section \ref{cooperative} of this
paper. In this setting, Monfroy's ideas are used to provide various primitives for solver combination,
and the $CFLP$ scheme allows to embed the resulting solvers into a Functional and Logic Programming language.
In contrast to our proposal, Marin's approach allows for
higher-order unification, which  leads both to greater expressivity
and to less efficient implementations.
Another difference w.r.t. our approach is the intended application domain.
The instance of $CFLP$ implemented by Marin and others \cite{MIS99b} combines four
solvers over a constraint domain for algebraic symbolic computation.
This line of research has been continued in works such as
\cite{DBLP:conf/aplas/KobayashiMI01,kobayashi+:Open-cflp-wflp02,kobayashi+:collaborative-cflp-ieice-2003,kobayashi1}.
These papers describe a collaborative $CFLP$ system, called Open $CFLP$, which solves symbolic constraints by collaboration between
distributed constraint solvers in an open environment such as Internet.
The solvers act as providers of constraint solving services, and
Open $CFLP$ is able to use them without knowing their location and implementation details.
The common communication
infrastructure (i.e., the protocol) and the specification language
were implemented using CORBA and MathML respectively.

Another recent proposal for the combination of solvers in a declarative programming language
can be found in  \cite{banda+:building-cs-in-hal-clp01}.
This paper deals with the construction of solvers  in the HAL system,
which supports the extension of existing solvers and the
construction of hybrid ones.  HAL provides semi-optional type, mode
and determinism declarations for predicates and functions as well as
a system of type classes over which constraint solvers' capabilities are specified.
In particular, HAL type classes can require that the types belonging to them must
have a suitable associated constraint solver.

A quite general scheme for solver cooperation fitting the
interoperability category has been proposed by Hofstedt et al. in
\cite{hofstedt:tigher-cooperation-cl00,DBLP:conf/cp/Hofstedt00,Hofstedt:phd-thesis-2001,HP07}.
Here, constraint domains are formalized by using $\Sigma$-structures
in a sorted language, constraints are modeled as n-ary relations,
and cooperation of solvers is achieved  by two mechanisms:
Constraint propagation, that submits a  constraint belonging to its
corresponding store; and projection of constraint  stores, that
consults  the contents of a given store  $S_{\mathcal{D}}$ and
deduces constraints for another domain. Relying on these mechanisms,
different constraint solvers (possibly working over different
domains, and implemented in various languages) can be used as
components of  an overall system, whose architecture provides a
uniform interface for constraint solvers which allows a fine-grain
formal specification of information exchange between them. This
approach has been implemented  in the system META-S
\cite{DBLP:conf/ki/FrankHM03,DBLP:conf/flairs/FrankHM03,DBLP:conf/wlp/FrankHR05},
 that supports the dynamic integration of arbitrary external (stand-alone) solvers to enable the
collaborative processing of constraints.
Some analogies and differences between this approach and our own
have been discussed already at several places of this paper;
see the Introduction, and sections \ref{dcr} and \ref{performance}.

As a more theoretical line of work related to the interoperability category,
there are a number of formal  approaches to the
combination of constraint solvers on domains modeled as algebraic structures.
This kind of research stems from a seminal paper by Nelson and Oppen \cite{NO79}.
More recent relevant work includes several papers by Baader and Schulz.
For instance, \cite{baader+:combination-cp95} provides an abstract framework to
combine constraint languages and constraint solvers and focuses on
ways in which different and independently defined solvers may be
combined. This paper does not really deal with the constraint
cooperation mechanism, but it focuses in defining algebraic properties needed
for the combination of constraint languages and solvers.
Later on, \cite{baader+:combination-cs-free-quasi-struct-tcs98}
generalized a proposal from a previous paper \cite{baader+:unification-disjoint-equ-theories-jsc96}
and presented a general method for the combination of constraint systems,
which is is applicable to so-called {\em quasi-structures}.
This general notion comprises various instances, such as
(quotient) term algebras, rational trees, lists, sets, etc.
The methods proposed in
\cite{baader+:unification-disjoint-equ-theories-jsc96,baader+:combination-cs-free-quasi-struct-tcs98}
can be seen as extensions of previous approaches to the combination of unification
algorithms for equational theories, viewing them as instances of constraint solvers \cite{KR92,KR94}.
As pointed out in \cite{kepser+:optimisation-combining-solvers-frocos98},
a weak point of these approaches is the lack of practical use.

Our proposal can clearly be catalogued in the interoperability category,
because it aims at the cooperation of several constraint domains equipped with their respective solvers.
Our main communication mechanism, namely bridges, has the
advantage of  syntactic simplicity, while being compatible with the static type systems used by many declarative languages.
Moreover, our notion of coordination domain allows us to use a generic scheme for $CFLP$ programming as a formal foundation.
\vspace*{-0.2cm}

\subsection{Combining Methods and/or Solvers based on Different Algorithms} \label{da}

One popular approach to cooperation consists of
combining solvers or methods based on different algorithms. In
this category we include the integration of different paradigms in
one language. In the following, we provide a (non-exhaustive) list
of proposals of this kind.

For instance, one of the initial forms of cooperating constraint
solving consisted of using different problem solvers (viewed as
algorithms) to work individually over different sub-parts of an overall problem.
This was the approach used  in \cite{durfee+:hai89} in order to integrate within
a network a number of individual solvers intended to work  over different parts of a problem.
In a similar way,
\cite{khedro+:cooperation-distributed-csp-ecai94} proposed a
multi-agent model where each agent acts independently to solve a
distributed set of constraints that constitutes a distributed
constraint satisfaction problem.
The paper  \cite{hong:tech-report94} also  studied the confluence of solvers
to solve a common problem, suggesting to manage a set of algorithms
each of which should be repeatedly applied on the problem until reaching a stable form.

Within the area of Constraint Programming,
\cite{benhamou:heterogeneous-cs-alp96} described a unified framework
for heterogeneous constraint solving. Here, the cooperation comes
from the combination of different algorithms, possibly defined over
distinct structures. The main idea is to represent the solvers as
constraint narrowing operators (CNO) that are closure operators, and
to use a generalized notion of arc-consistency. Conditions on the
CNOs needed to ensure the main properties of the principal algorithm
are identified. Solver communication involving  shared common
variables and sending and receiving information to each other is
described. The paper also gives a fixed point semantics to describe
the cooperation process. One of the main drawbacks of this proposal
is that termination of the central algorithm  relies on the
finiteness of an {\em approximate domain} $A$ built as a subset of
the powerdomain $\wp(D)$ of the domain $D$ under consideration,
including $D$ among its members and closed under intersection. For
instance, termination  cannot be guaranteed in case that $D$ is  the
domain of sets of real numbers, which is useful for dealing with
real interval constraints.

In relation to the problem  of solving real constraints,
\cite{benhamou+:revising-hull-box-iclp99} has proposed the combination
of hull consistency and box consistency with the objective of
reducing the computation time achieved by using box consistency alone.
This idea was reflected in DecLic~\cite{benhamou+:declic-programming-wic97,goualard+:hybrid-jflp99},
a $CLP$ language that mixes boolean, integer and real constraints in
the framework of intervals. This system was shown to be fairly
efficient on classical benchmarks but at the expense of decreasing
the declarativity of the language as a consequence of allowing the
programmer to choose the best kind of consistency to use for each
constraint.

The combination of interval techniques for solving
non-linear systems is also tackled in  \cite{granvilliers+:combination-interval-solvers-realiable2001},
Granvilliers describes a cooperative strategy to combine the interval-based local consistencies methods
(i.e., box and hull consistency) with the multi-dimensional interval Newton
method and shows the efficiency of the main algorithm.

Another proposal for developing a cooperation technique for
solving large scale combinatorial optimization problems was
described in \cite{DBLP:conf/sbia/CastroMR04}. This paper
introduces a framework for designing cooperative strategies,
describing a scheme for the sequential cooperation
between Forward Checking and Hill­-Climbing. A set of benchmarks
for the Capacity Vehicle Routing Problem shows the advantages of
using this framework, that always outperforms a single solver.

The combination of linear programming solvers and interval
solvers has also been specially fertile in the last decades
\cite{marti95distributed}. Many of the cooperating systems resulting
from this combination have been implemented as (prototype) declarative systems,
as e.g., ICE \cite{DBLP:books/el/beierleP95/BeringerB95},
Prolog IV \cite{n'dong:prologIV-jfplc97},
CIAL \cite{DBLP:journals/rc/ChiuL02}
and CCC \cite{DBLP:journals/rc/RueherS97}, among others.

The integration of mathematical programming techniques in the $CLP$ scheme
\cite{vanHoeve:integrating-clp-mp2000} may be considered another form of cooperation that
has been treated extensively in the literature;
see  e.g.  the integration of Mixed Integer programming and $CLP$
\cite{rodosek+:integrating-MIL-CLP-annaoper97,harjunkoski+:hybrid-mil-CLP-cce2000,DBLP:conf/cp/Thorsteinsson01},
the combination of $CLP$ and Integer Programming
\cite{bockmayr+:unifyinf-integer-fd-informs98} and the combination
of $CLP$ and Linear Programming \cite{vandecasteele-modelling}, among others.

The domain cooperation framework presented in this paper is quite generic,
 and its current implementation in $\toy$ relies on the availability of black box solvers provided by SICStus Prolog \cite{SP}.
Therefore, it cannot be catalogued into the cooperation category considered in this subsection,
which is very specific and relies on a  detailed control of the techniques and solvers involved.
Nevertheless, the work described in this subsection points to combination techniques
which lead to improved performance and may be useful for future implementations of our approach.
\vspace*{-0.3cm}
%
%

\section{Conclusions and Future Work}
\label{conclusions}

The work presented in this paper is aimed as a contribution to the efficient use
of constraint domains and solvers in declarative languages and systems.
We have investigated foundational and practical issues
concerning a computational framework for the cooperation of
constraint domains in Constraint Functional Logic Programming,
using constraint projection guided by bridge constraints as the main cooperation tool.
Taking a generic scheme as a formal basis,
we have focused on a particular case of practical importance,
namely the cooperation among the symbolic Herbrand domain $\herbrand$
and the two numeric domains $\rdom$ and $\fd$.

The relation to our previous related work
and some pointers to related work by other researchers
have been presented in Section \ref{introduction},
and a more detailed discussion of the state-of-the-art concerning cooperation
of constraint domains can be found   in Section \ref{relatedwork}.
In the rest of this section we give a  summary of the main results presented in the other sections of the paper,
followed by some considerations concerning current limitations and possible lines of future work.
\vspace*{-0.3cm}

\subsection{Summary of Main Results}

Our results include a formal computation model  for cooperative goal solving in  Constraint Functional Logic Programming,
the development of an implemented system,
and  experimental evidence on the implementation's performance
and its comparison with the closest related system we are aware of.
More precisely:

\begin{itemize}
\item
In Section \ref{coordination} we have presented a formal
framework for the cooperation of constraint domains in an improved version of
an already existing $CFLP$ scheme for Constraint Functional Logic Programming.
We have formalized a notion of constraint solver suitable for $CFLP$ programming,
as well as a mathematical construction of coordination domains,
a special kind of hybrid domains built as a combination of several pure domains intended to
cooperate. In addition to the facilities provided by their components,
coordination domains supply special primitives for building bridge constraints to allow communication between
different component domains. As particular case of practical
interest, we have formalized a coordination domain
$\ccdom = \mdom \oplus \herbrand \oplus \fd \oplus \rdom$ tailored to the
cooperation of three useful pure domains: the Herbrand domain
$\herbrand$ which supplies equality and disequality constraints over
symbolic terms, the domain $\rdom$ which supplies arithmetic
constraints over real numbers, and the domain $\fd$ which supplies
finite domain constraints over integer numbers. Practical
applications involving more that one of these pure domains can be
naturally treated within the $\cflp{\ccdom}$ instance of the $CFLP$ scheme.
From a programmer's viewpoint, the domain $\herbrand$ supports generic equality and disequality constraints
over arbitrary user defined datatypes, while $\rdom$ and $\fd$ provide more specific numeric constraints.

 \item
Section \ref{cooperative} presents a formal calculus for cooperative goal solving in $\cflp{\ccdom}$.
The main programming features available to  $\cflp{\ccdom}$ programmers
include a Milner's like polymorphic type system,
lazy and possibly higher-order functions, predicates,
and the cooperation of the three domains witihin $\ccdom$.
The goal solving calculus is presented as a set of goal transformation rules
for reducing initial goals into solved forms.
There are rules that use  lazy narrowing to process program defined function calls in a demand-driven way,
domain cooperation rules dealing among other things  with bridges and projections,
and constraint solving rules to invoke the solvers of the various pure domains involved in the cooperation.
The section concludes with theoretical results ensuring soundness and completeness of the goal solving calculus,
where completeness is guaranteed for well-typed solutions as far as permitted by the completeness of the underlying solvers
and some other more technical requirements.

\item
Section \ref{implementation} presents the implementation of the cooperative goal solving calculus for
$\cflp{\ccdom}$  in a state-of-the-art declarative programming system.
In addition to describing general aspects such as the software architecture, 
we have focused on the implementation of domain cooperation mechanisms,
illustrating the correspondence between code generation in the implemented system and
the goal transformation rules for cooperation formalized in the previous section.

\item
Section \ref{performance} is devoted to performance analysis by means of a set of benchmarks.
The experimental results obtained lead us to several conclusions.
Firstly, we conclude  that the activation of the domain cooperation mechanisms between
$\fd$ and $\rdom$ does not penalize the execution time in
problems which can be solved by using the domain $\fd$ alone.
Secondly, we also conclude that the cooperation mechanism using projections
helps to speed-up the execution time in problems where a real cooperation between $\fd$ and $\rdom$ is needed.
Thirdly, our experiments show a good performance of our implementation
with respect to the closest related system we are aware of.
In summary, we conclude that our approach to the cooperation of constraint domains
has been effectively implemented in a practical system,
that is distributed as a free open-source Sourceforge project
({\tt http://toy.sourceforge.net}) and runs on several platforms.
\end{itemize}
\vspace*{-0.3cm}

\subsection{Some Current Limitations and Planned Future Work}

In the future we would like to improve some of the limitations of our current approach to domain cooperation,
concerning both the formal foundations and the implemented system.
More precisely:

\begin{itemize}

\item
The cooperative goal solving calculus $\cclnc{\ccdom}$ presented in Section \ref{cooperative} should be generalized to allow for
an arbitrary coordination domain  $\ccdom$ in place of the  concrete choice $\mdom \oplus \herbrand \oplus \fd \oplus \rdom$.
This is a straightforward task. However, for the purposes of the present paper we found more appropriate to deal just with the
coordination domain supported by the current implementation.

\item
The implemented system should be expanded to support some of these more general coordination domains,
which could include specific domains for boolean values, sets, or different types of numeric values.
More efficient and powerful constraint solvers for such domains should be also integrated within the implementation.

\item
$\cclnc{\ccdom}$ should be also expanded to allow the computation of projections
from the primitive constraints placed within the constraint stores.
These more powerful projections were allowed in the preliminary version
of $\cclnc{\ccdom}$ presented in \cite{DBLP:journals/entcs/MartinFHRV07},
but they were not implemented  and no completeness result was given.
Currently, projections are computed only from  the constraints placed in the constraint pool
(see rule {\bf PP} in Table \ref{table4} in Subsection \ref{dcr}) and the $\toy$ implementation only
supports this kind of projections.
Allowing projections to act over stored constraints will require to solve new problems both
on the formal level (where some substantial difficulties are expected for proving a completeness result)
and on the implementation level (where the current system will have to be modified to enable a transparent access to the constraint stores).

\item
As a consequence of the previous improvement, the cooperative goal solving process will show more complicated patterns of interaction
among solvers. Therefore, some means to describe goal solving strategies should be provided to enable users to specify some desired 
sequences of goal transformation rules, especially with regard to the activation of solvers and projections.
In addition to being implemented as part of  the practical system, goal solving strategies are expected to be helpful for proving 
the completeness of a cooperative goal solving calculus improved as described in the previous item.

\item
The experimentation with benchmarks and application cases should be further developed.

\item
Last but not least, the implemented system should be properly maintained and improved in various ways.
In particular, library management should be standardized, both with respect to loading already existing libraries and
with respect to developing new ones.

\end{itemize}

\vspace*{-0.3cm}


%
%

\bibliographystyle{acmtrans}
\bibliography{ccs}

\appendix

\appendix\label{appendix}

\section{[Auxiliary Results and Proofs]}\label{proofs}

This Appendix collects proofs of the results stated in Section
\ref{coordination} and Section \ref{cooperative} omitted from the
main text. Some of them rely on previously stated auxiliary results,
especially Lemmata \ref{tpl} and \ref{trl} from Subsection
\ref{expressions} and  Lemma \ref{sl} from Subsection \ref{dom}.
In addition, some other auxiliary results will be included at the
proper places.

\subsection{Properties of Constraint Solvers and Coordination Domains} \label{pSolvCdom}

The first part of the Appendix includes the proofs of the main results stated in
Section \ref{coordination}. First we present the proof of Lemma
\ref{psts}, about general properties of proof transformation
systems.

\begin{proof}[Proof of Lemma \ref{psts}]
\begin{enumerate}
\item
The transition relation $\sts{\cdom}{\varx}$ of the $sts$  generates
a tree with root $\Pi\, \Box\, \varepsilon$, whose leaves correspond
to the stores belonging to $\mathcal{SF}_{\cdom}(\Pi,\varx)$. Since
$\sts{\cdom}{\varx}$ is finitely branching and terminating, this
tree is locally finite and has no infinite branches. By so-called
{\em K$\ddot{o}$nig's Lemma} (see  \cite{BN98}, Section 2.2) the
tree must be finite. Therefore, it must have finitely many leaves,
and $\mathcal{SF}_{\cdom}(\Pi,\varx)$ is finite. For later use, we
remark that $solve^{\cdom}(\Pi,\varx)$ can be characterized as
$$\bigvee \{\exists \overline{Y'} (\Pi'\, \Box\, \sigma') \mid
\Pi\, \Box\, \varepsilon \vdash\!\!\vdash_{\cdom,\, \varx}!\,\,
\Pi'\, \Box\, \sigma',\, \overline{Y'} = var(\Pi'\, \Box\, \sigma')
\setminus var(\Pi)\}$$
\item
Assume that the $sts$ has the fresh local variables property and the safe bindings property.
Due to the remark at the end of item 1.,
for each $\varx$-solved form $\exists \overline{Y'} (\Pi'\, \Box\, \sigma')$ computed by the call
$solve^{\cdom}(\Pi,\varx)$ there is some sequence of $\sts{\cdom}{\varx}$ steps
$$\Pi\, \Box\, \varepsilon =
\Pi'_0\, \Box\, \mu'_0\, \vdash\!\!\vdash_{\cdom,\, \varx}\,
\Pi'_1\, \Box\, \mu'_1\, \vdash\!\!\vdash_{\cdom,\, \varx}\,
\ldots\, \vdash\!\!\vdash_{\cdom,\, \varx}\,
\Pi'_n\, \Box\, \mu'_n$$
such that $\Pi'_n\, \Box\, \mu'_n = \Pi'\, \Box\, \sigma'$ is irreducible,
and the following conditions hold for all $1 \leq i \leq n$:
 $\Pi'_i\, \Box\, \mu'_i$ is a store with fresh local variables
 $\overline{Y'_i}  = var(\Pi'_i\, \Box\, \mu'_i) \setminus var(\Pi'_{i-1}\, \Box\, \mu'_{i-1})$;
 $\mu'_i = \mu'_{i-1} \mu_i$ for some substitution $\mu_i$ verifying
 $vdom(\mu_i) \cup vran(\mu_i)$  $\subseteq var(\Pi'_{i-1}) \cup  \overline{Y'_i}$; and
 $\mu_i(X)$ is a constant for all $X \in \varx \cap vdom(\mu_i)$.
 Then, $\overline{Y'} = \overline{Y'_1},\, \ldots,\, \overline{Y'_n}$, and
 an easy induction on $n$ allows to prove that
 $vdom(\sigma') \cup vran(\sigma') \subseteq var(\Pi) \cup  \overline{Y'}$
 and that $\sigma'(X)$ is a constant for all $X \in \varx \cap vdom(\sigma')$.
 Therefore, the solver $solve^{\cdom}$ also
 satisfies the  fresh local variables property and the safe bindings property.
\item
Assume that the $sts$ is locally sound. Because of the remark in
item 1., to prove soundness of $solve^{\cdom}$ it is sufficient to
show that the union
$$ \bigcup
\{Sol_{\cdom}(\exists\overline{Y'}(\Pi' \, \Box\, \sigma')) \mid \Pi
\, \Box \,  \sigma \vdash\!\!\vdash_{\cdom,\, \varx} ! \, \, \Pi'\,
\Box\, \sigma' , \overline{Y'}  = var(\Pi' \, \Box \, \sigma')
\setminus var(\Pi \,  \Box \, \sigma)\}$$ is a subset of
$\sol{\cdom}{\Pi\, \Box\, \sigma}$. In order to show this, we assume
$$\Pi\, \Box\, \sigma \vdash\!\!\vdash^{n}_{\cdom,\, \varx}!\,\,
\Pi'\, \Box\, \sigma',\, \overline{Y'}  = var(\Pi'\, \Box\,
\sigma')\setminus var(\Pi\, \Box\, \sigma)$$ and prove
$\sol{\cdom}{\exists \overline{Y'} (\Pi'\, \Box\, \sigma')} \subseteq
\sol{\cdom}{\Pi\, \Box\, \sigma}$ by induction  on $n$:

\underline{$n = 0$}: In this case $\overline{Y'} = \emptyset$, $\Pi'
\, \Box\, \sigma' = \Pi\, \Box\, \sigma$. The inclusion to be proved
is trivial.

\underline{$n > 0$}: In this case
 $\Pi\, \Box\, \sigma \vdash\!\!\vdash_{\cdom,\, \varx} \Pi'_1 \, \Box\, \sigma'_1
\vdash\!\!\vdash^{n-1}_{\cdom,\, \varx}! \, \, \Pi'\, \Box\,
\sigma'$ for some store $\Pi'_1 \, \Box\, \sigma'_1$. Let
$\overline{Y'_1}  = var(\Pi'_1 \Box \sigma'_1)\setminus var(\Pi\,
\Box\, \sigma)$ and $\overline{Y''}  = var(\Pi'\, \Box\,
\sigma')\setminus var(\Pi'_1 \Box \sigma'_1)$. Then $\overline{Y'} =
\overline{Y'_1}, \overline{Y''} = var(\Pi'\, \Box\,
\sigma')\setminus var(\Pi\, \Box\, \sigma)$. By induction
hypothesis, we can assume $\sol{\cdom}{ \exists \overline{Y''}
(\Pi'\, \Box\, \sigma')} \subseteq \sol{\cdom}{\Pi'_1\, \Box\, \sigma'_1}$.
Then, for any given $\eta \in
Sol_{\cdom}(\exists\overline{Y'} (\Pi' \, \Box\, \sigma'))$ we can
prove $\eta \in  Sol_{\cdom}(\Pi \, \Box\, \sigma)$ by the following
reasoning: by definition of $Sol_{\cdom}$, there is $\eta' \in
Sol_{\cdom}(\Pi' \, \Box\, \sigma')$ such that $\eta'  =_{\setminus
\overline{Y'}} \eta$ and hence $\eta'  =_{var(\Pi\, \Box\, \sigma)}
\eta$. Trivially, it follows that $\eta' \in
Sol_{\cdom}(\exists\overline{Y''} (\Pi' \, \Box\, \sigma'))$, which
implies
 $\eta' \in  Sol_{\cdom}(\Pi'_1 \, \Box\, \sigma'_1)$
by induction hypothesis. Trivially again, it follows that $\eta' \in
Sol_{\cdom}(\exists\overline{Y'_1} (\Pi'_1 \, \Box\, \sigma'_1))$
which implies $\eta' \in  Sol_{\cdom}(\Pi \, \Box\, \sigma)$ due to
local soundness. Since $\eta'  =_{var(\Pi\, \Box\, \sigma)} \eta$,
we can conclude that $\eta \in  Sol_{\cdom}(\Pi \, \Box\, \sigma)$.
\item
Assume now a selected set $\mathcal{RS}$ of $str$s such that the
$sts$ is locally complete for $\mathcal{RS}$-free steps. Because of
the remark in item 1., to prove completeness of $solve^{\cdom}$ for
$\mathcal{RS}$-free invocations it is sufficient to show that
$WTSol_{\cdom}(\Pi\, \Box\, \sigma)$ is a subset of the union
$$ \bigcup
\{WTSol_{\cdom}(\exists\overline{Y'}(\Pi' \, \Box\, \sigma')) \mid
\Pi \, \Box \,  \sigma \vdash\!\!\vdash_{\cdom,\, \varx} ! \, \,
\Pi'\, \Box\, \sigma',\, \overline{Y'}  = var(\Pi' \, \Box \,
\sigma') \setminus var(\Pi \,  \Box \, \sigma)\}$$ under the
additional assumption that $\Pi\, \Box\, \sigma$ is hereditarily
$\mathcal{RS}$-irreducible.
This can be viewed as a property of the store $ \Pi\, \Box\, \sigma$
that can be proved by {\em well-founded induction} (see  again
\cite{BN98}, Section 2.2) on the terminating  store transformation
relation $\sts{\cdom}{\varx}$:

\underline{Base Case}: $\Pi\, \Box\, \sigma$ is irreducible w.r.t. $
\vdash\!\!\vdash_{\cdom,\, \varx}$. In this case, the union reduces
to the set $WTSol_{\cdom}(\Pi\, \Box\, \sigma)$ and the inclusion to
be proved is trivial.

\underline{Inductive Case}: $\Pi\, \Box\, \sigma$ is reducible
w.r.t. $\vdash\!\!\vdash_{\cdom,\, \varx}$.
 In this case, since   $\Pi\, \Box\, \sigma$ is hereditarily $\mathcal{RS}$-irreducible and
 the $sts$ is locally complete for $\mathcal{RS}$-free steps,
for any $\eta \in  WTSol_{\cdom}(\Pi \, \Box\, \sigma)$ there is
some  hereditarily $\mathcal{RS}$-irreducible $(\Pi'_1 \, \Box\,
\sigma'_1)$ such that $\Pi\, \Box\, \sigma
\vdash\!\!\vdash_{\cdom,\, \varx} \Pi'_1 \, \Box\, \sigma'_1$ and
$\eta \in WTSol_{\cdom}(\exists\overline{Y'_1}(\Pi'_1 \, \Box\,
\sigma'_1))$ where  $\overline{Y'_1}  = var(\Pi'_1 \Box
\sigma'_1)\setminus var(\Pi\, \Box\, \sigma)$. Then, by definition
of  $Sol_{\cdom}$, there is $\eta'_1 \in  WTSol_{\cdom}(\Pi'_1 \,
\Box\, \sigma'_1)$ such that $\eta' _1 =_{\setminus \overline{Y'_1}}
\eta$. The induction hypothesis can be assumed for $\Pi'_1 \, \Box\,
\sigma'_1$, and there must be some $\Pi' \, \Box\, \sigma'$ such
that $\Pi'_1 \, \Box \,\sigma'_1 \vdash\!\!\vdash_{\cdom,\, \varx} !
\, \Pi'\, \Box\, \sigma'$, $\overline{Y''}  = $ $var(\Pi'\, \Box\,
\sigma')\setminus var(\Pi'_1 \Box \sigma'_1)$ and $\eta'_1 \in
WTSol_{\cdom}(\exists\overline{Y''}(\Pi' \, \Box\, \sigma'))$. By
definition of $Sol_{\cdom}$, there is $\eta' \in
WTSol_{\cdom}(\Pi'\, \Box\, \sigma')$ such that $\eta' =_{\setminus
\overline{Y''}}  \eta'_1$. Moreover,   we get $\Pi \, \Box \,
\sigma \vdash\!\!\vdash_{\cdom,\, \varx} ! \, \, \Pi'\, \Box\,
\sigma'$ and $\overline{Y'} = \overline{Y'_1}, \overline{Y''} =
var(\Pi' \, \Box \, \sigma') \setminus var(\Pi \,  \Box \, \sigma)$
such that $\eta' =_{\setminus \overline{Y'}} \eta$, and thus $\eta
\in  WTSol_{\cdom}(\exists\overline{Y'}(\Pi' \, \Box\, \sigma'))$.
\end{enumerate}
This completes the proof of the Lemma.\hfill
\end{proof}

Table \ref{hordering} displayed in the next page and the two
auxiliary Lemmata stated and proved immediately afterwards will be
used in the subsequent proof of Theorem \ref{hsolver}, the main
result in this subsection. It ensures that $solve^{\herbrand}$
satisfies the requirements for solvers listed in Definition
\ref{defSolver} (except for a technical limitation concerning
completeness). The proof of this theorem also relies on  Lemma
\ref{psts}.

\begin{lemma}[Auxiliary Soundness Lemma]\label{ASL}

Assume $\Pi \subseteq  PCon_{\cdom}$ and $\sigma, \sigma_1 \in
Sub_{\cdom}$ such that $\sigma$ is idempotent and  $\Pi\sigma =
\Pi$. Then $Sol_{\cdom}(\Pi \sigma_1) \cap
Sol_{\cdom}(\sigma\sigma_1) \subseteq Sol_{\cdom}(\Pi) \cap
Sol_{\cdom}(\sigma)$.
\end{lemma}
 \vspace*{-.3cm}
\begin{proof}[Proof of Lemma \ref{ASL}]
The hypothesis of the Lemma say that $\sigma = \sigma\sigma$ and
$\Pi\sigma = \Pi$. On the other hand, because of the Substitution
Lemma \ref{sl} and the definition of $Sol_{\cdom}$, any $\eta \in
Val_{\cdom}$ verifies $\eta \in Sol_{\cdom}(\Pi \sigma_1) \cap
Sol_{\cdom}(\sigma\sigma_1)$ iff $\sigma_1\eta \in \sol{\cdom}{\Pi}$
and $\sigma\sigma_1\eta = \eta$. Therefore, to prove the lemma it
suffices to assume

\begin{center}
$(a)~ \sigma = \sigma\sigma$ \qquad $(b)~ \Pi\sigma = \Pi$ \qquad
$(c)~ \sigma_1\eta \in \sol{\cdom}{\Pi}$ \qquad $(d)~
\sigma\sigma_1\eta = \eta$
\end{center}

\noindent and to deduce from these assumptions that $\eta \in
Sol_{\cdom}(\Pi) \cap Sol_{\cdom}(\sigma)$.

First we prove that $\eta \in Sol_{\cdom}(\Pi)$ as follows: by $(c)$
and $(b)$, we obtain  $\sigma_1 \eta \in Sol_{\cdom}(\Pi\sigma)$,
which amounts to $\sigma\sigma_1 \eta \in Sol_{\cdom}(\Pi)$ by the
Substitution Lemma. By $(d)$, this  is the same as $\eta \in
Sol_{\cdom}(\Pi)$.

Next, we note that $\eta \in Sol_{\cdom}(\sigma)$ is equivalent to
$\sigma \eta = \eta$, which can be proved by the following chain of
equalities: $\sigma\eta =_{(d)} \sigma\sigma\sigma_1 \eta =_{(a)}
\sigma\sigma_1\eta  =_{(d)} \eta$.\hfill
 \end{proof}

 \begin{lemma}[Auxiliary Completeness Lemma]\label{ACL}
Assume $\Pi \subseteq  PCon_{\cdom}$, $\sigma, \sigma_1 \in
Sub_{\cdom}$ and $\eta, \eta' \in  Val_{\cdom}$ such that $\eta \in
Sol_{\cdom}(\Pi) \cap Sol_{\cdom}(\sigma)$, $\sigma_1 \eta' = \eta'$
and $\eta' =_{\setminus \overline{Y'}} \eta$, where $\overline{Y'}$
are fresh variables away from $var(\Pi) \cup vdom(\sigma) \cup
vran(\sigma)$. Then $\sigma \eta' = \eta'$ and  $\eta' \in
Sol_{\cdom}(\Pi \sigma_1) \cap Sol_{\cdom}(\sigma\sigma_1)$.
\end{lemma}

\begin{proof}[Proof of Lemma \ref{ACL}]
In what follows  we can assume $\sigma \eta = \eta$ due to the
hypothesis $\eta \in Sol_{\cdom}(\sigma)$.

We prove $\sigma \eta' = \eta'$ by showing that $X \sigma \eta' = X
\eta'$ holds for any variable $X \in \var$. This is trivial for $X
\notin vdom(\sigma)$. For $X \in vdom(\sigma)$, we can assume that
$\overline{Y'}$ is away from $X$ and  $var(X\sigma)$; therefore
$\eta' =_{X, var(X\sigma)} \eta$ and hence $X \sigma \eta' = X
\sigma \eta = X  \eta =X  \eta'$ (where the assumption $\sigma \eta
= \eta$ has been used at the 2nd step).

Now we prove  $\eta' \in Sol_{\cdom}(\Pi \sigma_1)$. Because of the
Substitution Lemma \ref{sl}, this is equivalent to $\sigma_1\eta'
\in Sol_{\cdom}(\Pi)$, which amounts to $\eta \in Sol_{\cdom}(\Pi)$
due to the hypothesis $\sigma_1 \eta' = \eta'$, $\eta' =_{\setminus
\overline{Y'}} \eta$ and $\overline{Y'}$ away from $var(\Pi)$. But
$\eta \in Sol_{\cdom}(\Pi)$ is also ensured by the hypothesis.

Finally, $\eta' \in Sol_{\cdom}(\sigma\sigma_1)$ is equivalent to
$\sigma \sigma_1 \eta' = \eta'$, which can be proved as follows:
$\sigma \sigma_1 \eta' = \sigma  \eta' =  \eta'$ (where the 1st step
relies on the  assumption $\sigma_1\eta' = \eta'$ and the 2nd step
relies on a previously proved equality).\hfill
\end{proof}

\begin{proof}[Proof of Theorem  \ref{hsolver}]

Consider the $sts$ for $\herbrand$ stores with transition relation
$\sts{\herbrand}{\varx}$ as specified  in Table \ref{htable} in
Subsection \ref{hdom}, implicitly assuming that the notation used
for the various $str$s is exactly the same as there. We prove that
this $sts$ satisfies the six properties enumerated in Definition
\ref{defpsts}. The last  one (namely {\bf Local Completeness}) holds
for $\mathcal{URS}$-free steps, where
$\mathcal{URS} = \{{\bf OH3}, {\bf OH7}, {\bf H13}\}$
is the set of unsafe $\herbrand$-$str$s, as explained in Subsection \ref{hdom}.

\begin{enumerate}
\item
{\bf  Fresh Local Variables Property:} The specification of  the
$str$s in Table \ref{htable} clearly guarantees this property.
\item
{\bf Safe Bindings Property:} An inspection of Table \ref{htable}
shows that the $str$s {\bf H1} and {\bf H2} bind a variable to a
constant, and the other $str$s never bind a variable $X \in \varx$.
Therefore, this property is also satisfied.
\item
{\bf Finitely Branching Property:} This property holds because those
$str$s that  allow a non-deterministic choice of the next store
provide only finitely many possibilities.
\item
{\bf Termination Property:} Given  a $\herbrand$ store $\Pi\, \Box\,
\sigma$ and a set $\varx \subseteq cvar(\Pi)$, we define a 5-tuple
of natural numbers $||\Pi\, \Box\, \sigma||_{\varx}  =_{def} (P_1,
P_2, P_3, P_4, P_5) \in \mathbb{N}^5$ where
\begin{itemize}
\item [\bf P$_1$] is the  number of occurrences of atomic constraints in $\Pi$ which are {\em unsolved} w.r.t. $\varx$.
In this context, an atomic constraint $\pi$ occurring in $\Pi$ is
said to be unsolved w.r.t. $\varx$ iff some of the $str$s can be
applied taking $\pi$ as the selected atomic constraint.
\item [\bf P$_2$]  is the sum of the depths of all the occurrences of variables $X \in \varx$ within patterns in $\Pi$.
\item [\bf P$_3$]  is the sum of the syntactical sizes of all the patterns occurring in $\Pi$.
\item [\bf P$_4$]  is the number of {\em unsolved} occurrences of obviously demanded variables in $\Pi$.
In this context, an occurrence of an obviously demanded variable $X$
in $\Pi$ is called solved iff $X$ occurs in a constraint of the form
$X$ {\tt ==} $X$, and unsolved otherwise.
\item [\bf P$_5$]  is the number of occurrences of {\em misplaced} variables in $\Pi$.
In this  context, misplaced occurrences of $X$ in $\Pi$ are those
occurrences of the form $t$ {\tt  ==} $X$ or $t$ {\tt /=} $X$, with
$t \in \var$ and $X \neq t$.
\end{itemize}

\begin{table}[h]
\begin{center}
{\scriptsize
\begin{tabular}{lccccc}
\hline
{\bf RULES} ~~~~~~& ~~P$_1$~~ & ~~P$_2$~~ & ~~P$_3$~~ & ~~P$_4$~~ & ~~P$_5$\\
\hline
{\bf H$_1$} & $\geq$ & $\geq$ & $\geq$ &  $>$ &  \\
\hline
{\bf H$_2$} & $\geq$ & $\geq$ & $\geq$ &  $>$ &  \\
\hline
{\bf H$_3$} & $\geq$ & $\geq$ & $>$ &    &  \\
\hline
{\bf H$_4$} & $\geq$ & $\geq$ & $\geq$ &  $\geq$ & $>$\\
\hline
{\bf H$_5$} & $>$ &   &   &   &  \\
 \hline
{\bf H$_6$} & $>$ &  &  &    &  \\
\hline
{\bf H$_7$} & $\geq$ & $\geq$ & $>$ &  &  \\
 \hline
{\bf H$_8$} & $>$ &  &   &   &  \\
\hline
{\bf H$_9$} & $>$ &   &   &    &  \\
\hline
{\bf H$_{10}$} & $\geq$ & $\geq$ & $\geq$ &  $\geq$ & $>$\\
\hline
{\bf H$_{11a}$} & $\geq$ & $>$ &   &   &  \\
\hline
{\bf H$_{11b}$} & $>$ &  &   &   &  \\
\hline
{\bf H$_{12}$} & $>$ &  &   &  & \\
\hline
{\bf H$_{13}$} & $>$ &  &   &   &  \\
\hline
\end{tabular}
\caption{Well-founded progress ordering for
$>_{lex}$}\label{hordering} }
\end{center}
\end{table}
 \vspace*{-.3cm}

Let $>_{lex}$ be the lexicographic ordering induced by
$>_\mathbb{N}$ over $\mathbb{N}^5$. We claim that:
$$(\star)\,\, \Pi\, \Box\, \sigma \vdash\!\!\vdash_{\herbrand,\, \varx}\Pi'\, \Box\, \sigma'\,
\Rightarrow\, ||\Pi\, \Box\, \sigma||_{\varx} >_{lex} ||\Pi'\,
\Box\, \sigma'||_{\varx}$$ This is justified by Table
\ref{hordering}, which shows the behaviour of the different $str$s
w.r.t.  $>_{lex}$. In order to understand the table, note that two
different cases have been distinguished for the application of the
$str$ {\bf H$_{11}$}, namely:
\begin{itemize}
\item {\bf H$_{11a}$} Application of {\bf H$_{11}$} choosing a value of $i$ such that $\varx \cap var(t_i) \neq \emptyset$.
\item {\bf H$_{11b}$} Application of {\bf H$_{11}$} choosing a value of $i$ such that $\varx \cap var(t_i) = \emptyset$.
\end{itemize}
Since $>_{lex}$ is a well-founded ordering, termination of
$\sts{\herbrand}{\varx}$ can be concluded from $(\star)$. The reader
is referred to Section 2.3 in  \cite{BN98} for more information on
this proof technique.
\item
{\bf Local Soundness Property:} Given a $\herbrand$ store $\Pi\,
\Box\, \sigma$ and a set $\varx \subseteq odvar_{\herbrand}(\Pi)$,
we must prove that the union
$$\bigcup \{Sol_{\herbrand}(\exists\overline{Y'}(\Pi' \, \Box\, \sigma')) \mid
\Pi \, \Box \,  \sigma \vdash\!\!\vdash_{\herbrand,\, \varx}  \Pi'\,
\Box\, \sigma',\, \overline{Y'}  = var(\Pi' \, \Box \,  \sigma')
\setminus var(\Pi \,  \Box \, \sigma)\}$$ is a subset of
$\sol{\cdom}{\Pi\, \Box\, \sigma}$. Obviously, it suffices to prove
the inclusion
$$(\dagger)\, Sol_{\herbrand}(\exists\overline{Y'}(\Pi' \, \Box\, \sigma')) \subseteq Sol_{\herbrand}(\Pi\, \Box\, \sigma)$$
for each $\Pi' \, \Box\, \sigma'$ such that $\Pi \, \Box \,  \sigma
\vdash\!\!\vdash_{\herbrand,\, \varx} ! \, \, \Pi'\, \Box\, \sigma'
$ with $\overline{Y'}  = var(\Pi' \, \Box \, \sigma') \setminus
var(\Pi \,  \Box \, \sigma)$. However $(\dagger)$ is an easy
consequence of
$$(\dagger \dagger)\, Sol_{\herbrand}(\Pi' \, \Box\, \sigma') \subseteq Sol_{\herbrand}(\Pi\, \Box\, \sigma)$$
In fact, assuming $(\dagger \dagger)$ and an arbitrary $\eta \in
Sol_{\herbrand}(\exists\overline{Y'}(\Pi' \, \Box\, \sigma'))$,
there must be some $\eta' \in Sol_{\herbrand}(\Pi' \, \Box\,
\sigma')$ such that $\eta =_{\setminus\overline{Y'}} \eta'$. Then,
$\eta' \in Sol_{\herbrand}(\Pi \, \Box\, \sigma)$ because of $(\dagger
\dagger)$, and thus $\eta \in Sol_{\herbrand}(\Pi \, \Box\, \sigma)$
because $\eta =_{\setminus \overline{Y'}} \eta'$ and $\overline{Y'}
\cap var(\Pi \,  \Box \, \sigma) = \emptyset$.

Having proved that $(\dagger \dagger)$ entails $(\dagger)$, we proceed to
prove $(\dagger \dagger)$ by a case distinction according to the $str$
used in the step $\Pi \, \Box \,  \sigma
\vdash\!\!\vdash_{\herbrand,\, \varx}\, \Pi'\, \Box\, \sigma'$. In
each case, we assume that the stores $\Pi\, \Box\, \sigma$ and
$\Pi'\, \Box\, \sigma'$ occurring in $(\dagger \dagger)$ have exactly
the form displayed for the corresponding transformation in the Table
\ref{htable}  displayed in Subsection \ref{hdom}.  For instance, in
the case of transformation {\bf H1} we write $(t$ {\tt ==} $s)$
$\to!$ $R,~\Pi$ $\Box$ $\sigma$ in place of $\Pi\, \Box\, \sigma$.
Moreover, in all the cases we silently use the fact that the
constraints and variables within any store are not affected by the
substitution kept in that store.

\begin{itemize}
\item [\bf H1]
Assume $\eta \in Sol_{\herbrand}((t$ {\tt ==} $s,\, \Pi)\sigma_1\,
\Box\, \sigma\sigma_1)$. Then $\eta \in Sol_{\herbrand}((t$ {\tt ==}
$s,\,\Pi)\sigma_1)\, \cap\, Sol_{\herbrand}(\sigma\sigma_1)$. We
must prove $\eta \in Sol_{\herbrand}((t$ {\tt ==} $s)\, \to!\, R,\,
\Pi\, \Box\, \sigma)$.

Since $(t$ {\tt ==} $s,\Pi) = (t$ {\tt ==} $s,\,\Pi)\sigma$, we can
infer $\eta \in Sol_{\herbrand}(t$ {\tt ==} $s,~\Pi)\cap
Sol_{\herbrand}(\sigma)$ from our assumptions and Lemma \ref{ASL}.

It remains to prove that $\eta \in Sol_{\herbrand}((t$ {\tt ==}
$s)\, \to!\, R)$. Since we already know that $\eta \in
Sol_{\herbrand}(t$ {\tt ==} $s$), it suffices to prove that $R\eta
=$ {\tt true}. But $\eta \in Sol_{\herbrand}(\sigma\sigma_1)$ means
$\sigma\sigma_1\eta = \eta$, and therefore $R\eta = R
\sigma\sigma_1\eta = R \sigma_1\eta =$ {\tt true} $\eta =$ {\tt true}.
\item [\bf H2]
Very similar to {\bf H1}.
\item [\bf H3]
Trivial. Clearly, $Sol_{\herbrand}(\overline{t_m {\tt ==} s_m}) =
Sol_{\herbrand}(h\, \overline{t_m}$ {\tt ==} $\, h\, \overline{s_m})$.
\item [\bf H4]
Trivial. Clearly, $Sol_{\herbrand}(X$ {\tt  ==} $t) =
Sol_{\herbrand}(t$ {\tt  ==} $X)$.
\item [\bf H5]
Assume $\eta \in Sol_{\herbrand}(tot(t),\,\Pi\sigma_1 \Box
\sigma\sigma_1)$. Then $t\eta$ is a total pattern and $\eta \in
Sol_{\herbrand}(\Pi\sigma_1) \cap  Sol_{\herbrand}(\sigma\sigma_1)$.
We must prove $\eta \in Sol_{\herbrand}(X$ {\tt ==} $t,\,\Pi\,\Box
\sigma)$.

Since $\Pi = \Pi \sigma$, we can infer $\eta \in
Sol_{\herbrand}(\Pi) \cap  Sol_{\herbrand}(\sigma)$ from our
assumptions and Lemma \ref{ASL}. It remains to prove that $\eta \in
Sol_{\herbrand}(X$ {\tt ==} $t)$. But $\eta \in
Sol_{\herbrand}(\sigma\sigma_1)$ means $\sigma\sigma_1\eta = \eta$.
Thus, $X\eta = X\sigma\sigma_1\eta = X \sigma_1\eta = t\eta$, which
implies $\eta \in Sol_{\herbrand}(X$ {\tt ==} $t)$, because $t\eta$
is total.
\item [\bf H6]
Trivial, because $\eta \in  Sol_{\herbrand}(\blacksquare)$ is false
for any $\eta$.
\item [\bf H7]   Trivial.
Clearly, $Sol_{\herbrand}(\overline{t_i\,  {\tt /=}\, s_i})
\subseteq Sol_{\herbrand}(h\, \overline{t_m}$ {\tt /=} $h\,
\overline{s_m})$.
\item [\bf H8]
Trivial, because $\eta \in Sol_{\herbrand}(h\, \overline{t_n}$ {\tt
/=} $h'\, \overline{s_m})$ holds for any $\eta$.
\item [\bf H9]
Trivial, for the same reason as {\bf H6}.
\item [\bf H10]
Trivial. Clearly, $Sol_{\herbrand}(X$ {\tt  /=} $t) =
Sol_{\herbrand}(t$ {\tt  /=} $X)$.
\item [\bf H11]
Assume $\eta \in Sol_{\herbrand}((Z_i$ {\tt /=}
$t_i,\,\Pi)\sigma_1\,\Box\,\sigma\sigma_1)$. Then $\eta \in
Sol_{\herbrand}((Z_i$ {\tt /=} $t_i)\sigma_1)$ and $\eta \in
Sol_{\herbrand}(\Pi \sigma_1)\, \cap\,
Sol_{\herbrand}(\sigma\sigma_1)$. We must prove  $\eta \in
Sol_{\herbrand}(X$ {\tt  /=} $c\, \tpp{t}{n},\,\Pi\,\Box\,\sigma)$.

Since $\Pi = \Pi \sigma$, we can infer $\eta \in
Sol_{\herbrand}(\Pi)\, \cap\, Sol_{\herbrand}(\sigma)$ from our
assumptions and Lemma \ref{ASL}. It remains to prove that $\eta \in
Sol_{\herbrand}(X$ {\tt  /=} $c\, \tpp{t}{n})$. Because of $\eta \in
Sol_{\herbrand}(\sigma\sigma_1)$, we know that $\sigma\sigma_1\eta =
\eta$. Therefore, it suffices to prove $\sigma\sigma_1\eta \in
Sol_{\herbrand}(X$ {\tt  /=} $c\, \tpp{t}{n})$,
which can be reasoned as follows:\\

\hspace*{1.5cm} $\sigma\sigma_1\eta \in Sol_{\herbrand}(X$ {\tt  /=}
$c\, \tpp{t}{n})$
$\Leftrightarrow_{(1)}$ $\eta \in Sol_{\herbrand}(X$ {\tt  /=} $c\, \tpp{t}{n})\sigma\sigma_1$ \\
\hspace*{1.5cm} $\Leftrightarrow$ $\eta \in Sol_{\herbrand}(X$ {\tt
/=} $c\, \tpp{t}{n})\sigma_1$
$\Leftarrow_{(2)} $ $\eta \in Sol_{\herbrand}(Z_i$ {\tt  /=} $t_i\sigma_1)$ \\
\hspace*{1.5cm} $\Leftarrow_{(3)}$ $\eta \in Sol_{\herbrand}(Z_i$ {\tt  /=} $t_i)\sigma_1$\\

where $(1)$ holds because of the Substitution  Lemma \ref{sl}, $(2)$ and $(3)$ hold by
construction of $\sigma_1$, and $\eta \in Sol_{\herbrand}(Z_i$ {\tt
/=} $t_i)\sigma_1$ holds because of the assumptions of this case.
\item [\bf H12]
Assume $\eta \in Sol_{\herbrand}(\Pi\sigma_1 \Box \sigma\sigma_1)$.
Then $\eta \in Sol_{\herbrand}(\Pi\sigma_1) \cap
Sol_{\herbrand}(\sigma\sigma_1)$. We must prove  $\eta \in
Sol_{\herbrand}(X$ {\tt  /=} $c\, \tpp{t}{n},\,\Pi\,\Box\,\sigma)$.

Since $\Pi = \Pi \sigma$, we can infer $\eta \in
Sol_{\herbrand}(\Pi)\, \cap\, Sol_{\herbrand}(\sigma)$ from our
assumptions and Lemma \ref{ASL}. It remains to prove that $\eta \in
Sol_{\herbrand}(X$ {\tt  /=} $c\, \tpp{t}{n})$. This is the case
because $X\eta = X\sigma\sigma_1\eta = X \sigma_1\eta = (d\,
\tpp{Z}{m})\eta$, where the first equality holds because of the
assumption $\eta \in  Sol_{\herbrand}(\sigma\sigma_1)$ and the third
equality holds by construction of $\sigma_1$.
\item [\bf H13]
Trivial, for the same reason as {\bf H6}.
\end{itemize}
\item
{\bf Local Completeness Property} for $\mathcal{URS}$-free steps:
Recall the set of unsafe $str$s $\mathcal{URS} = \{{\bf OH3}, {\bf
OH7}, {\bf H13}\}$ defined in Subsection \ref{hdom}. Assume a
$\herbrand$ store $\Pi\, \Box\, \sigma$ and a set $\varx \subseteq
odvar_{\herbrand}(\Pi)$, such that $\Pi\, \Box\, \sigma$ is
$\mathcal{URS}$-irreducible but  not in $\varx$-solved form. We must
prove that $\wtsol{\cdom}{\Pi\, \Box\, \sigma}$ is a subset of the
union
$$\bigcup \{WTSol_{\cdom}(\exists\overline{Y'}(\Pi' \Box\sigma'))\mid\Pi\Box\sigma
\vdash\!\!\vdash_{\herbrand,\varx}\Pi'\Box \sigma',
\overline{Y'}=var (\Pi' \Box\sigma')\setminus var(\Pi\Box\sigma)\}$$
Given any well-typed solution $\eta \in WTSol_{\herbrand}(\Pi \,
\Box \, \sigma)$ (which satisfies in particular $\sigma\eta =
\eta$), we must find $\Pi'\, \Box\, \sigma'$ and $\eta'$ such that
$$(\ddagger)\,\, \Pi \, \Box \,  \sigma \vdash\!\!\vdash_{\herbrand,\varx}  \, \Pi'\, \Box\, \sigma',\,\,
\eta' \in WTSol_{\herbrand}(\Pi' \, \Box \, \sigma'),\,\, \eta
=_{\setminus \overline{Y'}} \eta'$$ so that  $\eta \in
WTSol_{\herbrand}(\exists \overline{Y'}(\Pi' \, \Box \, \sigma'))$
will be ensured. Because of the assumptions on $\Pi\, \Box\,
\sigma$, there must be some $str$ {\bf H}$_i \notin \mathcal{URS}$
that can be used to transform $\Pi \, \Box \, \sigma$. Below we
analyze all the  possibilities for {\bf H$_i$}, considering all the
$str$s shown in Table \ref{htable} in Subsection \ref{hdom} except
{\bf OH3}, {\bf OH7} and  {\bf H13}. In all the cases we conclude
that the conditions $(\ddagger)$  displayed above can be ensured. When
considering different $str$s that can be alternatively applied to
one and the same store (as e.g. {\bf H1} and {\bf H2}) we group all
the possibilities within the same case, arguing that some rule in
the group can be chosen to transform $\Pi \, \Box \, \sigma$
ensuring $(\ddagger)$. In all the cases, we assume that the stores
$\Pi\, \Box\, \sigma$ and $\Pi'\, \Box\, \sigma'$ occurring in
$(\ddagger)$ have exactly the form displayed for the corresponding
transformation in Table \ref{htable}, we note the selected atomic
constraint as $\pi$, and we silently use the fact that the
constraints and variables within any store are not affected by the
substitution kept in that store.

\begin{itemize}
\item [\bf H1, H2]
In this case $\pi$ is $(t$ {\tt ==} $s)\, \to!\, R$, $\eta \in
WTSol_{\herbrand}(t$ {\tt ==} $s\, \to!\, R,\, \Pi\, \Box\, \sigma)$
and $\overline{Y'} = \emptyset$. Because of $\eta \in
WTSol_{\herbrand}(\pi)$, one of the two following subcases must
hold:
\begin{itemize}
\item[(a)] $\eta(R) =$ {\tt true} and $\eta \in WTSol_{\herbrand}(t$ {\tt ==} $s)$ or else
\item[(b)] $\eta(R) =$ {\tt false} and $\eta \in WTSol_{\herbrand}(t$ {\tt /=} $s)$
\end{itemize}
Assume that subcase (a) holds. Then, $(\ddagger)$  can be ensured by
transforming the given store with {\bf H1} and  proving $\eta' =
\eta \in WTSol_{\herbrand}(\Pi\sigma_1\, \Box\, \sigma\sigma_1)$.
Note that Lemma \ref{ACL} can be applied with $\overline{Y'} =
\emptyset$, $\eta' = \eta$ and $\sigma_1 = \{R \mapsto true\}$,
because the condition $\sigma_1 \eta = \eta$ follows trivially from
$\eta(R) =$ {\tt true}. Then, $\eta \in Sol_{\herbrand}(\Pi\sigma_1)
\cap Sol_{\herbrand}(\sigma\sigma_1)$ is ensured by Lemma \ref{ACL},
and $\eta$ obviously remains a well-typed solution.

Assume now that subcase (b) holds. Then a similar argument can be
used, but choosing {\bf H2} instead of {\bf H1}.
\item [\bf H3]
In this case $\pi$ is $h\, \tpp{t}{m}$ {\tt ==} $h\, \tpp{s}{m}$ and
$(\ddagger)$ can be ensured by choosing to transform the given store
with {\bf H3} and taking $\overline{Y'} = \emptyset$ and $\sigma ' =
\sigma$. Note that $h$ must be $m$-transparent because of the
$\mathcal{URS}$-freeness assumption, and the Transparency Lemma
\ref{trl} can be applied to ensure that $\eta$ remains a well-typed
solution of the new store.
\item [\bf H4]
In this case $\pi$ is $t$ {\tt ==} $X$, where $t$ is not a variable,
and $(\ddagger)$ can be trivially ensured by choosing to transform the
given store with {\bf H4} and taking $\overline{Y'} = \emptyset$ and
$\sigma ' = \sigma$.
\item [\bf H5]
In this case $\pi$ is $t$ {\tt ==} $X$, with $X \notin \varx, X
\notin var(t), X \neq t$. Moreover, $\eta \in WTSol_{\herbrand}(X$
{\tt ==} $t,\, \Pi\, \Box\, \sigma)$ and $\overline{Y'} =
\emptyset$. Then $(\ddagger)$ can be ensured by transforming the given
store with {\bf H5} and  proving $\eta' = \eta \in
WTSol_{\herbrand}(tot(t), \Pi\sigma_1\, \Box\, \sigma\sigma_1)$. The
assumption $\eta \in WTSol_{\herbrand}(\pi)$ means that $\eta(X) =
t\eta$ is a total pattern, so that $\eta(Y)$ is also a total pattern
for each variable $Y \in var(t)$. In these conditions, $\eta \in
\sol{\herbrand}{tot(t)}$ and $\sigma_1 \eta = \eta$ holds for
$\sigma_1 = \{X \mapsto t\}$. This allows to apply Lemma \ref{ACL}
with $\overline{Y'} = \emptyset$, $\eta' = \eta$ and $\sigma_1$,
ensuring that $\eta \in Sol_{\herbrand}(\Pi\sigma_1) \cap
Sol_{\herbrand}(\sigma\sigma_1)$. Obviously, $\eta$ remains a
well-typed solution.
\item [\bf H7]
In this case, $\pi$ is $h\, \tpp{t}{m}$ {\tt /=} $h\, \tpp{s}{m}$.
Because of $\eta \in WTSol_{\herbrand}(\pi)$, there must be some
index $i$ such that $1 \leq i \leq m$ and $\eta \in
WTSol_{\herbrand}(t_i$ {\tt /=} $s_i)$. Then $(\ddagger)$ can be
ensured by choosing to transform the given store with {\bf H7} and
this particular value of $i$, and taking $\overline{Y'} =
\emptyset$, $\sigma ' = \sigma$. Note that $h$ must be
$m$-transparent because of the $\mathcal{URS}$-freeness assumption,
and the Transparency Lemma \ref{trl} can be applied to ensure that
$\eta$ remains a well-typed solution of the new store.
\item [\bf H8]
In this case $\pi$ is $h\, \tpp{t}{n}$ {\tt  /=} $h'\, \tpp{s}{m}$
with $h \neq h'$ or $n \neq m$, and $(\ddagger)$ can be trivially
ensured by choosing to transform the given store with {\bf H8},
taking $\overline{Y'} = \emptyset$ and  $\sigma' = \sigma$.
\item [\bf H10]
This is a trivial case, similar to {\bf H4}.
\item [\bf H11, H12]
In this case $\pi$ is $X$ {\tt /=} $c\, \tpp{t}{n}$, with $X \notin
\varx,\, c \in DC^n$ and  $\varx \cap var(c\, \tpp{t}{n}) \neq
\emptyset$, $\eta \in WTSol_{\herbrand}(X$ {\tt /=} $c\,
\tpp{t}{n}\, \Pi\, \Box\, \sigma)$. Because of $\eta \in
WTSol_{\herbrand}(\pi)$, one of the two following subcases must hold
for $\eta(X)$:
\begin{itemize}
\item[(a)] $\eta(X) = c\, \tpp{s}{n}$, where $s_i${\tt /=} $t_i \eta$ holds for some $1\leq i \leq n$.
\item[(b)] $\eta(X) = d\, \tpp{s}{m}$, where $d \in DC^m$ belongs to the same datatype as $c$,
but $d \neq c$.
\end{itemize}

Assume that subcase (a) holds. Then $(\ddagger)$ can be ensured by
choosing to transform the given store with {\bf H11} and a
particular value of $i$ such that $s_i${\tt /=} $t_i \eta$ holds,
taking  $\overline{Y'} = \tpp{Z}{n}$, defining $\eta'$ as the
valuation that satisfies $\eta'(Z_j) = s_j$ for all $1\leq j \leq n$
and $\eta'(Y) = \eta(Y)$ for any other variable $Y$ and proving
$\eta' \in WTSol_{\herbrand}((Z_i$ {\tt /=} $t_i,\, \Pi)\sigma_1\,
\Box\, \sigma\sigma_1)$.

Obviously, $\eta =_{\setminus \overline{Y'}} \eta'$. Moreover,
$\sigma_1 \eta' = \eta'$, since $X\sigma_1 \eta' = (c\,
\tpp{s}{n})\eta' = c\, \tpp{s}{n} = X \eta = X \eta'$ and $Y\sigma_1
\eta' = Y \eta'$ for any variable $Y \neq X$. Therefore, Lemma
\ref{ACL} can be applied to ensure that $\eta' \in
Sol_{\herbrand}(\Pi \sigma_1) \cap Sol_{\herbrand}(\sigma\sigma_1)$.
Since $\eta$ was a well-typed solution and data constructors have
the transparency property (see Subsection \ref{types}), $\eta'$ can
be also well-typed under appropriated type assumptions for the new
variables $\overline{Y'} = \tpp{Z}{n}$ introduced by the
transformation step.  It only remains to prove that $\eta' \in
Sol_{\herbrand}((Z_i${\tt /=}$t_i)\sigma_1)$.
This can be reasoned by a chain of equivalences, ending with the condition known to hold in subcase (a): \\

\hspace*{1.5cm} $\eta' \in Sol_{\herbrand}((Z_i${\tt /=}$t_i)\sigma_1)$
$\Leftrightarrow_{(1)}$ $\sigma_1\eta' \in Sol_{\herbrand}(Z_i${\tt /=}$t_i)$ $\Leftrightarrow_{(2)}$ \\
\hspace*{1.5cm} $\eta' \in Sol_{\herbrand}(Z_i${\tt /=}$t_i)$
$\Leftrightarrow$ $\eta'(Z_i)${\tt /=}$t_i \eta'$
$\Leftrightarrow_{(3)}$
$s_i${\tt /=}$t_i \eta'$ \\

Note that $(1)$ holds because of Lemma \ref{sl}, $(2)$ holds because
$\sigma_1 \eta' = \eta'$, and $(3)$ holds by construction of
$\eta'$. This finishes the proof for this subcase.

Finally, assume now that subcase (b) holds. Then $(\ddagger)$ can be
ensured by choosing to transform the given store with {\bf H12} and
the particular data constructor $d \in DC^m$ for which we know that
$\eta(X) = d \, \tpp{s}{m}$, taking  $\overline{Y'} = \tpp{Z}{m}$,
defining $\eta'$ as the valuation that satisfies $\eta'(Z_j) = s_j$
for all $1\leq j \leq m$ and $\eta'(Y) = \eta(Y)$ for any other
variable $Y$ and proving  $\eta' \in WTSol_{\herbrand}(\Pi\sigma_1\,
\Box\, \sigma\sigma_1)$.

Obviously, $\eta =_{\setminus \tpp{Z}{m}} \eta'$. Moreover,
$\sigma_1 \eta' = \eta'$ can be easily checked, as in subcase (a).
Therefore, Lemma \ref{ACL} can be applied to ensure that $\eta' \in
Sol_{\herbrand}(\Pi \sigma_1) \cap Sol_{\herbrand}(\sigma\sigma_1)$.
Finally,  since $\eta$ was a well-typed solution, $\eta'$ is clearly
also well-typed under appropriated type assumptions for the new
variables $\overline{Y'} = \tpp{Z}{n}$ introduced by the
transformation step.
\end{itemize}
 \end{enumerate} 
\vspace*{-0.1cm}

Using items 1. to 6. above and Lemma \ref{psts}, we can now claim
that $solve^{\herbrand}$ satisfies the requirements for solvers
enumerated in  Definition \ref{defSolver}, except that the {\bf
Completeness Property} is guaranteed to hold only for safe (i.e.,
$\mathcal{URS}$-free) solver invocations and the {\bf Discrimination
Property} has not  been proved  yet.

The remark in item 1. of the proof of Lemma \ref{psts} allows to
rephrase the {\em Discrimination Property} as follows: if a given
$\herbrand$ store $\Pi\, \Box\, \sigma$ satisfies neither $(a)\,\,
\varx\, \cap\, odvar(\Pi) \neq \emptyset$ nor $(b)\,\, \varx\,
\cap\, var(\Pi) = \emptyset$ then $\Pi\, \Box\, \sigma$ can be
reduced by some $\sts{\herbrand}{\varx}$ transformation. Assume that
$\Pi\, \Box\, \sigma$ satisfies neither $(a)$ nor $(b)$. Because of
$\lnot\, (b)$, there must be some $\pi \in \Pi$ such that $(c)\,\,
\varx\, \cap\, var(\pi) \neq \emptyset$. Because of $\lnot\, (a)$,
this $\pi$ must satisfy $(d)\,\, \varx\, \cap\, odvar(\pi) =
\emptyset$, which together with $(c)$ entails  $(e)\,\, \varx\,
\cap\, cvar(\pi) \neq \emptyset$. Using $(d), (e)$ and reasoning by
case distinction on the syntactic form of $\pi$, we find in all the
cases some $\sts{\herbrand}{\varx}$ transformation which can be used
to transform the store $\Pi\, \Box\, \sigma$ taking $\pi$ as the
selected atomic constraint. The cases are as follows:

\begin{itemize}
\item
$\pi$ is $(t$ {\tt ==} $s)$ $\to!$ $R$. In this case the store can
be transformed by means of {\bf H1} or {\bf H2}.
\item
$\pi$ is  $h\, \tpp{t}{m}$ {\tt ==} $h\, \tpp{s}{m}$. In this case
the store can be transformed by means of {\bf H3}.
\item
$\pi$ is $t$ {\tt  ==} $X$ with $t \notin \var$. In this case the
store can be transformed by means of {\bf H4}.
\item
$\pi$ is $X$ {\tt ==} $t$ with $X \notin var(t)$, $X \neq t$.
Because of $(d)$ above we know that $X \notin \varx$, and the store
can be transformed by means of {\bf H5}.
\item
$\pi$ is $X$ {\tt ==} $t$ with  $X \in var(t)$, $X \neq t$. In this
case the store can be transformed by means of {\bf H6}.
\item
$\pi$ is $h\, \tpp{t}{m}$ {\tt  /=} $h\, \tpp{s}{m}$. In this case
the store can be transformed by means of {\bf H7}.
\item
$\pi$ is $h\, \tpp{t}{n}$ {\tt  /=} $h'\, \tpp{s}{m}$ with $h \neq
h'$ or $n \neq m$. In this case the store can be transformed by
means of {\bf H8}.
\item
$\pi$ is $t$ {\tt  /=} $t$ with $t \in \mathcal {V}\!\!ar \cup DC
\cup DF \cup SPF$. In this case the store can be transformed by
means of {\bf H9}.
\item
$\pi$ is $t$ {\tt  /=} $X$ with $t \notin \var$. In this case the
store can be transformed by means of {\bf H10}.
\item
$\pi$ is $X$ {\tt  /=} $c\, \tpp{t}{n}$, with $c \in DC^n$. Because
of $(d), (e)$  above we know that $X \notin \varx$ and $\varx \cap
var(c\, \tpp{t}{n}) \neq \emptyset$. Therefore, the store can be
transformed by means of {\bf H11} or {\bf H12}.
\item
$\pi$ is $X$ {\tt  /=} $h\, \tpp{t}{m}$ with $h \notin DC^m$.
Because of $(d), (e)$  above we know that $X \notin \varx$ and
$\varx \cap var(h\, \tpp{t}{m}) \neq \emptyset$. Therefore, the
store can be transformed by means of {\bf H13}.
\end{itemize}
This completes the proof of the Discrimination Property and the
Theorem.\hfill
\end{proof}

We refrain to include in this Appendix a proof of Theorem
\ref{msolver}, stated in Subsection \ref{ourcdom} and ensuring the
properties required for the solver $solve^{\mdom}$. The proof would
follow exactly the same pattern as the previous one, but with much
simpler arguments, since the $sts$ for $\mdom$-stores involves no
decompositions. Actually, this $sts$ is finitely branching,
terminating, locally sound and locally complete. Therefore, Lemma
\ref{psts} can be applied to ensure all the properties required for
solvers, including unrestricted completeness.

We end this subsection with the proof of Theorem
\ref{sumProperties}, ensuring that the amalgamated sums presented in
Subsection \ref{cdomains} are well defined domains behaving as a
conservative extension of their components.

\begin{proof}[Proof of Theorem  \ref{sumProperties}]

Assume $\sdom = \cdom_1 \oplus \cdots \oplus \cdom_n$ of signature
$\Sigma$ constructed as the amalgamated sum of $n$  pairwise
joinable domains $\cdom_i$ of signatures $\Sigma_i$, $1 \leq i \leq
n$. Note that the information ordering $\sqsubseteq$ introduced in
Subsection \ref{expressions} has the same syntactic definition for
any specific domain signature. Note also that any arguments
concerning well typing needed for this  proof can refer to the
principal type declarations within signature $\Sigma$, which
includes those within signature $\Sigma_i$ for all $1 \leq i \leq n$.
Let us now prove the four claims of the theorem in order.

\begin{enumerate}
\item
$\mathcal{S}$ is well-defined as a constraint domain; i.e., the
interpretations of primitive  function symbols $p \in SPF$ in
$\mathcal{S}$ satisfy the four conditions listed in Definition
\ref{dcdom} from Subsection \ref{dom}. We consider them one by  one,
assuming that $p$ is not the primitive {\tt ==} except in the fourth
condition.
\begin{enumerate}
\item {\bf Polarity}: Assume $p \in SPF^m$ and  $\overline{t}_m, \overline{t'}_m, t, t' \in  \mathcal{U}_{\sdom}$
such that $p^{\sdom}\, \overline{t}_m \to t$, $\overline{t}_m
\sqsubseteq \overline{t'}_m$ and $t \sqsupseteq t'$. In case that
$t$ is $\bot$, we trivially conclude $p^{\sdom}\, \overline{t'}_n
\to t'$ because $t'$ must be also $\bot$. Otherwise, by the first
assumption and the definition of $p^{\sdom}$, there must be some $1
\leq i \leq n$ and some $\overline{t''}_m, t'' \in
\mathcal{U}_{\cdom_i}$ such that $\overline{t''}_m \sqsubseteq
\overline{t}_m$, $t'' \sqsupseteq t$ and $p^{\cdom_i}\,
\overline{t''}_m \to t''$. Since $\overline{t''}_m \sqsubseteq
\overline{t}_m \sqsubseteq \overline{t'}_m$ and $t'' \sqsupseteq t
\sqsupseteq t'$, $p^{\cdom_i}\, \overline{t''}_m \to t''$ implies
$p^{\sdom}\, \overline{t'}_m \to t'$ by definition of $p^{\sdom}$.
\item {\bf Radicality}:  Assume $p \in SPF^m$ and $\overline{t}_m, t \in \mathcal{U}_{\sdom}$
such that $p^{\sdom}\, \overline{t}_m \to t$ and $t$ is not $\bot$.
By the definition of $p^{\sdom}$ there must be some $1 \leq i \leq
n$ and some $\overline{t''}_m,\, t'' \in  \mathcal{U}_{\cdom_i}$
such that $\overline{t''}_m \sqsubseteq \overline{t}_m$, $t''
\sqsupseteq t$ and $p^{\cdom_i}\, \overline{t''}_m \to t''$. By the
radicality condition for  $\cdom_i$, there must be some total $t'
\in \mathcal{U}_{\cdom_i}$ such that $p^{\cdom_i}\, \overline{t''}_m
\to t' \sqsupseteq t''$. Note that $t' \sqsupseteq t'' \sqsupseteq
t$, and because of $\overline{t''}_m \sqsubseteq \overline{t}_m$ and
$t' \sqsupseteq t'$, $p^{\cdom_i}\, \overline{t''}_m \to t'$ implies
$p^{\sdom}\, \overline{t}_m \to t'$ by definition of $p^{\sdom}$.
\item {\bf Well-typedness}: Assume $p \in SPF^m$,
a monomorphic instance $\overline{\tau'}_m \to \tau'$ of $p$'s
principal type and $\overline{t}_m, t \in \mathcal{U}_{\sdom}$ such
that $\Sigma\, \vdash_{WT} \overline{t}_m :: \overline{\tau'}_m$ and
$p^{\sdom}\, \overline{t}_m \to t$. In case that $t$ is $\bot$, the
type judgement $\Sigma\, \vdash_{WT} \bot :: \tau'$ holds trivially.
Otherwise,  by the assumption $p^{\sdom}\, \overline{t}_m \to t$ and
the definition of $p^{\sdom}$ there exist $1 \leq i \leq n$ and
$\overline{t'}_m, t' \in  \mathcal{U}_{\cdom_i}$ such that
$\overline{t'}_m \sqsubseteq \overline{t}_m$, $t' \sqsupseteq t$ and
$p^{\cdom_i}\, \overline{t'}_m \to t'$. Moreover, since
$\overline{t'}_m \sqsubseteq \overline{t}_m$ the assumption
$\Sigma\, \vdash_{WT} \overline{t}_m :: \overline{\tau'}_m$ and the
Type Preservation Lemma \ref{tpl} imply $\Sigma\, \vdash_{WT}
\overline{t'}_m :: \overline{\tau'}_m$ Then, the well-typedness
assumption for $\cdom_i$ guarantees $\Sigma\, \vdash_{WT} t' ::
\tau'$, which implies $\Sigma\, \vdash_{WT} t :: \tau'$ because of
$t \sqsubseteq t'$ and Lemma \ref{tpl}.
\item {\bf Strict Equality}: The primitive {\tt ==} (in case that it belongs to $SPF$)
is interpreted as {\em strict equality} over $\mathcal{U}_{\sdom}$.
This is automatically guaranteed by the amalgamated sum
construction.
\end{enumerate}
\item
Given an index $1 \leq i \leq n$, a primitive function symbol $p \in
SPF_i^m$ and values $\overline{t}_m, t \in  \mathcal{U}_{\cdom_i}$,
we must prove: $p^{\cdom_i}\, \overline{t}_m \to t$ iff $p^{\sdom}\,
\overline{t}_m \to t$. By definition of $p^{\sdom}$, we know that
$p^{\sdom}\, \overline{t}_m \to t$ holds iff there are some
$\overline{t'}_m, t' \in  \mathcal{U}_{\cdom_i}$ such that
$\overline{t'}_m \sqsubseteq \overline{t}_m$, $t' \sqsupseteq t$ and
$p^{\cdom_i}\, \overline{t'}_m \to t'$. But this condition is
equivalent to $p^{\cdom_i} \overline{t}_m \to t$ because
$p^{\cdom_i}$ satisfies the polarity property.
\item
Given an index $1 \leq i \leq n$, a set of primitive constraints
$\Pi  \subseteq  APCon_{\cdom_i}$ and a valuation  $\eta  \in
Val_{\cdom_i}$, we will prove: $\eta \in Sol_{\cdom_i}(\Pi)
\Leftrightarrow \eta \in Sol_{\mathcal{S}}(\Pi)$.
The corresponding equivalence for the case of well-typed solutions follows then easily.
Since
$$\eta \in Sol_{\cdom_i}(\Pi)
\Leftrightarrow \forall \pi \in \Pi:\, \eta \in Sol_{\cdom_i}(\pi)
\Leftrightarrow  \forall \pi \in \Pi:\, \eta \in
Sol_{\sdom}(\pi) \Leftrightarrow\eta \in Sol_{\sdom}(\Pi)$$
it suffices to prove the equivalence
$$(\star)\,  \eta \in Sol_{\cdom_i}(\pi) \Leftrightarrow \eta \in Sol_{\sdom}(\pi)$$
for a fixed $\pi \in \Pi$.
Note that $\pi$ must have the form $p\, \overline{t}_m\, \to!\, t$ for some $p \in SPF_i^m$, $\overline{t}_m
\in  Pat_{\cdom_i}$ and total $t \in Pat_{\cdom_i}$. In case that
$p$ is {\tt ==}, $(\star)$ is trivially true because $t_1\eta${\tt
==}$^{\cdom_i} t_2\eta\, \to! t\eta$ and $t_1\eta${\tt ==}$^{\sdom}
t_2\eta\, \to! t\eta$ hold under the same conditions, as specified
in Definition \ref{dcdom} from Subsection \ref{dom}. In case that
$p$ is not {\tt ==}, let $\overline{t'}_m = \overline{t}_m\eta$ and
$t' = t\eta$. If $t'$ is not a total pattern, then neither $\eta \in
Sol_{\cdom_i}(\pi)$ nor $\eta \in Sol_{\sdom}(\pi)$ hold. Otherwise,
$$\eta \in Sol_{\cdom_i}(\pi)
\Leftrightarrow p^{\cdom_i} \overline{t'}_m\, \to\, t'
\Leftrightarrow_{(\star\star)} p^{\sdom} \overline{t'}_m\, \to\, t'
\Leftrightarrow \eta \in Sol_{\sdom}(\pi)$$ where the $(\star\star)$
step holds by the second item of this theorem, because
$\overline{t'}_m, t' \in \mathcal{U}_{\cdom_i}$.
\item
Given an index $1 \leq i \leq n$, a set of $\cdom_i$-specific primitive constraints
$\Pi  \subseteq  APCon_{\cdom_i}$ and a valuation  $\eta  \in
Val_{\sdom}$, we will prove:
$\eta \in Sol_{\sdom}(\Pi)\, \Leftrightarrow\, \trunc{\eta}{\cdom_i} \in Sol_{\cdom_i}(\Pi)$.
The corresponding equivalence for the case of well-typed solutions follows then easily.

First we prove $\eta \in Sol_{\sdom}(\Pi)\, \Leftarrow\, \trunc{\eta}{\cdom_i} \in Sol_{\cdom_i}(\Pi)$.
Assume $\trunc{\eta}{\cdom_i} \in Sol_{\cdom_i}(\Pi)$. Applying the previous item of this theorem,
we obtain $\trunc{\eta}{\cdom_i} \in Sol_{\sdom}(\Pi)$. Since $\trunc{\eta}{\cdom_i} \sqsubseteq \eta$,
we can apply the Monotonicity Lemma \ref{ml} and get $\eta \in Sol_{\sdom}(\Pi)$, as desired.

Now we prove $\eta \in Sol_{\sdom}(\Pi)\, \Rightarrow\, \trunc{\eta}{\cdom_i} \in Sol_{\cdom_i}(\Pi)$.
Assume $\eta \in Sol_{\sdom}(\Pi)$. Since $\Pi$ is $\cdom_i$-specific,
we can also assume that $\eta(X) \in \uni{\cdom_i}$ for all $X \in var(\Pi)$.
Then $\eta(X) = \trunc{\eta}{\cdom_i}(X)$ holds for all $X \in var(\Pi)$,
and therefore $\trunc{\eta}{\cdom_i} \in Sol_{\sdom}(\Pi)$,
which implies $\trunc{\eta}{\cdom_i} \in Sol_{\cdom_i}(\Pi)$,
again because of the previous item of this theorem.
\end{enumerate}
\hfill
\end{proof}
\vspace*{-0.5cm}

\subsection{Properties of the $CCLNC(\cdom)$ Calculus} \label{PropertiesCalculus}

The second part of the Appendix  includes the proofs of the main results stated in Subection \ref{SC}.
First we present an auxiliary result which is not stated in the main text of the article.
The $(WT)Sol$ notation is intended to indicate that the lemma holds both for plain solutions and
for well-typed solutions.

\begin{lemma}[Auxiliary Result for Checking Goal Solutions] \label{cgs}
Let $G$ $\equiv$ $\exists \overline{U}.$ $P$ $\Box$ $C$ $\Box$ $M$ $\Box$ $H$ $\Box$ $F$ $\Box$ $R$
be an admissible goal for  a given $\cflp{\mathcal{C}}$-program $\prog$.
Assume new variables $\overline{Y'}$ away from $\overline{U}$ and the other variables in $G$,
and two valuations $\mu, \hat{\mu} \in Val_{\ccdom}$ such that
$\hat{\mu} =_{\backslash \overline{U}, \overline{Y'}} \mu$
and $\hat{\mu} \in (WT)Sol_{\prog}(P\, \Box\, C \, \Box \,M\,\Box\,H\,\Box\,F\,\Box\,R)$.
Then $\mu \in (WT)Sol_{\prog}(G)$.
\end{lemma}
\vspace*{-0.4cm}
\begin{proof}
Consider $\hat{\hat{\mu}} \in Val_{\ccdom}$ univocally defined by the two conditions
$\hat{\hat{\mu}} =_{\backslash \overline{Y'}} \hat{\mu}$ and $\hat{\hat{\mu}} =_{\overline{Y'}} \mu$.
By hypothesis,  $\hat{\mu} \in (WT)Sol_{\prog}(P\, \Box\, C \, \Box \,M\,\Box\,H\,\Box\,F\,\Box\,R)$
and the variables $\overline{Y'}$ do not occur in $G$.
Therefore,
$\hat{\hat{\mu}} \in (WT)Sol_{\prog}(P\, \Box\, C \, \Box \,M\,\Box\,H\,\Box\,F\,\Box\,R)$
is ensured by the construction of $\hat{\hat{\mu}}$.
Recalling Definition \ref{defGoalSol} (see Subsection \ref{SC}),
we only need to prove $\hat{\hat{\mu}} =_{\backslash \overline{U}} \mu$
in order to conclude $\mu \in (WT)Sol_{\prog}(G)$.
In fact, given any variable $X \notin \overline{U}$, either $X \in \overline{Y'}$ or $X \notin \overline{Y'}$.
In the first case, $\hat{\hat{\mu}}(X) = \mu(X)$ by construction of $\hat{\hat{\mu}}$.
In the second case, $\hat{\hat{\mu}}(X) = \hat{\mu}(X)$ by construction of $\hat{\hat{\mu}}$
and $\hat{\mu}(X) = \mu(X)$ because of one of the hypothesis.
\end{proof}

Next we present the proof of Theorem \ref{localSC} which guarantees {\em local} soundness and completeness
for the {\em one-step} transformation of a given goal.

\begin{proof}[Proof of Theorem \ref{localSC}] \label{localSCProof}

Assume a given program $\prog$, an admissible goal $G$ for $\prog$ which is not in solved form,
and a rule {\bf RL} applicable to a selected part $\gamma$ of $G$.
The claim that   there are finitely many possible transformations $G \red_{{\bf RL},\gamma, \prog}G'_j$  ($1 \leq j \leq k$)
can be trivially checked by inspecting all the rules in Tables \ref{table3}, \ref{table4}, \ref{table7} and \ref{Stable} one by one.
We must prove two additional claims:
\begin{enumerate}
\item {\bf Local Soundness:} $Sol_{\prog}(G) \supseteq \bigcup_{j=1}^{k} Sol_{\prog}(G'_j)$.
\item {\bf Limited Local Completeness:} $WTSol_{\prog}(G) \subseteq \bigcup_{j=1}^{k} WTSol_{\prog}(G'_j)$,
provided that the application of {\bf RL}  to the selected part $\gamma$ of $G$ is {\em safe}, i.e.
it is neither an opaque application of {\bf DC}
nor an application of a rule from Table \ref{Stable} involving an incomplete solver invocation.
\end{enumerate}
Claims 1. and 2. must be proved for each {\bf RL} separately. In case that {\bf RL} is some of the rules displayed in
Table \ref{table3}, proving 1. and 2. involves building suitable witnesses as multisets of $\crwl{\ccdom}$ proof trees,
using techniques originally stemming from \cite{GHLR96,GHLR99} and later extended to $CFLP$ programs without domain cooperation
in \cite{LMR04}.  In case that {\bf RL} is some of the rules shown in Tables  \ref{table4}, \ref{table7} and \ref{Stable},
proving 1. and 2. requires almost no work with building witnesses.

We will consider rules {\bf DF} and {\bf FC} as representatives for Table  \ref{table3},
and most of the rules from Tables  \ref{table4}, \ref{table7} and \ref{Stable},
which are the main novelty in this paper.
When dealing with each rule {\bf RL}, we will assume that
$G$ resp. $G'_j$ are exactly as the original resp. transformed goal as displayed
in the presentation of {\bf RL} in Subsection \ref{lnr}, \ref{dcr} or \ref{csr}.
In our reasonings we will  use the notation $\mathcal{M} : \prog \vdash_{\crwl{\mathcal{C}}}(P\,\Box\,C)\mu'$
to indicate that the witness $\mathcal{M}$ is a multiset of $\crwl{\ccdom}$ proof trees that prove
$(P\,\Box\,C)\mu'$ from program $\prog$, using the inference rules of the $\crwl{\ccdom}$ logic
presented in \cite{LRV07}.
\vspace*{-0.4cm}

\subsubsection{Selected rules from Table \ref{table3}} \label{rt3}

\noindent
Rule {\bf DF, Defined Function}.
In this case, $\gamma$ is a production $f\, \tpp{e}{n} \to t$.
\begin{enumerate}
\item {\bf Local Soundness:}
Assume $\mu \in Sol_{\prog}(G'_j)$ for some $1 \leq j \leq k$.
Then there exists $\mu' =_{\backslash \overline{Y}, \overline{U}} \mu$
such that $\mu' \in Sol_{\prog}(\overline{e_n \to t_n},\, r \to t,\, P\, \Box\, C', C\, \Box\, M\, \Box\, H\, \Box\, F\, \Box\, R)$.
From this we deduce that
$\mu' \in Sol_{\ccdom}(M\, \Box\, H\, \Box\, F\, \Box\, R)$
and $\mathcal{M'} : \prog \vdash_{\crwl{\mathcal{C}}}(\overline{e_n \to t_n},\, r \to t,\, P\, \Box\, C', C)\mu'$
for a suitable witness $\mathcal{M'}$.
A part of $\mathcal{M'}$ proves $(\overline{e_n \to t_n},\, r \to t,\, C')\mu'$, which allows to deduce
$(f\, \tpp{e}{n} \to t)\mu'$ using the $\crwl{\ccdom}$ inference rule which deals with defined functions.
Therefore, $\mathcal{M'}$ can be used to build another witness
$\mathcal{M} : \prog \vdash_{\crwl{\mathcal{C}}}(f\, \tpp{e}{n} \to t,\, P\, \Box\, C)\mu'$.
Since $\mu' =_{\backslash \overline{U}} \mu$, we can conclude that $\mu \in Sol_{\prog}(G)$.
\item {\bf Limited  Local Completeness:}
Assume $\mu \in WTSol_{\prog}(G)$.
Then there is some $\mu' =_{\backslash \overline{U}} \mu$
such that $\mu' \in WTSol_{\prog}(f\, \tpp{e}{n} \to t,\, P\, \Box\, C\, \Box\, M\, \Box\, H\, \Box\, F\, \Box\, R)$.
Then, $\mu' \in WTSol_{\ccdom}(M\, \Box\, H\, \Box\, F\, \Box\, R)$
and $\mathcal{M} : \prog \vdash_{\crwl{\mathcal{C}}}(f\, \tpp{e}{n} \to t,\, P\, \Box\, C)\mu'$
for a suitable witness $\mathcal{M}$.
Note that $\mathcal{M}$ must include a $\crwl{\ccdom}$ proof tree
$\mathcal{T}$ proving  the production $(f\, \tpp{e}{n} \to t)\mu'$ using some instance of
$Rl : f\, \overline{t_n} \to r$ $\Leftarrow$ $C'$, suitably chosen as a variant of some $\prog$ rule
with new variables $\overline{Y} = var(Rl)$.
Let us choose $j$ so that $G'_j$ is the result of applying {\bf DF}  with $f\, \tpp{e}{n} \to t$ as the selected part of $G$ and
$Rl$ as the selected $\prog$ rule for $f$.
Consider a well typed $\mu'' \in Val_{\ccdom}$ that instantiates the variables in $\overline{Y}$ as required by the proof tree $\mathcal{T}$,
and instantiates any other variable $V$ to $\mu'(V)$.
By suitably reusing parts of $\mathcal{M}$, it is possible to build a witness
$\mathcal{M'} : \prog \vdash_{\crwl{\mathcal{C}}}(\overline{e_n \to t_n},\, r \to t,\, P\, \Box\, C', C)\mu''$.
Since $\mu'' =_{\backslash \overline{Y}, \overline{U}} \mu$, we can conclude that $\mu \in WTSol_{\prog}(G'_j)$.
\end{enumerate}
\noindent
Rule {\bf FC, Flatten Constraint}.
In this case, $\gamma$ is an atomic constraint  $p\, \tpp{e}{n} \to!\, t$ such that some $e_i$ is not a pattern and $k = 1$.
We write $G'$ instead of $G'_1$.
For the sake of simplicity, we consider $p\, e_1\, t_2\, \to!\, t$, where $e_1$ is not a pattern.
The presentation of the rule is then as in Table \ref{table3} with $n = 2$, $m = 1$.
\begin{enumerate}
\item {\bf Local Soundness:}
Assume $\mu \in Sol_{\prog}(G')$.
Then there exists $\mu' =_{\backslash V_1, \overline{U}} \mu$
such that $\mu' \in Sol_{\prog}(e_1 \to V_1,\, P\, \Box\,  p\, V_1\, t_2\, \to!\, t,\, C\, \Box\, M\, \Box\, H\, \Box\, F\, \Box\, R)$.
Then, we get
$\mu' \in Sol_{\ccdom}(M\, \Box\, H\, \Box\, F\, \Box\, R)$
and $\mathcal{M'} : \prog \vdash_{\crwl{\mathcal{C}}}(e_1 \to V_1,\, P\, \Box\,  p\, V_1\, t_2\, \to!\, t,\, C)\mu'$
for a suitable witness $\mathcal{M'}$.
A part of $\mathcal{M'}$ proves $(e_1 \to V_1,\,   p\, V_1\, t_2\, \to!\, t)\mu'$, which allows to deduce
$(p\, e_1\, t_2\, \to!\, t)\mu'$ using the $\crwl{\ccdom}$ inference rule which deals with primitive functions.
Therefore, $\mathcal{M'}$ can be used to build another witness
$\mathcal{M} : \prog \vdash_{\crwl{\mathcal{C}}}(P\, \Box\, p\, e_1\, t_2\, \to!\, t,\, C)\mu'$.
Since $\mu' =_{\backslash \overline{U}} \mu$, we can conclude that $\mu \in Sol_{\prog}(G)$.
\item {\bf Limited  Local Completeness:}
Assume $\mu \in WTSol_{\prog}(G)$.
Then there is some $\mu' =_{\backslash \overline{U}} \mu$
such that $\mu' \in WTSol_{\prog}(P\, \Box\, p\, e_1\, t_2\, \to!\, t,\, C\, \Box\, M\, \Box\, H\, \Box\, F\, \Box\, R)$.
Then, $\mu' \in WTSol_{\ccdom}(M\, \Box\, H\, \Box\, F\, \Box\, R)$
and $\mathcal{M} : \prog \vdash_{\crwl{\mathcal{C}}}(P\, \Box\, p\, e_1\, t_2\, \to!\, t,\, C)\mu'$
for a suitable witness $\mathcal{M}$.
Note that $\mathcal{M}$ must include a $\crwl{\ccdom}$ proof tree
$\mathcal{T}$ proving  the atomic constraint  $(p\, e_1\, t_2\, \to!\, t)\mu'$.
A part of $\mathcal{T}$ must prove a production of the form $e_1\mu' \to t_1$ for some suitable pattern $t_1$.
Consider a well typed $\mu'' \in Val_{\ccdom}$ such that $\mu''(V_1) = t_1$ and $\mu'' =_{\backslash V_1} \mu'$.
By suitably reusing parts of $\mathcal{M}$, it is possible to build a witness
$\mathcal{M'} : \prog \vdash_{\crwl{\mathcal{C}}}(e_1 \to V_1,\, P\, \Box\,  p\, V_1\, t_2\, \to!\, t,\, C)\mu''$.
Since $\mu'' =_{\backslash V_1, \overline{U}} \mu$, we can conclude that $\mu \in WTSol_{\prog}(G')$.
\end{enumerate}
\vspace*{-0.5cm}

\subsubsection{Rules from Table \ref{table4}} \label{rt4}

\noindent
Rule {\bf SB, Set Bridges}.
In this case, $\gamma$ is a primitive atomic constraint  $\pi$ which can be used to compute bridges, and $k = 1$.
We write $G'$ instead of $G'_1$. The application of the rule computes
$\exists \overline{V'}\, B' = bridges^{\mathcal{D} \to \mathcal{D'}}( \pi, B_M) \neq \emptyset$,
where $\mathcal{D} = \fd$ and $\mathcal{D'} = \rdom$ or vice versa,
according to the two cases (i) and (ii) explained in Table  \ref{table4}.
\begin{enumerate}
\item {\bf Local Soundness:}
Assume $\mu \in Sol_{\prog}(G')$.
Then there exists $\mu' =_{\backslash \overline{V'}, \overline{U}} \mu$ such that
$\mu' \in Sol_{\prog}(P\, \Box\,  \pi,\, C\, \Box\, M'\, \Box\, H\, \Box\, F\, \Box\, R)$.
Therefore, $\mu' \in Sol_{\ccdom}(M'\, \Box\, H\, \Box\, F\, \Box\, R)$
and $\mathcal{M'} : \prog \vdash_{\crwl{\mathcal{C}}}(P\, \Box\,  \pi,\, C)\mu'$
for a suitable witness $\mathcal{M'}$.
Since $M'$ is $B', M$, we get
$\mu' \in Sol_{\ccdom}(M\, \Box\, H\, \Box\, F\, \Box\, R)$
and $\mathcal{M'} : \prog \vdash_{\crwl{\mathcal{C}}}(P\, \Box\,  \pi,\, C)\mu'$,
which implies $\mu \in Sol_{\prog}(G)$ because of Lemma \ref{cgs}.
\item {\bf Limited  Local Completeness:}
Assume $\mu \in WTSol_{\prog}(G)$.
Then there is some $\mu' =_{\backslash \overline{U}} \mu$
such that $\mu' \in WTSol_{\prog}(P\, \Box\,  \pi,\, C\, \Box\, M\, \Box\, H\, \Box\, F\, \Box\, R)$.
Therefore, $\mu' \in WTSol_{\ccdom}(M\, \Box\, H\, \Box\, F\, \Box\, R)$
and $\mathcal{M} : \prog \vdash_{\crwl{\mathcal{C}}}(P\, \Box\, \pi,\, C)\mu'$
for a suitable witness $\mathcal{M}$.
Since $\pi$ is primitive, these conditions imply
$\mu' \in WTSol_{\mathcal{C}}(\pi \wedge B_M)$.
By item 2. of Proposition \ref{propBP} from Subsection \ref{dcr}, we know that
$\overline{V'}$ are new fresh variables and
$WTSol_{\mathcal{C}}(\pi \wedge B_M) \subseteq
WTSol_{\mathcal{C}}(\exists \overline{V'}(\pi \wedge B_M \wedge B'))$.
From this we can conclude that
$\mu' \in WTSol_{\mathcal{C}}(\exists \overline{V'}(\pi \wedge B_M \wedge B'))$
and therefore there is some $\mu'' =_{\backslash \overline{V'}} \mu'$ such that
$\mu''' \in WTSol_{\mathcal{C}}(\pi \wedge B_M \wedge B')$. Since $ \overline{V'}$ are
new variables not occurring in $G$, it is easy to check that
$\mu'' \in WTSol_{\mathcal{C}}(M'\, \Box\, H\, \Box\, F\, \Box\, R)$ and
$\mathcal{M} : \prog \vdash_{\crwl{\mathcal{C}}}(P\, \Box\, \pi,\, C)\mu''$, which ensures
$\mu \in WTSol_{\prog}(G')$.
\end{enumerate}
\noindent
Rule {\bf PP, Propagate Projections}.
In this case, $\gamma$ is a primitive atomic constraint  $\pi$ which can be used to compute projections, and $k = 1$.
We write $G'$ instead of $G'_1$. The application of the rule obtains $G'$ from $G$ by computing
$\exists \overline{V'}\, \Pi' = proj^{\mathcal{D} \to \mathcal{D'}}( \pi, B_M) \neq \emptyset$,
where $\mathcal{D} = \fd$ and $\mathcal{D'} = \rdom$ or vice versa,
according to the two cases (i) and (ii) explained in Table  \ref{table4}.
The reasonings for local soundness and limited local completeness are quite similar to those
used in the case of rule {\bf SB},
except that  item 3. of Proposition \ref{propBP} must be used in place of item 2. \\

\noindent
Rule {\bf SC, Submit Constraints}.
In this case, $\gamma$ is a primitive atomic constraint  $\pi$ and $k = 1$.
We write $G'$ instead of $G'_1$.
\begin{enumerate}
\item {\bf Local Soundness:}
Assume $\mu \in Sol_{\prog}(G')$.
Then there exists $\mu' =_{\backslash \overline{U}} \mu$ such that
$\mu' \in Sol_{\prog}(P\, \Box\, C\, \Box\, M'\, \Box\, H'\, \Box\, F'\, \Box\, R')$.
Therefore, $\mu' \in Sol_{\ccdom}(M'\, \Box\, H'\, \Box\, F'\, \Box\, R')$
and $\mathcal{M'} : \prog \vdash_{\crwl{\mathcal{C}}}(P\, \Box\, C)\mu'$
for a suitable witness $\mathcal{M'}$.
Due to the syntactic relationship between $G$ and $G'$ (see Table \ref{table4}),
$\mu' \in Sol_{\ccdom}(M'\, \Box\, H'\, \Box\, F'\, \Box\, R')$
amounts to $\mu' \in Sol_{\ccdom}(M\, \Box\, H\, \Box\, F\, \Box\, R)$
and $\mu' \in Sol_{\ccdom}(\pi)$.
Due to $\mu' \in Sol_{\ccdom}(\pi)$, the witness $\mathcal{M'}$ can be expanded
to another witness $\mathcal{M} : \prog \vdash_{\crwl{\mathcal{C}}}(P\, \Box\, \pi,\, C)\mu'$.
Thanks to this new witness we obtain
$\mu' \in Sol_{\prog}(P\, \Box\, \pi,\, C\, \Box\, M\, \Box\, H\, \Box\, F\, \Box\, R)$
and thus $\mu \in Sol_{\prog}(G)$.
\item {\bf Limited  Local Completeness:}
Assume $\mu \in WTSol_{\prog}(G)$.
Then there is some $\mu' =_{\backslash \overline{U}} \mu$
such that $\mu' \in WTSol_{\prog}(P\, \Box\,  \pi,\, C\, \Box\, M\, \Box\, H\, \Box\, F\, \Box\, R)$.
Therefore, $\mu' \in WTSol_{\ccdom}(M\, \Box\, H\, \Box\, F\, \Box\, R)$
and $\mathcal{M} : \prog \vdash_{\crwl{\mathcal{C}}}(P\, \Box\, \pi,\, C)\mu'$
for a suitable witness $\mathcal{M}$.
Because of the syntactic relationship between $G$ and $G'$ and the fact that $\pi$ is primitive,
we can conclude that $\mu' \in WTSol_{\ccdom}(M'\, \Box\, H'\, \Box\, F'\, \Box\, R')$.
Let $\mathcal{M'}$ be the witness constructed from $\mathcal{M}$ by omitting the $\crwl{\ccdom}$
proof tree for $\pi \mu'$ which is part of $\mathcal{M}$.
Then $\mathcal{M'} : \prog \vdash_{\crwl{\mathcal{C}}}(P\, \Box\, C)\mu'$.
This allows to conclude $\mu' \in WTSol_{\prog}(P\, \Box\, C\, \Box\, M'\, \Box\, H'\, \Box\, F'\, \Box\, R')$
and thus $\mu \in WTSol_{\prog}(G')$.
\end{enumerate}
\vspace*{-0.3cm}

\subsubsection{Rules from Table \ref{table7}} \label{rt7}

\noindent
Rule {\bf IE, Infer Equalities}.
This rule includes two similar cases. Here we will treat only the first one, the second one being completely analogous.
The selected part $\gamma$ is a pair of bridges of the form
$X$ {\tt \#==} $RX$, $X'$ {\tt \#==} $RX$ and $k = 1$.
We write $G'$ instead of $G'_1$.
\begin{enumerate}
\item {\bf Local Soundness:}
Assume $\mu \in Sol_{\prog}(G')$.
Then there exists $\mu' =_{\backslash \overline{U}} \mu$ such that
$\mu' \in Sol_{\prog}(P\, \Box\, C\, \Box\, X$ {\tt \#==} $RX,\, M\, \Box\, H\, \Box\, X$ {\tt ==} $X',\, F\, \Box\, R)$.
This implies two facts: firstly,  $\mathcal{M'} : \prog \vdash_{\crwl{\mathcal{C}}}(P\, \Box\, C)\mu'$
for a suitable witness $\mathcal{M'}$;
and secondly, $\mu' \in Sol_{\ccdom}(X$ {\tt \#==} $RX,\, M\, \Box\, H\, \Box\, X$ {\tt ==} $X',\, F\, \Box\, R)$.
The second fact clearly implies
$\mu' \in Sol_{\ccdom}(X$ {\tt \#==} $RX,\, X'$ {\tt \#==} $RX,\, M\, \Box\, H\, \Box\, F\, \Box\, R)$.
Along with the witness $\mathcal{M'}$, this condition guarantees
$\mu' \in Sol_{\prog}(P\, \Box\, C\, \Box\, X$ {\tt \#==} $RX,\, X'$ {\tt \#==} $RX,\, M\, \Box\, H\, \Box\, F\, \Box\, R)$
and hence  $\mu \in Sol_{\prog}(G)$.
\item {\bf Limited  Local Completeness:}
Assume $\mu \in WTSol_{\prog}(G)$.
Then there is some $\mu' =_{\backslash \overline{U}} \mu$
such that $\mu' \in WTSol_{\prog}(P\, \Box\, C\, \Box\, X$ {\tt \#==} $RX,\, X'$ {\tt \#==} $RX,\, M\, \Box\, H\, \Box\, F\, \Box\, R)$.
This implies two facts: firstly,  $\mathcal{M} : \prog \vdash_{\crwl{\mathcal{C}}}(P\, \Box\, C)\mu'$
for a suitable witness $\mathcal{M}$;
and secondly, $\mu' \in WTSol_{\ccdom}(X$ {\tt \#==} $RX,\, X'$ {\tt \#==} $RX,\, M\, \Box\, H\, \Box\, F\, \Box\, R)$.
The second fact clearly implies
$\mu' \in WTSol_{\ccdom}(X$ {\tt \#==} $RX,\, M\, \Box\, H\, \Box\, X$ {\tt ==} $X',\, F\, \Box\, R)$.
Then, $\mu' \in WTSol_{\prog}(P\, \Box\, C\, \Box\, X$ {\tt \#==} $RX,\, M\, \Box\, H\, \Box\, X$ {\tt ==} $X',\, F\, \Box\, R)$
holds thanks to the same witness $\mathcal{M}$,
and therefore  $\mu \in Sol_{\prog}(G')$.
\end{enumerate}

\noindent
Rule {\bf ID, Infer Disequalities}.
This rule includes two similar cases. Here we consider only the first one, the second one being completely analogous.
The selected part $\gamma$ is an antibridge of the form
$X$ {\tt \#/==} $u'$ placed within the $M$ store, and $k = 1$.
We write $G'$ instead of $G'_1$.
The application of the rule obtains $G'$ from $G$ by  dropping $X$ {\tt \#/==} $u'$ from $M$ and adding a semantically equivalent
disequality constraint $X$ {\tt /=} $u$ to the $F$ store.
The reasonings for local soundness and limited local completeness are very similar to those
used in the case of rule {\bf IE}. 
\vspace*{-0.3cm}

\subsubsection{Rules from Table \ref{Stable}} \label{rst}

Here we present only the proofs concerning the two rules {\bf FS} and {\bf SF}.
Note that the soundness and completeness properties of the $\fd$ solver
refer to valuations over the universe $\uni{\fd}$,  that must be related to valuations
over the universe $\uni{\ccdom}$ by means of Theorem  \ref{sumProperties} from Subsection \ref{cdomains},
as we will see below. The same technique can be applied to the rules {\bf MS} and {\bf RS}.
Rule {\bf HS} can be also  handled similarly to {\bf FS},  but in this case  Theorem \ref{sumProperties}
is not needed because the soundness and completeness properties
of the extensible $\herbrand$-solver refer directly to valuations over the universe $\uni{\ccdom}$. \\

\noindent
Rule {\bf FS $\fd$-Constraint Solver ({\em black-box})}.
The selected part $\gamma$ is the $\fd$-store $F$.
\begin{enumerate}
\item {\bf Local Soundness:}
Let us choose $G'$ as one of the finitely many goals $G'_j$ such that  $G \red_{{\bf FS},\gamma, \prog}G'_j$. Then
$G' = \exists \overline{Y'}, \overline{U}.
(P\, \Box\, C\, \Box\, M\, \Box\, H\, \Box\, (\Pi'\, \Box\, \sigma_F)\, \Box\,R)@_{\fd}\sigma'$
for some $\exists \overline{Y'} (\Pi'\, \Box\, \sigma')$ chosen as one of the alternatives computed by the $\fd$-solver,
i.e.  such that $\Pi_F \vdash\!\!\vdash_{solve^{\fd}} \exists \overline{Y'} (\Pi'\, \Box\, \sigma')$.
Assume now $\mu \in Sol_{\prog}(G')$.
Then there exists $\mu' =_{\backslash \overline{Y'}, \overline{U}} \mu$ such that
$$\mu' \in Sol_{\prog} ((P\, \Box\, C\, \Box\, M\, \Box\, H\, \Box\, (\Pi'\, \Box\, \sigma_F)\, \Box\,R)@_{\fd}\sigma')$$
for some $\exists \overline{Y'} (\Pi'\, \Box\, \sigma')$
such that $\Pi_F \vdash\!\!\vdash_{solve^{\fd}} \exists \overline{Y'} (\Pi'\, \Box\, \sigma')$.
Since $\Pi'\, \Box\, \sigma'$ is a store, we can assume $\Pi' \sigma' = \Pi'$ and
deduce the following conditions:
$$(0)\, \mathcal{M'} : \prog \vdash_{\crwl{\mathcal{C}}}(P\, \Box\, C)\sigma'\mu'\,
\textnormal{for a suitable witness}\, \mathcal{M'}$$
$$(1)\, \mu' \in Sol_{\ccdom}(\Pi_M\sigma'\, \Box\, \sigma_M\star\sigma')\,\,\,\,\,
(2)\, \mu' \in Sol_{\ccdom}(\Pi_H\sigma'\, \Box\, \sigma_H\star\sigma')$$
$$(3)\, \mu' \in Sol_{\ccdom}(\Pi'\sigma'\, \Box\, \sigma_F\sigma'),\,  \textnormal{where}\,  \Pi' \sigma' = \Pi'\,\,\,\,\,
(4)\, \mu' \in Sol_{\ccdom}(\Pi_R\sigma'\, \Box\, \sigma_R\star\sigma')$$

In particular, $(3)$ implies $\mu' \in Sol(\sigma_F\sigma'$), i.e.
$$(5)\, \sigma_F\sigma'\mu' = \mu'$$
In order to conclude that $\mu \in Sol_{\prog}(G)$, we show that the hypothesis of the
auxiliary Lemma \ref{cgs} hold for $\hat{\mu} = \mu'$.
Clearly, $\hat{\mu} =_{\backslash \overline{U}, \overline{Y'}} \mu$
and the new variables $\overline{Y'}$ are away from $\overline{U}$ and the other variables in $G$.
We still have to prove that $\mu' \in Sol_{\prog}(P\, \Box\, C\, \Box\, M\, \Box\, H\, \Box\, F\, \Box\,R)$.

\begin{itemize}
\item
Proof of $\mu' \in Sol_{\prog}(P\, \Box\, C)$:
Due to the invariant properties of admissible goals, $(P\, \Box\, C) = (P\, \Box\, C)\sigma_F$.
Using this equality and  $(5)$ we get
$(P\, \Box\, C)\sigma'\mu' = (P\, \Box\, C)\sigma_F\sigma'\mu' = (P\, \Box\, C)\mu'$.
Therefore, $\mathcal{M'} : \prog \vdash_{\crwl{\mathcal{C}}}(P\, \Box\, C)\mu'$
follows from $(0)$.
\item
Proof of $\mu' \in Sol_{\ccdom}(S)$, $S$ being any of the stores $M$, $H$, $R$:
According to the choice of $S$ we can use $(1)$, $(2)$ or $(4)$ to conclude
$$(6)\, \mu' \in Sol_{\ccdom}(\Pi_S\sigma')\,\,\,\, \textnormal{and}\,\,\,\,
(7)\, \mu' \in Sol(\sigma_S\star\sigma')\, \textnormal{i.e.}\, (\sigma_S\star\sigma')\mu' = \mu'$$
\begin{itemize}
\item
Proof of $\mu' \in Sol_{\ccdom}(\Pi_S)$:
Due to the invariant properties of admissible goals, $\Pi_S = \Pi_S\sigma_F$.
Then $(6)$ is equivalent to $\mu' \in Sol_{\ccdom}(\Pi_S\sigma_F\sigma')$.
By applying the Substitution Lemma \ref{sl} we deduce
$\sigma_F\sigma'\mu' \in Sol_{\ccdom}(\Pi_S)$,
which amounts to $\mu' \in Sol_{\ccdom}(\Pi_S)$ because of $(5)$.
\item
Proof of $\mu' \in Sol(\sigma_S)$:
Assume any variable $X \in vdom(\sigma_S)$. Then
$$X\mu' = X\sigma_S\sigma'\mu' = X\sigma_S\sigma_F\sigma'\mu' = X\sigma_S\mu'$$
where the first equality holds because of $(7)$,
the second equality holds because the admissibility properties of $G$ guarantee $\sigma_S\star\sigma_F = \sigma_S$,
and the third equality holds because of $(5)$.
\end{itemize}
\item
Proof of $\mu' \in Sol_{\ccdom}(F)$: First, we claim that
$$(8) \, \trunc{\sigma' \mu'}{\fd} \in Sol(\sigma') \, \textnormal{  i.e.  } \,
\sigma' \trunc{\sigma' \mu'}{\fd} = \trunc{\sigma' \mu'}{\fd}$$
To prove the claim, assume any $X \in vdom(\sigma')$. Because of Postulate \ref{fsolver} there are two possible cases:
\begin{enumerate}
\item[(a)] $\sigma'(X)$ is an integer value $n$. Then:
$$X \sigma' \trunc{\sigma' \mu'}{\fd} = n = \trunc{X\sigma' \mu'}{\fd} = X \trunc{\sigma' \mu'}{\fd}$$
\item[(b)] $X \in var(\Pi_F)$ and $\sigma'(X)$ is a variable $X' \in var(\Pi_F)$. Then $\sigma'(X') = X'$ because
$\sigma'$ is idempotent, and:
$$X\sigma'\trunc{\sigma' \mu'}{\fd} = X' \trunc{\sigma' \mu'}{\fd} = \trunc{X'\sigma' \mu'}{\fd} =$$
$$\trunc{X' \mu'}{\fd} = \trunc{X\sigma' \mu'}{\fd} = X \trunc{\sigma' \mu'}{\fd}$$
\end{enumerate}
We continue our reasoning using $(8)$.
\begin{itemize}
\item
Proof of $\mu' \in Sol_{\ccdom}(\Pi_F)$:
From $(3)$ and the Substitution Lemma \ref{sl} we get $\sigma'\mu' \in Sol_{\ccdom}(\Pi')$.
Because of Postulate \ref{fsolver} we can assume that all the constraints belonging to $\Pi'$ are $\fd$-specific.
Then, item 4. of Theorem \ref{sumProperties} can be applied to conclude $\trunc{\sigma'\mu'}{\fd} \in Sol_{\fd}(\Pi')$.
Using $(8)$ we get $\trunc{\sigma'\mu'}{\fd} \in Sol_{\fd}(\Pi'\, \Box\,\sigma')$,
which trivially implies $\trunc{\sigma'\mu'}{\fd} \in Sol_{\fd}(\exists \overline{Y'} (\Pi'\, \Box\,\sigma'))$.
Because of the soundness property of the $\fd$-solver (see Definition \ref{defSolver} and Postulate \ref{fsolver})
we obtain  $\trunc{\sigma'\mu'}{\fd} \in Sol_{\fd}(\Pi_F)$.
Applying again   item 4. of Theorem \ref{sumProperties}, we get $\sigma'\mu' \in Sol_{\ccdom}(\Pi_F)$.
Since $\Pi_F\, \Box\, \sigma_F$ is a store, $\Pi_F = \Pi_F\sigma_F$ and therefore
$\sigma'\mu' \in Sol_{\ccdom}(\Pi_F\sigma_F)$.
Then, the Substitution Lemma \ref{sl} yields $\sigma_F\sigma'\mu' \in Sol_{\ccdom}(\Pi_F)$,
which is the same as $\mu' \in Sol_{\ccdom}(\Pi_F)$ because of $(5)$.
\item
Proof of $\mu' \in Sol(\sigma_F)$:
$\mu' = \sigma_F\mu'$ follows from the following chain of equalities, which relies on $(5)$ and the idempotency of $\sigma_F$:
$$\mu' = \sigma_F\sigma'\mu' = \sigma_F\sigma_F\sigma'\mu' = \sigma_F\mu'$$
\end{itemize}
\end{itemize}

\item {\bf Limited  Local Completeness:}
At this point we assume that rule {\bf FS} can be applied to $G$ in a safe way, i.e. that the solver invocation
$solve^{\fd}(\Pi_F)$ satisfies the completeness property for solvers stated in Definition \ref{defSolver}
(see Subsection \ref{csolvers}).
Assume $\mu \in WTSol_{\prog}(G)$.
Then there is some $\mu' =_{\backslash \overline{U}} \mu$
such that $\mu' \in WTSol_{\prog}(P\, \Box\, C\, \Box\, M\, \Box\, H\, \Box\, F\, \Box\,R)$.
Consequently, we can assume:
$$(9)\, (P\, \Box\, C)\mu'\,  \textnormal{is well-typed and}\,
\mathcal{M} : \prog \vdash_{\crwl{\mathcal{C}}}(P\, \Box\, C)\mu'\,  \textnormal{for some witness}\, \mathcal{M}$$
$$(10)\, \mu' \in WTSol_{\ccdom}(M)\,\,\,\,\,
(11)\, \mu' \in WTSol_{\ccdom}(H)$$
$$(12)\, \mu' \in WTSol_{\ccdom}(F)\,\,\,\,\,
(13)\, \mu' \in WTSol_{\ccdom}(R)$$
In particular, $(12)$ implies $\mu' \in WTSol_{\ccdom}(\Pi_F)$. Thanks to Postulate \ref{fsolver} we can assume that
$\Pi_F$ is $\fd$-specific and apply item 4 of Theorem \ref{sumProperties} to conclude $\trunc{\mu'}{\fd} \in WTSol_{\fd}(\Pi_F)$. By completeness of the solver invocation $solve^{\fd}(\Pi_F)$ there is
some alternative $\exists \overline{Y'} (\Pi'\, \Box\, \sigma')$ computed by the solver
(i.e. such that  $\Pi_F \vdash\!\!\vdash_{solve^{\fd}} \exists \overline{Y'} (\Pi'\, \Box\, \sigma')$)
verifying
$$(14)\, \trunc{\mu'}{\fd}  \in WTSol_{\fd}(\exists \overline{Y'} (\Pi'\, \Box\, \sigma'))$$
Then $G' = \exists \overline{Y'}, \overline{U}.
(P\, \Box\, C\, \Box\, M\, \Box\, H\, \Box\, (\Pi'\, \Box\, \sigma_F)\, \Box\,R)@_{\fd}\sigma'$
is one of the the finitely many goals $G'_j$ such that $G \red_{{\bf FS},\gamma, \prog}G'_j$.
In the rest of the proof we will show that $\mu \in WTSol_{\prog}(G')$ by finding
$\mu'' =_{\backslash \overline{Y'}, \overline{U}} \mu$ such that
$$ (\dag)\, \mu'' \in Sol_{\prog} ((P\, \Box\, C\, \Box\, M\, \Box\, H\, \Box\, (\Pi'\, \Box\, \sigma_F)\, \Box\,R)@_{\fd}\sigma')$$
Because of $(14)$ there is $\hat{\hat{\mu}} \in Val_{\fd}$ such that
$$(15)\, \trunc{\mu'}{\fd} =_{\backslash \overline{Y'}} \hat{\hat{\mu}} \in WTSol_{\fd}(\Pi'\, \Box\, \sigma')$$
Let $\mu'' \in Val_{\ccdom}$ be univocally defined by the conditions
$\mu'' =_{\overline{Y'}} \hat{\hat{\mu}}$ and $\mu'' =_{\backslash \overline{Y'}} \mu'$.
Since $\mu =_{\backslash \overline{U}} \mu'$, it follows that $\mu'' =_{\backslash \overline{Y'}, \overline{U}} \mu$.
Moreover, $\trunc{\mu''}{\fd} = \hat{\hat{\mu}}$, because for any variable $X \in \var$ there are two possible cases:
either $X \in \overline{Y'}$ and then  $\trunc{\mu''}{\fd}(X) = \trunc{\hat{\hat{\mu}}}{\fd}(X) =  \hat{\hat{\mu}}(X)$,
since $\hat{\hat{\mu}} \in Val_{\fd}$;
or else $X \notin \overline{Y'}$ and then $\trunc{\mu''}{\fd}(X) = \trunc{\mu'}{\fd}(X) =  \hat{\hat{\mu}}(X)$,
since $\trunc{\mu'}{\fd} =_{\backslash \overline{Y'}} \hat{\hat{\mu}}$.
From $(15)$ and $\trunc{\mu''}{\fd} = \hat{\hat{\mu}}$ we  obtain $\mu'' \in WTSol_{\fd}(\Pi')$ by applying
item 4 of Theorem \ref{sumProperties}. We now claim:
$$(16)\, \mu'' \in WTSol_{\fd}(\Pi'\, \Box\, \sigma')$$
To justify this claim it is sufficient to prove $\mu'' \in Sol(\sigma')$, i.e. $\sigma'\mu'' = \mu''$.
In order to prove this let us assume any $X \in vdom(\sigma')$.
Because of Postulate \ref{fsolver}, there are two possible cases:
\begin{enumerate}
\item[(a)] $\sigma'(X)$ is an integer value $n$. From  $(15)$ we know $\hat{\hat{\mu}} \in Sol(\sigma')$
and therefore $\hat{\hat{\mu}}(X) = n$. Since $\trunc{\mu''}{\fd} = \hat{\hat{\mu}}$, it follows that $\mu''(X) = n$, and then
$X\sigma'\mu'' = n = X\mu''$.

\item[(b)] $X \in var(\Pi_F)$ and $\sigma'(X)$ is a variable $X' \in var(\Pi_F)$. Then:

\begin{tabular}{llr}
        & $X\sigma' \mu'' = X' \mu''$ & \\
$= $ & $X' \mu'$ & (using $\mu'' =_{\backslash \overline{Y'}} \mu'$ and $X' \notin \overline{Y'}$)\\
$=$ & $\trunc{X'\mu'}{\fd}$ & (using the fact that $\Pi_F$ is $\fd$-specific and $(12)$)\\
$=$ & $X'\trunc{\mu'}{\fd}$ &\\
$= $ & $X' \hat{\hat{\mu}}$ & (using $(15)$ and $X' \notin \overline{Y'}$)\\
$= $ & $X\sigma' \hat{\hat{\mu}} = X\hat{\hat{\mu}}$ & (using $(15)$)\\
$= $ & $X\trunc{\mu'}{\fd}$ & (using $(15)$ and $X \notin \overline{Y'}$)\\
$=$ & $\trunc{X\mu'}{\fd} \,= X \mu'$ & (using the fact that  $\Pi_F$ is $\fd$-specific and $(12)$)\\
$= $ & $X \mu''$ & (using $\mu'' =_{\backslash \overline{Y'}} \mu'$ and $X \notin \overline{Y'}$)\\
\end{tabular}
\end{enumerate}

 We are now in a position to prove $(\dag)$, thereby finishing the proof:
\begin{itemize}
\item
Proof of $\mu'' \in WTSol_{\prog}(P\, \Box\, C)\sigma'$:
Because of the Substitution Lemma \ref{sl}, this is equivalent to
$\sigma'\mu'' \in WTSol_{\prog}(P\, \Box\, C)$.
Because of $(16)$, $\sigma'\mu'' = \mu''$.
Since $\mu' =_{\backslash \overline{Y'}} \mu''$ and the variables $\overline{Y'}$ do not occur in $P\, \Box\, C$,
$\mu'' \in WTSol_{\prog}(P\, \Box\, C)$ is equivalent to $\mu' \in WTSol_{\prog}(P\, \Box\, C)$,
which is ensured by the same witness $\mathcal{M}$ given by $(9)$.
\item
Proof of $\mu'' \in WTSol_{\ccdom}(S\star\sigma')$, $S$ being any of the stores $M$, $H$, $R$:
According to the choice of $S$ we can use $(10)$, $(11)$ or $(13)$ to conclude
$$(17)\, \mu' \in WTSol_{\ccdom}(\Pi_S)\,\,\,\, \textnormal{and}\,\,\,\,
(18)\, \mu'\in Sol(\sigma_S)\, \textnormal{i.e.}\, \sigma_S\mu' = \mu'$$
\begin{itemize}
\item
Proof of $\mu'' \in WTSol_{\ccdom}(\Pi_S\sigma')$:
Since $\mu'' =_{\backslash \overline{Y'}} \mu'$ and the variables $\overline{Y'}$ do not  occur in $\Pi_S$,
$(17)$ implies $\mu'' \in WTSol_{\ccdom}(\Pi_S)$, which is equivalent  to $\sigma'\mu'' \in WTSol_{\ccdom}(\Pi_S)$
because of $(16)$. Then,  $\mu'' \in WTSol_{\ccdom}(\Pi_S\sigma')$ follows from the Substitution Lemma \ref{sl}.
\item
Proof of $\mu'' \in WTSol(\sigma_S\star\sigma')$:
Assume any variable $X \in vdom(\sigma_S)$. Then
$$X\sigma_S\sigma'\mu'' = X\sigma_S\mu'' = X\sigma_S\mu' = X\mu' = X\mu''$$
where the first equality holds because of $(16)$,
the second equality holds because $\mu'' =_{\backslash \overline{Y'}} \mu'$
and the variables $\overline{Y'}$ do not  occur in $X\sigma_S$,
the third equality holds because of $(18)$,
and the fourth equality holds because $\mu'' =_{\backslash \overline{Y'}} \mu'$
and the variables $\overline{Y'}$ do not  include $X$.
\end{itemize}
\item
Proof of $\mu'' \in WTSol_{\ccdom}(\Pi'\sigma'\, \Box\, \sigma_F\sigma')$:
\begin{itemize}
\item
Proof of $\mu'' \in WTSol_{\ccdom}(\Pi'\sigma')$:
This is a trivial consequence of $(16)$, since $\Pi'\sigma' = \Pi'$ (because $\Pi'\, \Box\, \sigma'$ is a store).
\item
Proof of $\mu'' \in Sol(\sigma_F\sigma')$:
Because of $(16)$ we can assume that $\mu'' \in Sol(\sigma')$, i.e. $\sigma'\mu'' = \mu''$.
We must prove $\sigma_F\sigma'\mu'' = \mu''$. Assume any variable $X \in vdom(\sigma_F\sigma')$.
Because of the invariant properties of admissible goals, there are three possible cases:
\begin{enumerate}
\item[(a)]
$X \in vdom(\sigma_F)$ and $\sigma_F(X)$ is an integer value $n$.
Because of $(12)$, we know that $\mu' \in Sol(\sigma_F)$ and hence $X\sigma_F\mu' = n = X\mu'$.
Moreover, $X\mu'' = X\mu' = n$ because $\mu'' =_{\backslash \overline{Y'}} \mu'$
and the variables $\overline{Y'}$ do not  include $X$.
Then we can conclude that $X\sigma_F\sigma'\mu'' = n = X\mu''$.
\item[(b)]
$X \in vdom(\sigma_F)$ and $\sigma_F(X) = X' \in var(\Pi_F)$. Then:

\begin{tabular}{llr}
        & $X\sigma_F\sigma' \mu'' = X' \sigma'\mu''$ & \\
$= $ & $X' \mu''$ & (using $(16)$)\\
$= $ & $X' \mu'$ & (using $\mu'' =_{\backslash \overline{Y'}} \mu'$ and $X' \notin \overline{Y'}$)\\
$=$ & $X\sigma_F \mu' = X\mu'$ & (using $(12)$)\\
$= $ & $X \mu''$ & (using $\mu'' =_{\backslash \overline{Y'}} \mu'$ and $X \notin \overline{Y'}$)\\
\end{tabular}

\item[(c)]
$X \notin vdom(\sigma_F)$.
Then $X\sigma_F = X$, and we can use  $\mu'' \in Sol(\sigma')$ to deduce that
$X\sigma_F\sigma'\mu'' = X\sigma'\mu'' = X\mu''$.
\end{enumerate}
\end{itemize}
\end{itemize}

\end{enumerate}

\noindent
Rule {\bf SF, Solving Failure}.
The selected part $\gamma$ is one of the four stores of the goal,
the number $k$ of possible transformations $G \red_{{\bf RL},\gamma, \prog}G'_j$
of $G$ into a non-failed goal $G'_j$ is $0$,
and therefore $\bigcup_{j=1}^{k} WTSol_{\prog}(G'_j) = \emptyset$.
\begin{enumerate}
\item {\bf Local Soundness:}
The inclusion $Sol_{\prog}(G) \supseteq \emptyset$ holds trivially.
\item {\bf Limited  Local Completeness:}
The inclusion $WTSol_{\prog}(G) \subseteq \emptyset$
is equivalent to $WTSol_{\prog}(G) = \emptyset$.
In order to prove this, we assume that the  application of {\bf SF} to $G$ has relied
on a complete invocation of the $\cdom$ solver.
Since the invocation of the solver has failed
(i.e., $\Pi_S \vdash\!\!\vdash_{solve^{\cdom}_{\varx}} \blacksquare$)
but it is assumed to be complete, we know that
$WTSol_{\cdom}(\Pi_S) = \emptyset$.
From this we can conclude $WTSol_{\ccdom}(\Pi_S) = \emptyset$,
using item 4 of Theorem \ref{sumProperties} in case that $\cdom$ is not $\herbrand$.
Finally, $WTSol_{\prog}(G) = \emptyset$ is a trivial consequence of 
$WTSol_{\ccdom}(\Pi_S) = \emptyset$.
\end{enumerate}
\hfill
\end{proof}
\vspace*{-0.5cm}

\subsubsection{Proof of the Progress Lemma} \label{ProofProgress}

In this Subsection  we  prove the  Progress Lemma \ref{progress} used in Subsection \ref{SC}
to obtain the  Global Completeness Theorem \ref{globalC}.
First, we define a {\em well-founded progress ordering} $\vartriangleright$
between pairs $(G,\mathcal{M})$ formed by an admissible goal $G$
without free occurrences of higher-order variables and a
witness $\mathcal{M} = \{\mathcal{T}_1, \ldots, \mathcal{T}_n\}$  for the fact that $\mu \in Sol_{\prog}(G)$.
Given such a pair, we define a $7$-tuple
$||(G,\mathcal{M})|| =_{def} (O_1, O_2, O_3, O_4, O_5, O_6, O_7)$
(where $O_1$ is a finite multiset of natural numbers and $O_2,\, \ldots,\, O_7$ are natural numbers)
as follows:

\begin{itemize}
\item [\bf O$_1$] is the {\em restricted size of the witness} $\mathcal{M}$,
defined as the multiset of natural numbers $\{\rsize{\mathcal{T}_1},\,  \ldots,\, \rsize{\mathcal{T}_n}\}$,
where $\rsize{\mathcal{T}_i}$ $(1\leq i \leq n)$
denotes the {\em res\-tricted size} of the $\crwl{\mathcal{C}}$ proof tree $\mathcal{T}_i$
as defined in \cite{LRV07}, namely as the number of nodes in $\mathcal{T}_i$ corresponding to
$\crwl{\mathcal{C}}$ inference steps that depend on the meaning of primitive functions $p$
(as interpreted in the coordination domain $\ccdom$)
plus the number of nodes in $\mathcal{T}_i$ corresponding to
$\crwl{\mathcal{C}}$ inference steps that depend on the meaning of user-defined  functions $f$
(according to the current program $\prog$).
\item [\bf O$_2$] is the sum of $||p\, \tpp{e}{n}||$ for all the total applications $p\, \tpp{e}{n}$ of primitive functions
$p \in PF^n$ occurring in the parts $P$ and $C$ of $G$,
where $||p\, \tpp{e}{n}||$ is defined as the number of argument expressions $e_i\, (1 \leq i \leq n)$ that are not patterns.
\item [\bf O$_3$] is the number of occurrences of rigid and
passive expressions $h\, \tpp{e}{n}$ that are not patterns in the productions $P$ of $G$.
\item [\bf O$_4$] is the sum of the syntactic sizes  of the right hand sides
of all the productions occurring  in $P$.
\item [\bf O$_5$] is the sum $\mathit{sf}_M + \mathit{sf}_H + \mathit{sf}_F + \mathit{sf}_R$ of the
 {\em solvability flags} of the four constraint stores occurring in $G$.
 The  solvability flag  $\mathit{sf}_M$ takes the value $1$ if rule {\bf MS} from Table  \ref{Stable}
 can be applied to $G$, and $0$ otherwise. The other three flags are defined analogously.
\item [\bf O$_6$] is the number of bridges occurring in the mediatorial store $M$ of $G$.
\item [\bf O$_7$] is the number of antibridges occurring  in the mediatorial store $M$ of $G$.
\end{itemize}

Let $>_{lex}$ be the lexicographic product  of the  $7$ orderings $>_i\, (1 \leq i \leq 7)$,
where $>_1$ is the multiset ordering  $>_{mul}$ over multisets of natural numbers,
and $>_i$ is the ordinary ordering $>$ over natural numbers for $2 \leq i \leq 7$.
Finally, let us define the progress ordering $\vartriangleright$ by the condition
$(G,\mathcal{M}) \vartriangleright (G',\mathcal{M'})$ iff $||(G,\mathcal{M})||  >_{lex} ||(G',\mathcal{M'})||$.
As proved in \cite{BN98}, $>_{mul}$ is a well-founded ordering and the lexicographic
product of well-founded orderings is again a well-founded ordering.
Therefore, $\vartriangleright$ is well-founded.

Now we can  prove the  Progress Lemma \ref{progress}.

\begin{proof}[Proof of Lemma \ref{progress}] \label{progressproof}

Consider  an admissible goal
$G$ $\equiv$ $\exists \overline{U}.$ $P$ $\Box$ $C$ $\Box$ $M$ $\Box$ $H$ $\Box$ $F$ $\Box$ $R$
for a program $\prog$,
a well-typed solution $\mu \in WTSol_{\prog}(G)$ and a
witness $\mathcal{M}$ for the fact that $\mu \in Sol_{\prog}(G)$.
Assume that neither $\prog$ nor $G$ have free occurrences of higher-order variables,
and that $G$ is not in solved form.
\begin{enumerate}
\item Let us prove that there must be some rule {\bf RL} applicable to $G$ which is not a failure rule.
Since $G$ is not in solved form, we know that either $P \neq \emptyset$ , or else $C \neq \emptyset$,
or else some of the transformations displayed in Tables \ref{table7} and \ref{Stable} can be applied to $G$.
Note that {\bf CF} cannot be applied to $G$ because $G$ has got solutions. Moreover, if the
failing rule {\bf SF} would be applicable to $G$, then some of the other rules in Table \ref{Stable} would be applicable also.
Let ${\bf \mathcal{PR}}$ be  the set of those transformation rules  displayed in Table \ref{table3}
which are different of {\bf CF}, {\bf EL} and {\bf FC}.
In the following items we analyze different cases according to the from of $G$. In each case we either find some
rule {\bf RL} that can be applied to $G$ or we make some assumption that can be used to reason in the subsequent cases.
In the last item we conclude that rule {\bf EL} can be applied, if no previous item has allowed to prove the applicability of another rule.

\begin{enumerate}
\item
If some of the transformation rules in Tables \ref{table7} and \ref{Stable} can be applied to $G$, then we are ready.
In the following items we assume that this is not the case.
\item
If $P \neq \emptyset$ and  some rule {\bf RL} $\in$ ${\bf \mathcal{PR}}$ can be applied to $G$, then we are ready.
In the following items we assume that this is not the case.
\item
Due to the hypothesis that $G$ has no free occurrences of higher-order variables,
from  this point on we can assume that each production occurring in $P$ must have one of the three following forms:
\begin{enumerate}
\item
$h\, \tpp{e}{m}\, \to\, X$, with $h\, \tpp{e}{m}$ passive but not a pattern.
\item
$f\tpp{e}{n}\, \tpp{a}{k}\, \to\, X$, with $f \in DF^n$ and $k \geq 0$.
\item
$p\, \tpp{e}{n}\, \to\, X$, with $p \in PF^n$.
\end{enumerate}
If this were not the case, then  $P$ would include some production $e\, \to\, t$ of some other form,
and a simple case analysis of the syntactic form of $e\, \to\, t$ would lead to the conclusion that some
rule {\bf RL} $\in$ ${\bf \mathcal{PR}}$ could be applied to it.
\item
If $C \neq \emptyset$ and includes some atomic constraint $\alpha$ that is not primitive,
then the rule {\bf FC} from Table \ref{table3} can be applied to $\alpha$, and we are ready.
In the following items we assume that this is not the case.
\item
If $C \neq \emptyset$ and only includes primitive atomic constraints $\pi$,
then at least rule rule {\bf SC} from Table \ref{table4} (and maybe also rules {\bf SB} and {\bf PP})
can be applied to $G$ taking $\pi$ as the selected part, and we are ready.
In the following items we assume that $C = \emptyset$.
\item
At this point, if there would be some variable $X \in pvar(P) \cap odvar(G)$,
this $X$ would be the right-hand side of some production in $P$ with one of the three
forms  i, or ii or  iii displayed in item (c) above, and one of the three rules
{\bf IM} or {\bf DF} or {\bf PC} could be applied, which contradicts the assumptions made at item (b).
From  this point on, we can assume that $pvar(P) \cap odvar(G) = \emptyset$.
\item
Let $S = \Pi_S\, \Box\, \sigma_S$ be any of the four stores, let $\cdom$ be the corresponding domain, and
let $\chi = pvar(P) \cap var(\Pi_S)$. Because of the assumptions made at item (a), $S$ must be in $\chi$-solved form
and the discrimination property of the solver $solve^{\cdom}$ ensures that one of the two following conditions must hold:
\begin{enumerate}
\item
$\chi \cap odvar_{\cdom}(\Pi_S) \neq \emptyset$, i.e.  $pvar(P) \cap var(\Pi_S) \cap odvar_{\cdom}(\Pi_S) \neq \emptyset$.
\item
$\chi \cap var_{\cdom}(\Pi_S) = \emptyset$, i.e.  $pvar(P) \cap var(\Pi_S) = \emptyset$.
\end{enumerate}
Since i contradicts the assumption $pvar(P) \cap odvar(G) = \emptyset$ made at item (f),
ii must hold for the four stores. On the other hand, the invariant properties of admissible goals
guarantee that produced variables cannot occur in the answer substitutions $\sigma_S$.
\item
At this point, because of the assumptions made at the previous items,
we can assume that $C = \emptyset$, the four stores are in solved form and include no produced variables,
and all the productions occurring in $P$ have the form $e \to X$, where $X$ is a variable.
Since $G$ is not solved, it must be the case that $P \neq \emptyset$.

Note that $pvar(P)$ is finite and not empty. Moreover, the transitive closure $\gg_{P}^{+}$
of the production relation $\gg_{P}$ between produced variables must be irreflexive,
due to the invariant properties of admissible goals. Therefore, there is some production
$(e \to X) \in P$ such that $X$ is minimal w.r.t. $\gg_{P}$.

The variable $X$ cannot occur in $e$
because this would imply $X \gg_{P} X$, contradicting the irreflexivity of $\gg_{P}^{+}$.
For any other production $(e' \to X') \in P$, $X$ must be different of $X'$ because of the invariant
properties of admissible goals, and $X$ cannot occur in $e'$ because this would imply $X \gg_P X'$,
contradicting the minimality of $X$ w.r.t. $\gg_P$.
Moreover, $X$ cannot occur in the stores because they include no produced variables.

Therefore, $X$ does not occur in the rest of the goal, and the rule {\bf EL} can be applied to eliminate $e \to X$.
\end{enumerate}

\item Assume now  any choice of a  rule {\bf RL} (not a failure rule) and a part $\gamma$ of $G$,
such that {\bf RL} can be applied  to $\gamma$ in a safe manner, i.e. involving
neither an opaque application  of {\bf DC} nor an incomplete solver invocation.
We must prove the existence of a finite computation
$G \red_{{\bf RL},\gamma, \prog}^{+} G'$ and a witness
$\mathcal{M'} : \mu \in WTSol_{\prog}(G')$
such that $(G,\mathcal{M}) \vartriangleright (G',\mathcal{M}')$.
Due to the Limited Local  Completeness of $\cclnc{\ccdom}$
(Theorem \ref{localSC}, item 2.), there is one step
$G\red_{{\bf RL}, \gamma, \prog} G'_1$ such that $\mathcal{M'} : \mu \in WTSol_{\prog}(G')$ with a witness $\mathcal{M'}$
constructed as we have sketched  in the proof of Theorem \ref{localSC}.
We define the desired finite computation by distinction of cases, as follows:
\begin{enumerate}
\item
If  {\bf RL} is different from the two rules {\bf SB} and {\bf PP},
then the finite computation is chosen as $G \red_{{\bf RL}, \gamma, \prog} G'_1$ and $G'$ is $G'_1$.
\item
If {\bf RL} is {\bf SB} and {\bf PP} is applicable to $\gamma$, then the finite computation is chosen as
$G \red_{{\bf SB}, \gamma, \prog} G'_1 \red_{{\bf PP}, \gamma, \prog} G'_2 \red_{{\bf SC}, \gamma, \prog} G'_3$
and $G'$ is $G_3$.
\item
If {\bf RL} is {\bf SB} and {\bf PP} is not applicable to $\gamma$, then the finite computation is chosen as
$G \red_{{\bf SB}, \gamma, \prog} G'_1 \red_{{\bf SC}, \gamma, \prog} G'_2$
and $G'$ is $G_2$.
\item
If {\bf RL} is {\bf PP} and {\bf SB} is applicable to $\gamma$, then the finite computation is chosen as
$G \red_{{\bf PP}, \gamma, \prog} G'_1 \red_{{\bf SB}, \gamma, \prog} G'_2 \red_{{\bf SC}, \gamma, \prog} G'_3$
and $G'$ is $G_3$.
\item
If {\bf RL} is {\bf PP} and {\bf SB} is not applicable to $\gamma$, then the finite computation is chosen as
$G \red_{{\bf PP}, \gamma, \prog} G'_1 \red_{{\bf SC}, \gamma, \prog} G'_2$
and $G'$ is $G_2$.
\end{enumerate}

\begin{table}[h]
\begin{center}
{\scriptsize
\begin{tabular}{lccccccc}
\hline
{\bf RULES} ~~~~~~& ~~{\bf O$_1$}~~ & ~~{\bf O$_2$}~~ & ~~{\bf O$_3$}~~ & ~~{\bf O$_4$}~~ & ~~{\bf O$_5$}~~ & ~~{\bf O$_6$}~~ & ~~{\bf O$_7$} \\
\hline
{\bf DC}                & $\geq_{mul}$  & $\geq$ & $\geq$  & $>$ &  &  & \\
\hline
{\bf SP}                 & $\geq_{mul}$  & $\geq$ & $\geq$  & $>$ &  &  & \\
\hline
{\bf IM}                 & $\geq_{mul}$  & $\geq$ & $>$        &         &  &  & \\
\hline
{\bf EL}                & $\geq_{mul}$  & $\geq$ & $\geq$   & $>$ & &  & \\
\hline
{\bf DF}                & ${>_{mul}}$    &                 &         &  & & \\
\hline
{\bf PC}                & $\geq_{mul}$ & $\geq$ & $\geq$   & $>$ & & & \\
\hline
{\bf FC}                & $\geq_{mul}$  & $>$ &                 &         &  & & \\
\hline
{\bf (b),(c),(d),(e)} & ${>_{mul}}$    &                 &         &  & & \\
\hline
{\bf IE}                  & $\geq_{mul}$ & $\geq$ & $\geq$ & $\geq$ & $\geq$ & $>$  &  \\
\hline
{\bf ID}                  & $\geq_{mul}$ & $\geq$ & $\geq$ & $\geq$ & $\geq$ &  $\geq$  & $>$ \\
 \hline
{\bf MS}                & $\geq_{mul}$ & $\geq$ & $\geq$ & $\geq$ & $>$ &  &\\
\hline
{\bf HS}                 & $\geq_{mul}$ & $\geq$ & $\geq$ & $\geq$ & $>$ &  &\\
\hline
{\bf FS}                  & $\geq_{mul}$ & $\geq$ & $\geq$ & $\geq$ & $>$ &  &\\
\hline
{\bf RS}                 & $\geq_{mul}$ & $\geq$ & $\geq$ & $\geq$ & $>$ &  &\\
\hline
\end{tabular}
\caption{Well-founded progress ordering $\vartriangleright$ for
$CCLNC(\mathcal{C})$}\label{ordering}}
\end{center}
\end{table}

Note that cases (b), (c), (d), and (e) above refer to the rules in Table \ref{table4}.
In all these cases, the Limited Local  Completeness of $\cclnc{\ccdom}$
allows to find all the computation steps and a witness $\mathcal{M'} : \mu \in WTSol_{\prog}(G')$.
In all the cases, we claim that $(G,\mathcal{M}) \vartriangleright (G',\mathcal{M}')$,
i.e. $||(G,\mathcal{M})|| >_{lex} ||(G',\mathcal{M}')||$.
This can  be justified by Table \ref{ordering}.
Each file of this table corresponds to a possibility for the rule {\bf RL} used in
a one-step  finite computation $G \red_{{\bf RL},\gamma, \prog}^{+} G'$ of type (a),
except for one file which corresponds to a finite computation $G \red_{{\bf RL},\gamma, \prog}^{+} G'$
of type (b), (c), (d) or (e).
Each column $1 \leq i \leq 7$ shows the variation in  $O_i$ according to $>_i$
when going from $||(G,\mathcal{M})||$ to $||(G',\mathcal{M'})||$ by means of the corresponding finite computation.
For instance, the file for {\bf IE} shows that the application of this rule
does not increase $O_i$ for $1 \leq i \leq 5$  and decreases $O_6$.

It only remains to show that the information displayed in  Table \ref{ordering} is correct.
Here we limit ourselves to explain the key ideas. A more precise proof could be presented on the basis of a more detailed construction of
the  witnesses  $\mathcal{M'} : \mu \in WTSol_{\prog}(G')$.

\begin{itemize}
\item
For every  rule {\bf RL}, the application of {\bf RL} does not increase $O_1$, as shown by the first column of the table.
This happens because the witness  $\mathcal{M'}$ can be constructed from  $\mathcal{M}$ in such a way that all the inference steps
within $\mathcal{M'}$ dealing with primitive and defined functions are borrowed from $\mathcal{M}$.
\item
The application of any of the rule {\bf DF} strictly decreases $O_1$, as seen in the table.
The reason is that the witness $\mathcal{M}$ includes a $\crwl{\ccdom}$ proof tree $\mathcal{T}$ for an appropriate instance of
a production of the form $f\, \tpp{e}{n}\, \to t$. The root inference of this proof tree contributes to the restricted size of
$\mathcal{M}$ and disappears in the witness $\mathcal{M'}$ constructed from $\mathcal{M}$ as sketched in Subsection \ref{rt3}.
Therefore, the restricted size of $\mathcal{M'}$ decreases by one w.r.t. the restricted size of $\mathcal{M}$.
\item
The table also shows that finite computations of type (b), (c), (d) or (e) strictly decrease $O_1$.
The reason is that such finite computations always work with a fixed primitive atomic constraint $\pi$
which is ultimately moved from the constraint pool $C$ of $G$ to one of the stores in $G'$ when performing
the last  {\bf SC} computation step. The witness $\mathcal{M} : \mu \in WTSol_{\prog}(G)$ includes  a $\crwl{\ccdom}$ proof tree
for an appropriate instance of $\pi$, while no corresponding proof tree is needed in the witness $\mathcal{M'}$.
Therefore, the restricted size of $\mathcal{M'}$ decreases by some positive amount.
\item
The application of rule {\bf FC} decrements $O_2$, because $G$ includes a production $p\, \tpp{e}{n} \to t$ with
$||p\, \tpp{e}{n}|| > 0$, which is replaced  in $G'$ by a primitive atomic constraint  $p\, \tpp{t}{n} \to!\, t$ with $||p\, \tpp{t}{n}|| = 0$
and some new productions $e_i \to V_i$ whose contribution to the $O_2$ measure of $G'$ must  be smaller than $||p\, \tpp{e}{n}||$.
\item
The application of rule {\bf IE} decreases $O_6$ and does not increment $O_i$ for $1 \leq i \leq 5$.
This is because in this case the witness $\mathcal{M'}$ can be chosen as $\mathcal{M}$ itself,
the measures $O_2, O_3, O_4$ and $O_5$ are obviously not affected by {\bf IE}, and the measure $O_6$ obviously decreases by $1$
when {\bf IE}  is applied.
\item
Because of similar reasons, the application of rule {\bf ID} decreases $O_7$ and does not increment $O_i$ for $1 \leq i \leq 6$.
\item
Let {\bf RL} be any of the four constraint solving transformations {\bf MS}, {\bf HS}, {\bf FS} and {\bf RS}.
The witness $\mathcal{M'} : \mu \in WTSol_{\prog}(G')$ can be guaranteed to exist only if the solver invocation has been a complete one.
In this case, $\mathcal{M'}$ can be chosen as the same witness $\mathcal{M}$,
and therefore the $O_1$ measure does not increase when going from $G$ to $G'$.
Measures $O_2$, $O_3$ and $O_4$ are not affected by the bindings created by the solver invocations (since they substitute patterns for variables).
Measure $O_5$ obviously decreases, since the solvability flag $\mathit{sf}_S$ for the store that has been solved descends from $1$ to $0$.
\end{itemize}
\hfill
\end{enumerate}
\end{proof}

\end{document}